\documentclass[english,11pt]{article}
\usepackage{verbatim}

\usepackage{tabularx,lipsum,environ,amsmath,amssymb}

\usepackage{varwidth}
\usepackage{subcaption}
\usepackage{tikz}
\usetikzlibrary{decorations.pathmorphing}
\usetikzlibrary{decorations.markings}
\usetikzlibrary{shapes.gates.logic.US,trees,positioning,arrows,shapes}

\tikzset{
	myblue/.style ={circle, white, draw=blue, fill=blue, inner sep =2.2},
	arn/.style = {circle, white, draw=black, fill=black, inner sep = 1.4},
	arn_l/.style = {circle, white, draw=black, fill=black, inner sep = 2.5},
	photon/.style={draw=black, thin},
	electron/.style={draw=black, very thick, line width=0.09cm},
	solidline/.style={draw=black, very thick, solid, line width=0.08cm},
	tr/.style={buffer gate US,thick,draw,fill=gray!60,rotate=90,	anchor=east,minimum width=2.25cm},
	br/.style={buffer gate US,thick,draw,fill=gray!60,rotate=90,	anchor=east,minimum width=4.5cm},
	brr/.style={buffer gate US,draw,fill=gray!60,rotate=90,	anchor=east,minimum width=4.5cm, opacity = 0.6},
	trr/.style={buffer gate US,thick,draw,fill=gray!40,rotate=90,	anchor=east,minimum width=1.9cm, opacity = 0.5},
	trrr/.style={buffer gate US,draw,fill=gray!40,rotate=90,	anchor=east,minimum width=3.6cm, opacity = 0.5}
}
\usepackage[utf8]{inputenc}
\usepackage{amssymb,amsmath,amsthm}
\usepackage[hypertexnames=false]{hyperref}
\usepackage{dsfont}
\usepackage{algorithm}
\usepackage[noend]{algpseudocode}
\usepackage{algpseudocode}
\usepackage{thm-restate}
\usepackage[capitalise]{cleveref}
\usepackage{amsfonts}

\makeatletter
\def\BState{\State\hskip-\ALG@thistlm}
\makeatother
\usepackage{graphicx}
\usepackage[utf8]{inputenc}
\usepackage[english]{babel}
\usepackage{color}
\usepackage{verbatim}

\usepackage{float}
\usepackage{array}
\usepackage{setspace}
\usepackage[margin=1in,footskip=0.25in]{geometry}
\usepackage{bm}
\usepackage{bbm}
\usepackage{xcolor}

\newtheorem{theorem}{Theorem}
\newtheorem{lemma}{Lemma}

\newtheorem{corollary}{Corollary}
\newtheorem{conjecture}{Conjecture}
\newtheorem{claim}{Claim}

\newtheorem{assumption}{Assumption}

\newenvironment{proofof}[1]{\noindent{\bf Proof of #1:}}{$\qed$\par}
\newenvironment{proofsketch}{\noindent{\bf Proof outline:}}{$\qed$\par}

\theoremstyle{definition}
\newtheorem{definition}{Definition}

\DeclareMathOperator{\poly}{poly}

\let\emptyset\varnothing

\newcommand{\high}{\mathrm{highbudget}~}

\newcommand{\lef}{\mathrm{left}}
\newcommand{\righ}{\mathrm{right}}

\newcommand{\leaves}{\textsc{Leaves}}
\newcommand{\hleaves}{\textsc{HeavyLeaves}}
\newcommand{\head}{\textsc{head}}

\newcommand{\Estimated}{\mathrm{Est}}

\newcommand{\found}{\mathsf{Found}}
\newcommand{\steps}{\mathrm{Steps}}

\newcommand{\Path}{\textsc{SideTree}}

\newcommand{\idnt}{\mathrm{Identified}}

\newcommand{\excl}{\mathsf{Excluded}}
\newcommand{\supp}{\mathrm{supp}}
\newcommand{\tfull}{T^{\mathrm{full}}}
\newcommand{\scheap}{\mathrm{Cheap}}
\newcommand{\sident}{\mathrm{Marked}}
\newcommand{\tree}{\mathrm{Tree}}

\newcommand{\Anc}{\mathrm{Anc}}
\newcommand{\Estimate}{\textsc{Estimate}}
\newcommand{\ZeroTest}{\textsc{ZeroTest}}

\newcommand{\marked}{\text{Marked}}

\bmdefine{\aaa}{a}
\bmdefine{\jj}{j}
\bmdefine{\rr}{r}
\bmdefine{\lv}{l}
\bmdefine{\sv}{s}
\bmdefine{\tv}{t}
\bmdefine{\ff}{f}
\bmdefine{\gg}{g}
\bmdefine{\hh}{h}
\bmdefine{\tt}{t}
\bmdefine{\qq}{q}
\bmdefine{\vv}{v}
\bmdefine{\ww}{w}
\bmdefine{\phib}{\phi}

\global\long\def\ZZ{\mathbb{Z}}

\bmdefine{\alphav}{\alpha}
\bmdefine{\betav}{\beta}
\bmdefine{\bv}{B}
\bmdefine{\bb}{b}

\newcommand{\wh}{\widehat}
\newcommand{\wt}{\widetilde}

\newcommand{\E}{\mathbb{E}}

\newcommand{\argmin}{\mathrm{argmin}}
\newcommand{\true}{\mathrm{True}}
\newcommand{\false}{\mathrm{False}}
\newcommand{\frontier}{\textsc{Frontier}}
\newcommand{\correct}{\textsc{IsCorr}}

\DeclareMathOperator{\subtree}{\mathrm{FreqCone}}

\newcommand{\C}{{\mathbb C}}

{\makeatletter
	\gdef\xxxmark{%
		\expandafter\ifx\csname @mpargs\endcsname\relax 
		\expandafter\ifx\csname @captype\endcsname\relax 
		\marginpar{xxx}
		\else
		xxx 
		\fi
		\else
		xxx 
		\fi}
	\gdef\xxx{\@ifnextchar[\xxx@lab\xxx@nolab}
	\long\gdef\xxx@lab[#1]#2{{\bf [\xxxmark #2 ---{\sc #1}]}}
	\long\gdef\xxx@nolab#1{{\bf [\xxxmark #1]}}
}

\definecolor{light}{rgb}{0.5, 0.5, 0.5}
\def\lighttitle#1{{\color{light}#1}}

\makeatletter
\newcommand{\problemtitle}[1]{\gdef\@problemtitle{#1}}
\newcommand{\probleminput}[1]{\gdef\@probleminput{#1}}
\newcommand{\problemquestion}[1]{\gdef\@problemquestion{#1}}
\NewEnviron{problem}{
  \problemtitle{}\probleminput{}\problemquestion{}
  \BODY
  \par\addvspace{.5\baselineskip}
  \noindent
  \begin{tabularx}{\textwidth}{@{\hspace{\parindent}} l X c}
    \multicolumn{2}{@{\hspace{\parindent}}l}{\@problemtitle} \\
    \lighttitle{Input:} & \@probleminput \\
    \lighttitle{Question:} & \@problemquestion
  \end{tabularx}
  \par\addvspace{.5\baselineskip}
}
\makeatother

\usepackage{enumitem}

\newcommand{\e}{\epsilon}
\renewcommand{\Pr}{\mathrm{Pr}}

\usepackage{thmtools}

\begin{document}

\title{Traversing the FFT Computation Tree for Dimension-Independent Sparse Fourier Transforms}
\author{Karl Bringmann\\Saarland Uni. \& MPI \and Michael Kapralov\\EPFL \and Mikhail Makarov\\EPFL \and Vasileios Nakos\\Saarland Uni. \& MPI \and Amir Yagudin\\MIPT \and Amir Zandieh\\MPI}

\maketitle
\thispagestyle{empty}

\begin{abstract}


We are interested in the well-studied Sparse Fourier transform problem, 
where one aims to quickly recover an approximately Fourier $k$-sparse domain vector $\widehat{x} \in \mathbb{C}^{n^d}$ from observing 
its time domain representation $x$. In the exact $k$-sparse case the best known dimension-independent algorithm 
runs in near cubic time in $k$ and it is unclear whether a faster algorithm like in low dimensions is possible. Beyond that, all known approaches either suffer from an exponential dependence of 
their runtime on the dimension $d$ or can only tolerate a trivial amount of noise. This is in sharp contrast with the classical FFT algorithm of Cooley and Tukey, 
which is stable and completely insensitive to the dimension of the input vector: its runtime is $O(N\log N)$ in any dimension $d$ for $N=n^d$. 
Our work aims to address the above issues.

First, we provide a translation/reduction of the exactly $k$-sparse Sparse FT problem to a concrete tree exploration task which asks to recover $k$ leaves 
in a full binary tree under certain exploration rules. Subsequently, we provide (a) an almost quadratic in $k$ time 
algorithm for the latter task, and (b) evidence that obtaining a strongly subquadratic time for Sparse FT via this approach is likely to be impossible.
 We achieve the latter by proving a conditional quadratic time 
lower bound on sparse polynomial multipoint evaluation (the classical non-equispaced sparse Fourier transform problem) which is a core routine in the aforementioned translation.
Thus, our results combined can be viewed as an almost complete understanding of this approach, which is the only known approach that yields 
sublinear time dimension-independent Sparse FT algorithms.

Subsequently, we provide a robustification of our algorithm, yielding a robust cubic time algorithm under bounded $\ell_2$ noise. This requires
proving new structural properties of the recently introduced adaptive aliasing filters combined with a variety of new techniques and ideas. 
Lastly, we provide a preliminary experimental evaluation comparing the runtime of our algorithm to FFTW and SFFT 2.0.  
\end{abstract}
\setcounter{page}{0}
\newpage
\thispagestyle{empty}
\tableofcontents

\newpage
\setcounter{page}{1}

\section{Introduction.}

Computing the largest in magnitude Fourier coefficients of a function without computing all of its Fourier transform, or reconstructing a sparse vector/signal $x$ from partial Fourier measurements are common and well-studied tasks across science and engineering, as they appear in a variety of disciplines. Possibly the earliest work on the topic was by Gaspard de Prony in 1795, who showed that any $k$-sparse vector can be efficiently reconstructed from its first $2k$ Discrete Fourier transform (DFT) coefficients. These ideas have been re-discovered/used both in the context of decoding BCH codes~\cite{wolf1967decoding}, as well as in the context of computer algebra by Ben-Or and Tiwari~\cite{ben1988deterministic}. In the context of learning theory, and in particular learning decision trees, Kushilevitz and Mansour~\cite{KM} devised an algorithm that detects the largest Fourier coefficients of a function defined over the Boolean hypercube, building upon~\cite{GL}. The work of~\cite{AGS} uses sparse Fourier transform techniques in cryptography, namely for proving hard-core predicates for one-way functions. In 2002, a sublinear-time efficient algorithm for learning the $k$ largest DFT coefficients was proposed in~\cite{GGIMS}; this line of work has resulted in (near-)optimal algorithms~\cite{GMS,hikp12a,k16,k17} for the DFT case. In terms of its applications to signal processing and reconstruction, arguably the most prominent is the work of Candes, Donoho, Romberg, and Tao~\cite{Don,CTao,CRT}, which has far-reaching applications in fields such as medical imaging and spectroscopy~\cite{MRICS,nmrCS}, and created the area of \emph{compressed sensing}; the reader may consult the text~\cite{FR} for a thorough view on the topic.

Formally, the Sparse Fourier Transform problem is the following. Given oracle access to a size $N$ $d$-dimensional vector $x$, find a vector $\wh \chi$ such that

\[  \|\wh x - \wh \chi \|_p \leq C \cdot \mathrm{min}_{k\text{-sparse vectors } \wh z} \| \wh{x} - \wh z\|_q,       \]
where $C$ is the approximation factor, and $\|\cdot\|_p, \|\cdot\|_q$ are norms. The number of oracle accesses to $x$ shall be referred to as \emph{sample complexity}. The most well studied case in the literature is the case where $C=1+\epsilon$ (or constant) and $p = q = 2$, referred to as the $\ell_2/\ell_2$ guarantee. Other well-studied cases are the so-called $\ell_\infty/\ell_2$ guarantee, where $C=\frac{1}{\sqrt{k}}, p = \infty, q = 2$, as well as the $\ell_2/\ell_1$ guarantee, see~\cite{CTao,IK,NSW19}. Our focus in this paper is the $\ell_2/\ell_2$ guarantee. Frequently, the $k$ largest in magnitude coordinates of $\wh{x}$ are referred to as the \emph{head} of the signal, while all the other coordinates are referred to as the \emph{tail} of the signal, or as \emph{noise}. With this vocabulary, the $\ell_2/\ell_2$ guarantee asks to recover the head of $\wh{x}$ with error up to $(1+\epsilon)$ times the noise level.

The research on the topic, especially over the last fifteen years, has been extensive~\cite{KM,LMN,FFTmachinelearning,FFTBoollearning,Mansour95,GGIMS,GMS,CTao,IGS,Iwen10,Ak,CGV,hikp12a,hikp12b,BoufounosCGLS12,pawar2013computing,ikp14,pawar2014, Bourgain2014,IK,OngPR15, ps15,JanakiramanENR15, ChenKPS16, HR16, k16,cevher2017adaptive, k17,CheraghchiI17,merhi2017new,kapralov2019dimension,amrollahi2019efficiently,NSW19,HR19,jin2020robust}.
Our understanding of the sample complexity of this problem is quite good: we know that $O(k \poly(\log N))$ samples are sufficient for finding in time near linear in $N$ a vector $\wh \chi$ satisfying any of the aforementioned guarantees~\cite{CTao,HR16,NSW19}. Regarding the particularly interesting case of $d=1$, the research effort of the community has produced time-efficient algorithms as well. The fastest algorithm, due to the celebrated work of Hassanieh, Indyk, Katabi, and Price~\cite{hikp12a}, runs in time $O(k \log (N/k)  \log N)$ and achieves the same sample complexity as well. We know also how to achieve $O(k \log N)$ sample complexity and $O(k \poly(\log N))$ running time~\cite{k17}. On the other extreme, when $d=\log N$, i.e. in the case of the Walsh-Hadamard transform, almost optimal running time is known to be achievable, even deterministically~\cite{CheraghchiI17}.

Along with the running time, the sample complexity, and the error guarantee, of particular interest is also the sensitivity of the algorithm to the underlying field. When we are concerned with Fourier transforms over $\mathbb{Z}_n^d$\footnote{This is the case with the groups of interest in the Sparse FT literature. Furthermore, these are the groups on which the FFT algorithm of Cooley and Tukey operates. For general finite groups $G$, the fastest FT algorithm runs in time almost $|G|^{\omega/2}~\cite{umans2019fast}$, where $\omega$ is the matrix multiplication exponent.}, this corresponds to the sensitivity to the dimension $d$. Indeed, virtually all Sparse Fourier transform algorithms have a running time that suffers from an \emph{exponential} dependence on $d$ (in particular $\log^{\Omega(d)} N$), and the techniques either in dimension $d=1$ or $d= \log N$ heavily rely on the structure of the corresponding group. At the same time, given that the Cooley-Tukey FFT algorithm itself is completely dimension-independent, a natural question is whether this independence transfers also to the Sparse Fourier transform setting. Concretely, is the curse of dimensionality an inherent problem, or an artifact of previous techniques? A major practical motivation is that a quest for removing the curse of dimensionality can ultimately lead to new insights for designing empirical, efficient algorithms in dimensions $d=3,4$, which are mostly relevant in applications in NMR-spectroscopy and MRI imaging. Thus, an algorithm with better dependence on the $d$ and $k$ could thus be of practical importance as well.

 A step towards dimension-independence was taken in~\cite{kapralov2019dimension}, by giving a $O(k^3 \cdot \poly(\log N))$-time algorithm which recovers \emph{exactly} $k$-sparse signals in any dimension. Their approach is based on pruning the FFT computation graph, using a new tool called adaptive aliasing filters. 
However, the aforementioned algorithm had two disadvantages: i) the time was cubic and there was no evidence whether this was optimal under some reasonable assumption, and ii) was not able to go beyond the barrier of exactly $k$-sparse signals (or, noise level $\poly(N)$ times smaller than the energy of the head). 
Somewhat relevant is an algorithm due to Mansour~\cite{Mansour95}, which performs breadth-first search on the Cooley-Tukey FFT computation tree, and can get $\poly(k)$ running time for exactly $k$-sparse signals, but pays an additional multiplicative \emph{signal to noise ratio} factor for general signals~\cite{Mansour95}. We also mention a beautiful $O(k \cdot \poly(\log N))$-time algorithm for exactly $k$-sparse signals from~\cite{GHIKPS}, which requires  a distributional assumption on the support of the input signal in Fourier domain and unfortunately suffers from the restriction $k= O(N^{1/ d})$; already in dimension $d = O(\log N / \log \log N)$, this guarantees correctness only for $k \leq \poly(\log N)$.

\paragraph{Our results.} 
First, we translate the exactly $k$-Sparse FT problem using the machinery developed in~\cite{kapralov2019dimension} to a tree exploration problem 
that is accessible without any knowledge on Fourier transform. Our first main result is an almost complete understanding of this line of attack.

\begin{itemize}
\item The tree exploration task can be solved in almost quadratic time, and hence the exact $k$-Sparse FT problem can be solved in almost quadratic time.
This shaves off almost a factor of $k$ from the previous best sublinear-time, dimension-independent algorithm of~\cite{kapralov2019dimension}.
\item The quadratic time is most likely impenetrable by any explorative algorithm which successively peels off elements. That implies that overcoming this quadratic time barrier will likely require a major paradigm shift in Sparse FFT technology. This is based on a lower bound
on sparse polynomial multipoint evaluation and is interesting in its own right as the problem is well-studied under the name of non-equispaced Fourier transform.

\end{itemize}

In the robust case, we obtain a quadratic sample complexity, sublinear-time, dimension-independent algorithm that recovers the head of the signal under bounded $\ell_2$ noise, i.e. when every frequency in the head is larger than the energy of the tail. Even under this seemingly restricted noise model, designing an efficient algorithm turns out to be non-trivial, requiring a constellation of new techniques. Previous algorithms were either i) robust and dimension-independent but not sublinear-time~\cite{CTao,IK,NSW19}, ii) sublinear-time and robust but not dimension-independent~\cite{GMS,hikp12a,k16}, or iii) sublinear-time and dimension-independent but not robust to any form of noise~\cite{kapralov2019dimension}. We also discuss all the barriers we have faced, including the barrier to handling noise of larger magnitude, in Section~\ref{sec:barriers}.

\section{Computational Tasks and Formal Results Statement.}\label{sec:results}

This section contains the computational tasks studied in this paper, our results, and a preparations section for the lower bound, namely Theorem~\ref{thm:lower_bound}. We will be concerned with $N$-length $d$-dimensional vectors $x: [n]^d \rightarrow \mathbb{C}$, where $N=n^d$ and $n$ is a power of $2$. Thus, $N,n,d$ will remain unaltered throughout the paper. We will use the notation $[n]$ to denote the set of integer numbers $\{0, 1, \dots, n-1 \}$. We will use a \textbf{non-standard} notation $\widetilde{O}(f) = O( f \poly(\log N))$, where $f$ is some parameter and $N$ is the size of our underlying vector $x$. For a vector $x$, we denote $\|\wh{x}\|_0 = \left| \left\{ \bm f\in [n]^d: \wh{x}_{\bm f} \neq 0\right\} \right|$, and $\wh{x}_T$, for a set $T \subseteq [n]^d$, to be the vector that results from zeroing out every coordinate of $x$ outside of $T$. We let $\wh{x}_{-k}$ be the vector that occurs after zeroing out the top $k$ coordinates in magnitude, breaking ties arbitrarily. All logarithms are base $2$. For the algorithm we present, we shall assume exact arithmetic operations over $\mathbb C$ in unit time throughout the paper, although the analysis goes through with $\frac{1}{\poly(N)}$ precision as well.


\begin{figure}
\begin{tabular}{|m{\dimexpr \linewidth - 5em}|m{5em}|}
  \hline
  Task & Result \\
  \hline
  \hline
  \begin{minipage}{\linewidth}
\begin{problem}
  \problemtitle{$\star$ Sparse Fourier Transform in the exact case}
  \probleminput{Integers $n,d,k$ and $N=n^d$, and oracle access to a vector $x \in \mathbb{C}^{n^d}$ satisfying $\|\wh{x}\|_0 \leq k$.}
  \problemquestion{Compute  $\wh x$.}
\end{problem}
\end{minipage}
&Theorem~\ref{thm:exact_quadratic}\\
\hline
\begin{minipage}{\linewidth}
\begin{problem}
  \problemtitle{$\star$ $\ell_2/\ell_2$ Sparse Fourier Transform}
  \probleminput{Integers $n,d,k$ and $N=n^d$, parameter $\epsilon <1$, and oracle access to a vector $x \in \mathbb{C}^{n^d}$.}
  \problemquestion{Compute a vector $\wh{\chi} \in \mathbb{C}^{n^d}$ such that $\|\wh{x} - \wh{\chi}\|_2 \leq (1+\epsilon) \|\wh{x}_{-k}\|_2$.}
\end{problem}
\end{minipage}
&Theorem~\ref{thrm:near-quad-robust}\\
\hline
\begin{minipage}{\linewidth}
\begin{problem}
  \problemtitle{$\star$ Non-Equispaced Fourier Transform}
  \probleminput{Integers $n,d$, parameter $\epsilon<1$, two sets $F, T \subseteq [n]^d$ with $|F|=|T|=k$, and a vector $x \in \mathbb{C}^{n^d}$ supported on $T$.}
  \problemquestion{Compute additive $\pm \epsilon \|\wh{x}\|_2$ approximations to each of $\wh{x}_{\ff}$, for $ \bm{f} \in F$.}
\end{problem}
\end{minipage}
&Theorem~\ref{thm:lower_bound}\\
\hline
\begin{minipage}{\linewidth}
\begin{problem}
  \problemtitle{$\star$ Sparse Polynomial Multipoint Evaluation}
  \probleminput{Integers $n,k$, parameter $\epsilon<1$, a polynomial $p$ of degree $n$ and sparsity $k$, i.e. $k$ non-zero coefficients, each of which is of magnitude $1$, as well as points $a_1,a_2,\ldots,a_k \in \mathbb{C}^n$ of magnitude $1$.}
  \problemquestion{Compute additive $\pm \epsilon $ approximations to each of $p(a_i)$, for all $i=1,2,\ldots,k$.}
\end{problem}
\end{minipage}
&Theorem~\ref{thm:multipoint}\\
\hline
\begin{minipage}{\linewidth}
\begin{problem}
\problemtitle{$\star$ Orthogonal vectors, $\textsc{OV}_{k,d}$}
\probleminput{$A,B \subseteq \{0,1\}^d$, with $|A|=|B|= k$}
\problemquestion{Determine whether there exists $a \in A, b \in B$ such that $\langle a,b \rangle = 0$.}
\end{problem}
\end{minipage}
&\\
\hline
\end{tabular}
\caption{Computational tasks considered in this paper.}\label{fig-problem-defs-all}
\end{figure}


We start by summarizing the formal definitions of all relevant computational problems in Figure~\ref{fig-problem-defs-all}. With these definitions in place we can state our results as follows:

\begin{theorem}[Almost-Quadratic Time Exact $k$-Sparse FFT]
  \footnote{proved as Theorem~\ref{thm:exact_sparseft} in \cref{sec:tree_problem} and \cref{sec:exact-k-sparse-reduction}}
\label{thm:exact_quadratic}
Given oracle access to $x: [n]^d \to \C$ with $\|\wh{x}\|_0 \leq k$, we can find $\wh{x}$ in deterministic time
	\[ \widetilde O\left( k^2 \cdot 2^{8\sqrt{\log k \cdot \log \log N}} \right).	\]
\end{theorem}
We formally show that the exact Sparse FFT problem can be reduced to a tree exploration problem and show how to solve the tree exploration in almost quadratic time, and thus prove the above theorem, in \cref{sec:tree_problem} and \cref{sec:exact-k-sparse-reduction} as Theorem~\ref{thm:exact_sparseft}.

\begin{conjecture} (Orthogonal Vectors Hypothesis(OVH)~\cite{Williams05,AbboudWW14})
\label{ovc}
For every $\epsilon > 0$, there exists a constant $c \geq 1$ such that $\textsc{OV}_{k,d}$ (see Figure~\ref{fig-problem-defs-all}) requires $\Omega(k^{2-\epsilon})$ time whenever $d \geq c\log k$.
\end{conjecture}

It is known that a collapse of the Orthogonal Vectors Hypothesis would have groundbreaking implications in algorithm design, see~\cite{GaoIKW19} and~\cite{AbboudBDN18}.

\begin{theorem}[Lower Bound for Non-Equispaced Fourier Transform]
\footnote{proved as~Theorem~\ref{thm:detailed_lb} in~\cref{sec:lb}}
\label{thm:lower_bound}
Assume that for all $k<n$ and $\epsilon>0$ there exists an algorithm that solves the Non-Equispaced Fourier Transform in time $O(k^{2-\delta} \poly(\log(n/\epsilon)))$ for some constant $\delta > 0$. Then the Orthogonal Vectors hypothesis fails.
\end{theorem}
\begin{proofsketch}
 Given sets of vectors \[A=\{a_0,a_1,\ldots,a_{k-1}\},~B=\{b_0,b_1,\ldots,b_{k-1}\} \subseteq\{0,1\}^d,\] we build $|A|$ points in time domain and $|B|$ points in frequency domain as follows. We pick sufficiently large $M,q,N$ (for details see Section~\ref{sec:lb}) and define for $j \in [k]$:
  \begin{align*}
    t_j := \sum_{r\in[d]} a_j(r) \cdot M^{r q}, \qquad
    f_j := \sum_{r\in [d]} b_j(r) \cdot \frac{N}{M^{r q +  1}}, 
  \end{align*}

Subsequently, we look at the indicator vector of the set $\{t_0,t_1,\ldots,t_{k-1}\}$, let it be $x$. Asking for the values $\wh x_{f_0},\wh{x}_{f_1},\ldots,\wh{x}_{f_{k-1}}$ corresponds exactly to the non-equispaced Fourier transform problem. 
Using the aforementioned evaluations we show that it is possible to extract the values

\[ V_{j,h} := \sum_{\ell \in [k]} \langle a_\ell, b_j \rangle^h, \text{ for $j\in[n],h \in [d]$.} \]
For a fixed $j$, the values of $V_{j,h}$ can be expresed in terms of $Z_r := \left|\{ \ell \in [k] \mid \langle a_\ell, b_j \rangle = r \}\right|$, via multiplication by a $d\times d$ Vandermonde matrix. Since the entries involved in this matrix and $V_{j,h}$ have $\poly(d,\log k)$ bits, we can then solve for $Z_0$ in $\poly(d,\log k)$ time, where $Z_0$ corresponds to the number of vectors $a \in A$ which are orthogonal to $b_j$. Repeating this over all $j \in [k]$ yields whether there exists a pair of orthogonal vectors. 

{\em Of course, the overview presented above completely ignores how we actually extract the values of $V_{j, h}$ from evaluations of the Fourier transform. This carefully exploits periodicity of complex exponentials -- see Section~\ref{sec:lb} for more details.}
\end{proofsketch}
A lower bound for sparse polynomial multipoint evaluation (Figure~\ref{fig-problem-defs-all}) also follows immediately.
\begin{theorem}[Lower bound for Sparse Polynomial Multipoint Evaluation over $\mathbb C$]
\label{thm:multipoint}
Assume that for all $k<n$ and $\epsilon$ there exists an algorithm for sparse polynomial multipoint evaluation which runs in time $k^{2-\delta} \poly(\log(n/\epsilon))$.  Then the Orthogonal Vector Hypothesis fails.
\end{theorem}



\paragraph{Significance of our lower bound for computational Fourier Transforms.} Non-equispaced Fourier transform falls into a class of Fourier transforms referred to as \emph{non-uniform}. These transforms are an extensively studied topic in signal processing and numerical analysis~\cite{greengard1987fast,fessler2003nonuniform,greengard2004accelerating}, with numerous applications in imaging, signal interpolation and solutions of differential equations; the reader may consult the texts~\cite{bagchi1996nonuniform,potts2001fast,bagchi2012nonuniform}. 
\newline

To present our robust Sparse FFT results, we first quantify the notion of ``bounded $\ell_2$ noise''.
\paragraph{High SNR model.} A vector $x: [n]^d \to \C$ satisfies the $k$-high SNR assumption, if there exist vectors $w,\eta: [n]^d \to \C$ such that i) $\wh{x}=\wh{w} + \wh{\eta}$, ii) $\supp(\wh{w}) \cap \supp(\wh{\eta}) = \emptyset$, iii) $|\supp(\wh{w})| \leq k$ and iv) $|\wh{w}_f| \geq 3 \cdot \|\wh{\eta}\|_2$\footnote{The constant $3$ is arbitrary, and can be driven down to $(1+\zeta)$, for any $\zeta>0$.}, for every $f \in \supp(\wh{w})$.


\begin{restatable}[Robust Sparse Fourier Transform with Near-quadratic Sample Complexity]{theorem}{recursivesfftrobust}
\footnote{proved as Theorem~\ref{thm:core} in \Cref{sec:robust_first}}
\label{thrm:near-quad-robust}
	Given oracle access to $x: [n]^d \to \C$ in the $k$-high SNR model and parameter $\epsilon >0$, we can solve the $\ell_2/\ell_2$ Sparse Fourier Transform problem with high probability in $N$ using 
	\[m = \widetilde{O}\left( \frac{k^2}{\epsilon} + k^2 \cdot 2^{\Theta\left( \sqrt{\log k \cdot \log \log N} \right)} \right)\] samples from $x$ and $\widetilde{O}\left( \frac{k^3}{\epsilon} \right)$ running time.
\end{restatable}
This theorem is restated as Theorem~\ref{thm:core} in \Cref{sec:robust_first} followed by the proof. We re-iterate that even though the noise model we consider might seem restrictive, it turns out to be quite challenging requiring whole new constellation of ideas. The starting point here is the observation that the adaptive aliasing filters constructed by~\cite{kapralov2019dimension} in fact form an orthogonal basis (see Lemma~\ref{lem:gram} in Section~\ref{sec:filters_new}), and while the norms of individual filters in the family are not the same, the sum of their squares is equal to $1$ at every point in time domain (see Lemma~\ref{filter-robust-multidim} in Section~\ref{sec:filters_new}). The combination of these new facts allows us to argue noise stability of our algorithm in Section~\ref{sec:robust_first}.

Additionally, in Section~\ref{sec:barriers} we explain how we are led to consider this particular notion of high SNR regime, and why handling lower SNR is a hard barrier for algorithms which explore a pruned Cooley-Tukey FFT computation tree (which is also the only known class of algorithms that enables sublinear and dimension-independent recovery). The discrepancy between the running time and sample complexity provided by Theorem~\ref{thrm:near-quad-robust} is due to the fact that we used non-uniform Sparse Fourier Transform to subtract recovered frequencies from time domain in our algorithm, which requires quadratic time as per Theorem~\ref{thm:lower_bound}.

\paragraph{Experimental Evaluation.} Lastly, we present our experimental evaluation in Section~\ref{sec:experiments}, where we compare our method to the highly optimized software packages such as FFTW and SFFT 2.0.
The source code of our implementation is available at \url{https://bitbucket.org/michaelkapralov/sfft-experiments}.

\section{Technical overview.}
\label{sec:tree_problem}

In this section we first present (in Section~\ref{sec:near-isometry}) a new near-isometry property of adaptive aliasing filters of~\cite{kapralov2019dimension}, which underlies our robust high dimensional Sparse FFT algorithm. We then present (in Section~\ref{sec:tree-exp}) an abstract formulation of the Sparse Fourier transform algorithms which work based on these adaptive aliasing filters as an abstract \emph{Tree Exploration Problem}. Such a formulation allows us to present the key ideas behind our quadratic time dimension-independent Sparse FFT algorithm in a concise way, avoiding unnecessary Fourier analytic formalism.  The formal connection the tree exploration problem and the adaptive aliasing filter-based Sparse FFT is presented in \cref{sec:exact-k-sparse-reduction}. 

\subsection{A Near-Isometry Property of Adaptive Aliasing Filters.}
\label{sec:near-isometry}

Recall that given a signal $x: [n]^d \rightarrow \mathbb{C}$, the execution of the FFT algorithm produces a binary tree, referred to as $\tfull_N$. The root of $\tfull_N$ corresponds to the universe $[n]^d$, while the children of the root correspond to $\left[n/2 \right]\times [n]^{d-1}$; note that FFT recurses by peeling off the least significant bit. Every node $v$ has a \emph{label} ${\bm f}_v \in \ZZ_{n}^d$ associated to it, defined according to the following rules.

\begin{enumerate}
\item The root has label ${\bm f}_{\text{root}} = (\underbrace{0,0,\ldots,0}_{d~\mathrm{entries}})$, and corresponds to the universe $[n]^d$.
\vspace{-8pt}
\item The children $v_\lef,v_\righ$ of a node $v$ which corresponds to the universe $[n/2^l]\times [n]^{d'}$, with $0 \leq d' \leq d-1, 0\leq l \leq \log n-1$, have the following properties. Both correspond to universe $[n/2^{l+1}] \times [n]^{d'}$, and $v_\righ$ has label $\bm{f}_{v_\righ} = \bm{f}_v $, while $v_\lef$ has label $\bm{f}_{v_\lef} = \bm{f}_v + (\underbrace{0,0, \ldots ,0}_{d'}, 2^{l} ,\underbrace{0,0,\ldots,0}_{d-d'-1})$.
\vspace{-8pt}
\item The children of a node $v$ corresponding to universe $[1]\times[n]^{d'}$ with $d'>0$, are $v_\lef, v_\righ$, corresponding to universe $[n/2]\times [n]^{d'-1}$ and have labels $\bm{f}_{v_\righ} = \bm{f}_v$ and $\bm{f}_{v_\lef} = \bm{f}_v + (\underbrace{0,0,\ldots ,0}_{d' -1 },1,\underbrace{0,0,\ldots,0}_{d-d'})$ respectively.
\vspace{-8pt}
\item A node $v$ corresponding to universe $[1]$ is called a \emph{leaf} in $\tfull_N$.
\end{enumerate}
The above rules create a binary tree of depth $\log N$, which corresponds to the FFT computation tree. The labels of the leaves of $\tfull_N$ represent the set $[n]^d$ of all possible frequencies of any signal $x: [n]^d \rightarrow \mathbb{C}$ in the Fourier domain. We demonstrate $\tfull_N$ that corresponds to the $2$-dimensional FFT computation on universe $[4] \times [4]$ in Figure~\ref{Tfull-labeling}. 
Subtrees $T$ of $\tfull_N$ can be defined as usual. For every node $v \in T$, the \emph{level} of $v$, denoted by $l_T(v)$, is the distance from the root to $v$. 
We denote by $\leaves(T)$ the set of all leaves of tree $T$, and for every $v\in \leaves(T)$, its \emph{weight} $w_T(v)$ \emph{with respect to} $T$ is the number of ancestors of $v$ in tree $T$ with two children. The levels (distances from the root) on which the aforementioned ancestors lie will be called $\Anc(v,T)$. Furthermore, the sub-path of $v$ with respect to $T$ will be the children of the aforementioned ancestors which are not ancestors of $v$. Additionally, for a node $v \in T$ we denote the subtree of $T$ rooted at $v$ by $T_v$.

The following definition will be particularly important for our algorithms.
\begin{definition}[Frequency cone of a leaf of $T$]
For every subtree $T$ of $\tfull_N$ and every node $v\in T$, we define the {\em frequency cone of $v$} with respect to $T$ as,
\[	\subtree_T(v):=\left\{ \ff_u : \text{ for every leaf } u \text{ in subtree of }\tfull_N \text{ rooted at }v \right\}.\]
Furthermore, we define $\supp{(T)} : = \bigcup_{u\in \leaves(T)} \subtree_T(u)$.
\end{definition}

The \emph{splitting tree} of a set $S \subseteq [n]^d$ is the subtree of $\tfull_N$ that contains all nodes $v \in \tfull_N$ such that $S \cap \subtree_{\tfull_N}(v) \neq \emptyset$.

The main technical innovation of \cite{kapralov2019dimension} is the introduction of adaptive aliasing filters, a new class of filters that allow to isolate a given frequency from a given set of $k$ other frequencies using $O(k)$ samples in time domain and in $O(k \log N)$ time -- see Section~\ref{sec:filters_old} for a more detailed account of this prior work.

\begin{definition}[$(v, T)$-isolating filter, see Definition~\ref{def:v-t-isolating}] \label{def:v-t-isolating-overview}
	Consider a subtree $T$ of $\tfull_N$, and a leaf $v$ of $T$. A filter $G:[n]^d \to \C$ is called \emph{$(v,T)$-isolating} if the following conditions hold:
	\begin{itemize}
		\item For all $\bm f\in \subtree_T(v)$, we have $\wh{G}(\bm f)=1$.
		\item For every $\bm f'\in \bigcup_{\substack{ u \in \leaves(T) \\u \neq v }}\subtree_T(u) $, we have $\wh{G}(\bm f')=0$.
	\end{itemize}
\end{definition}

As shown in~\cite{kapralov2019dimension}, for a given tree $T$ and a node $v$ one can construct isolating filters $G$ such that $\|G\|_0 = O(2^{w_T(v)})$, and $\wh{G}(\bm f)$ is computable in $\widetilde{O}(1)$ time (see also Lemma~\ref{lem:isolate-filter-highdim}). The sparsity of $G$ in time domain, i.e. $\|G\|_0$, corresponds to the number of accesses to $x$ needed in order to get our hands on $(\wh{G} \cdot \wh{x})_{\bm f}$ for a fixed $\bm f$.

Unfortunately, as we have already pointed out, the algorithm in~\cite{kapralov2019dimension} works only for exactly $k$-sparse signals, and also demands cubic time and sample complexity. Our new toolkit shows that all three limitations can be remedied (though not completely simultaneously). The key observation underlying our new techniques is a new near-isometry property of adaptive aliasing filters.

\paragraph{Collectively, adaptive aliasing filters act as near-isometries.}

Adaptive aliasing filters as used in~\cite{kapralov2019dimension} are particularly effective for \emph{non-obliviously} isolating elements of the head with respect to each other. However, in standard sparse recovery tasks, one desires control of the tail energy that participates in the measurement. This is a relatively easy (or at least well-understood) task in Sparse Fourier schemes which operate via $\ell_\infty$-box filters~\cite{hikp12a,hikp12b,ikp14,IK,k17}, but a non-trivial task using adaptive aliasing filters. The reason is that the tail via the latter filtering is hashed in a \emph{non-uniform} way. The hashing depends on the arithmetic structure of the elements used to construct the filters, as well as their arithmetic relationship with the elements in the tail. This non-uniformity is essentially the main driving reason for the ``exactly $k$-sparse'' assumption in~\cite{kapralov2019dimension}. Our starting point is the observation that for every tree $T \subseteq \tfull_N$, the $(v,T)$-isolating filters for $v \in \leaves(T)$, satisfy the following orthonormality condition in dimension one, see subsection~\ref{sec:filters_onedim}.
\begin{lemma}(Gram Matrix of adaptive aliasing filters in $d=1$)
Let $T \subseteq \tfull_n$, let $G_v$ be the $(v,T)$-isolating filter of leaf $v\in \leaves(T)$, as per \eqref{eq:isolat-filter-1d}. Let $v$ and $v'$ be two distinct leaves of $T$. 
Then,  {\bf (1)} $\|\wh{G}_v \|_2^2 := \sum_{ \xi \in [n]} |\wh{G}_v(\xi) |^2 = n \cdot 2^{-w_T(v)}$
and {\bf (2)}  the adaptive aliasing filters corresponding to $v$ and $v'$ are orthogonal, i.e. 
$\langle \wh G_v, \wh G_{v'} \rangle :=  \sum_{ \xi \in [n]} \wh G_v(\xi) \cdot \overline{\wh G_{v'}(\xi)} = 0.$
\end{lemma}

This already postulates that adaptive aliasing filters are relatively well-behaved: for a signal $x$ with tree $T$ all leaves of which have roughly the same weight, it must be the case that $x\mapsto \{ \langle \wh G_v,\wh x \rangle \}_{v \in \leaves(T)}$ is a near-orthonormal transformation. Of course, this is too much to ask in general. The crucial property that we will make use of is captured in the following Lemma, see Subsection~\ref{sec:filters-d}.

\begin{lemma}(see Lemma~\ref{filter-robust-multidim})
	Consider a tree $T \subseteq T_N^{full}$. For every leaf $v$ of $T$ we let $\wh{G}_v$ be a Fourier domain $(v,T)$-isolating filter. Then for every ${\bm\xi} \in [n]^d$,
$\sum_{ v\in \leaves(T)} |\wh{G}_v({\bm\xi})|^2 = 1.$
\end{lemma}

Using standard arguments, the above gives the following Lemma.

\begin{lemma}
For $z: [n]^d \rightarrow \mathbb{C} $, let $z^{\rightarrow \bm a}$ be the cyclic shift of $z$ by $a$, i.e. $z^{\rightarrow {\bm a}}(\bm f) := z(\bm f - \bm a )$, where the subtraction happens modulo $n$ in every coordinate. For a tree $T \subseteq \tfull_N$, 

\[	\mathbb{E}_{\bm a \sim U_{[n]^d}} \left[\sum_{ v \in \leaves(T)}  |\langle \wh{G_{v}}, \wh{z^{\rightarrow {\bm a}}} \rangle|^2 \right] = \|z\|_2^2,	\]

i.e. on expectation over a random shift the total collection of filters is an isometry.

\end{lemma}

The above property is the key new observation that underlies our analysis, but several other technical ideas are needed to obtain our robust result -- a more detailed overview is given in Section~\ref{sec:techniques}. We also note that the our ultimate robust algorithm does not achieve the standard $\ell_2/\ell_2$ sparse recovery guarantees, which allow for recovery of a good approximation to the signal if the energy of the top $k$ coefficients is larger than the energy of the tail. Instead, we show that recovery is possible when every one of the top $k$ coefficients dominates the cumulative energy of the tail of the signal. It is weaker, but one must note that several recent works on Fourier sparse recovery are only known to tolerate inverse polynomial amounts of noise~\cite{DBLP:conf/nips/HuangK15,DBLP:conf/stoc/Moitra15}, so our robustness guarantee appears to be a strong first step.

\subsection{A Tree Exploration Problem.}\label{sec:tree-exp}
We now present an abstract formulation of the Sparse Fourier transform algorithms which work based on the adaptive aliasing filters of~\cite{kapralov2019dimension} as an abstract Tree Exploration Problem. We will lay down the formal connection between this problem and the adaptive aliasing filter-based Sparse FFT in Section 9. For now, we focus on the tree problem without any reference to Fourier transforms.

\paragraph{The setup.} In this problem we are given a full binary tree $\tfull_N$ with $N$ nodes where each leaf of this tree has a (potentially complex) number written on it, known as the \emph{value} of the leaf. Suppose that $T$ is an instance of $\tfull_N$ with at most $k$ of its leaves having non-zero values. The goal of the tree exploration problem is to find those $k$ leaves and estimate their corresponding values. 
It has been formally shown in \cite{kapralov2019dimension} that there is a bijective correspondence between any $k$-sparse $\wh{x}:[n]^d \to \C$ ($\|\wh x \|_0 \le k$) and such tree $T$ with $k$ non-zero valued leaves, so the Sparse FFT problem can be formulated as learning the tree. We will show this formally in \cref{sec:exact-k-sparse-reduction}.

\begin{definition}[$\leaves$ and $\hleaves$] \label{def:hleaves-def} For a node $v \in T$ we let $\leaves(v)$ be the leaves of $\tfull_N$ which are the descendants of $v$ (including $v$ itself in case $v$ is a leaf). We let $\hleaves(v)$ denote the leaves in $\leaves(v)$ with non-zero values.
\end{definition}

For a set of vertices $S \subset \tfull_N$ we will denote by $T(S)$ the subtree of $\tfull_N$ with minimum number of vertices containing $S$ and the root.

\begin{definition}[Weight of a vertex] For a binary tree $T$, and a vertex $v \in T$, the weight $w_T(v)$ of $v$ is equal to the number of ancestors of $v$ in $T$ with two children. For a set $S\subseteq \tfull_N$ and $v \in \tfull_N$, we denote $w_S(v) := w_{T(S \cup \{ v\})}(v)$ for simplicity.
\end{definition}

It turns out that this weight function is subadditive for a fixed $v$.
\begin{lemma}\label{lem:weight_subadditive}
For any two sets $S_1, S_2 \subseteq \tfull_N$ and $v \in \tfull_N$,
\[
    w_{S_1 \cup S_2}(v) \leq w_{S_1}(v) + w_{S_2}(v)
\]
\end{lemma}
You can find the proof in \Cref{appndx-weight-subadditive}.

The tree $T$ is unknown to us and we can explore it indirectly only using two primitives; $\ZeroTest$, which can answer queries about whether $\hleaves(v)$ is empty, and $\Estimate$, which can estimate the value of a leaf. In this section we are not concerned about the internal working of these primitives and treat them as oracles. The implementation of these routines in the context of Sparse FFT is given in \cref{sec:exact-k-sparse-reduction}.
In order to design efficient versions of these primitives to virtually prune the tree $\tfull_N$ and avoid operating on $\tfull_N$ as a whole, we need to introduce two extra parameters, $\found$ and $\excl$. 

\paragraph{The $\found$ and $\excl$ parameters.} $\found$ is an associative array of already recovered leaves with some estimates for their values. We say that the estimates in $\found$ are correct if for all leaves $v$ in $\found$, $\found(v)$ is equal to the value of $v$ in tree $T$. 
We denote by $|\found|$ the number of leaves with non-zero estimates in $\found$. For two associative arrays, $\found_1$ and $\found_2$, $\found_1 + \found_2$ denotes their union. In our applications we will never take a union of two arrays with intersecting key sets. Our algorithm will virtually subtract the values in $\found$  from the corresponding leaves of tree $T$, essentially deleting them if the estimates are correct.

$\excl$ is a subset of nodes of $\tfull_N$. Our algorithm will virtually delete all subtrees with roots in $\excl$ from the tree. For technical reasons, the procedure only accept the set excluded in the form of a tree $T(\excl)$. For simplicity, we equate the set $\excl$ and its tree $T(\excl)$.
One can pictorially see examples of these notions in \Cref{fig:explanation_tree}.
This motivates the following definition:

\begin{definition}\label{def:isolated-found-excl}
	A vertex $v \in \tfull_N$ is \emph{isolated} by $\found$ and $\excl$ if
	\[
		\leaves(v) \cap \leaves(\excl) =\emptyset
	\]
	 and for all leaves $\ell$ of $\tfull_N$ not in $\leaves(v)$, either the value of $\ell$ is zero, the estimate $\found(\ell)$ for the value of $\ell$ is correct, or $\ell \in \leaves(\excl)$.
\end{definition}

Now we are ready to present the interfaces of $\ZeroTest$ and $\Estimate$ as well as their runtime.
\begin{assumption} \label{assmptn-zerotest-estimate}
There exist procedures $\ZeroTest$ and $\Estimate$ with the following properties,
\begin{enumerate}
\item  $\ZeroTest(\found, \excl, v, b)$, where $b$ is a positive integer, representing budget, and $v$ is a vertex in $\tfull_N$. This routine checks if there are any leaves with non-zero value in the subtree of $v$. More formally, if $v$ is isolated by $\found$ and $\excl$ and 
$|\hleaves(v)| \leq b$,
then the routine returns True if for every $\ell \in \leaves(v)$ either the value of $\ell$ is zero or the estimate $\found(\ell)$ for the value of $\ell$ is correct, and returns False otherwise; if $v$ is not isolated, it can return either True or False. 
The runtime of $\ZeroTest$ is $ \wt{O}\left(2^{w_\excl(v)}b + |\found|\cdot b  \right)$.

\item $\Estimate(\found, \excl,\ell)$, where $\ell$ is a leaf of $\tfull_N$. This routine estimates the value of $\ell$ correctly if $\ell$ is isolated by $\found$ and $\excl$. If not, the estimate is arbitrary.

The time complexity of $\Estimate$ is $\wt{O}\left(2^{w_{\excl}(\ell)} + |\found|\right)$.
\end{enumerate}
The above running time bounds still hold regardless of the correctness of the inputs.
\end{assumption}

\begin{figure}
	\vspace{-8pt}
	\centering
		\scalebox{.8}{
			\begin{tikzpicture}[level/.style={sibling distance=60mm/#1,level distance = 1.5cm}]
			\node [arn_l, label=\large{$r$}] (z){}
			child {node [arn_l] {}edge from parent [draw=black,very thin]
				child {node [arn_l] {}edge from parent [draw=black,very thin]
					child {node [arn_l, fill=none, label=below:1] {}edge from parent [draw=black,very thin]}
					child {node [arn_l, diamond, label=below:2] {}edge from parent [draw=black,very thin]}
				}
				child {node [arn_l] {}edge from parent [draw=black,very thin]
					child {node [arn_l, diamond, label=below:3] (a) {}edge from parent [draw=black,very thin]}
					child {node [arn_l, fill=none, label=below:4] {}edge from parent [draw=black,very thin]}
				}
			}
			child { node [arn_l] {} edge from parent [draw=black,line width=2.5pt]
				child {node [arn_l,  regular polygon,regular polygon sides=3, label=\Large{$v$}, very thin] {}edge from parent [draw=black,line width=2.5pt]
					child {node [arn_l, label=below:5, very thin] {}edge from parent [draw=black,very thin]}
					child {node [arn_l, label=below:6, very thin] (b) {}edge from parent [draw=black,very thin]}}
				child {node [arn_l, label=\Large{$u$}, very thin] {}edge from parent [draw=black,line width=2.5pt]
					child {node [arn_l, fill=none, label=below:7, very thin] (c){}edge from parent [draw=black,very thin]}
					child {node [arn_l, label=below:8, very thin] {} edge from parent [draw=black,very thin]}
				}
			};
			
			\draw[draw=black, ->] (-3.1,-5.6) -- (-3.5,-5.2);
			\draw[draw=black, ->] (-2.9,-5.6) -- (-2.5,-5.2);
			\node [] at (-3,-5.8) [label=center:\Large{$\found$}]	{};

			\draw[draw=black, ->] (0,-2.3) -- (1.3,-3);
			\node [] at (0,-2) [label=center:\Large{$\excl$}]	{};

			\draw[draw=black, ->] (5.5,-3) -- (4.8,-3);
			\node [] at (5.4,-3) [label=right:\Large{$u$ is isolated}]	{};
			
			\draw[draw=black, ->, >=stealth] (3.5,-0.15) -- (1.5,-0.5);
			\node [] at (4.5,0) [label=center:\Large{tree $T$}]	{};

			\node [] at (5,-2) [label=right:\Large{Set $S = \{ u, v \}$}]	{};
			
			\end{tikzpicture}
		}
		\par
		\caption{Example of a tree $\tfull_8$. The leaves are labeled by integers $1$ to $8$. The unfilled leaves have zero associated value, hence $\hleaves(u) = \{ 8 \}$, $\leaves(u) = \{ 7, 8 \}$. Also the diamond-shaped leaves $2$ and $3$ are assumed to lie in the $\found$ together with their correct estimates. Thus if the triangle-shaped node $v$ is in the set $\excl$, then $u$ is isolated by $\found$ and $\excl$. Thick edges represent the tree $T$, and the set of its leaves is $S = \{u, v\}$. If we pick $u$ as the next vertex in the vanilla algorithm, the set $\excl$ will again contain only the vertex $v$.}
\label{fig:explanation_tree}
\end{figure}
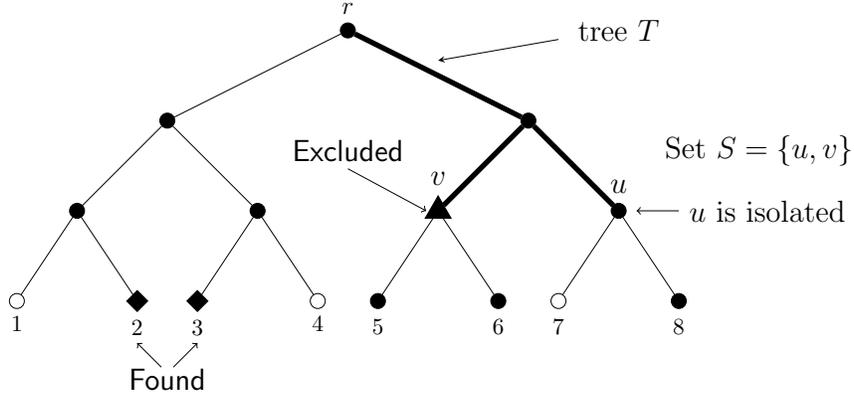

\paragraph{The vanilla algorithm in~\cite{kapralov2019dimension}.} Using the above translation of Sparse FFT into a tree learning problem, the algorithm in~\cite{kapralov2019dimension} does the following. At all times it maintains a tree $T \subseteq \tfull_N$ and a set $\found$ of leaf nodes with their \emph{perfect} estimates such that the following invariants hold:
 \[
	\found  \subseteq \hleaves(\mathrm{root}),
\]  
\[
	 \hleaves(\mathrm{root}) \setminus \found \subseteq \leaves(T).
\]
Initially, $T$ contains only the root of the tree and $\found = \emptyset$.

While $T$ is not empty, the algorithm picks a leaf $u \in T$ with the smallest weight $w_T(u)$. It now needs to find a set $\excl$ such that $u$ would be isolated by $\found$ and  $\excl$. As it turns out, it is enough for $\excl$ to contain the children of nodes in $T$ on the path from $u$ to the root, except for the nodes that are on this path themselves --- see Fig.~\ref{fig:explanation_tree} for an illustration. At every point the algorithm calls 
$\ZeroTest(\found, \excl, u, k)$ to test whether the subtree rooted at $u$ contains a non-recovered leaf, in order to avoid exploring empty subtrees.  The budget of \ZeroTest~ is chosen to be $k$ in order to avoid false negatives, i.e. never miss a non-empty subtree.  This is the main inefficiency in~\cite{kapralov2019dimension} which we address here, obtaining a quadratic time algorithm. Before outlining our main algorithmic technique, which allows us to handle frequent false negatives in \ZeroTest~ by a novel error correction mechanism, we note that quadratic time is a natural barrier for any algorithm that iteratively recovers the input signal. Indeed, suppose that the algorithm has recovered a constant fraction of coefficients and recurses on the rest. The most common approach here is to subtract the recovered elements from time domain samples, reducing to the same problem with a smaller number of coefficients --  but this requires nearly quadratic time by our lower bound (see Theorem~\ref{thm:lower_bound} in \cref{sec:lb}).

\subsection{Obtaining Almost Quadratic Runtime via Hierarchical Error Correction.} \label{sec:almost-quadratic-overview}
Our main insight is that not all calls to $\ZeroTest$ need to succeed. Instead, one can try to assign varying budgets to the nodes to be explored and then perform \emph{hierarchical error correction} in order to detect errors in exploration which are caused by the failures of $\ZeroTest$ due to incorrect budget assignments. We explain the main underlying idea next.

Suppose that we optimistically explore a subtree of $\tfull_N$ with a budget $b\ll k$ for $\ZeroTest$. If the budget assigned to the subtree was correct, i.e. number of heavy leaves in that subtree are no larger than $b$, then we correctly recover the leaves locally and the resulting speed-up will be a multiplicative $k/b$ factor. 
On the other hand, if the assigned budget was wrong, exploration could be misguided and estimates could be incorrect. 
The idea is that this error can be detected at an ancestor of the subtree if we assign that ancestor a large enough budget so that the invocation of $\ZeroTest$ does not get fooled. As a sanity check, note that this is definitely the case for the root, where we use a budget of size $k$. 
Once the error is detected at an ancestor of the subtree in question, the algorithm will re-explore that subtree with increased budget. Roughly speaking, the algorithm tries to learn the correct budget of each subtree by performing \emph{backtracking}: guess a budget, explore the subtree, determine whether it is wrong upon backtracking to an ancestor and subsequently increase the budget to explore that subtree, so on so forth. In this approach there are two things that need to be carefully balanced. On the one hand, one needs to use the smallest possible budgets, close to the actual sizes of the corresponding subtrees (so that calls to $\ZeroTest$ are cheap) and on the other hand one needs to control the amount of backtracking the algorithm performs; a smaller budget leads to a larger number of required backtracking steps. A careful analysis reveals that, roughly speaking, increasing the budget by a factor of $1/\alpha$ for $\alpha := 2^{-2\sqrt{ \log k \cdot \log \log N}}$ whenever an adjustment is needed ensures nearly quadratic runtime. 

The pseudocode of the recursive recovery algorithm is presented in \Cref{alg:exact_sparse_simple}. The algorithm is passed budget $s$, which is assumed to be the sparsity of the subtree of $v$. If the sparsity is low it defaults to the cubic algorithm of \cite{kapralov2019dimension}. Otherwise it starts the inner loop in lines \ref{line:inner_loop_start} to \ref{line:inner_loop_end}, where it explores the tree. It maintains the tree $T$ containing all of unexplored heavy leaves of $v$ as a set $S$ of its leaves. On each iteration of inner loop, the algorithms picks a vertex $z$ from $S$ with minimum weight in line~\ref{line:pick_min_leaf}, for which it first checks if there is any heavy leaves in $\leaves(z)$. If not, it discards the vertex and continues to the next iteration. Otherwise, it tries to guess that the sparsity of each child of $z$ are at most $\alpha \cdot s$ and runs itself recursively on both children of $z$ with a this decreased budget in lines~\ref{line:exact-recursive-call1} and~\ref{line:exact-recursive-call2}. For each child it then checks using \ZeroTest{} in lines~\ref{line:exact_zerotest1} and~\ref{line:exact_zerotest2} if the guess was correct and the values were found correctly, they get added to $\found_{out}$, otherwise the corresponding child is added to $S$. Finally, if the $z$ is the leaf of the $\tfull_N$, the algorithm runs \Estimate{} on it instead and then adds the recovered value to $\found_{out}$.

\begin{algorithm}[!h]
	\caption{$\textsc{ExactSparseRecovery}(\found, \excl, v, s, k)$}
	\label{alg:exact_sparse_simple}
	\begin{algorithmic}[1]
	\If { $s \leq 1/\alpha$} \Comment{This is the base case --- run the qubic time algorithm}\label{line:svs1alpha}
		\State $\found_{out} \leftarrow \textsc{SlowExactSparseRecovery}(\found,\excl, v, s)$\label{line:base-case-slow-call}
		\State \textbf{if} $|\found_{out} | \leq s$ \textbf{return} $\found_{out}$ \textbf{else} return $\emptyset$
	\EndIf
	\State $\found_{out} \leftarrow \emptyset$	,  $S \leftarrow \left\{v\right\}$ ,  $\steps \gets 1$ \Comment{Construct a subtree tree of $\tfull_N$ rooted at $v$,}\\
	\Comment{with $S$ as the set of leaves (so $S$ is initialized as $\{v\}$)}
	\Repeat \label{line:inner_loop_start}

		\State $z \leftarrow$ vertex in $S$ with the minimum weight with respect to $S$ \label{line:pick_min_leaf}
		\State $S = S \setminus \{z\}$,  $\steps \leftarrow \steps +1$
		\State $\excl' \gets \excl \cup S$,  $\found' \gets \found + \found_{out}$\label{line:exclpfoundp}
		\State \textbf{if} {$\textsc{ZeroTest}(\found', \excl',z, s)$} \textbf{then continue} \label{line:quadratic_first_simple_zero_test}\\
		\Comment{No heavy leaves in the subtree of $z$, so we remove it}

		\If {$z$ is a leaf in $\tfull_N$}
			\State $\found_{out}(z) \leftarrow \textsc{Estimate}(\found',\excl', z)$ \label{line:quadratic_first_simple_estimate}
			\State \textbf{continue} \label{line:quadratic_estimate_continue}
		\EndIf\\
		\State $s_{desc} \gets \min(\alpha \cdot s, k - |\found'|)$ \label{line:exact_simple_found_bound}
		\Comment{$s_{desc}$ is the budget for the children of $z$}
		\State $z_{\lef},z_{\righ} \leftarrow$ left and right child of $z$ in $\tfull_N$ respectively.\\
		\hspace*{\fill} \linebreak \Comment{Try to estimate the values in the subtrees of the children of $z$ with a smaller budget $s_{desc}$}
		\State $\found_{\lef} \gets \textsc{ExactSparseRecovery}(\found' ,\excl' \cup\{z_{\righ}\},z_{\lef},s_{desc}, k)$ \label{line:exact-recursive-call1}
		\State $\found_\righ \gets \textsc{ExactSparseRecovery}(\found' ,\excl' \cup\{z_{\lef\}},\,z_{\righ},s_{desc}, k)$	    \label{line:exact-recursive-call2}\\
		\hspace*{\fill} \linebreak \Comment{Check if the values were correctly recovered}	
		\State $\mathrm{IsZero}_{\lef} \leftarrow \textsc{ZeroTest}(\found' + \found_\lef,\excl' \cup\{z_{\righ}\},z_{\lef}, s)$		\label{line:exact_zerotest1}
		\State $\mathrm{IsZero}_{\righ} \leftarrow \textsc{ZeroTest}(\found' + \found_\righ, \excl' \cup\{z_{\lef}\},z_{\righ},s)$	\label{line:exact_zerotest2}\\
		\State 
		 \hspace*{\fill} \linebreak\Comment{If the estimates appear correct under current budget, save them}
		 \hspace*{\fill} \linebreak\Comment{Otherwise, we must increase the budget for the child, so we add it to the set $S$}
		
		\State {\bf If~} {$\mathrm{IsZero}_{\lef}$} {\bf~then~} $\found_{out} \leftarrow \found_{out} + \found_\lef${\bf~else~}$S \gets S \cup \{z_\lef \}$
		\State {\bf If~} {$\mathrm{IsZero}_{\righ}$} {\bf~then~}$\found_{out} \leftarrow \found_{out} + \found_\righ$ {\bf~else~}$S \gets S \cup \{z_\righ \}$\\
		
	\Until { $S = \emptyset$,  $\steps > \frac{6 \log N}{\alpha}$ or $|\found_{out}| > s$}~\label{line:until_simple}\label{line:inner_loop_end}\\
	
	\If {$S = \emptyset$, $\steps \leq \frac{6 \log N}{\alpha}$ and $|\found_{out}| \leq s$}
	\State \Return $\found_{out}$ 
	\Else
	\State \Return $\emptyset$
	\EndIf
	\end{algorithmic}
\end{algorithm}

\subsection{Analysis of $\textsc{ExactSparseRecovery}$.}
In this section we shall present analysis of \cref{alg:exact_sparse_simple}. The key idea is that instead of using a large budget that is sufficient for $\ZeroTest$ to succeed every time, we try to explore the subtrees with a lower budget and then check with a larger budget whether the subtrees have been recovered correctly or not.

To do the analysis we will need the correctness guarantee for the base \Cref{alg:cubic-exact}.

Recall that $b$ is the input parameter of $\textsc{SlowExactSparseRecovery}$.
\begin{theorem}[Correctness of \Cref{alg:ksparsefft_simple}]
	\label{thm:kexactrecovery_simple_corr}
	If $|\hleaves(v)| \leq b$ and $v$ is isolated by $\found$ and $\excl$, then the procedure $\textsc{SlowExactSparseRecovery}$ returns the correct estimates for all $\hleaves(v)$.
\end{theorem}

The proof can be found in \Cref{appndx-slow-tree-explor}.

\begin{theorem}[Correctness of \cref{alg:exact_sparse_simple}] \label{thm:exact_sparse_corr}
Consider a call to  primitive \textsc{ExactSparseRecovery}($\found$, $\excl$, $v, s, k$) for any vertex $v \in \tfull_N$, budget $s \leq k$ and sets $\found$ and $\excl$ such that $v$ is isolated by them (see \Cref{def:isolated-found-excl}) and $\found \cap \leaves(v)$ \footnote{The algorithm works even without this requirement, however it somewhat simplifies the proof. }. If the sparsity of subtree rooted at $v$ is less than the allowed budget, that is $|\hleaves(v)| \leq s$, then  \textsc{ExactSparseRecovery} returns a correct estimate for every leaf in $\hleaves(v)$. In particular, if $v = \mathrm{root}$, $\excl = \emptyset$, $\found = \emptyset$ and $|\hleaves(r)| \leq s$, the procedure correctly recovers the entire tree.
\end{theorem}
\begin{proof}
	We will show correctness by induction on the budget $s$. The {\bf base case} is provided by $s \leq 1/\alpha$. In that case the procedure calls \textsc{SlowExactSparseRecovery} (see line~\ref{line:svs1alpha} and line~\ref{line:base-case-slow-call}), and thus the output is correct by Theorem~\ref{thm:kexactrecovery_simple_corr}.
	We now provide the {\bf inductive step}. 

	Now, we show by induction on the number of iterations of the \textbf{repeat} loop in line~\ref{line:inner_loop_start} that for a fixed $s$ the returned frequencies are correct. The set $S$ is the set of leaves of the tree $T$, described earlier in \Cref{sec:almost-quadratic-overview}. We will show that the following invariants hold: 
	\begin{description}
	\item[{\bf (1)}] $\found_{out} \subseteq \hleaves(v)$ and the estimated values in $\found_{out}$ are all correct
	\item[{\bf (2)}] $\hleaves(v) \setminus \found_{out} \subseteq \leaves(S)$\footnote{Recall the definition of $\hleaves$ -- Definition~\ref{def:hleaves-def}.}
	\end{description}
	Then the correctness follows from the fact that $S = \emptyset$ at the end of execution.

	Consider an iteration of the \textbf{repeat} loop where the above invariants hold and let $z$ be the vertex extracted from $S$ in line~\ref{line:pick_min_leaf}. By the inductive hypothesis, $z$ is isolated by $\found'$ and $\excl'$ (see line~\ref{line:exclpfoundp}), and by the theorem assumption, $|\hleaves(v)| \leq s$. Therefore, the calls to $\ZeroTest$ and $\Estimate$ in lines \ref{line:quadratic_first_simple_zero_test} and \ref{line:quadratic_first_simple_estimate} return correct answers. Therefore, in case $\hleaves(z)$ is empty the node $z$ can be removed from $S$, or in case $z$ is a leaf in $\tfull_N$ $\Estimate$ returns a correct estimate for $z$ and it again can be removed from $S$ and the invariants still hold.

	Otherwise, the algorithm recursively calls itself on $z_\lef$ and $z_\righ$. Assume that we virtually add $z_\lef$ and $z_\righ$ to $S$, and remove $z$ from it.
	Then the invariants still hold. We will only discuss the correctness of operations with $z_{\lef}$, since they are symmetric. $z_{\lef}$ is isolated by $\found'$ and $\excl'\cup \{z_\righ\}$, since $z$ is isolated by $\found'$ and $\excl'$. 
	Therefore, if $|\hleaves(z_\lef)| \leq \alpha s$, $\found_\lef$ will contain correct estimates of $\hleaves(z_\lef)$, by the inductive hypothesis. Because the algorithm doesn't know if that is the case, it checks whether the recovered values are correct or not in line~\ref{line:exact_zerotest1} by running $\ZeroTest$. $\mathrm{IsZero}_{\lef}$ would be True only if $\hleaves(z_\lef)$ were correctly recovered. 
	From this fact it follows that if $\hleaves(z_\lef)$ were correctly recovered, we can remove $z_\lef$ from $S$ and update $\found$ by adding $\found_\lef$ to it without violating the invariants, and if not, we can just discard $\found_\lef$.
	Notice that under the theorem's assumption on sparsity and by invariant $1$ it can never happen that $|\found_{out}| > s$.

	Finally, notice that if $|\hleaves(z)| \leq \alpha \cdot s$, the subtree of $z$ will be completely recovered by a recursive call. Therefore, only nodes with $|\hleaves(z)| \geq \alpha \cdot s$ and their children get added to $S$. By an averaging argument, the number of such nodes is at most $\frac{\log N}{ \alpha}$, thus, the maximum number of nodes that would ever be added to $S$ is $3 \frac{\log N}{ \alpha}$. Because at each iteration at least one vertex is removed from $S$, $\steps \leq \frac{6\log N }{\alpha}$ and at the end of the loop $S = \emptyset$.
\end{proof}

We finish this section with the runtime analysis of \cref{alg:exact_sparse_simple}. 
To do it, we will need to use the correctness guarantee for \Cref{alg:cubic-exact}, proof of which can be found in \Cref{appndx-slow-tree-explor}.

\begin{theorem}[Running time of \Cref{alg:ksparsefft_simple}]
	\label{thm:kspase_simple_runtime}
	If  $\leaves(v) \cap \leaves(\excl) = \emptyset$, the runtime of 
	 $\textsc{SlowExactSparseRecovery}$ is bounded by 
		$\wt{O} \left( |\found| \cdot b^2 + 2^{w_{\excl}(v)} \cdot b^3 \right)$.
\end{theorem}

\begin{theorem}[Running time of \textsc{ExactSparseRecovery}] \label{thm:exact_sparse_time}
	Let $\log N > 6 \sqrt{5}$ and $k \leq N$. The running time of $\textsc{ExactSparseRecovery}(\emptyset, \emptyset, \mathrm{root}, k, k)$ is bounded by $\wt{O}(k^2 \cdot 2^{8\sqrt{\log k \log \log N}})$.
\end{theorem}
\begin{proof}
	Recall that $\alpha := 2^{-2\sqrt{ \log k \cdot \log \log N}}$.
First, we make several observations:
\begin{itemize}
	\item Each call to \textsc{ExactSparseRecovery} (Algorithm~\ref{alg:exact_sparse_simple}) returns $\found_{out}$ of size at most $s$.
	\item At most $6\frac{ \log N}{\alpha}$ vertices are inserted in $S$ in a single invocation of \textsc{ExactSparseRecovery}, and there the number of recursive calls to \textsc{ExactSparseRecovery} is bounded by $12 \frac{\log N}{\alpha}$.
	\item Similarly, at all times during an invocation of \textsc{ExactSparseRecovery} one has $|\found_{out}| \leq \alpha s \cdot \frac{12 \log N}{\alpha} = \wt{O}(s)$. Also, because $s_{desc} \leq k - |\found'|$ from line \ref{line:exact_simple_found_bound}, it is guaranteed that $|\found'| \le k+ \widetilde{O}(\alpha s)$.
	\item For each recursive call in lines~\ref{line:exact-recursive-call1} and \ref{line:exact-recursive-call2} as well as the call to \textsc{SlowExactSparseFFT}, we have $\leaves(v) \cap \leaves(\excl) = \emptyset$. This means that the precondition of Theorem~\ref{thm:kspase_simple_runtime} is satisfied and that for each picked $z$, 
	\[
		w_{\excl \cup S}(z) \leq w_\excl(z) + w_S(z) \leq w_\excl(z) + \log(6 \log N / \alpha).
	\]
    where the first inequality follows from \Cref{lem:weight_subadditive}.
\end{itemize}

Recall that bound our basic setup an invocation of $\ZeroTest$ takes $\wt{O}(2^{w_\excl(z)}\cdot b + |\found|\cdot b)$ time,
an invocation of $\Estimate$ takes $\wt{O}(2^{w_\excl(z)} + |\found|)$ time and an invocation of $\textsc{SlowExactSparseFFT}$ takes
$\wt{O}(2^{w_\excl(z)}\cdot b^3 + |\found|\cdot b^2)$ time respectively. We will separately bound the time dependent on $|\found|$ and $2^{w_\excl(v)}$.
Formally, we say that there are two runtime pools, first and second, and the aforementioned primitives spend $\wt{O}(2^{w_\excl(v)}\cdot b)$, $\wt{O}(2^{w_\excl(v)})$ and $\wt{O}(2^{w_\excl(v)}\cdot b^3)$ from the first pool, and $\wt{O}(|\found|\cdot b)$, $\wt{O}(|\found|)$ and $\wt{O}(|\found|\cdot b^2)$ from the second one, respectively. We now bound the sizes of both pools.
Using the notation $A := \frac{6 \log N}{\alpha}$, we have

\paragraph{The first pool.}

Let $T_1[s, W, l]$ be the time the procedure spends from the first pool where for convenience of notation $W = 2^{w_{\excl}(v)}$ and $l$ is the distance to the node $v$ from the root. Notice that the algorithm makes at most $A$ iterations, since at most $A$ vertices are added to $S$. That also means that $|S| \leq A$. Also notice that for each picked $z$, since $z$ is the vertex with the smallest weight in $S$, by \Cref{lem:weight_subadditive} and by \Cref{lem:kraft_averaging}
\[
	2^{w_{\excl \cup S}(z)} \leq 2^{w_{\excl}(z)} \cdot 2^{w_{S}(z)} \leq W \cdot A.
\]
Notice that because we only call \Estimate{} in line~\ref{line:quadratic_first_simple_estimate}, there always a call to \ZeroTest{} in line~\ref{line:quadratic_first_simple_zero_test} that precedes it. Since they are called with the same set of parameters, by Assumption~\ref{assmptn-zerotest-estimate} the time to run \ZeroTest{} dominates that of \Estimate{}. Similarly, it is easy to see that the time to perform all other operations except for recursive calls is also dominated by \ZeroTest{}, so there exists an $f = \wt{O}(1)$ such that,
\begin{itemize}
\item If $l = \log N$, then we are in the case where $v$ is a leaf of $\tfull_N$. The algorithm makes only one iteration of the \textbf{repeat} loop in line~\ref{line:inner_loop_start} where it executes \Estimate{} and stops at line~\Cref{line:quadratic_estimate_continue}. The \ZeroTest{} call takes time $\wt{O}(2^{w_{\excl}(v)}s) = \wt{O}(s W) \leq s W f$, so
	\[T_1[s, W, l] \leq s W f.\] 
\item Else, if $s \leq 1/\alpha$, then we are in the base case where \Cref{alg:cubic-exact} is called in line~\ref{line:base-case-slow-call}. By Theorem~\ref{thm:kspase_simple_runtime} it takes time $\wt{O}(s^3 2^{w_{\excl}(v)}) \leq s^3 W f$ from the first pool to run it, so
	\[T_1[s, W, l] \leq s^3 W f.\]
\item Else the algorithm proceeds with the recursive mode of operation. As was discussed before, it makes at most $A$ iterations of the \textbf{repeat} loop where it runs \ZeroTest{} and calls itself recursively in lines~\ref{line:exact-recursive-call1} and~\ref{line:exact-recursive-call2}. Now \ZeroTest{} is called with set $\excl$ being equal to $\excl \cup S$, so its runtime from the first pool is bounded by $WA s f$. For each recursive call, similarly, the set $\excl$ becomes $\excl \cup S \cup \{z'\}$, where $z'$ is the other child of $z$ (see lines~\ref{line:exact-recursive-call1} and~\ref{line:exact-recursive-call2}), hence the new $W$ is upper bounded by $2 W A$. The distance $l$ also increases by at least $1$. Finally, by observing that $T_1[s, W, l]$ is monotonically non-decreasing with respect to its parameters, we get the following formula
	\[ T_1[s, W, l] \leq A (W A s f + 2 T_1[\alpha \cdot s, 2 W \cdot A, l + 1]). \]
\end{itemize}

We can now show by induction that 
\begin{equation}\label{eq:ind-bound}
	T_1[s, W, l] \leq (5 \alpha A^2)^{\frac{\log s}{\log 1/\alpha}} W \cdot s \cdot f / \alpha^2.
\end{equation}
The base case corresponds to $s \leq 1 / \alpha $ or $l = \log N$ for which we get by the above inequalities that the runtimes respectively are
$s^3 W f \leq s Wf / \alpha^2$ and $s W f$, both of which are not greater than the right hand side of \Cref{eq:ind-bound}.

Suppose now that for $s > 1/\alpha$ and $l < \log N$, the inductive hypothesis holds for smaller values of $s$ or larger values of $l$. Then $\beta \geq 1$ and 
\begin{equation}\label{eq:rec}
	T_1[s, W, l]  \leq A (W A s f + 2 T_1[\alpha s, 2 W \cdot A, l + 1]),
\end{equation}
where we have 
\begin{equation*}
\begin{split}
T_1[\alpha s, 2 W \cdot A, l + 1]&\leq 2(5 \alpha A^2)^{\frac{\log \alpha s}{\log 1/\alpha}} W A \cdot \alpha s\cdot f / \alpha^2\\
&\leq 2(5 \alpha A^2)^{\frac{\log s}{\log 1/\alpha}-1} W A \cdot s \cdot f / \alpha\\
\end{split}
\end{equation*}
by the inductive hypothesis. Substituting this bound into~\eqref{eq:rec}, we get
\begin{align*}
T_1[\alpha s, 2 W \cdot A, l + 1]	&\leq  A (W A s f + 4(5 \alpha A^2)^{\frac{\log s}{\log 1/\alpha}-1} W A \cdot s \cdot f / \alpha)  \\
	& = W s f (A^2 + (5\alpha A^2)^{\frac{\log s}{\log 1/\alpha} - 1} \cdot 4A^2 / \alpha)\\
	&= W s f (\alpha^2 A^2 + (5\alpha A^2)^{\frac{\log s}{\log 1/\alpha} - 1} \cdot 4\alpha A^2) / \alpha^2 \\
	&\leq W s f (\alpha^2 A^2 + (5\alpha A^2)^{\frac{\log s}{\log 1/\alpha} - 1} \cdot 4\alpha A^2) / \alpha^2 \\
	&\leq W s f (\alpha^2 A^2 + (5\alpha A^2)^{\frac{\log s}{\log 1/\alpha} - 1} \cdot 4\alpha A^2) / \alpha^2 \\	
	& \leq (5\alpha A^2)^{\beta} Wsf / \alpha^2,
\end{align*}
where in the last transition we used the fact that $(5\alpha A^2)^{\frac{\log s}{\log 1/\alpha} - 1}= (5(6\log N)^2/\alpha)^{\frac{\log s}{\log 1/\alpha} - 1}\geq 5^{\frac{\log s}{\log 1/\alpha} - 1}\geq 1$ (the latter bound holds since $\frac{\log s}{\log 1/\alpha} - 1\geq 0$, as $s\geq 1/\alpha$ in the inductive step). This completes the inductive step, establishing~\eqref{eq:ind-bound}. Substituting the values for $\alpha, W, f$ and $A$ and simplifying, we get
\begin{equation}\label{eq:t1-bound-root}
\begin{split}
	T_1[k, 1, 0] & \leq (5 \alpha A^2)^{\frac{\log k}{\log 1/\alpha}} W \cdot k \cdot f / \alpha^2\\
	&\leq (5 \alpha A^2)^{\frac{\log k}{\log 1/\alpha}} 2^{w_{\excl}(v)} \cdot k \cdot f / \alpha^2\\
	&=\wt{O}( (5 \alpha A^2)^{\frac{\log k}{\log 1/\alpha}} k / \alpha^2),
\end{split}
\end{equation}
where we used the fact that $2^{w_{\excl}(v)}=1$ when $v$ is the root of $\tfull_N$.

It remains to upper bound $(5\alpha A^2)^{\frac{\log k}{\log 1/\alpha}} = (5(6 \log N)^2/\alpha)^{\frac{\log k}{\log 1/\alpha}}$. We  bound the logarithm of this value:
\begin{equation*}
\begin{split}
	\frac{\log k}{\log 1/\alpha} \log (5(6 \log N)^2/ \alpha) & \leq \frac{\log k}{\log 1/\alpha} (\log ( \log^4 N) + \log (1/\alpha))\\
	& \leq 4 \frac{\log k \log \log N}{\log 1/\alpha} + \log k \\
	&\leq 2 \sqrt{\log k \log \log N} + \log k.
\end{split}
\end{equation*}
Substituting this into~\eqref{eq:t1-bound-root} and recalling that $\alpha = 2^{-2\sqrt{ \log k \cdot \log \log N}}$, we get
$$
T_1[k, 1, 0]=\wt{O}(k^2\cdot 2^{8\sqrt{\log k \log \log N}})
$$
as required.

\paragraph{The second pool.} We bound $|\found'|$ by $\wt{O}(k)$. Let $T_2[s, l]$ denote the upper bound on the runtime from the second pool, where $l$ is the distance from $v$ to root. Again, the runtime is dominated by the call to $\ZeroTest$, so for some $f = \wt{O}(1)$ the following relations hold:
\begin{itemize}
	\item If $l = \log N$, then
		$T_2[s, l] \leq k f$.
	\item Else, if $s \leq 1 / \alpha$, then
	$T_2[s, l] \leq k s^2 f \leq k f / \alpha^2$.
	\item Else,
		$T_2[s, l] \leq A (k s f + 2 T_2[\alpha s, l + 1])$.
\end{itemize}
Similarly to the first pool, one can show by induction that
	$T_2[s, l] \leq (3 \alpha A)^\beta k s f / \alpha^2$,
where $\beta = \lfloor \frac{\log s}{\log 1/\alpha} \rfloor$. Using the fact that 
	$(3 \alpha A)^\beta k s f \leq 2^{2 \sqrt{\log k \log \log n}}$
we have $T_2[k, 0] = \wt{O}(k^2 \cdot 2^{8\sqrt{\log k \log \log n}})$.
Summing up runtimes of both pools yields total runtime of $\wt{O}(k^2 \cdot 2^{8\sqrt{\log k \log \log n}})$.

Finally, the time spent on maintaining $\excl'$ is negligible, since, similarly to \textsc{SlowExactSparseRecovery}, it can be constructed once at the beginning of the algorithm, and on each iteration we modify it by adding and removing a constant number of vertices to or from it. Hence, by the same proof as in Theorem~\ref{thm:kspase_simple_runtime} the used time is $\wt{O}(s + w_{\excl}(v))$.

\end{proof}

\newpage
\section{Preliminaries and Notations.} \label{sec:prelim}

\paragraph{Fourier transform basics.}
We will often identify $[n]^d \to \C$ with $\C^{n^d}$ for convenience and use the two interchangeably depending on the context. 

\begin{definition}[Fourier transform] \label{def:ft}
For any positive integers $d$ and $n$, the \emph{Fourier transform} of a signal $x \in \C^{n^d}$ is denoted by $\wh x$, where 
$\wh x_\ff = \sum_{\tt \in [n]^d} x_\tt e^{-2\pi i \frac{\ff^\top \tt}{n}}$ for any $\ff \in [n]^d$. Here $\ff^\top \tv = \sum_{q=0}^{d-1} f_q t_q$.
\end{definition}
Recall that by Parseval's theorem we have $\|\wh x\|_2^2 = n^d \cdot \| x\|_2^2$. Furthermore, recall the convolution-multiplication duality $\wh{(x \star y )} = \wh x \cdot \wh y$,
where $x \star y \in C^{n^d}$ is the convolution of $x$ and $y$ and defined by the formula
$(x \star y )_{\tv} = \sum_{\bm{\tau} \in [n]^d} x_{\bm{\tau}}\cdot y_{(\tv - \bm\tau \mod n)}$ for all $\tv \in [n]^d$, where the modulus is taken coordinate-wise. We will also need the following well-known theorem on the Fourier subsampled matrices.

\begin{theorem}(Restricted Isometry Property of Subsampled Fourier Matrices,~\cite[Theorem 3.7]{haviv2017restricted})
\label{RIP-thrm}
Let $q = \Theta(s \log ^3 N)$. Then with high probability in $N$, the time domain points $\left\{x_{\bm{t}} \right\}_{\bm{t} \in Q}$ for a random multiset $Q \subseteq [n]^d$ with $q$ uniform samples are sufficient to $(1\pm\epsilon)$-approximate the energy of all $s$-sparse vectors $\wh{x}$, where $\epsilon>0$ is some absolute constant. Formally, simultaneously for all $s$-sparse vectors: $\frac{N^2}{q}\sum_{\bm{t} \in Q} |x_{\bm{t}}|^2 \in \left[ (1-\epsilon)\|\wh{x}\|_2^2, (1+\epsilon)\|\wh{x}\|_2^2\right]$.
\end{theorem}

\subsection{Notation for Manipulating FFT Computation Trees.}
\label{sec:manipulate_FFT}

Recall that given a signal $x: [n]^d \rightarrow \mathbb{C}$, the execution of the FFT algorithm produces a binary tree, referred to as $\tfull_N$. The root of $\tfull_N$ corresponds to the universe $[n]^d$, while the children of the root correspond to $\left[n/2 \right]\times [n]^{d-1}$; note that FFT recurses by peeling off the least significant bit. Every node $v$ has a \emph{label} ${\bm f}_v \in \ZZ_{n}^d$ associated to it, defined according to the following rules.

\begin{enumerate}
\item The root has label ${\bm f}_{\text{root}} = (\underbrace{0,0,\ldots,0}_{d~\mathrm{entries}})$, and corresponds to the universe $[n]^d$.
\vspace{-8pt}
\item The children $v_\lef,v_\righ$ of a node $v$ which corresponds to the universe $[n/2^l]\times [n]^{d'}$, with $0 \leq d' \leq d-1, 0\leq l \leq \log n-1$, have the following properties. Both correspond to universe $[n/2^{l+1}] \times [n]^{d'}$, and $v_\righ$ has label $\bm{f}_{v_\righ} = \bm{f}_v $, while $v_\lef$ has label $\bm{f}_{v_\lef} = \bm{f}_v + (\underbrace{0,0, \ldots ,0}_{d'}, 2^{l} ,\underbrace{0,0,\ldots,0}_{d-d'-1})$.
\vspace{-8pt}
\item The children of a node $v$ corresponding to universe $[1]\times[n]^{d'}$ with $d'>0$, are $v_\lef, v_\righ$, corresponding to universe $[n/2]\times [n]^{d'-1}$ and have labels $\bm{f}_{v_\righ} = \bm{f}_v$ and $\bm{f}_{v_\lef} = \bm{f}_v + (\underbrace{0,0,\ldots ,0}_{d' -1 },1,\underbrace{0,0,\ldots,0}_{d-d'})$ respectively.
\vspace{-8pt}
\item A node $v$ corresponding to the universe $[1]$ is called a \emph{leaf} in $\tfull_N$.
\end{enumerate}
The above rules create a binary tree of depth $\log N$, which corresponds to the FFT computation tree. The labels of the leaves of $\tfull_N$ represent the set $[n]^d$ of all possible frequencies of any signal $x: [n]^d \rightarrow \mathbb{C}$ in the Fourier domain. We demonstrate $\tfull_N$ that corresponds to the $2$-dimensional FFT computation on universe $[4] \times [4]$ in Figure~\ref{Tfull-labeling}. 
Subtrees $T$ of $\tfull_N$ can be defined as usual. For every node $v \in T$, the \emph{level} of $v$, denoted by $l_T(v)$, is the distance from the root to $v$. 
We denote by $\leaves(T)$ the set of all leaves of tree $T$, and for every $v\in \leaves(T)$, its \emph{weight} $w_T(v)$ \emph{with respect to} $T$ is the number of ancestors of $v$ in tree $T$ with two children. The levels (distances from the root) on which the aforementioned ancestors lie will be called $\Anc(v,T)$. Furthermore, the sub-path of $v$ with respect to $T$ will be the children of the aforementioned ancestors which are not ancestors of $v$. Additionally, for a node $v \in T$ we denote the subtree of $T$ rooted at $v$ by $T_v$.

The following definition will be particularly important for our algorithms.
\begin{definition}[Frequency cone of a leaf of $T$]\label{def:iso-t-highdim}
For every subtree $T$ of $\tfull_N$ and every node $v\in T$, we define the {\em frequency cone of $v$} with respect to $T$ as,
\[	\subtree_T(v):=\left\{ \ff_u : \text{ for every leaf } u \text{ in subtree of }\tfull_N \text{ rooted at }v \right\}.\]
Furthermore, we define $\supp{(T)} : = \bigcup_{u\in \leaves(T)} \subtree_T(u)$.
\end{definition}

The \emph{splitting tree} of a set $S \subseteq [n]^d$ is the subtree of $\tfull_N$ that contains all nodes $v \in \tfull_N$ such that $S \cap \subtree_{\tfull_N}(v) \neq \emptyset$.

\section{Techniques and Comparison with the Previous Technology.}
\label{sec:techniques}

This section is devoted to highlighting the differences between previous work and our technical contributions.

\subsection{Previous Techniques.}

Most previous sublinear-time Sparse Fourier transform algorithms~\cite{GMS,hikp12a,k16,k17} rely on emulating the hashing of signal $\wh{x}$ by picking a structured set of samples (in low dimensions, the samples correspond to arithmetic progressions) and processing them with the help of bandpass filters, i.e. functions which approximate the $\ell_\infty$ box in frequency domain and are simultaneously sparse in time domain. However, while those filters are particularly efficient in low dimensions, their performance deteriorates when the number of dimensions increases: indeed, a $d$-dimensional $\ell_\infty$ box has $2^{d}$ faces, and hence this approach suffers inevitably from the curse of dimensionality. On the other hand, an unstructured collection of $O(k \cdot \poly(\log N))$ samples~\cite{CTao,NSW19} suffice, showing that the sample complexity is dimension-independent; the cost that one needs to pay, however, is $\Omega(N)$ running time.

To (partially) remedy the aforementioned state of affairs, the approach of~\cite{kapralov2019dimension} departs from both the aforementioned approaches, and performs pruning in the Cooley-Tukey FFT computation graph, in a way that suffices for recovery of \emph{exactly} $k$-sparse vectors. Recall that as we explained in \cref{sec:tree_problem}, the exact Sparse FFT problem can be translated to a tree exploration problem.
What makes the exploration possible and is the main technical innovation of \cite{kapralov2019dimension} is the introduction of adaptive aliasing filters, a new class of filters that allow to isolate a given frequency from a given set of $k$ other frequencies using $O(k)$ samples in time domain and in $O(k \log N)$ time. Those filters are revised in Section~\ref{sec:filters_old}.

\begin{definition}[$(v, T)$-isolating filter, see Definition~\ref{def:v-t-isolating}]
	Consider a subtree $T$ of $\tfull_N$, and a leaf $v$ of $T$. A filter $G:[n]^d \to \C$ is called \emph{$(v,T)$-isolating} if the following conditions hold:
	\begin{itemize}
		\item For all $\bm f\in \subtree_T(v)$, we have $\wh{G}(\bm f)=1$.
		\item For every $\bm f'\in \bigcup_{\substack{ u \in \leaves(T) \\u \neq v }}\subtree_T(u) $, we have $\wh{G}_{v}(\bm f')=0$.
	\end{itemize}
\end{definition}

As shown in~\cite{kapralov2019dimension}, for a given tree $T$ and a node $v$ one can construct isolating filters $G$ such that $\|G\|_0 = O(2^{w_T(v)})$, and $\wh{G}(\bm f)$ is computable in $\widetilde{O}(1)$ time (see also Lemma~\ref{lem:isolate-filter-highdim}). The sparsity of $G$ in time domain, i.e. $\|G\|_0$, corresponds to the number of accesses to $x$ needed in order to get our hands on $(\wh{G} \cdot \wh{x})_{\bm f}$ for a fixed $\bm f$.

As was shown in \cref{sec:tree_problem}, the FFT tree exploration proceeds using two primitives that satisfy Assumption~\ref{assmptn-zerotest-estimate} and both of these primitives can be efficiently constructed given the above filters. 
The first one is a primitive for performing a \emph{zero test} on a subtree, i.e., checking whether $\wh{x}_{\subtree(v)} \equiv 0$. This check can be performed efficiently using a (deterministic) collection of $O(k \log^3 N)$ samples which satisfy the Restricted Isometry Property (RIP) of order $k$; its pseudocode, named $\textsc{ZeroTest}$, is depicted in Algorithm~\ref{alg:ZeroTest}. The sample complexity of $\textsc{ZeroTest}$ is then 
\[	O(2^{w_T(v)} \cdot k \cdot \poly(\log N)),	\]
namely, one needs to multiply the time domain support size of the isolating filter $G$ with the number of samples needed to satisfy \textsc{RIP} of order $k$. 
The second primitive is used when $\ell$ is a leaf in $\tfull_N$, i.e. a node at depth $\log N$, in which case the algorithm needs to \emph{estimate} $\wh{x}_{\bm{f}_{\ell}}$ using the $(\ell,T)$-isolating filter, see Algorithm~\ref{alg:estimate_freq} for a pseudocode. This requires only $O\left( 2^{w_T(\ell)} \right)$ samples.

Unfortunately, as we have already pointed out, the algorithm in~\cite{kapralov2019dimension} works only for exactly $k$-sparse signals, and also demands cubic time and sample complexity. Our new toolkit shows that all three limitations can be remedied (though not completely simultaneously).

We also mention that a modified version of \cite{Mansour95} can be employed to recover exactly $k$-sparse signals in $\widetilde{O}(k^3)$ time. The algorithm presented in \cite{Mansour95} performs breadth-first search in the Cooley-Tukey FFT computation graph, rather than exploring by picking the lowest weight leaf. Opposed to~\cite{kapralov2019dimension}, the algorithm in~\cite{Mansour95} uses Dirac comb filters to learn all the non-empty frequency cones in the same level at once. However, the techniques in that paper cannot go beyond cubic time for $k$-sparse signals, and as can be seen in \cite[Section~6]{Mansour95}, extending the result to robust signals pays a multiplicative signal-to-noise ratio factor on top of $k^3$.

\subsection{Our Techniques.}

Our first technique is a way to traverse the Cooley-Tukey FFT computation graph in almost quadratic time complexity. This was presented in detail in \cref{sec:tree_problem}. Here we give a quick summary of our FFT tree exploration.
\paragraph{FFT backtracking.} The first crucial observation is that the vanilla FFT traversal algorithm given in \cite{kapralov2019dimension} performs a \emph{zero test} with \textsc{RIP} of order $k$ to decide whether a subtree contains a non-zero frequency, and this might be unnecessary. Indeed, if we are at a node $v$ for which $\|\wh{x}_{\subtree(v)}\|_0 = O(1)$, i.e. there are at most $O(1)$ elements in $\subtree(v)$, we only need to perform \textsc{RIP} of order $O(1)$. Thus, maybe there is a way to approximately learn $\|\wh{x}_{\subtree(v)}\|_0$, for nodes $v$ explored during the execution of the algorithm, and perform a \emph{low-budget} zero test accordingly?

We have demonstrated that this intuition is correct in \cref{sec:tree_problem}. The idea is to assign varying budgets to the nodes to be explored and then perform the hierarchical error detection in order to detect errors in the exploration which are caused by the failures of $\ZeroTest$ due to incorrect budget assignments. 
The algorithm maintains at all times a subtree $T$, as well as a vector $\wh{\chi}$, such that $\supp(\wh{x} - \wh{\chi})\subseteq \cup_{u \in T} \subtree(u)$, and $\supp(\wh{\chi}) \subseteq \supp(\wh{x})$. 
The algorithm explores the tree by considering values $b_1,b_2\ldots$, corresponding to the possible assumptions on the sparsity of $\wh{x}_{\subtree(v)}$, for some node $v$ picked during the execution of the algorithm. For a parameter $\alpha<1$ we use thresholds $b_0 := k, b_1 := \alpha k, b_2 := \alpha^2 k, \ldots, b_{ \frac{\log k} {\log (1/\alpha)} } = O(1)$. 
Our algorithm recursively explores various subtrees $T_v$ with some budget $b:=b_j$, i.e. under the assumption $\|\wh{x}_{\subtree(v)}\|_0 \leq b$. The algorithm maintains a subtree $T_v$, initialized at $\{v\}$ and proceeds by picking the minimum weight node $z \in T_v$ and considering the two children of $z$, let them be $z_\lef,z_\righ$.
Then, it runs itself recursively on $T_{z_\lef},T_{z_\righ}$ with budget $b_{j+1} = \alpha b$. When the recursive calls return, yielding candidate vectors $\wh{\chi}_\lef,\wh{\chi}_\righ$, it performs a zero test on each of $z_\lef,z_\righ$ with \textsc{RIP} of order $b$, in order to check whether $\wh{x}_{\subtree(z_\lef)} - \wh{\chi}_\lef$ is the all zeros vector (similarly for the right child). If the zero test on $z_\lef$ is $\false$, we add $z_\lef$ to $T_v$; similarly for $z_\righ$. If both zero tests are $\true$, then we remove $z$. This continues either until $T_v = \emptyset$ or until the number of nodes that have ever been inserted in $T_v$ becomes too large (in particular if there is $\Omega(b / \alpha)$ leaves). In the first case, the algorithm returns the found vector, otherwise it returns the all zeros vector, since insertion of too many nodes into $T_v$ means that we have underestimated the sparsity of $\wh{x}_{\subtree(v)}$, as we argued in \cref{sec:tree_problem}. 

Upon performing a call with arguments a node $v$ and a budget $b$, it could be the case that $\|\wh{x}_{\subtree(v)}\|_0 \leq b$ does not hold; however, this misassumption is not detected by that call, and a vector which is not equal to $\wh{x}_{\subtree(v)}$ is returned to the above recursion level. Nevertheless, although undetectable at the time, this discrepancy will be detected in some recursion level above, where we make use of higher budget; definitely at the very first level where we perform \textsc{RIP} of order $k$. 
We proved the correctness of the above process in \cref{sec:tree_problem} using induction on the tree.

\paragraph{Robust Algorithm.} Our tree exploration technique works well for solving the exact Sparse FFT problem. For designing a robust algorithm we need a collection of new techniques in addition to the FFT backtracking. In what follows we explain the techniques needed for \emph{robustifying} our Sparse FFT algorithm.

First of all, in the robust case we should substitute $\textsc{ZeroTest}$ with an analogous $\textsc{HeavyTest}$ routine. The role of this routine is to determine whether $\|(\wh{x} - \wh{\chi})_{\subtree(v)}\|_2 \geq \|\wh{\eta}\|_2$, where $v$ is any node that appears during the execution of the algorithm. If the latter inequality holds, this means that there are elements of the head of $\wh{x}$ inside $\subtree(v)$ that are yet to be recovered. Pseudocode for this routine is presented in Algorithm~\ref{alg:zero-test}, and the guarantees of this routine are spelled out in Lemma~\ref{lem:guarantee_heavy_test}. The algorithm is very similar to $\textsc{ZeroTest}$, with the difference that we now need to take a collection of random samples, since a deterministic collection of samples sastisfying \textsc{RIP} does not suffice to control the non-sparse component, i.e. the contribution of the tail under filtering. Furthermore, what is demanded is a control on how a $(v,T)$-isolating filter $\wh{G}$ acts on $\wh{x}_{ \cup_{u \in T\setminus\{v\}}\subtree(u) }$, i.e. on parts of the signal living inside frequency cones which $u$ is \emph{not} isolated from. In words, one would like to appropriately control the energy of $\left(\wh{G} \cdot \wh{x}_{ \cup_{u \notin T}\subtree(u) }\right)$, where $\cdot$ corresponds to element-wise vector multiplication.
\newline

\paragraph{Collectively, adaptive aliasing filters act as near-isometries.}

Adaptive aliasing filters are particularly effective for \emph{non-obliviously} isolating elements of the head with respect to each other. However, in standard sparse recovery tasks, one desires control of the tail energy that participates in the measurement. This is a relatively easy (or at least well-understood) task in Sparse Fourier schemes which operate via $\ell_\infty$-box filters~\cite{hikp12a,hikp12b,ikp14,IK,k17}, but a non-trivial task using adaptive aliasing filters. The reason is that the tail via the latter filtering is hashed in a \emph{non-uniform} way. The hashing depends on the arithmetic structure of the elements used to construct the filters, as well as their arithmetic relationship with the elements in the tail. This non-uniformity is essentially the main driving reason for the ``exactly $k$-sparse'' assumption in~\cite{kapralov2019dimension}. Our starting point is the observation that for every tree $T \subseteq \tfull_N$, the $(v,T)$-isolating filters for $v \in \leaves(T)$, satisfy the following orthonormality condition in dimension one, see subsection~\ref{sec:filters_onedim}.
\begin{lemma}(Gram Matrix of adaptive alliasing filters in $d=1$)
Let $T \subseteq \tfull_n$, let $G_v$ be the $(v,T)$-isolating filter of leaf $v\in \leaves(T)$, as per \eqref{eq:isolat-filter-1d}. Let $v$ and $v'$ be two distinct leaves of $T$. 
Then, 
\begin{enumerate}
\item \[	\|\wh{G}_v \|_2^2 := \sum_{ \xi \in [n]} |\wh{G}_v(\xi) |^2 = \frac{n}{2^{w_T(v)}}.	\]

\item (cross terms) the adaptive aliasing filters corresponding to $v$ and $v'$ are orthogonal, i.e. 
\[	\langle \wh G_v, \wh G_{v'} \rangle :=  \sum_{ \xi \in [n]} \wh G_v(\xi) \cdot \overline{\wh G_{v'}(\xi)} = 0.	\]
\end{enumerate}

\end{lemma}

This already postulates that adaptive aliasing filters are relatively well-behaved: for a tree $T$ all leaves of which have roughly the same weight, it must be the case that $x\mapsto\quad \{ \langle \wh G_v,\wh x \rangle \}_{v \in \leaves(T)}$ is a near-orthonormal transformation. Of course, this is too much to ask in general. The crucial property that we will make use of is captured in the following Lemma, see Subsection~\ref{sec:filters-d}.

\begin{lemma}(see Lemma~\ref{filter-robust-multidim})
	Consider a tree $T \subseteq T_N^{full}$. For every leaf $v$ of $T$ we let $\wh{G}_v$ be a Fourier domain $(v,T)$-isolating filter. Then for every ${\bm\xi} \in [n]^d$,
\[ \sum_{ v\in \leaves(T)} |\wh{G}_v({\bm\xi})|^2 = 1.\]
\end{lemma}

Using standard arguments, the above gives the following Lemma.

\begin{lemma}
For $z: [n]^d \rightarrow \mathbb{C} $, let $z^{\rightarrow \bm a}$ be the cyclic shift of $z$ by $a$, i.e. $z^{\rightarrow {\bm a}}(\bm f) := z(\bm f - \bm a )$, where the subtraction happens modulo $n$ in every coordinate. For a tree $T \subseteq \tfull_N$, 

\[	\mathbb{E}_{\bm a \sim U_{[n]^d}} \left[\sum_{ v \in \leaves(T)}  |\langle \wh{G_{v}}, \wh{z^{\rightarrow {\bm a}}} \rangle|^2 \right] = \|z\|_2^2,	\]

i.e. on expectation over a random shift the total collection of filters is an isometry.

\end{lemma}

Thus, although the tail is hashed in a way that is dependent on the head of the signal, what we can prove is that in expectation over a random shift the total amount of noise is controllable. Using the last property we can ensure that $\textsc{HeavyTest}$ in the high-SNR regime we consider i) does not introduce \emph{false~positives}, i.e. does not engage in exploration in subtrees that contain no sufficient amount of energy, and ii) prevents false negatives. Guarantee i) translates to a bound on the running time of the algorithm, while ii) ensures correct execution of the algorithm. Note that due to the explorative nature of algorithm and the fact that missing a heavy element \emph{increases the total noise in the system} (since we stop isolating with respect to it afterwards, it contributes as noise in subsequent measurements), accumulation of false negatives can totally destroy the guarantees of our approach. We note that this phenomenon of the tail not hashed independently of the signal occurs also in one-dimensional continuous Sparse Fourier Transform~\cite{ps15}, although for a very different reason; in their setting handling such an irregularity is significantly easier, mostly due to the fact that errors do not accumulate as in our explorative algorithm.

\paragraph{Identification and estimation are interleaved.} In contrast to more standard sparse recovery tasks where usually identification and estimation can be decoupled, our algorithm needs to have a precise way to perform estimation upon identification of a coordinate. That happens due to the explorative nature of our algorithm, which does not allow us to perform estimation at the very end. This is relatively easy in the exactly $k$-sparse case, but in the robust case, due to the presence of noise it is much more challenging. Whenever we identify a frequency and isolate it from the other head elements, we can pick $\tilde{O}(k)$ random samples and estimate it up to $1/\sqrt{k}$ fraction of the tail energy. Although this precision is sufficient for our algorithm to go through, it would lead us to an undesirable \emph{cubic} sample complexity in total. The next two techniques are introduced in order to handle this situation.

\paragraph{Lazy Estimation.}

One additional crucial difference between the exactly $k$-sparse case and the robust case is estimation. In the former, when we had a tree $T$ and the minimim-weight leaf $v \in T$ was also a leaf in $\tfull_N$, we needed $\widetilde O(2^{w_T(v)})$ samples in order to perfectly estimate $\wh{x}_{\bm f_v}$. However, in the robust case, perfect estimation is impossible, and as is usual in sparse recovery tasks, we should estimate it up to additive error $O\left(\frac{1}{\sqrt{k}}\|\wh{\eta}\|_2\right)$ (recall that we write $\wh{x} = \wh w + \wh \eta$, where $\eta$ is the tail of the signal). One way to achieve this type of guarantee is to take $\widetilde O(k)$ random samples from $G_v \star x$, where $G_v$ is the $(v,T)$-isolating filter. This would yield $\widetilde O( k \cdot 2^{w_T(v)})$ samples for estimation, a $k$ factor worse than what is needed in the exactly $k$-sparse case. In total, the sample complexity (and running time) would be $k$ times more expensive, getting us back to $\widetilde O(k^3)$.

Let's see how it is possible to shave the aforementioned multiplicative $k$ factor in the sample complexity. Imagine that upon finding such a leaf $v$, our algorithm \emph{does not} estimate it immediately, but rather decides to postpone estimation for later. Instead, it marks it as a fully identified frequency, without removing it from $T$ and proceeds in exploring $T$ further. From now on, instead of picking the lowest weight leaf in $T$ at any time, it picks the lowest weight \emph{unmarked} leaf in $T$. Of course, it could be the case that this rule causes the leaf picked to have weight much more than $\log k$, significantly increasing the cost of filtering. Consider however the following strategy. While the minimum weight unmarked leaf in $T$ has weight at most $\log k + 2$, we pick and it and continue exploring. Whenever the aforementioned condition does not hold, the total Kraft mass\footnote{For a tree $T$ and a set $S\subseteq \leaves(T)$ we shall refer to the quantity $\sum_{v \in S}2^{-w_T(v)}$ as the Kraft mass occupied by $S$ in $T$, or just the Kraft mass of $S$ if it is clear from context.} occupied by the \emph{marked} leaves in $T$ is at least $1-k \cdot \frac{1}{2k}  = \frac 1 2 $. When this happens, we show that we can extract a large subset of the marked nodes, see Lemma~\ref{lemma:kraft-ISCheap}, which can be well-estimated \emph{on average} using only a polylogarithmic number of samples. This suffices for the $\ell_2/\ell_2$ guarantee, and furthermore reduces the number of marked nodes (and hence the Kraft mass occupied by marked nodes) causing our algorithm to proceed without increasing the cost of filtering. A more involved demonstration of this idea appears in section~\ref{sec:robust_first}.

\paragraph{Multi-scale Estimation.} The lazy estimation technique presented above can estimate $k$ heavy frequencies of $\wh x$ up to average additive error of $O\left(\frac{\|\wh \eta\|_2}{\sqrt{k}}\right)$ using quadratic samples only if we use the vanilla tree exploration strategy which always picks the lowest weight unmarked leaf of tree $T$ and explores its children. This exploration strategy ensures that leaves get identified and consequentky \emph{marked} in ascending weight order. Thus, there will be a point where the Kraft mass occupied by marked leaves is sufficiently large (recall that marked leaves have weight bounded by $\log k+2$). 
However, as we already mentioned, the tree exploration employed in \cite{kapralov2019dimension} results in cubic sample complexity even in the exactly $k$-sparse case. 
On the other hand, our new exploration strategy (FFT backtracking) does not necessarily guarantee that the identified leaves will have large Kraft mass and bounded weight at the same time. 

To make both lazy estimation and backtracking tree exploration techniques work together and achieve near quadratic total sample complexity, we devise a multi-scale estimation scheme. Our estimation strategy is to estimate every heavy frequency not once, but multiple times, each time to a different accuracy. More precisely, let's assume we are exploring a node $v \in T$ under the assumption that $\|\wh x_{\subtree(v)}\|_0 \le b$, and this assumption is correct. For every found frequency $\bm f$, we estimate $\wh{x}(\bm f)$, to precision $\frac{\|\wh \eta\|_2}{\sqrt{b}}$ instead of $\frac{\|\wh \eta\|_2}{\sqrt{k}}$, which would be the standard thing to do. However, sticking to this error precision will not give the desired $\ell_2/\ell_2$ guarantee: for small $b$, it blows up the error by a factor of $\sqrt{\frac{k}{b}}$, and it could be that all $\bm f \in \supp(\wh x)$ are estimated in a low-budget subproblem, due to recursion. Nevertheless, we can use these coarse-grained estimates to \emph{only locate} the support of $\wh{x}$ inside a subtree, and return it to the parent subproblem, i.e. to the above recursion level. The parent subproblem will mark those recovered frequencies, ignore their values, and continue its execution normally (pick the lowest leaf, perform lazy estimation etc). At some point, when the Kraft mass occupied by the parent subproblem is large enough, those frequencies will be estimated up to \emph{higher} precision, i.e. $\frac{\|\wh \eta\|_2}{\sqrt{b/\alpha}}$. When it finishes execution, it will return those elements to the above recursion level, so on so forth. This type of argumentation can be used to glue together lazy estimation and FFT backtracking. An illustration of this idea takes place in Section~\ref{sec:robust_sec}.

\begin{figure*}[t!]
	\centering
	\scalebox{.65}{
		\begin{tikzpicture}[level/.style={sibling distance=110mm/#1,level distance = 2cm}]
			\node [arn] (z){}
			child {node [arn] (a){}edge from parent [electron]
				child {node [arn] (b){}
					child [sibling distance=30mm]{node [arn] (c){}
						child[sibling distance=15mm]{node [arn] (cl){}}
						child[sibling distance=15mm]{node [arn] (cr){}}
					}
					child [sibling distance=30mm]{node [arn] (d){}
						child[sibling distance=15mm]{node [arn] (dl){}}
						child[sibling distance=15mm]{node [arn] (dr){}}
					}
				}
				child {node [arn] (e){}
					child [sibling distance=30mm]{node [arn] (f) {}
						child[sibling distance=15mm]{node [arn] (fl) {}}
						child[sibling distance=15mm]{node [arn] (fr) {}}
					}
					child [sibling distance=30mm]{node [arn] (g){}
						child[sibling distance=15mm]{node [arn] (gl) {}}
						child[sibling distance=15mm]{node [arn] (gr) {}}
					}
				}
			}
			child { node [arn] (h){}edge from parent [electron]
				child {node [arn] (i){}
					child [sibling distance=30mm]{node [arn] (j){}
						child[sibling distance=15mm]{node [arn] (il){}}
						child[sibling distance=15mm]{node [arn] (ir){}}
					}
					child [sibling distance=30mm]{node [arn] (k) {}
						child[sibling distance=15mm]{node [arn] (kl){}}
						child[sibling distance=15mm]{node [arn] (kr){}}
					}
				}
				child {node [arn] (l){}
					child [sibling distance=30mm]{node [arn] (m){}
						child[sibling distance=15mm]{node [arn] (ml){}}
						child[sibling distance=15mm]{node [arn] (mr){}}
					}
					child [sibling distance=30mm]{node [arn] (n){} 
						child[sibling distance=15mm]{node [arn] (nl){}}
						child[sibling distance=15mm]{node [arn] (nr){}}
					}
				}
			};
			
			\node []	at (z.north)	[label=right:\Large{$(0,0) \qquad\qquad\qquad\qquad\qquad\qquad\qquad\qquad~~ \to {\text{universe }\bf [4] \times [4]}$}]	{};
			
			\node []	at (a)	[label=left:\Large{$(0,1)$}]	{};
			\node []	at (b)	[label=left:\Large{$(0,3)$}]	{};
			\node []	at (c)	[label=left:\Large{$(1,3)$}] {};
			\node []	at (cl)	[label=below:\Large{$(3,3)$}] {};
			\node []	at (cr)	[label=below:\Large{$(1,3)$}] {};
			\node []	at (d)	[label=left:\Large{$(0,3)$}] {};
			\node []	at (dl)	[label=below:\Large{$(2,3)$}] {};
			\node []	at (dr.west)	[label=below:\Large{$(0,3)$}] {};
			\node []	at (e)	[label=right:\Large{$(0,1)$}] {};
			\node []	at (f)	[label=left:\Large{$(1,1)$}] {}; 
			\node []	at (fl.east)	[label=below:\Large{$(3,1)$}] {}; 
			\node []	at (fr)	[label=below:\Large{$(1,1)$}] {}; 
			\node []	at (g)	[label=left:\Large{$(0,1)$}] {}; 
			\node []	at (gl)	[label=below:\Large{$(2,1)$}] {}; 
			\node []	at (gr.west)	[label=below:\Large{$(0,1)$}] {}; 
			\node []	at (h)	[label=right:\Large{$(0,0)\qquad\qquad~~~~~ \to {\text{universe }\bf[2] \times [4]}$}] {}; 
			\node []	at (i)	[label=left:\Large{$(0,2)$}] {}; 
			\node []	at (j)	[label=left:\Large{$(1,2)$}] {}; 
			\node []	at (il.east)	[label=below:\Large{$(3,2)$}] {}; 
			\node []	at (ir)	[label=below:\Large{$(1,2)$}] {}; 
			\node []	at (k)	[label=left:\Large{$(0,2)$}] {}; 
			\node []	at (kl)	[label=below:\Large{$(2,2)$}] {}; 
			\node []	at (kr.west)	[label=below:\Large{$(0,2)$}] {}; 		
			\node []	at (l)	[label=right:\Large{$(0,0) \to {\text{universe }\bf [1] \times [4]}$}] {}; 
			\node []	at (m)	[label=left:\Large{$(1,0)$}] {}; 
			\node []	at (ml.east)	[label=below:\Large{$(3,0)$}] {}; 
			\node []	at (mr)	[label=below:\Large{$(1,0)$}] {}; 
			\node []	at (n)	[label=left:\Large{$(0,0)$}] {}; 
			\node []	at (n)	[label=right:\Large{$~~~~~~ \to {\text{universe }\bf [2]}$}] {}; 
			\node []	at (nl)	[label=below:\Large{$(2,0)$}] {}; 
			\node []	at (nr)	[label=below:\Large{$(0,0)$}] {}; 
			\node[] at (nr)[label=right:\Large{$~ \to {\text{universe }\bf [1]}$}] {};

			\node [] at (-6,0) [label=right:\LARGE${\tfull_{16}}$]	{};
			\draw[draw=black,very  thick, ->] (-4.2,-0.1) -- (-2,-0.6);
			
		\end{tikzpicture}
	}

	\caption{An example of the FFT binary tree $\tfull_{N}$ with $n=4$ and dimension $d=2$, (thus $N=16$). The universe corresponding to the nodes at each level of the tree is shown on the right side and the labeles of each node appears next to it.} \label{Tfull-labeling}
\end{figure*}
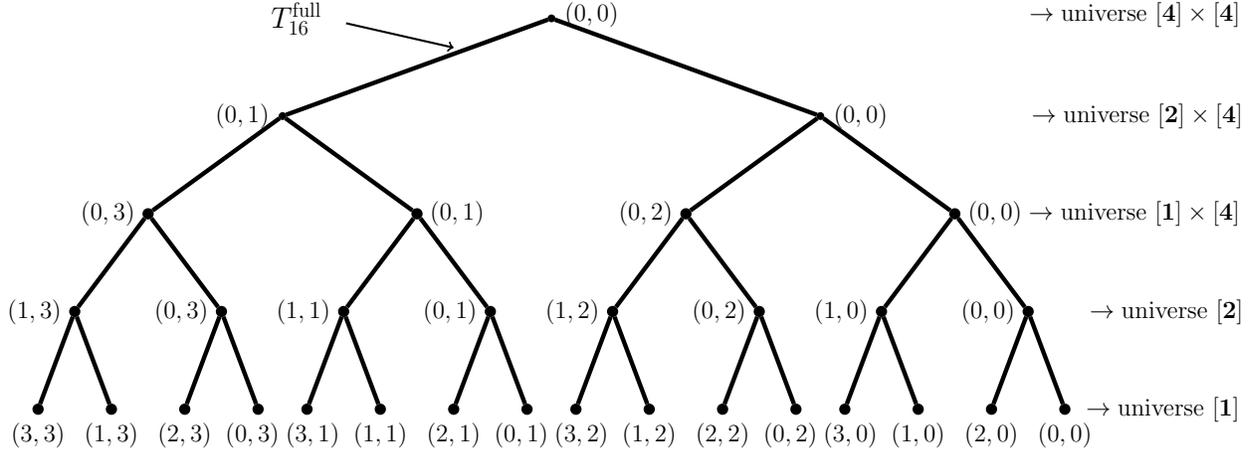

\subsection{Explanation of the barriers faced.}
\label{sec:barriers}

\paragraph{Discussion on the limits of the explorative approach, or why the quadratic barrier is impenetrable.} On a high level, the explorative approach we take maintains a vector $\wh{\chi}$ such that $\supp(\wh{\chi}) \subseteq \supp(\wh{x})$ at all times\footnote{In fact, this is an oversimplification of our approach (as well as slightly inaccurate), but for the sake of discussion let us assume that this is the case.}. Whenever the algorithm reaches a leaf $v \in\tfull_N$ (see definitions in the Preliminaries Section), it estimates it and adds it to $\wh{\chi}$. Subsequently, it proceeds by trying to recover the residual vector $\wh{x}-\wh{\chi}$. Now, imagine that we have recovered a constant fraction, say $1/10$, of $\wh x$, and want to proceed further in order to recover the remaining part of $x$, i.e. $\wh x - \wh \chi$, which is an $\Omega(k)$-sparse vector. In order even to test whether $\wh x - \wh \chi$ is the zero vector, we need to pick a set of $\Omega(k)$ random samples, satisfying for example the Restricted Isometry Property of order $k$, from $x - \chi$. In turn, this means that we need to compute the values $\chi_{\bm t}$ for all $\bm t$ in the aforementioned collection of random samples, and subtract them from the corresponding values of $x$. Since both $\supp(\wh \chi)$ and the samples needed for \textsc{RIP} are in principle unstructured sets of size $\Omega(k)$, the computation of the relevant $\chi_{\bm t}$ is exactly the classical non-equispaced Fourier transform, for which no strongly subquadratic algorithm in available. We explain this unavailability by providing a quadratic lower bound on this task based on the well-established Orthogonal Vectors hypothesis, see Theorem~\ref{thm:lower_bound}. This also provides evidence that the quadratic time barrier is the limit of our explorative approach. Indeed, at all times we need to decide whether to explore a subtree or not by testing whether $\wh x - \wh \chi$ is the zero vector projected on that subtree. Since subtracting the effect of $\wh \chi$ from the measurements, i.e. evaluating $\chi$ on an unstructured set of samples, cannot be done in strictly subquadratic time unless OVH fails, a subquadratic algorithm for exactly $k$-sparse FFT by traversing a pruned Cooley-Tukey FFT computation tree would most likely yield a subquadratic algorithm for the Orthogonal Vectors problem.

\paragraph{Discussion on the high-SNR regime.} We shall illustrate a potential scenario where we might miss most frequencies in the head of the signal if we run our algorithm on an input signal that is not in the high-SNR regime. Note that throughout the exploration algorithm, we always maintain a set of nodes, such that the union of the frequency cones of those nodes covers the head of the signal. The frequencies which are not covered are essentially treated as \emph{noise}, and we do not isolate with respect to them. Due to the fact that the adaptive aliasing filters hash the noise in a non-uniform way, it could be that our \textsc{HeavyTest} primitive misclassifies a subtree as ``frequency-inactive'', i.e. no head element inside it, although it contains one. In such a scenario, it is natural to abandon exploration inside the subtree. This would cause the noise in the system to increate by the magnitude of the missed head element (since we shall not isolate with respect to it anymore). Subsequently, this can potentially lead to a chain reaction, leading to successively missing head elements, and successively increasing the noise in the system, ending up to not recovering anything. 
However, our \textsc{HeavyTest} primitive is strong and ensures that we never miss a heavy frequency of signals that are not in the high-SNR regime as long as we perform oversampling by a factor $k$.
 
 On the other hand, note that in order to achieve the $\ell_2/\ell_2$ guarantee on signals that are not in high-SNR regime, we need to set the threshold of \textsc{HeavyTest} to $1/k$ fraction of the tail norm as opposed to the tail norm. Hence, another conceivable bad scenario is that, with such low threshold, the tail of the signal can make some frequency-inactive cones to appear heavy, introducing false positives. This can blow up the running time of the algorithm to super-polynomial in $k$.\newline

\paragraph{Discrepancy between the runtime of our robust algorithm and its sample complexity.} The only way we know how to perform dimension-independent estimation is via random sampling, as implemented in the \textsc{HeavyTest} routine. If we perform standard (non-lazy estimation) this would yield an additional multiplicative $k$ factor, as claimed in the first paragraph of Techniques III. Remedying this via lazy estimation shaves the multiplicative $k$ factor from the sample complexity, but does not do so in the running time. In particular, we run again into the same issue of subtracting $\wh \chi$ from the buckets (which corresponds to an unstructured set of samples), i.e. the solution of a non-equispaced Fourier transform instance. As we've proven a quadratic time lower bound for the latter problem, this indicates that this discrepancy is most likely unavoidable with this approach.

\section{Roadmap.}

The roadmap of this paper is the following. We follow an incremental approach, trying to introduce the techniques one by one, to the extent that is possible. In Section~\ref{sec:filters_old} we revise adaptive aliasing filters from~\cite{kapralov2019dimension}. In Section~\ref{sec:kraft} we give the facts related to Kraft's inequality which we are going to use throughout our algorithms. In \cref{sec:exact-k-sparse-reduction} we formally prove that the exact Sparse FFT problem can be translated and reduced to the tree exploration problem and prove our first main result, i.e., Theorem~\ref{thm:exact_quadratic}.
In Section~\ref{sec:lb} we give the conditional lower bound on non-equispaced Fourier transform. In Section~\ref{sec:filters_new}, the new structural properties of adaptive aliasing filters are inferred. In section~\ref{sec:robust_first} we introduce our first robust Sparse Fourier transform algorithm, illustrating techniques II-III and partly technique I. Lastly, in Section~\ref{sec:robust_sec} we obtain our final robust Sparse FT algorithm, which uses techniques I-IV. For that reason, the algorithm is presented last.

\section{Machinery from Previous work: Adaptive Aliasing Filters.}
\label{sec:filters_old}
In this section, we recall the class of adaptive aliasing filters that were introduced in \cite{kapralov2019dimension}. These filters form the basis of our sparse recovery algorithm. For simplicity, we begin by introducing the filters in one-dimensional setting and then show how they naturally extend to the multidimensional setting (via tensoring). 

\subsection{One-dimensional Fourier transform.}\label{sec:filters-1d}
Our algorithm extensively relies on binary partitioning the frequency domain. In $d=1$, the following definitions are the one-dimensional analogues (special cases) of the ones in \cref{sec:manipulate_FFT}. We re-iterate them here, for completeness.
The following is a re-interpretation of the \emph{splitting tree} of a set in dimension $1$.
\begin{definition}[Splitting tree] \label{def:splittree}
For every $S \subseteq [n]$, the \emph{splitting tree} $T=\tree(S, n)$ of a set $S$ is a binary tree that is the subtree of $\tfull_n$ that contains, for every $j\in [\log n]$, all nodes $v \in \tfull_n$ at level $j$ such that $\left\{ f \in S : f \equiv f_v \pmod{2^j} \right\} \neq \emptyset $. 
\end{definition}


Our Sparse FFT algorithm requires a filter $G$ that satisfies a refined isolating property due to the fact that throughout the execution of the algorithm, the identity of $\supp (\wh{x})$ is only partially known.  The following is a re-interpretation of the frequency cone of a node in dimension $1$. 
\begin{definition}[Frequency cone of a leaf of $T$]\label{def:iso-t}
Consider a subtree $T$ of $\tfull_n$, and vertex $v\in T$ which is at level $l_T(v)$ from the root, the {\em frequency cone of $v$ with respect to $T$} is defined as,
\[	\subtree_T(v):=\left\{ f_u : \text{ for every leaf } u \text{ in subtree of }\tfull_n \text{ rooted at }v \right\}.\]
\end{definition}

Note that under this definition, the frequency cone of a vertex $v$ of $T$ corresponds to the subtree rooted at $v$ when $T$ is embedded inside $\tfull_n$. Next we present the definition of an isolating filter, introduced in \cite{kapralov2019dimension}.
\begin{definition}[$(v, T)$-isolating filter] \label{def:v-t-isolating}
Consider a subtree $T$ of $\tfull_n$, and leaf $v$ of $T$, a filter $G:[n] \to \C^n$ is called \emph{$(v,T)$-isolating} if the following conditions hold:
	\begin{itemize}
		\item For all $f\in \subtree_T(v)$, we have $\wh{G}(f)=1$.
		\item For every $f'\in \bigcup_{\substack{ u \in \leaves(T) \\u \neq v }}\subtree_T(u) $, we have $\wh{G}_{v'}(f')=0$.
	\end{itemize}
\end{definition}
Note that in particular, for all signals $x \in \C^n$ with $\supp(\wh{x}) \subseteq \bigcup_{u\in \leaves(T)} \subtree_T(u)$ and $t\in[n]$,
\[ \sum_{j\in[n]} x(j) G_v(t-j) = \frac{1}{n} \sum_{f \in \subtree_T(v)}\wh{x}_f e^{2\pi i \frac{f t}{n}}. \]

The main technical construction of~\cite{kapralov2019dimension} is captured by the following Lemma.

\begin{lemma}[Filter properties, \cite{kapralov2019dimension}]\label{lem:filter-isolate}
Let $n$ be an integer power of two, $T$ a subtree of $\tfull_n$, $v$ a leaf in $T$. Let $f := f_v$ be the label of node $v$. Then the filter $G_v : [n] \rightarrow \mathbb{C}$ with Fourier Transform

\begin{equation}\label{eq:isolat-filter-1d}
\wh{G}_v(\xi) = \frac{1}{2^{w_T(v)}}\prod_{\ell \in \Anc(v,T)} \left( 1+ e^{2\pi i \frac{(\xi - f)}{2^{\ell+1}} } \right),
\end{equation}

is a $(v,T)$-isolating filter. Furthermore, 

\begin{itemize} 
\item $|\supp(G_v)| = 2^{w_T(v)}$, and the filter $G$ can be constructed in $O(2^{w_T(v)} + \log n)$ time (in the time domain).
\item Computing $\wh{G}_v(\xi)$ for $\xi \in [n]$ can be done in $O( \log n)$ time.	
\end{itemize}

\end{lemma}

\subsection{$d$-dimensional Fourier transform.}
In this subsection, we present the extension of adaptive aliasing filters to higher dimensions (by tensoring). It was shown in \cite{kapralov2019dimension} that multidimensional construction of these filters is extremely efficient and incurs no loss in the dimensionality.  

\begin{definition}[Multidimensional $(v, T)$-isolating filter]\label{def:v-t-isolating-highdim}
For every subtree $T$ of $\tfull_N$ and vertex $v\in T$, a filter $G_v\in \C^{n^d}$ is called \emph{$(v,T)$-isolating} if $\wh{G}_v({\ff})=1$ for every $\ff\in \subtree_T(v)$ and $\wh{G}_v({\ff'})=0$ for every $\ff'\in \supp{(T)} \setminus \subtree_T(v)$.
	
	In particular, for every signal $x \in \C^{n^d}$ with $\supp(\wh{x}) \subseteq \supp{(T)}$ and for all $\tv\in[n]^d$,
	\[ \sum_{\jj\in[n]^d} x(\jj) G_v(\tv-\jj) = \frac{1}{N} \sum_{\ff \in \subtree_T(v)}\wh{x}_{\ff} e^{2\pi i \frac{\ff^T \tv}{n}}. \]
\end{definition}

We need the following lemma which is the main result of this section and shows that isolating filters can be constructed efficiently.
\begin{lemma}[Construction of a multidimensional isolating filter --  Lemma 4.2 of \cite{kapralov2019dimension}]\label{lem:isolate-filter-highdim}
Let $T$ of $\tfull_N$, and consider $v \in \leaves(T)$. There exists a deterministic construction of a $(v,T)$-isolating filter $G_v$ such that 
\begin{enumerate}
\item $|\supp{(G_v)}| = 2^{w_T(v)}$.
\item $G_v$ can be constructed in time $O\left(2^{w_T(v)} + \log N\right)$.
\item For any frequency $\bm\xi \in [n]^d$, $\wh G_v(\bm\xi)$, i.e. the Fourier transform of $G_v$ at frequency $\bm\xi$, can be computed in time $O(\log N)$.
\end{enumerate}
\end{lemma}

\section{Kraft-McMillan inequality and averaging claims.}
\label{sec:kraft}

For our needs, we are going to make use of the following standard claim from coding theory, referred to as Kraft's or Kraft-McMillan inequality. The most general version is an inequality, but in the case of binary trees (complete codes in coding theory vocabulary), it becomes an equality.
\begin{theorem}[Kraft's equality]
\label{thm:kraft_eq}
Let $T\subseteq \tfull_N$, it holds that

\[	\sum_{u \in \leaves(T)} 2^{-w_T(u)} = 1.		\]
\end{theorem}

For a tree $T \subseteq \tfull_N$ and a set $S\subseteq \leaves(T)$, we shall refer to the \emph{Kraft mass} of $S$ with respect to $T$ as the quantity $\sum_{u \in S} 2^{-w_T(u)}$.

We shall frequently use the following straightforward Lemma, which we shall refer to as Kraft averaging. This ideas has appeared in~\cite{kapralov2019dimension}.

\begin{lemma}[Kraft averaging]
\label{lem:kraft_averaging}
Let $T\subseteq \tfull_N$, with $L$ leaves. Then there exists a $u^\ast \in \leaves(T)$ such that $w_T(u^\ast) \leq \log_2 L$.
\end{lemma}

The following fine-grained version of Kraft averaging is an indispensable building block of our lazy estimation technique, and constitutes one of the important departures from the approach in~\cite{kapralov2019dimension}. The reader may postpone reading it at the moment, since its first usage will be in section~\ref{sec:robust_first}. Neverthless, we decided to keep all the claims regarding Kraft's inequality in a separate section, for compactness reasons.

\begin{lemma}[Fine-grained Kraft Averaging]\label{lemma:kraft-ISCheap}
	Consider a subtree $T$ of $\tfull_N$ and a positive integer $b$ such that $|\leaves(T)| \le b$. Let $S:= \left\{ v\in \leaves(T): 2^{w_T(v)} \le 2b \right\}$, i.e. the leaves of $T$ with weight at most $\log_2(2b)$. Then there exists a subset $L \subseteq S$ such that 
\[ \frac{\max_{v\in L} 2^{w_T(v)} }{|L|} \le \frac{1}{\theta}, \]
where $\theta \le \frac{1}{4+ 2\log_2b}$.
\end{lemma}

Informally (but somewhat imprecisely), the claim postulates that for any subtree $T$ of $\tfull_N$ with $|\leaves(T)| = k$, there exist either 1 node of weight $1$, or $2$ nodes of weight of $2$, or \ldots at least $2^j/ \log k $ nodes of weight $j$, or $\ldots k / \log k$ nodes of weight $\log k$.  We now proceed with its proof.
\begin{proof}
	First note that one can show the preconditions of claim imply that $\sum_{u \in S} 2^{-w_T(u)} \ge \frac{1}{2}$.
	For every $j = 0, 1, \dots \lceil \log_2(2b) \rceil$, let $L_j$ denote the subset of $S$ defined as $L_j:= \{ u: u\in S, w_T(u) = j\}$. We can write,
	\[ \sum_{u \in S} 2^{-w_T(u)} = \sum_{j=0}^{\lceil \log_2(2b) \rceil} \frac{|L_j|}{2^j}\]
	Therefore by the assumption of the claim, we have that there must exist a $j \in \{0, 1, \dots \lceil \log_2(2b) \rceil\}$ such that $\frac{|L_j|}{2^j} \ge \frac{1}{2\lceil \log_2(2b) \rceil}$. Because $\theta \le \frac{1}{4+2\log_2|S|}$, there must exist a set $L \subseteq S$ such that $|L| \ge \theta \cdot \max_{v\in L} 2^{w_T(v)}$.
\end{proof}

\section{Translation of Exactly $k$-sparse FFT to Tree Exploration.} \label{sec:exact-k-sparse-reduction}

This section is devoted to solving the exact Sparse FFT problem through translation and reduction of this problem to the tree exploration problem detailed in \cref{sec:tree_problem} and invoking \cref{alg:exact_sparse_simple}. 
Recall that in this problem, we try to recover the \emph{values} written on the leaves of a full binary tree $\tfull_N$ using two procedures, \textsc{ZeroTest} and \textsc{Estimate} which satisfy the properties given in Assumption~\ref{assmptn-zerotest-estimate}. 
For solving the Sparse FFT problem we let the full binary tree $\tfull_N$ be defined as per \cref{sec:manipulate_FFT} and each of its leaf values be the Fourier coefficients of the input signal $x$ associated with the frequency labels of the corresponding leaves. Given this tree construction we can implement the procedures \textsc{ZeroTest} and \textsc{EstimateFreq} in a similar fashion to \cite{kapralov2019dimension}.

\begin{algorithm}[!h]
	\caption{$\textsc{ZeroTest}(T, x, \wh{\chi}, v, s)$} \label{alg:ZeroTest}
	\begin{algorithmic}[1]
		\State $\bm f := \bm f_v$
		\State $G_v \leftarrow $ the $(v,T)$ isolating filter as in 
		\cref{lem:isolate-filter-highdim}
		\State $\bm{\textsc{RIP}}_s :=$ a set of $O(s\log^3 N)$ samples, which suffice for $s$-\textsc{RIP} condition, see Theorem~\ref{RIP-thrm}
		\State $h_{\bm f}^{\Delta} \leftarrow \sum_{\bm \xi \in [n]^d } \left(e^{2\pi \frac{\bm \xi^\top \Delta}{n}} \cdot \wh{\chi}(\bm \xi) \cdot  \wh{G}_v(\bm \xi)\right)$, for all $\Delta \in \bm{\textsc{RIP}}_s$
		\State $H_{\bm f}^{\Delta} \leftarrow  \sum_{\bm j \in [n]^d} x(\bm j)G_v(\Delta-\bm j) - h^{\Delta}_{\bm j}$, for all $\Delta \in \bm{\textsc{RIP}}_s$
		\If { $\sum_{\Delta \in \bm{\textsc{RIP}}_s } |H_{\bm f}^\Delta|^2 =0 $ }
		\State {\bf return} $\true$
		\Else
		\State {\bf return} $\false$
		\EndIf
	\end{algorithmic}
	
\end{algorithm}

\begin{algorithm}[!h]
	\caption{$\textsc{EstimateFreq}(T, x, \wh{\chi}, v)$} \label{alg:estimate_freq}
	\begin{algorithmic}[1]
		\State $f := \bm f_v$
		\State $G_v \leftarrow$ the $(v,T)$ isolating filter as in \cref{lem:isolate-filter-highdim}
		\State $h_{\bm f} \leftarrow \sum_{ \xi \in [n]^d } \left( \wh{\chi}(\bm \xi) \cdot  \wh{G}_v(\bm \xi)\right)$
		\State Return $N \cdot \sum_{\bm j \in [n]^d} x(\bm j)G_v(-\bm j) - h_{\bm j}$
	\end{algorithmic}
\end{algorithm}

For detailed proofs we refer the reader to~\cite{kapralov2019dimension}. First we present the performance guarantee of the procedure \textsc{EstimateFreq} given in \cref{alg:estimate_freq} which can be proved by the filter isolation properties given in \cref{lem:isolate-filter-highdim}.
\begin{lemma}(Estimation)\label{lem:perfect_estimation}
	For any signals $x,\wh{\chi}$, any tree $T\subseteq \tfull_N$ such that $\supp(\wh{x} - \wh{\chi}) \subseteq \supp{(T)}$, and any leaf $\ell \in T$ which is also a leaf of $\tfull_N$, the procedure $\textsc{EstimateFreq}(T, x, \wh{\chi}, \ell)$ returns $(\wh{x} - \wh{\chi} )(\bm f_\ell)$. Furthermore, the routine requires
	\begin{itemize}
		\item $O\left(2^{w_T(\ell)}\right)$ sample complexity, and
		\item $\widetilde{O}( \|\wh{\chi}\|_0 +2^{w_T(\ell)} )$ running time.
	\end{itemize}
\end{lemma}

Next we present the performance guarantee of the procedure \textsc{ZeroTest} given in \cref{alg:ZeroTest} which was also proved in~\cite{kapralov2019dimension}.
\begin{lemma}(Testing whether a subtree is empty, see also~\cite[Lemma~7]{kapralov2019dimension})
	\label{lem:zerotest}
	For any signals $x,\wh{\chi}$, any tree $T\subseteq \tfull_N$ such that $\supp(\wh{x} - \wh{\chi}) \subseteq \supp{(T)}$, and any leaf $v \in T$, if $ \|\left(\wh{x}- \wh{\chi}\right)_{\subtree_T(v)}\|_0 \leq s$, then $\textsc{ZeroTest}(T, x, \wh{\chi}, v, s)$ determines correctly whether $\wh{x}_{\subtree_T(v)} = \wh{\chi}_{\subtree_T(v)}$ or not. Iff $\wh{x}_{\subtree_T(v)} = \wh{\chi}_{\subtree_T(v)}$, the primitive returns $\true$.
	Furthermore, the routine requires
	\begin{itemize}
		\item $O(2^{w_T(v)} \cdot |\textsc{RIP}_s|)$ sample complexity, and
		\item $\widetilde{O}\left(\|\wh{\chi}\|_0 \cdot |\textsc{RIP}_s| + 2^{w_T(v)}\cdot |\textsc{RIP}_s|\right)$ running time~\footnote{The $2^{w_T(v)}\cdot |\textsc{RIP}_s|$ correspond to the number of accesses on $x$, and $\|\wh{\chi}\|_0 \cdot |\textsc{RIP}_s|$ corresponds to the time needed to subtract $\wh{\chi}$ from the measurements. Lemma~7 in~\cite{kapralov2019dimension} has an additional third component, which corresponds to the time needed to prepare the isolating filter $\wh{G}_v$. It is not hard to see that this third component can always be bounded by $O(\log N)$, and hence can be safely ignored.}.
	\end{itemize}
	Recall that $\textsc{RIP}_s$ is a set of samples satisfying $s$-\textsc{RIP}, see Theorem~\ref{RIP-thrm}, and $|RIP_s| = O(s \log^3 N)$.
\end{lemma}

Given these primitives we can reduce the Sparse FFT problem to the tree exploration problem. More specifically, for any given signal $x \in \C^{n^d}$ we let $\tfull_N$ be a binary tree wth $N = n^d$ leaves such that the values of each leaf is the Fourier coefficient of the signal $x$ at the frequency which corresponds to the lable of that leaf. The only catch here is that the tree exploration algorithm we developed in \cref{sec:tree_problem} relies on functions $\ZeroTest$ and $\Estimate$ which satisfy the properties given in Assumption~\ref{assmptn-zerotest-estimate}, i.e., take as inputs $\found$ and $\excl$ but the primitives in \cref{alg:ZeroTest} and \cref{alg:estimate_freq} operate on a tree $T$ and a signal $\wh{\chi} \in \C^{n^d}$. 
We show that in fact any $\found$ and $\excl$ can be very efficiently translated to a tree $T$ and a signal $\wh{\chi}$. Using this translation, we can just invoke \cref{alg:exact_sparse_simple} to solve the exact $k$-Sparse FFT problem in almost quadratic time and thus, prove Theorem~\ref{thm:exact_quadratic}.

\begin{theorem}[Theorem~\ref{thm:exact_quadratic}, restated]
	\label{thm:exact_sparseft}
	The sparse Fourier transform problem with an exactly $k$-(Fourier sparse) signal $x : [n]^d \to \C$, i.e., $\|\wh x\|_0 \le k$ can be solved in 
	\[m = \widetilde O\left( k^2 \cdot 2^{8\sqrt{ \log k \cdot \log \log N}} \right)\] time, deterministically.
\end{theorem}
\begin{proof}
	We prove the theorem by defining an appropriate binary tree $\tfull_N$ for any given $k$-Sparse signal $x \in \C^{n^d}$ and then invoking \cref{alg:exact_sparse_simple} on $\tfull_N$ and then transforming the output to get the Fourier transform $\wh{x}$. The first part is straightforward because for any given signal $x \in \C^{n^d}$ we let $\tfull_N$ be a binary tree wth $N = n^d$ leaves such that the values of each leaf $\ell \in \leaves(\tfull_N)$ equals $\wh{x}(\ff_v)$. Because, $\|\wh{x}\|_0 \le k$, we can readily see that $|\hleaves(\tfull_N)| \le k$. 
	
	Next we have to show how to invoke \cref{alg:exact_sparse_simple} to learn this tree efficiently. As was assumed in \cref{sec:tree_problem}, for this algorithm to operate correctly it needs to have access to two primitives $\ZeroTest$ and $\Estimate$ which satisfy the properties given in Assumption~\ref{assmptn-zerotest-estimate}. We show that the primitives given in \cref{alg:ZeroTest} and \cref{alg:estimate_freq} can be modified to satisfy the conditions of Assumption~\ref{assmptn-zerotest-estimate}.

	According to Assumption~\ref{assmptn-zerotest-estimate}, these primitives take as input $\found$ and $\excl$, however, \cref{alg:ZeroTest} and \cref{alg:estimate_freq} take as input a tree $T$ and signals $x, \wh{\chi}$. The signal $x$ is just the input signal in time domain and we can feed it to these procedures without any modifications. The signal $\wh{\chi}$ is going to be constructed from $\found$ very efficiently in time $O(|\found|)$ as follows,
	\[ \wh{\chi}_{\ff} := \begin{cases}
	\found(\ell) & \text{if } \ff = \ff_\ell \text{ for some leaf }\ell \in \leaves(\tfull_N)\\
	0 & \text{otherwise}
	\end{cases}. \]
	We can also construct the tree $T$ from $\excl$ efficiently as follows. We consider the path $p$ in $\tfull_N$ from node $v$ to the root. First we start with $T = p$. Then we iterate over every node $u \in \tfull_N$ which is a child of a node that belongs to the path $p$ we check whether $\leaves(u) \cap \excl \neq \emptyset$ and if so we will add $u$ to tree $T$. Thus, this tree can be constructed in time $\wt{O}(|\excl|)$.
	
	Now we show that with the above translation of $\found$ and $\excl$ to $\wh{\chi}$ and tree $T$, \cref{alg:ZeroTest} and \cref{alg:estimate_freq} satisfy the conditions of Assumption~\ref{assmptn-zerotest-estimate}. First note that the tree $T$ that was constructed above is a subtree of $T(\excl)$ therefore,
	\[ \supp(T(\excl)) \subseteq \supp(T). \]
	Using the above inequality and the way we defined the signal $\wh{\chi}$ we have that, if a node $v$ is isolated by $\found$ and $\excl$ as per \cref{def:isolated-found-excl} we have the following,
	\[ \supp(\wh{x} - \wh{\chi}) = \hleaves(\tfull_N) \setminus \leaves(\found) \subseteq \supp(T(\excl)) \subseteq \supp(T). \]
	Furthermore, under the assumption that node $v$ is isolated by $\found$ and $\excl$ we have
	\[ \|\left(\wh{x}- \wh{\chi}\right)_{\subtree_T(v)}\|_0 \le | \hleaves(v) |. \]
	Thus, using the above two inequalities we can invoke \cref{lem:zerotest} to conclude that $\ZeroTest(T, x, \wh{\chi}, v, b)$ given in \cref{alg:ZeroTest} satisfies the conditions of the first part of Assumption~\ref{assmptn-zerotest-estimate}. Note that the way we constructed the tree $T$ implies that $w_\excl(v) = w_T(v)$, and also from the construction of $\wh{\chi}$ one can easily see $|\found| = \| \wh{\chi}\|_0$. Therefore, the runtime of $\ZeroTest(T, x, \wh{\chi}, v, b)$ matches the desired runtime in Assumption~\ref{assmptn-zerotest-estimate}.
	
	The above argument also implies that for a leaf $\ell \in \tfull_N$ if $\ell$ is isolated by $\found$ and $\excl$ then $T, x, \wh{\chi}, \ell$ satisfy the preconditions of \cref{lem:perfect_estimation}, thus $\textsc{EstimateFreq}(T, x, \wh{\chi}, \ell)$ given in \cref{alg:estimate_freq} satisfies the conditions of the second part of Assumption~\ref{assmptn-zerotest-estimate}. Also the runtime of this procedure matches the desired runtime in Assumption~\ref{assmptn-zerotest-estimate}.
	
	Therefore, \Cref{alg:exact_sparse_simple} is applicable with our proposed translations, and by Theorem~\ref{thm:exact_sparse_corr} and Theorem~\ref{thm:exact_sparse_time}, this procedure finds $\wh{x}$ perfectly and outputs correct estimates of the frequencies in time $\widetilde O\left( k^2 \cdot 2^{8\sqrt{ \log k \cdot \log \log N}} \right)$.
\end{proof}


\section{Lower Bound on Non-Equispaced Fourier Transform.}
\label{sec:lb}

The main result of this section is the following theorem.


\begin{theorem}(Detailed version of Theorem~\ref{thm:lower_bound})\label{thm:detailed_lb}
For every $c>0$ larger than an absolute constant and every $\delta>0$ there exists $c'>0$ and $\delta'>0$ such that if for all $\e\in (0, 1/2)$, for all $N$ a power of two and all $k\leq 2^{c' (\log N)^{1/3}}$ there exists an algorithm that solves the $1$-dimensional non-equispaced Fourier Transform problem on universe size $N$, sparsity $k$ in time $k^{2-\delta'} \poly(\log(N/\epsilon))$, then there exists an algorithm which solves $\textsc{OV}_{k,d}$ with $d=c\log k$ in time $k^{2-\delta}$.
\end{theorem}

As also mentioned in the abstract of this paper, this answers one of the subproblems of Problem 21 from IITK Workshop on Algorithms for Data Streams, Kanpur 2006. Additionally, the following proof facilities gives also the lower bound on sparse multipoint evaluation, i.e. Theorem~\ref{thm:multipoint}.

\begin{proof}
Given an Orthogonal Vectors instance, we shall appropriately construct a non-equispaced Fourier transform instance, such that an algorithm for the non-equispaced Fourier transform with strongly subquadratic running time in $k$ implies a strongly subquadratic time algorithm for the Orthogonal Vectors problem.

  Let $A = \{a_0,\ldots,a_{k-1}\},B = \{b_0,\ldots,b_{k-1}\} \subseteq \{0,1\}^d$ be the input to an $\textsc{OV}_{k,d}$ instance with $d=c\log k$. We denote by $a_j(r)$ the $r$-th coordinate of vector $a_j \in A$. We first pick sufficiently large integers $N,M,q$ that are powers of 2 such that $M = k d^{C_1 d}$, $q = C_2 d$, and $N = M^{2dq}$, where $C_1,C_2$ are sufficiently large absolute constants.

Next, we define for $j \in [k]$:
  \begin{align*}
    t_j := \sum_{r\in[d]} a_j(r) \cdot M^{r q}, \qquad
    f_j := \sum_{r\in [d]} b_j(r) \cdot \frac{N}{M^{r q +  1}}, 
  \end{align*}
  and set $F = \{f_0,\ldots,f_{k-1}\}, T = \{t_0,\ldots,t_{k-1}\}$. Furthermore, we define vector $x \in \mathbb{C}^N$ such that $x_t =1$ if $t \in T$, and $0$ otherwise, and we pick $\epsilon = \frac{1}{N}$. Thus, to transform our initial $\textsc{OV}_{k,d}$ instance to an instance of non-equispaced Fourier transform, we show that from additive $\epsilon \|\wh{x}\|_2$-approximations of $\wh{x}_{f_0},\ldots,\wh{x}_{f_{k-1}}$ we can infer whether $(A,B)$ contains a pair of orthogonal vectors. It then follows that an algorithm for non-equispaced Fourier transform running in time $k^{2-\delta'} \poly(\log(N/\epsilon))$ would imply a strongly subquadratic time algorithm for Orthogonal Vectors.
  
  \medskip
Our first claim postulates that $\wh{x}_{f_j}$ corresponds to summing up $\exp\left( -2\pi i \cdot \frac 1M \langle a_\ell, b_j \rangle \right)$ for all $\ell \in [k]$, up to error terms in the exponent.

\begin{claim} \label{cla:intermediate__0}
For every $j \in [k]$ it holds that
	\[		\wh{x}_{f_j} = \sum_{\ell\in [k]} \exp\left( -2\pi i \cdot \left(\tfrac 1  M \langle a_\ell, b_j \rangle + \xi_{\ell,j} \right) \right),			\]
for a real number $\xi_{\ell,j}$ satisfying
	\[	|\xi_{\ell,j}| \leq {d \choose 2} M^{-q-1}.	\]
\end{claim}

\begin{proof}
Fix $j \in [k]$ and note that
\begin{align*}
&\wh{x}_{f_j} = \sum_{t \in T} \exp\left(-2\pi i \frac{f_j t }{N}\right) \\ 
&=\sum_{\ell\in[k]} \exp\left(-\frac{2\pi i }{N} \cdot  \left( \sum_{r'\in[d]} a_\ell(r) \cdot M^{r q} \right) \cdot \left(\sum_{r\in[d]} b_j(r') \cdot \frac{N}{M^{r' q +  1}} \right)\right) \\
&=\sum_{\ell\in[k]} \exp\left(-2\pi i \cdot \sum_{(r,r') \in [d] \times [d]} a_\ell(r) b_j(r') \cdot M^{(r-r')q -1} \right) \\
&=\sum_{\ell\in[k]} \prod_{(r,r')\in [d] \times [d]} \exp\left(-2\pi i \cdot a_\ell(r) b_j(r') \cdot M^{(r-r')q -1} \right)
\end{align*}

We now investigate the exponents of the complex exponentials, namely $a_\ell(r) b_j(r') \cdot M^{(r-r')q -1}$ for $\ell \in [k]$ and $(r,r') \in [d] \times [d]$. In particular, we find that:
\begin{enumerate}
\item For any pair $(r,r')$ with $r > r'$, we have $(r-r')q - 1 \ge 0$, meaning that the corresponding exponent is an integer multiple of $2\pi i$. In turn, the corresponding term in the product contributes $1$, so it can be ignored.
\item For any pair $(r,r')$ with $r < r'$ we have $(r-r')q-1 \le -q-1$. 
  For a fixed $\ell$, there are ${d \choose 2}$ such products, and hence their total contribution to the exponent of the $\ell$-th summand is at most ${d \choose 2} M^{-q-1}$ (in absolute value). 
\item The pairs $(r,r')$ with $r=r'$ contribute to the exponent of the $\ell$-th summand the term $-2\pi i \cdot M^{-1} \sum_{r \in [d]} a_\ell(r) b_j(r) = -2\pi i \cdot M^{-1} \langle a_\ell, b_j \rangle$. 
\end{enumerate}
Putting everything together we arrive at the proof of the claim.
\end{proof}

  In the remainder of this proof we write
  \[ V_{j,h} := \sum_{\ell \in [k]} \langle a_\ell, b_j \rangle^h. \]

  Next, we perform a series expansion and error analysis on the exponential function to obtain:

\begin{claim}\label{claim:fourier-taylor-series}
For every $j \in [k]$ it holds that
	\[		\wh{x}_{f_j} = \xi'_j + \sum_{h \ge 0} \big( -\tfrac{2\pi i}M \big)^h \tfrac {1}{h!} \cdot V_{j,h}, \]
for a complex number $\xi'_j$ satisfying
	\[	|\xi'_j| \leq M^{-q}. \]
\end{claim}
\begin{proof}
  Let $a,b$ be real numbers. 
  Starting from the basic fact $|\exp(-2 \pi i b) - 1| \le 2 \pi |b|$, we obtain $\exp(-2\pi i (a+b)) = \exp(-2\pi i a) + \exp(-2\pi i a) (\exp(-2\pi i b) - 1) = \exp(-2\pi i a) + \xi'_{a,b}$ with $|\xi'_{a,b}| \le 2 \pi |b|$. 
  In particular, with notation as in Claim~\ref{cla:intermediate__0}, we have
  \[ \exp\left( -2\pi i \cdot \left(\tfrac 1  M \langle a_\ell, b_j \rangle + \xi_{\ell,j} \right) \right) = \exp\left( -2\pi i \cdot \tfrac 1  M \langle a_\ell, b_j \rangle \right) + \xi'_{\ell,j}, \]
  with $|\xi'_{\ell,j}| \le 2 \pi |\xi_{\ell,j}| \le 2 \pi {d \choose 2} M^{-q-1} \le M^{-q}$. 
  
  Summing over all $\ell \in [k]$ now yields
  \[		\wh{x}_{f_j} = \sum_{\ell\in [k]} \exp\left( -2\pi i \cdot \left(\tfrac 1  M \langle a_\ell, b_j \rangle + \xi_{\ell,j} \right) \right) = \xi'_j + \sum_{\ell\in [k]} \exp\left( -2\pi i \cdot \tfrac 1  M \langle a_\ell, b_j \rangle \right),			\]
  with $|\xi'_j| \le 2 \pi k {d \choose 2} M^{-q-1}$. Using that $M = k d^{C_1 d}$ for a sufficiently large constant $C_1>0$, we obtain $|\xi'_j| \le M^{-q}$.
  
  Finally, we use the series expansion of $\exp(.)$ to obtain
  \[ \wh{x}_{f_j} = \xi'_j + \sum_{h \ge 0} \big( -\tfrac{2\pi i}M \big)^h \tfrac {1}{h!} \cdot \sum_{\ell \in [k]} \langle a_\ell, b_j \rangle^h. \]
\end{proof}

We now show that in our expression for $\wh{x}_{f_j}$ the summands $\big( -\tfrac{2\pi i}M \big)^h \tfrac {1}{h!} \cdot V_{j,h}$ lie sufficiently far apart, so that each summand can be reconstructed from an approximation of $\wh{x}_{f_j}$.

\begin{claim} \label{cla:sdfoasdj}
Let $j \in [k],\, H \in [d]$, and let $\widetilde{x}_{f_j}$ be an additive $\epsilon \|\wh{x}\|_2$ approximation of $\wh{x}_{f_j}$. Then 
	\[		\widetilde{x}_{f_j} - \sum_{h=0}^{H-1} \big( -\tfrac{2\pi i}M \big)^h \tfrac {1}{h!} \cdot V_{j,h} = \big( -\tfrac{2\pi i}M \big)^H \tfrac {1}{H!} \cdot \left( V_{j,H} + \xi''_{j,H} \right), \]
for a complex number $\xi''_{j,H}$ satisfying
	\[	|\xi''_{j,H}| < 1/3. \]
\end{claim}
\begin{proof}
  Note that by Parseval's identity, we have $\|\wh{x}\|_2 = \sqrt{N}\cdot \|x\|_2 = \sqrt{N \cdot k}$. Therefore, $|\widetilde{x}_{f_j} - \wh{x}_{f_j}| \le \epsilon \|\wh{x}\|_2 \le \epsilon \sqrt{N \cdot k} \le \sqrt{k/N}$ as $\epsilon = 1/N$. Since $N = M^{2dq}$ and $M \ge k$, we obtain $|\widetilde{x}_{f_j} - \wh{x}_{f_j}| \le M^{-(q-1)}$.
  
  Note that 
  \begin{align*} \bigg| \sum_{h > H} \big( -\tfrac{2\pi i}M \big)^h \tfrac {1}{h!} \cdot V_{j,h} \bigg| 
  \;&\le\; \sum_{h > H} \big( \tfrac{2\pi}M \big)^h \tfrac {1}{h!} \cdot \sum_{\ell \in [k]} \langle a_\ell, b_j \rangle^h
  \\ &\le\; \big( \tfrac{2\pi}M \big)^H \tfrac {1}{H!} \cdot \sum_{h > H} \big( \tfrac{2\pi}M \big)^{h-H} \cdot k \cdot d^h \\
  &= \; \big( \tfrac{2\pi}M \big)^H \tfrac {1}{H!} \cdot k d^H \cdot \sum_{h > H} \big( \tfrac{2\pi d}M \big)^{h-H}.
  \end{align*}
  Since $M$ is sufficiently larger than $d$, the latter sum can be bounded by $\frac{4 \pi d}M$, and hence
  \begin{align} \label{eq:ghhjsaf}
    \bigg| \sum_{h > H} \big( -\tfrac{2\pi i}M \big)^h \tfrac {1}{h!} \cdot V_{j,h} \bigg| 
  \;&\le\; \big( \tfrac{2\pi}M \big)^H \tfrac {1}{H!} \cdot \tfrac{4 \pi k d^{H+1}}M \;\le\; \frac{1}{10} \cdot \big( \tfrac{2\pi}M \big)^H \tfrac {1}{H!},
  \end{align}
  using the fact that $M = k d^{C_1 d}$ for a sufficiently large constant $C_1>0$ and $H \in [d]$.
  
  We now, using Claim~\ref{claim:fourier-taylor-series}, decompose:
  \begin{align*}
    &\widetilde{x}_{f_j} - \sum_{h=0}^{H-1} \big( -\tfrac{2\pi i}M \big)^h \tfrac {1}{h!} \cdot V_{j,h} \\
    &= \left( \widetilde{x}_{f_j} - \wh{x}_{f_j} \right) + \left( \wh{x}_{f_j} - \sum_{h=0}^{H-1} \big( -\tfrac{2\pi i}M \big)^h \tfrac {1}{h!} \cdot V_{j,h} \right)  \\
    &= \left( \widetilde{x}_{f_j} - \wh{x}_{f_j} \right) + \xi'_j + \big( -\tfrac{2\pi i}M \big)^H \tfrac {1}{H!} \cdot V_{j,H} + \sum_{h > H} \big( -\tfrac{2\pi i}M \big)^h \tfrac {1}{h!} \cdot V_{j,h}.
  \end{align*}
  Recall that $|\widetilde{x}_{f_j} - \wh{x}_{f_j}| \le M^{-(q-1)}$ and $|\xi'_j| \le M^{-q}$. We use $H \in [d]$ and our choice of $M = k d^{C_1 d}$ and $q = C_2 d$ for sufficiently large constants $C_1, C_2 > 0$ to conclude  that $M^{-(q-1)} \le \frac{1}{10} \cdot \big( \frac{2 \pi}M \big)^H \frac 1{H!}$. 
  Together with inequality (\ref{eq:ghhjsaf}), this gives
  \[		\widetilde{x}_{f_j} - \sum_{h=0}^{H-1} \big( -\tfrac{2\pi i}M \big)^h \tfrac {1}{h!} \cdot V_{j,h} = \big( -\tfrac{2\pi i}M \big)^H \tfrac {1}{H!} \cdot \left( V_{j,H} + \xi''_{j,H} \right), \]
for a complex number $\xi''_{j,H}$ with $|\xi''_{j,H}| < 1/3$.
\end{proof}

Repeatedly applying the above claim allows us to reconstruct the numbers $V_{j,0},\ldots,V_{j,d}$:
   \begin{claim}\label{claim_intermediate1}
Fix $j \in [k]$. Let $\epsilon = \frac{1}{N}$. Given an additive $\epsilon \|\wh{x}\|_2 =  \sqrt{k/N}$ approximation to $\wh{x}_{f_j}$ we can infer the exact values of 
\[V_{j,h} := \sum_{\ell\in[k]} \langle a_\ell, b_j \rangle^h,\] for any $h\in [d]$, in time $\poly(d, \log k)$.
  \end{claim}
  \begin{proof}
Suppose that we have already computed the sums $V_{j,h}$ for all $0 \le h < H$. Then we know the left hand side of Claim~\ref{cla:sdfoasdj}. Since $|\xi''_{j,H}| < 1/3$, there is a unique integer $V_{j,H} = \sum_{\ell\in[k]} \langle a_\ell, b_j \rangle^H$ that satisfies the equation in Claim~\ref{cla:sdfoasdj}. Hence, we can infer $V_{j,H}$. Therefore, we can iteratively compute $V_{j,0},V_{j,1},\ldots,V_{j,d-1}$.  

  Note that when evaluating expressions of the form $\big( -\tfrac{2\pi i}M \big)^h \tfrac {1}{h!}$, we can compute them up to precision $\epsilon$ in time $\poly(d,\log k)$, since it suffices to perform arithmetic on numbers with $\poly(d,\log k)$ digits. This yields another additive error in the same order of magnitude as in the proof of Claim~\ref{cla:sdfoasdj}. The same error analysis therefore shows that this precision is sufficient to compute the exact integers $V_{j,h}$. 
  \end{proof}
  
  The above claim postulates that we can infer the values $V_{j,h}$ for $h \in [d]$. We next show that these values allow us to determine whether there exists a pair of orthogonal vectors.

\begin{claim}\label{claim_intermediate2}
Given the values $V_h := V_{j,h}$ for all $h\in [d]$ and some fixed $j$, we can find out whether there exists an $\ell$ such that $\langle a_\ell,b_j \rangle = 0$, in time $\poly(d,\log k)$.
\end{claim}
\begin{proof}
This relies on the observation that we can write $V_h$ as 
\[V_h = \sum_{r=0}^{d-1} Z_r \cdot r^h\] 
for 
\[Z_r := \left|\{ \ell \in [k] \mid \langle a_\ell, b_j \rangle = r \}\right|.\]
In other words, the values $V_h$ are obtained from the values $Z_r$ by multiplication with a Vandermonde matrix. Since this $d\times d$ matrix is invertible and all elements of this matrix and $V_h$ are of value at most $k\cdot d^d$, we can infer the values $Z_r$ from the values $V_h$ in $\poly(d,\log k)$ time. Indeed, we can compute the inverse of this Vandermonde matrix multiplied by its determinant (so that the resulting matrix contains integer entries) using $\poly(d)$ operations on integers with $\poly(d,\log k)$ digits (each such operation takes $\poly(d, \log k)$ time). Multiplying the vector of $V_h$'s by this matrix yields $Z_r$'s multiplied by the determinant of the Vandermonde matrix, which can be computed and canceled using $\poly(d, \log k)$ operations by manipulating large integers with $\poly(d, \log k)$ number of digits. This yields the value $Z_0 = |\{ \ell \mid \langle a_\ell, b_j \rangle = 0 \}|$ and thus allows us to decide whether $b_j \in B$ is orthogonal to some vector in $A$. 
 \end{proof}

Using Claims~\ref{claim_intermediate1} and~\ref{claim_intermediate2} over all Fourier evaluations $\left\{\wh{x}_{f}\right\}_{f \in F}$ we can determine in time $k\cdot \poly(d, \log k)$ whether whether $(A,B)$ contains an orthogonal pair. Thus, for  $\delta\in (0, 1/2)$ an algorithm for non-equispaced Fourier transform running in time $k^{2-\delta'} \poly(\log(N/\epsilon))$ for $\e=1/N$, would imply the existence of a $k^{2-\delta'} \poly(d, \log k)$ time algorithm for $\textsc{OV}_{k,d}$, since $\log(N/\e)=2\log N=O(d^2) \log M=\poly(d, \log k)$ for any choice of constants $C_1, C_2>0$. 
For any constant $c>0$,  if dimension $d = c \log k$, this running time can be bounded by $O(k^{2 - \delta})$ as long as $\delta'\geq 2\delta$, contradicting the Orthogonal Vectors Hypothesis (Conjecture~\ref{ovc}). Finally, it remains to note that since $d=c\log k$ and 
$$
N=M^{2dq}=(k d^{C_1 d})^{2dq}=(c\log k)^{C_1C_2 c^3 \log^3 k},
$$
we have that $ 2^{ c'(\log N/\log\log N)^{1/3}} \le k \le 2^{ c''(\log N)^{1/3}} $ as long as $c'$ is sufficiently small as a function of $c, C_1, C_2$, and $c''$ is sufficiently large as required.

\end{proof}


\section{Robust analysis of adaptive aliasing filters.}
\label{sec:filters_new}

This section is devoted to our technical innovation regarding adaptive aliasing filters. This a delicate analysis of how the filters act on an arbitrary vector. Such a robustification will be useful in order to control the amount of energy a measurement receives from the elements outside of the head. The absence of the properties derived in this section constitutes the restriction that has driven the ``exactly $k$-sparse'' assumption in~\cite{kapralov2019dimension}.

\subsection{One-dimensional case.}
\label{sec:filters_onedim}	

We first develop the appropriate machinery for the one-dimensional case. Generalizing the idea to higher dimensions can be done using tensoring, as we shall show in the next subsection. We first present a standalone computation of the Gram matrix of adaptive aliasing filters corresponding to a specific tree $T \subseteq \tfull_n$.

\begin{lemma}(Gram Matrix of adaptive aliasing filters)\label{lem:gram}
Consider a tree $T \subseteq \tfull_n$, and two distinct leaves $v,v'$ of $T$. Let $G_v$ (resp. $G_{v'})$ be the $(v,T)$-isolating (resp. $(v',T)$-isolating) filter, as per \eqref{eq:isolat-filter-1d}. 
Then, 
\begin{enumerate}
\item (diagonal terms) the energy of the filter corresponding to $v$ is proportional to $2^{-w_T(v)}$. In particular,
\[	\|\wh{G}_v \|_2^2 := \sum_{ \xi \in [n]} |\wh{G}_v(\xi) |^2 = \frac{n}{2^{w_T(v)}}.	\]

\item (cross terms) the adaptive aliasing filters corresponding to $v$ and $v'$ are orthogonal, i.e. 
\[	\langle \wh G_v, \wh G_{v'} \rangle :=  \sum_{ \xi \in [n]} \wh G_v(\xi) \cdot \overline{\wh G_{v'}(\xi)} = 0.	\]
\end{enumerate}

\end{lemma}

\begin{proof}

We prove each bullet separately. Both bullets follow by symmetry considerations: cancellations that occur either by the fact that roots of unity cancel across a poset of a group, or by the sign change happening to specific complex exponentials at branching points of the tree $T$. The first one uses Kraft's equality.

\paragraph{Proof of Bullet 1.} Let $f:= f_v$ and $f': = f_{v'}$ denote the labels of $v$ and $v'$, respectively. By \eqref{eq:isolat-filter-1d}, we have  
\begin{align*}
|\wh{G}_v(\xi)|^2 &= 4^{-w_T(v)} \cdot  \prod_{ \ell \in \Anc(v,T)} \left( 1 + e^{2\pi  i \frac{\xi - f}{2^{\ell+1}}} \right) \cdot\left( 1 + e^{-2\pi  i \frac{ \xi-f}{2^{\ell+1}}} \right) \\
& =4^{-w_T(v)} \cdot  \prod_{ \ell \in \Anc(v,T)} \left(2 + e^{2\pi  i \frac{\xi - f}{2^{\ell+1}}} + e^{-2\pi  i \frac{\xi - f}{2^{\ell+1}}} \right) \\
&= 4^{-w_T(v)} \cdot \sum_{ \substack{S, T \subseteq \Anc(v,T)\\ S \cap T = \emptyset} }
 2^{|\Anc(v,T)| - |S \cup T|} \cdot e^{2\pi  i (\xi - f) \cdot \left( \sum_{ \ell \in S } \frac{1}{2^{\ell+1}} - \sum_{\ell \in T} \frac{1}{2^{\ell+1}} \right)}   \\
&= 4^{-w_T(v)} \cdot \left( 2^{ w_T(v)} + \sum_{ \substack{S, T \subseteq \Anc(v,T)\\ S \cap T = \emptyset, S \cup T \neq \emptyset }}  2^{w_T(v) - |S \cup T|} \cdot e^{2\pi  i (\xi - f) \cdot \left( \sum_{ \ell \in S } \frac{1}{2^{\ell+1}} - \sum_{\ell \in T} \frac{1}{2^{\ell+1}} \right)}  \right)
\end{align*}

Note that the expression $\mathrm{expr_{S,T}} =  \sum_{ \ell \in S } \frac{1}{2^{\ell+1}} - \sum_{\ell \in T} \frac{1}{2^{\ell+1}} $ inside the complex exponential can be $0$ if and only if $S = T$, which is precluded by the fact that $S \cap T = \emptyset, S \cup T \neq \emptyset$. Thus, this gives rise to the exponential $e^{ 2\pi i (\xi - f) \cdot \mathrm{expr_{S,T}}}$, which cancels out when summing over all $\xi$. Hence, we obtain that 
\begin{align*}
\sum_{ \xi \in [n]} |\wh{G}_v(\xi) |^2 = \sum_{ \xi \in [n] } 4^{-w_T(v)} \cdot \left( 2^{w_T(v)} + 0 \right) = \frac{n}{2^{w_T(v)}}.
\end{align*}

\paragraph{Proof of Bullet 2.} By \eqref{eq:isolat-filter-1d}, we have that 

\begin{align*}
\langle \wh{G}_v, \wh{G}_{v'} \rangle &= \sum_{\xi \in [n]} \wh{G}_v ( \xi ) \cdot \overline{\wh{G}_{v'} (\xi)} \\
&= \sum_{\xi \in [n]}\left( \frac{1}{2^{w_T(v)}} \prod_{\ell \in \Anc(v,T)} \left(1+e^{2\pi i \frac{\xi-f}{2^{\ell+1}}} \right)\right) \cdot \left(\frac{1}{2^{w_T(v')}} \prod_{\ell \in \Anc(v',T)} \left( 1+e^{-2\pi i \frac{\xi-f'}{2^{\ell+1}}} \right)\right) \\
&= 2^{-w_T(v) - w_T(v')}\cdot \sum_{\xi\in [n]} \sum_{\substack{S \subseteq \Anc(v,T)\\ S' \subseteq \Anc(v',T)}} e^{2\pi i (\xi -f ) \cdot \sum_{\ell \in S} \frac{1}{2^{\ell+1}} - 2\pi  i (\xi -f')\cdot \sum_{\ell \in S'} \frac{1}{2^{\ell+1}}} \\ 
&= 2^{-w_T(v) - w_T(v')} \cdot \sum_{\substack{S \subseteq \Anc(v,T)\\ S' \subseteq \Anc(v',T)}} \sum_{\xi \in [n]}   e^{2\pi i (\xi -f ) \cdot \sum_{\ell \in S} \frac{1}{2^{\ell+1}} - 2\pi  i (\xi -f')\cdot \sum_{\ell \in S'} \frac{1}{2^{\ell+1}}} \\
&:= 2^{-w_T(v) - w_T(v')}\cdot(A+B), 
\end{align*}
where $A$ is sum of the terms that satisfy $S \neq S'$, and $B$ is sum of terms satisfying $S=S'$. We will show that $A = B = 0 $ separately. The equality $A=0$ holds by a summation over all $\xi$ and the fact that roots of unity cancel across a poset of a subgroup, whereas the equality $B=0$ by a symmetry argument which exploits the sign change in the lowest common ancestor of $v$ and $v'$.

\paragraph{Computing $A$.} We will prove that if $S \neq S'$ then  
\[	\sum_{\xi \in [n]}  e^{2\pi i (\xi -f ) \cdot \sum_{\ell \in S} \frac{1}{2^{\ell+1}} - 2\pi  i (\xi -f')\cdot \sum_{\ell \in S'} \frac{1}{2^{\ell+1}}} = 0,	\]
which suffices to establish $A=0$. Note that
\begin{align*}
e^{2\pi i (\xi -f ) \cdot \sum_{\ell \in S} \frac{1}{2^{\ell+1}} - 2\pi  i (\xi -f')\cdot \sum_{\ell \in S'} \frac{1}{2^{\ell+1}}} = \\
e^{2\pi  i \xi \cdot (\sum_{\ell \in S} \frac{1}{2^{\ell+1}} - \sum_{\ell \in S'} \frac{1}{2^{\ell+1}} )} \cdot g,
\end{align*}
where $g = e^{2\pi i f' \cdot \sum_{\ell \in S'} \frac{1}{2^{\ell+1}} - 2\pi i f \cdot \sum_{\ell \in S} \frac{1}{2^{\ell+1}}}$ does not depend on $\xi$. Summing over all $\xi \in [n]$ and taking into account that $\sum_{\ell \in S} \frac{1}{2^{\ell+1}} - \sum_{\ell \in S'} \frac{1}{2^{\ell+1}}  \neq 0$ by the fact that $S \neq S'$, yields the desired result (the summation can also be viewed a summation of the roots of unity over $\frac{n}{2^{\mathrm{max}\{S\bigtriangleup S'\}}}$ copies of a poset of an additive subgroup of size $2^{\mathrm{max}\{S \bigtriangleup S'\}}$, where $\bigtriangleup$ denotes symmetric difference of sets).

\paragraph{Computing $B$.} This quantity contains only terms corresponding to $S = S'$. Note that in this case $S \subseteq \Anc(v,T) \cap \Anc(v',T)$, and we have
\begin{align*}
B = \sum_{ S \subseteq \Anc(v,T) \cap \Anc(v',T) } \sum_{\xi\in [n]} e^{2\pi i (f'-f) \cdot \sum_{\ell \in S} \frac{1}{2^{\ell+1}}} = \\
n \cdot \sum_{ S \subseteq \Anc(v,T) \cap \Anc(v',T) } e^{2\pi i (f'-f) \cdot \sum_{\ell \in S} \frac{1}{2^{\ell+1}}}.
\end{align*}
Let $u$ be the lowest common ancestor of $v,v'$ in tree $T$, i.e. the node on which the paths from the root to those two nodes split. Partition the powerset of $\Anc(v,T) \cap \Anc(v',T)$ to pair $\left(S, S \cup \{l_T(u)\} \right)$, where $l_T(u) \notin S$. We shall prove that
\[	e^{2\pi i(f'-f) \cdot \sum_{\ell \in S} \frac{1}{2^{\ell+1}}} + e^{2\pi i(f'-f) \cdot \sum_{\ell \in S \cup \{l_T(u)\}} \frac{1}{2^{\ell+1}}} = 0.	\]

Indeed, by definition of $u$ we have that $(f' - f) \equiv 2^{l_T(u)} \mod{2^{l_T(u)+1}}$, which in turn gives that $e^{2\pi  i (f '-f )\cdot \frac{1}{2^{l_T(u)+1}}} = e^{2 \pi  i \frac{2^{l_T(u)}}{2^{l_T(u)+1}}} = e^{ \pi i} = -1$. This gives 
\begin{align*}
e^{2\pi i(f'-f) \cdot \sum_{\ell \in S} \frac{1}{2^{\ell+1}}} + e^{2\pi i(f'-f) \cdot \sum_{\ell \in S \cup \{l_T(u)\}} \frac{1}{2^{\ell+1}}} = \\
e^{2\pi i(f'-f) \cdot \sum_{\ell \in S} \frac{1}{2^{\ell+1}}} \cdot \left( 1 + e^{2\pi  i (f'-f) \cdot \frac{1}{2^{l_T(u)+1}}}\right) = 0.
\end{align*}
Thus, we conclude that $B = 0$, which finishes the proof of this Lemma.
\end{proof}

The next lemma proves that for any tree $T$, the sum of squared values of adaptive aliasing filters corresponding to all leaves of $T$ is equal to $1$ at every frequency. The $(v,T)$-isolating filters for different leaves $v$ of $T$ can have very different behaviors and shapes in the Fourier domain, nevertheless, these filters collectively act as an isometry in the sense that the sum of their squared values is 1 everywhere in the Fourier domain.

\begin{lemma}(Total contribution of adaptive aliasing filters to one frequency)\label{filter-robust-1dim}
Consider a tree $T \subseteq \tfull_n$. For every leaf $v$ of $T$, let $G_v$ denote the $(v,T)$-isolating filter as per \eqref{eq:isolat-filter-1d}, then it holds that 
	\[	\forall \xi \in [n]: \sum_{v\in \leaves(T)} |G_v(\xi)|_2^2 =1 .	\]
\end{lemma}

\begin{proof}
Fix $\xi \in [n]$. By \eqref{eq:isolat-filter-1d}, we have 
\begin{align*}
&\sum_{v\in \leaves(T)} |\wh{G_{v}}(\xi)|^2  = \\
&\sum_{ v\in \leaves(T)} 4^{-w_T(v)} \cdot \prod_{\ell \in \Anc(v,T)} \left|1+e^{2\pi  i (\xi - f_v)/2^{\ell+1}} \right|^2 = \\
&\sum_{ v\in \leaves(T)} 4^{-w_T(v)} \cdot  \prod_{\ell \in \Anc(v,T)} \left( 2 + e^{2\pi  i (\xi - f_v ) / 2^{\ell+1}} + e^{-2\pi  i (\xi - f_v)/2^{\ell+1}} \right) = \\
&\sum_{ v\in \leaves(T)} 2^{-w_T(v)} \cdot \prod_{\ell \in \Anc(v,T)} \left( 1 + \mathrm{cos}\left( 2\pi(\xi - f_v)/2^{\ell+1} \right) \right) = \\
&\sum_{v\in \leaves(T)} 2^{-w_T(v)} \sum_{ S \subseteq \Anc(v,T)}  \prod_{\ell \in S} \mathrm{cos}\left(2\pi \frac{\xi - f_v}{2^{\ell+1}} \right).
\end{align*}

Thus, it suffices to prove that for all $\xi \in [n]$

\begin{equation}\label{eq:totalcontr-onefreq}
	\sum_{v\in \leaves(T)} 2^{-w_T(v)} \sum_{ S \subseteq \Anc(v,T)} \prod_{\ell \in S} \mathrm{cos}\left(2\pi \frac{\xi - f_v}{2^{\ell+1}} \right) = 1.
\end{equation}

We will implicitly interchange the summation between $v$ and $S$ in \eqref{eq:totalcontr-onefreq} and carefully group terms together so that most of them cancel out, due to the sign change in each branching point. In particular, fix a branching point, i.e. a node $u \in T$ with two children. We will estimate the contribution of all sets $S$ such that $\mathrm{max}(S) = l_T(u)$ in \eqref{eq:totalcontr-onefreq}. Let $u_l$ be the left child of $u$ in $T$, and let $u_r$ be the right child of $u$ in $T$. Note that,
\[	\forall f \in \subtree_T(u_l), f' \in \subtree_T(u_r): f - f' \equiv 2^{l_T(u)} \mod{2^{l_T(u)+1}} .	\]
In turn, this implies that for any $\xi \in [n]$ and any two $f,f'$ as above we have: $(\xi - f) \equiv (\xi - f') + 2^{l_T(u)} \mod{2^{l_T(u)+1}}$, which gives the desired change in the branching point:
\[\mathrm{cos}\left(2\pi \frac{\xi - f}{2^{l_T(u)+1}}\right) = - \mathrm{cos}\left(2\pi \frac{\xi-f'}{2^{l_T(u)+1}}\right). \] 
Thus, if we let $T_r$ and $T_l$ denote the subtrees of $T$ rooted at $u_r$ and $u_l$, respectively, then the total contribution of a set $S$ that satisfies $\mathrm{max}(S) = l_T(u)$ and $S\subseteq \Anc(v,T)$ for some leaf $v$ of $T$ to \eqref{eq:totalcontr-onefreq} can be expressed as
\begin{align*}	
& \prod_{\ell \in S \setminus \{ l_T(u) \} } 
\mathrm{cos} \left(2 \pi \frac{\xi - f_u}{2^{\ell+1}} \right) \cdot \left(\sum_{v \in T_{r}} \frac{1}{2^{w_T(v)}} - \sum_{v \in T_{l}} \frac{1}{2^{w_T(v)}}\right) \\
&\qquad = \prod_{\ell \in S \setminus \{ \mathrm{br}\} } 
\cos \left(2 \pi \frac{\xi - f_u}{2^{\ell+1}} \right) \cdot 2^{ -w_T(u)} \left(\sum_{v \in T_{r}} \frac{1}{2^{w_{T_r}(v)}} - \sum_{v \in T_l} \frac{1}{2^{w_{T_l}(v)}}\right)  = 0.
\end{align*}
The latter holds since $\sum_{v \in T_{r}} \frac{1}{2^{w_{T_r}(v)}} = 1$ by Kraft's equality; similarly $\sum_{v \in T_{l}} \frac{1}{2^{w_{T_l}(v)}} = 1$.

Thus, we will get cancellation of the contribution of all non-empty sets $S$ by summing over all branching points. On the other hand, the contribution of the empty set $S=\emptyset$ is exactly $2^{-w_T(v)}$, for each leaf $v$. The sum of all those contributions is $1$, again by Kraft's equality, giving the lemma.

\end{proof}

\subsection{Extension to $d$ dimensions.}
\label{sec:filters-d}	

We are now ready to proceed with the generalization of the robustness properties of the adaptive aliasing filters given in~Section\ref{sec:filters_onedim} to high dimensions. The following lemma states that the isolating filters constructed in Lemma~\ref{lem:isolate-filter-highdim}, collectively for all leaves, preserve (in particular, do not increase) the energy of a signal.
\begin{lemma}\label{filter-robust-multidim}
	Consider a tree $T \subseteq T_N^{full}$. If for every leaf $v$ of $T$ we let $\wh{G}_v$ be the Fourier domain $(v,T)$-isolating filter constructed in Lemma~\ref{lem:isolate-filter-highdim}, then for every ${\bm\xi} \in [n]^d$,
\[ \sum_{ v\in \leaves(T)} |\wh{G}_v({\bm\xi})|^2 = 1.\]
\end{lemma}
\begin{proof}
The proof is by induction on the dimension $d$. 

\paragraph{Base of induction:} Lemma~\ref{filter-robust-1dim} precisely proves the inductive claim for $d=1$.

\paragraph{Inductive step:} Suppose that the inductive hypothesis holds for $d-1$ dimensional isolating filters. Given this inductive hypothesis, we want to prove that the inductive claim holds for $d$ dimensional filters. Let $T$ be a subtree of $T_N^{full}$, where $N=n^d$. For every leaf $v$ of tree $T$, let $v_0, v_1, \cdots v_l$ denote the path from root to $v$ where $v_0$ is the root and $v_l=v$. We let $p_v$ denote a vertex in $T$, defined as
\[ p_v := \begin{cases}
v_{\log_2n} & \text{ if } l_T(v)\ge \log_2n\\
v & \text{ otherwise}
\end{cases}. \]

Now, we construct the tree $T^*$ by making a copy of the tree $T$ and then removing every node which is at distance more than $\log_2n$ from the root. Let the nodes of $T^*$ be labeled by projecting the labels of $T$ to their first coordinate as follows,
\[ \text{for every node } u \in T^*: f_u = f_1, \text{ where }(f_1,f_2, \cdots f_d) \text{ is the label of $u$ in }T. \]
One can easily verify that the set $P := \{ p_v: v \in \leaves(T) \}$ specifies the set $\leaves(T^*)$. For every $u \in P$ let $H_u$ be a $(u,T^*)$-isolating filter, constructed as in Lemma \ref{lem:isolate-filter-highdim}.

Moreover, for every leaf $u \in P$ we define $T_u$ to be a copy of the subtree of $T$ which is rooted at $u$. We label the nodes of the tree $T_u$ by projecting the labels of $T$ to their last $d-1$ coordintates as follows,
\[ \text{for every node } z \in T_u: \ff_z = (f_2,f_3, \cdots f_d), \text{ where }(f_1,f_2, \cdots f_d) \text{ is the label of $u$ in }T. \]
For every leaf $v$ of $T$, let $\wh Q_v$ be the Fourier domain $(v,T_{p_v})$-isolating filter constructed in Lemma~\ref{lem:isolate-filter-highdim}. Note that in case $p_v = v$, the tree $T_{p_v}$ will be empty and by convention we define our $(v,T_{p_v})$-isolating filter to be $\wh Q_v \equiv 1$. Therefore, using these definitions, for every leaf $v \in \leaves(T)$, the $(v,T)$-isolating filter $\wh G_v$ constructed in Lemma~\ref{lem:isolate-filter-highdim} satisfies
\[ \wh G_v(\bm \xi) \equiv H_{p_v}(\xi_1) \cdot Q_v(\xi_2,\xi_3, \dots \xi_d), \]
for every $\bm \xi = (\xi_1, \xi_2, \dots, \xi_d) \in [n]^d$. 
Hence, we can write
\begin{align*}
\sum_{ v\in \leaves(T)} \left|\wh{G}_v({\bm\xi})\right|^2 &= \sum_{ v\in \leaves(T)} \left|H_{p_v}(\xi_1) \cdot Q_v(\xi_2,\xi_3, \dots \xi_d)\right|^2\\
&= \sum_{ u \in P} \sum_{\substack{v\in \leaves(T) \\ \text{s.t. } p_v = u}} |H_{u}(\xi_1)|^2 \cdot \left| Q_v(\xi_2,\xi_3, \dots \xi_d)\right|^2\\
&= \sum_{ u \in P} |H_{u}(\xi_1)|^2 \sum_{\substack{v\in \leaves(T) \\ \text{s.t. } p_v = u}}  \left| Q_v(\xi_2,\xi_3, \dots \xi_d)\right|^2.
\end{align*}
We proceed by proving that for every $u \in P$, $\sum_{\substack{v\in \leaves(T) \\ \text{s.t. } p_v = u}}  \left| Q_v(\xi_2,\xi_3, \dots \xi_d)\right|^2 = 1$. Recall that for every leaf $v \in \leaves(T)$, $Q_v$ is a $(v,T_{p_v})$-isolating filter, constructed in Lemma~\ref{lem:isolate-filter-highdim}. Therefore, for every leaf $v$ of $T$ such that $p_v=u$, $Q_v$ is indeed a $(v,T_{u})$-isolating filter as per the construction of Lemma~\ref{lem:isolate-filter-highdim}. Hence,
\[ \sum_{\substack{v\in \leaves(T) \\ \text{s.t. } p_v = u}}  \left| Q_v(\xi_2,\xi_3, \dots \xi_d)\right|^2 = \sum_{\substack{v\in \leaves(T_u) }}  \left| Q_v(\xi_2,\xi_3, \dots \xi_d)\right|^2. \]
Now we can invoke the inductive hypothesis because $T_u$ is a subtree of $T_{N'}^{full}$ where $N'=n^{d-1}$. therefore,
\[ \sum_{\substack{v\in \leaves(T) \\ \text{s.t. } p_v = u}}  \left| Q_v(\xi_2,\xi_3, \dots \xi_d)\right|^2 = \sum_{\substack{v\in \leaves(T_u) }}  \left| Q_v(\xi_2,\xi_3, \dots \xi_d)\right|^2 = 1. \]
Consequently, we have,
\[ \sum_{ v\in \leaves(T)} \left|\wh{G}_v({\bm\xi})\right|^2 = \sum_{ u \in P} |H_{u}(\xi_1)|^2 = \sum_{ u\in \leaves(T^*)} |H_{u}(\xi_1)|^2 = 1, \]
where the last equality follows because $H_u$ is a $(u,T^*)$-isolating filter as per the construction of Lemma \ref{lem:filter-isolate} and hence by Lemma~\ref{filter-robust-1dim}, $\sum_{ u\in \leaves(T^*)} |H_{u}(\xi_1)|^2 = 1$. This completes the inductive proof and ergo the Lemma.
\end{proof}

We readily find that the following corollary of the above lemma holds,
\begin{corollary}\label{infty-norm-bound-islating-filter}
	The Fourier domain isolating filter $\wh G$ constructed in Lemma~\ref{lem:isolate-filter-highdim} satisfies $\|\wh G\|_\infty \le 1 $.
\end{corollary}

\newpage
\section{Robust Sparse Fourier Transform I.}
\label{sec:robust_first}

The section is devoted to proving our first result on robust Sparse Fourier transforms, which illustrates techniques II to IV and partially technique I. We first remind the reader about the high SNR regime we consider.

\paragraph{$k$-High SNR Regime.} A vector $x: [n]^d \to \C$ satisfies the $k$-high SNR assumption, if there exists vectors $w,\eta: [n]^d \to \C$ such that i) $\wh{x}=\wh{w} + \wh{\eta}$, ii) $\supp(\wh{w}) \cap \supp(\wh{\eta}) = \emptyset$, iii) $|\supp(\wh{w})| \leq k$ and iv) $|\wh{w}_f| \geq 3 \cdot \|\wh{\eta}\|_2$, for every $f \in \supp(\wh{w})$.
In the rest of this section we prove the following main theorem.

\begin{restatable}[Robust Sparse Fourier Transform]{theorem}{sfftrobust} \label{thm:core}
	Given oracle access to $x: [n]^d \to \C$ with $x= w+\eta$ in $k$-high SNR model and parameter $\epsilon >0$, we can find using 
\[m = \widetilde{O} \left( k^{7/3} + \frac{k^2}{\epsilon}\right)\]
 samples from $x$ and in $\widetilde{O}\left( \frac{k^3}{\epsilon} \right)$ time a signal $\wh{\chi}$ such that
	\[\|\wh{\chi} - \wh{x}\|_2^2\leq (1+\epsilon) \cdot \|\wh{\eta}\|_2^2, \]
with high probability in $N$.
\end{restatable}

For every tree $T$ and node $v \in T$, we let $\wh{x}_v$ be the vector $\wh{x}_{\subtree(v)}$, i.e. signal $\wh x$ supported on frequencies in the frequency cone of $v$ and zeroed out everywhere else.
At all times, for every $v \in T$, our algorithm maintains a signal $\wh \chi_v: [n]^d \to \C$ that is supported on $\subtree_{T}(v)$. This signal will serve as our estimate for $\wh{w}_v$. Initially, all these vectors are going to be $\{0\}^{n^d}$. The execution of our algorithm ensures that we can always keep sparse representations of those vectors. Parameters and variables $n, d$ and $N=n^d$ are treated as global. 

Furthermore, for any signal $y:[n]^d \to \C$ and parameter $\mu\ge0$ we define
\begin{equation}\label{def:head}
\head_\mu(y) := \left\{ \jj \in [n]^d:~ |y_\jj| \geq 3\mu\right\}.
\end{equation}
Under this notation, we are interested in recovering the set $\head_{\|\wh{\eta}\|_2}(\wh{x})$, as well as obtain accurate estimations for the values of $\wh x$ on frequencies in set $\head_{\|\wh{\eta}\|_2}(\wh{x})$. Using the notion of $\head_\mu(y)$, one can see that a signal $x$ is in the $k$-high SNR regime iff there exists a $\mu> 0$ such that $\left| \head_\mu(\wh x) \right| \le k$ and $\mu \ge \left\| \wh x - \wh x_{\head_\mu(\wh x)} \right\|_2$.

At all times, we keep a set $\Estimated$, corresponding to the coordinates in $\supp(\wh{w})$ that we have estimated. We define $L_v := \subtree_T(v) \cap \left( \supp(\wh{w}) \setminus \Estimated \right)$, which corresponds to the \emph{unestimated} coordinates in the support of $w$ that lie in the frequency cone of $v$.

Our main algorithm consists of an outer loop that we call \textsc{RobustSparseFFT} and an inner loop that we call \textsc{RobustPromiseSFT}. Our algorithm also makes use of an auxiliary primitive for estimating the values of located frequencies as well as a primitive for testing whether a signal is ``heavy'' (meaning that it contains a head element). In the rest of this section we first give the primitives \textsc{Estimate} and \textsc{HeavyTest} together with the guarantee on their performance. Then we present the main algorithm and prove its performance. The \textsc{HeavyTest} routine is analogous to $\textsc{ZeroTesT}$ from Section~6. However, the \textsc{RIP} property alone does not suffice (and hence we cannot pick a deterministic collection of samples). Instead, we use a random collection of samples, which suffices for upper bounding the contribution of the tail while simultaneously satisfying \textsc{RIP}.

\subsection{Computational Primitives for the Robust Setting.}

In this subsection we give some of the primitives that will be used in our algorithms. The proof of correctness of these primitives is postponed to subsection~\ref{sec:prove_correct_primitives}.

The very first primitive we present is \textsc{HeavyTest}, see Algorithm~\ref{alg:zero-test}. This primitive performs a test on the signal to detect whether a given frequency cone contains heavy elements or not. 
\begin{algorithm}[!h]
	\caption{Test whether $v$ is a frequency-active node, i.e. $ \|(\wh{x-\chi})_v \|_2 > 2 \|\wh{\eta}\|_2$}\label{alg:zero-test}
	\begin{algorithmic}[1]
		\Procedure{HeavyTest}{$x, \wh{\chi}, T, v, m, \theta$} 
		\State{$\ff \gets \ff_v$}

		\State{$(G_v,\wh{G}_v) \gets \textsc{MultiDimFilter}(T, v, n)$} \State \emph{//$(v,T)$-isolating filters as per Lemma~\ref{lem:isolate-filter-highdim}}
		
		\For{$z=1$ to $32 \log N$}
		\State{$\textsc{RIP}_m^z \gets$ Multiset of $m$ i.i.d. uniform samples from $[n]^d$} \State

		\State{$h^{z}_{\Delta} \gets \sum_{\bm\xi \in [n]^d} \left( e^{2\pi i \frac{\bm\xi^\top \Delta}{n}} \cdot \wh \chi({\bm\xi}) \cdot \wh G_v(\bm\xi)\right)$ for every $\Delta \in \textsc{RIP}_m^z$}\label{a7l6}

		\State{$H^{z} \gets \frac{1}{|\textsc{RIP}_m^z|} \sum_{\Delta \in \textsc{RIP}_m^z} \left| N\cdot \sum_{\jj \in [n]^d} G_v(\Delta-\jj) \cdot x({\jj})- h^{z}_{\Delta} \right|^2$}
		\EndFor

		\If{$\textsc{Median}_{z \in [32 \log N]} \left\{H^{z}\right\} \le \theta$} \State \emph{//$\theta = 5\|\wh{\eta}\|_2^2$.}\label{a7l8}
		\State{\textbf{return} False}
		
		\Else	
		\State{\textbf{return} True}
		
		\EndIf	
		\EndProcedure
	\end{algorithmic}
\end{algorithm}

\begin{lemma}[\textsc{HeavyTest} guarantee]\label{lem:guarantee_heavy_test}
	Consider signals $x , \wh{\chi}: [n]^d \to \C$ and an arbitrary subtree $T$ of $T_N^{full}$. For an arbitrary leaf $v$ of $T$, let $\wh y := \left(\wh x - \wh \chi\right) \cdot \wh G_v$, where $\wh G_v$ be the Fourier domain $(v,T)$-isolating filter constructed in Lemma~\ref{lem:isolate-filter-highdim}. Then the following statements hold, for any $\theta >0$:
	\begin{itemize}
		\item If there exists a set $S \subseteq [n]^d$ such that $\left\| \wh{y}_S \right\|_2^2 >\frac{11 \theta}{10} $, then \textsc{HeavyTest}$(x, \wh{\chi}, T, v,  m, \theta)$ (Algorithm~\ref{alg:zero-test}) outputs $\true$ with probability $1 - \frac{1}{N^{16}}$, provided that $m$ is a large enough integer satisfying 
		\[	m = \Omega \left( |S| \cdot \frac{\left\| \wh{y} \right\|_2^2}{\left\| \wh{y}_S \right\|_2^2} \cdot \log^2|S| \log N\right).\]
		\item If $\left\| \wh{y}\right\|_2^2 \le \theta/5$,
		then \textsc{HeavyTest} outputs $\false$ with probability $1 - \frac{1}{N^{5}}$.
		
		\item The sample complexity of this procedure is $\widetilde O\left(2^{w_T(v)} \cdot m\right)$.
		
		\item The runtime of the \textsc{HeavyTest} procedure is $\widetilde O \left( \|\wh{\chi}\|_0 \cdot m +  2^{w_T(v)} \cdot m\right)$.
	\end{itemize}
	
\end{lemma}

Next, we present the second auxiliary primitive \textsc{Estimate} in Algorithm~\ref{alg:zero-test}.

\begin{lemma}[\textsc{Estimate} guarantee]\label{est-inner-lem}
	Consider signals signals $x,\wh \chi:[n]^d \to \C$, a subtree $T$ of $T_N^{full}$, and an integer parameter $m$. For a subset $S \subseteq \leaves(T)$, the procedure \textsc{Estimate}$(x,\wh \chi, T, S, m)$ (see Algorithm~\ref{alg:high-dim-Est-robust}) outputs $\left\{ \wh H_v \right\}_{v \in S}$ such that
	\[ \Pr\left[ \sum_{v \in S} \left| \wh H_v - (\wh{x-\chi})({\ff_v}) \right|^2 \le \frac{16}{m} \sum_{\bm\xi \in [n]^d \setminus \supp{(T)} } \left| (\wh{x-\chi})({\bm{\xi}}) \right|^2 \right] \ge 1 - \frac{|S|}{N^8}. \]
	The sample complexity of this procedure is $\widetilde O\left(m  \cdot \sum_{v \in S} 2^{w_T(v)} \right)$ and the runtime of the procedure is $\widetilde O\left( m \cdot \sum_{v \in S} 2^{w_T(v)} + |S| \cdot m\cdot  \|\wh \chi\|_0 \right)$.
\end{lemma}
 
\begin{algorithm}[t]
	\caption{For $S \subseteq T$, estimates $\left(\wh{x} - \wh{\chi}\right)_S$ by isolating $S$ from every node in $T$.}\label{alg:high-dim-Est-robust}
	\begin{algorithmic}[1]
		
		\Procedure{Estimate}{$x,\wh \chi, T, S, m$}

		\For{$v \in S$}
		
		\State{${\ff} \gets {\ff}_v$}	 
		
		\State $(G_v,\wh{G}_v) \gets \textsc{MultiDimFilter}(T, v, n)$	
		\Comment{$(v,T)$-isolating filters as per Lemma~\ref{lem:isolate-filter-highdim}}
		
		\For{$z=1$ to $16 \log N$}
		\State $\textsc{RIP}_m^z \gets$ Multiset of $B$ i.i.d. uniform samples from $[n]^d$ 
		\State $h^{z}_v \gets \sum_{\Delta \in \textsc{RIP}_m^z} e^{-2\pi i \frac{\ff^\top \Delta}{n}}\sum_{\bm\xi \in [n]^d}  e^{2\pi i \frac{\bm\xi^\top \Delta}{n}} \cdot \wh \chi({\bm\xi}) \cdot \wh G_v(\bm\xi)$ \label{a8l7}

		\State $H^{z}_v \gets \frac{1}{|\textsc{RIP}_m^z|} \left(N \cdot \sum_{\Delta \in \textsc{RIP}_m^z} \left( e^{-2\pi i \frac{\ff^\top \Delta }{n}} \sum_{\jj \in [n]^d} G_v(\Delta -\jj) \cdot x({\jj})\right) - h^z_v \right)$
		\EndFor
		
		\State{$\wh{H}_v \gets \textsc{Median}_{z \in [16\log N]} \left\{H^{z}_v\right\}$} \label{a8l9} 
		\Comment {Median of real and imaginary parts separately}
		
		\EndFor
		
		\State \textbf{return} $ 
		\left\{\wh H_v \right\}_{v \in S} $
		
		\EndProcedure
		
	\end{algorithmic}
	
\end{algorithm}
Lastly, we need the following primitive whose objecive is to find a subset of identified leaves that are cheap to estimate \emph{on average}.
\begin{claim}[\textsc{ExtractCheapSubset} guarantee]\label{claim:findCheap}
	For every subtree $T$ of $\tfull_N$ and every subset $S\subseteq \leaves(T)$ that satisfies $\sum_{u \in S} 2^{-w_T(u)} \ge \frac{1}{2}$, the primitive \textsc{ExtractCheapSubset}$(T,S)$ (see bottom of Algorithm~\ref{alg:sfftrobust}) outputs a non-empty subset $L \subseteq S$ such that 
\[|L| \cdot \left(8 + 4 \log |S| \right) \ge \max_{v\in L} 2^{w_T(v)}.\]
\end{claim}

\subsection{Main Algorithm.}
In this subsection we present our main sparse FFT algorithm. The algorithms consists of an outer loop and an inner loop. The outer loop, called \textsc{RobustSparseFT}, always maintains a vector $\wh{\chi}$ a tree $\frontier$ such that 
	\[ \head_\mu(\wh{x} - \wh{\chi}) \subseteq \cup_{u \in \frontier} \subtree(u). \]
At every point in time, we explore the frequency cones of the low-weight $\frontier$ by running the \textsc{RobustPromiseSFT} algorithm. For the pseudocodes of the routines \textsc{RobustPromiseSFT} and \textsc{RobustSparseFT}, see Algorithms~\ref{alg:ksparsefft-inner} and \ref{alg:sfftrobust}, respectively.


%


\begin{figure}[t]
	\centering
	\scalebox{.76}{
		\begin{tikzpicture}[level/.style={sibling distance=80mm/#1,level distance = 1.5cm}]
			\node [arn] (z){}
			child {node [arn] {}edge from parent [photon]
				child[draw=white]
				child{node [arn]{}}
			}
			child { node [arn] {} edge from parent [photon]
				child {node [arn_l] (v) {}
					child[sibling distance=55mm]{node [arn] {}edge from parent [solidline]
						child[sibling distance=45mm]{node [arn] (e1) {}
							child{node [arn] {}edge from parent [dashed]
								child{node [arn] {}}
								child{node [arn] {}}
							}
							child[draw=white]
						}
						child{node [arn] {}
							child{node [arn] {}
								child{node [draw=blue] {}}
								child[draw=white]
							}
							child{node [arn] {}
								child{node [draw=blue] {}}
								child{node [draw=blue] {}}
							}
						}
					}
					child[sibling distance=40mm]{node [arn] {}edge from parent [solidline]
						child[sibling distance=35mm]{node [circle, white, draw=red, fill=red, inner sep = 1.4] {}edge from parent [draw=red, thin, dashed]
							child[draw = white]
							child{node [circle, white, draw=red, fill=red, inner sep = 1.4]{}
								child{node [circle, white, draw=red, fill=red, inner sep = 1.4] (r1){}}
								child{node [circle, white, draw=red, fill=red, inner sep = 1.4] (r2){}}
							}
						}
						child[sibling distance=45mm]{node [arn] {}
							child{node [arn] (e2) {}
								child{node [arn] {} edge from parent [dashed]}
								child[draw=white]
							}
							child{node [arn] (e3) {}
								child{node [arn] {} edge from parent [dashed]}
								child[draw=white]
							}
						}
					}
				}
				child {node [arn] {}
					child{node [arn]{}}
					child{node [arn]{}}
				}
			};
		
		\node []	at (v.north)	[label=left: \LARGE{$v$}]	{};
		
		\node [] at (-6.9,0.1) [label=right:{\Large{\Path}}]	{};
		\draw[draw=black,very  thick, ->] (-4,0) -- (-1.8,-0.5);

		\node [] at (-5,-3.5) [label=right:{\Large{subtree $T$}}]	{};
		\draw[draw=black,very  thick, ->] (-2.5,-3.6) -- (-0.5,-4);
		
		\node [] at (e1) [label=left:{\Large{yet to be explored subtree}}]	{};
		\node [trrr] at (e1.north) [] {};

		\node [trr] at (e2.north) [] {};
		
		\node [trr] at (e3.north) [] {};
		
		\node [] at (-7,-7.8) [label=north:\Large{$\sident$ leaves}]	{}
		edge[->, thick, bend right=35] (-1.4,-9.3)
		edge[->, thick, bend right=38] (0.1,-9.3)
		edge[->, thick, bend right=40] (1.4,-9.3);
		
		\node [] at (-6,-11) [label=north:\Large{recovered \& subtracted}] {};
		\node [] at (-3.5,-11.3) [label=left:\Large{leaves (frequencies)}]	{}
		edge[->, bend right=20, thick, dashed, draw=red] (2.3,-9.3)
		edge[->, bend right=20, thick, dashed, draw=red] (3.4,-9.3);

		\draw[draw=black,thick] (-5,-9.25) rectangle ++(13,0.55);
		\foreach \x in {-3.85, -2.3, -0.4 , 1.1, 2.1, 3, 4.3, 5.7} \draw[draw=black,very thick] (\x,-8.7) -- (\x,-9.25);
		\node []	at (9,-11)	[label=left:\Large $\head \cap \subtree_\Path(v)$] {}
		edge[->, bend right=50, very thick] (8,-8.95);
		
		\end{tikzpicture}
	}
	\par
	
	\caption{Illustration of an instance of \textsc{RobustPromiseSFT} (Algorithm~\ref{alg:ksparsefft-inner}). This procedure takes in a tree $\Path$ (shown with thin edges) together with a leaf $v \in \leaves(\Path)$ and adaptively explores/constructs the subtree $T$ rooted at $v$ to find all heavy frequencies that lie in $\subtree_\Path(v)$. If $\head$ denotes the set of heavy frequencies, then the algorithm finds $\head \cap \subtree_\Path(v)$ by exploring $T$. Once the identity of a leaf is fully revealed, the algorithm adds that leaf to the set $\sident$. When the number of marked leaves grows to the point where marked frequencies can be estimated cheaply, our algorithm estimates them all in a batch, subtracts off the estimated signal, and removes all corresponding leaves from $T$.}\label{fig:robust-promise}
\end{figure}
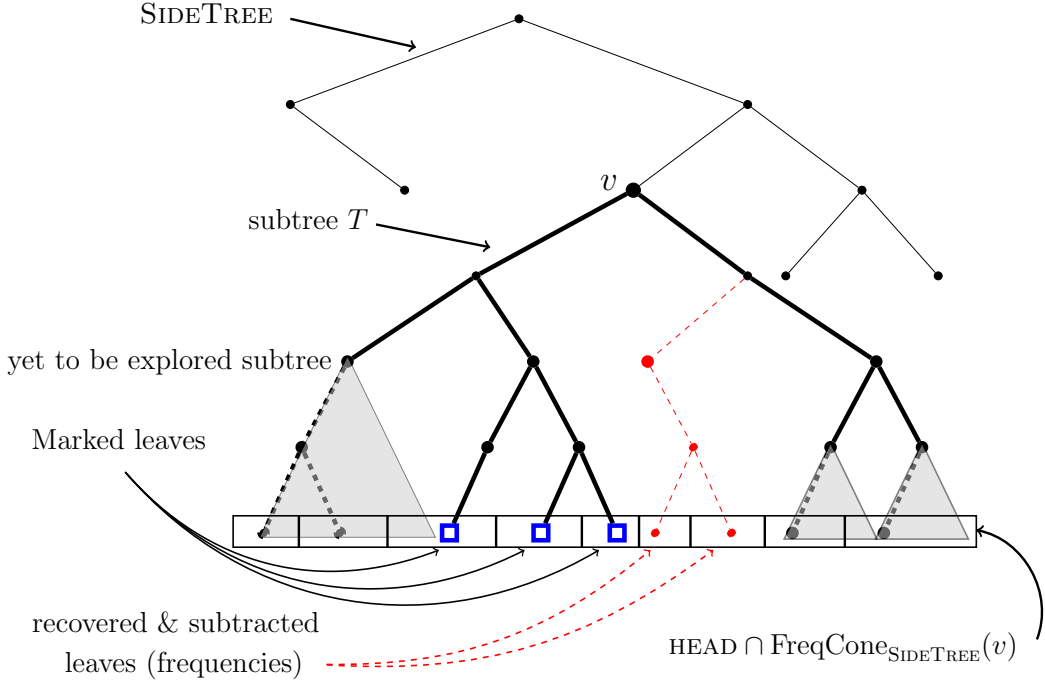

\paragraph{Overview of \textsc{RobustPromiseSFT} (Algorithm\ref{alg:ksparsefft-inner}):} Consider an invocation of \textsc{RobustPromiseSFT}$(x, \wh\chi_{in}, \Path, v, b, k, \mu)$. Suppose that $\wh y := \wh x - \wh \chi_{in}$ is a signal in the $k$-high SNR regime, i.e., $\wh y$ has $k$ heavy frequencies and the value of each such heavy frequency is at least $3$ times higher than the tail's norm. More formally, let $\head \subseteq[n]^d$ denote the set of heavy (head) frequencies of $\wh y$ and suppose that $|\head| \le k$, and the tail norm of $\wh y$ satisfies $\| \wh y - \wh y_{\head}\|_2 \le \mu$ and additionally suppose that $|\wh y(\ff)| \ge 3\mu$ for every $\ff \in \head$. If $\Path$ fully captures the heavy frequencies of $\wh y$, i.e., $\head \subseteq \supp{(\Path)}$, and the number of heavy frequencies in frequency cone of node $v$ is bounded by $b$, i.e., $|\head \cap \subtree_{\Path}(v)| \le b$, then \textsc{RobustPromiseSFT} finds a signal $\wh \chi_v$ such that $\supp{(\wh\chi_v)} = \head \cap \subtree_{\Path}(v) := S$ and $\|\wh y_S - \wh \chi_{v}\|_2^2 \le \frac{\mu^2}{20}$. An example of the input tree $\Path$ is illustrated in Figure~\ref{fig:robust-promise} with thin solid black edges. Additionally, one can see node $v$ which is a leaf of $\Path$ in this figure.

Algorithm~\ref{alg:ksparsefft-inner} recovers heavy frequencies in the subree of $v$, i.e., $S = \head \cap \subtree_{\Path}(v)$, by iteratively exploring the subtree of $\Path$ rooted at $v$, which we denote by $T$, and simultaneously updating $\wh \chi_v$. We show an example of subtree $T$ at some iteration of our algorithm in Figure~\ref{fig:robust-promise} with thick solid edges. 
Our algorithm, in all iterations, maintains a subtree $T$ such that the frequency cone of each of its leaves contain at least one head element, i.e.,
\begin{equation}\label{eq:leaves-non-empty}
\text{for every }u\in \leaves(T): \subtree_{\Path \cup T}(u) \cap \head \neq\emptyset. 
\end{equation}
We demonstrate, in Figure~\ref{fig:robust-promise}, the leaves that correspond to set $S = \head \cap \subtree_{\Path}(v)$ via leaves at bottom level of the subtree rooted at $v$.
One can easily verify \eqref{eq:leaves-non-empty} in this figure by noting that the frequency cone of each leaf of $T$ contains at least one element from the set $\head$.
Additionally, at every iteration of the algorithm, the union of all frequency cones of subtree $T$ captures all heavy frequencies that are not recovered yet, i.e., 
\begin{equation}\label{eq:tree-contain-all-heavies}
S \setminus \supp{(\wh \chi_v)} \subseteq \supp{\left(\Path\cup T\right)}.
\end{equation}
In Figure~\ref{fig:robust-promise}, we show the set of fully recovered leaves (frequencies), i.e., $\supp{(\wh \chi_v)}$, using red thin dashed subtrees. These frequencies are subtracted from the residual signal $\wh y - \wh \chi_v$ and their corresponding leaves are removed from subtree $T$, as well. One can verify that condition~\ref{eq:tree-contain-all-heavies} holds in the example depicted in Figure~\ref{fig:robust-promise}.
Moreover, the estimated value of every frequency that is recovered so far, is accurate up to an average error of $\frac{\mu}{\sqrt{20b}}$. More precisely, in every iteration of the algorithm the following property is maintained,
 \begin{equation}\label{eq:estimates-are-fine}
\frac{\sum_{\ff \in \supp{(\wh \chi_v)}}| \wh y(\ff) - \wh \chi_v(\ff) |^2}{|\supp{(\wh \chi_v)}|} \le \frac{\mu^2}{20b}.
 \end{equation}

At the start of the procedure, subtree $T$ is initialized to be the leaf $v$, i.e., $T=\{v\}$. Moreover, we initialize $\wh\chi_v \equiv 0$. Trivially, these initial values satisfy \eqref{eq:leaves-non-empty}, \eqref{eq:tree-contain-all-heavies}, and \eqref{eq:estimates-are-fine}. 
The algorithm also keeps a subset of leaves denoted by $\sident$ that contains the leaves of $T$ that are fully identified, that is the set of leaves that are at the bottom level and hence there is no ambiguity in their frequency content. Initially $\sident$ is empty. We show the set of marked leaves in Figure~\ref{fig:robust-promise} using blue squares.
The algorithm operates by picking the \emph{unmarked} leaf of $T$ that has the smallest weight. Then the algorithm explores the children of this node by running \textsc{HeavyTest} on them to detect if any heavy frequencies lie in their frequency cone. If a child passes the \textsc{HeavyTest} the algorithm updates tree $T$ by adding that child to $T$. As soon as a leaf of $T$ gets to the bottom level and becomes a leaf of $\tfull_N$, the algorithm marks it, i.e., adds that leaf to the $\sident$ set. It can be seen in Figure~\ref{fig:robust-promise} that all marked leaves are at the bottom level of the tree. The marked leaves need not be explored  any further because they are at the bottom level and their frequency content is fully identified.
These operations ensure that the invariants~\eqref{eq:leaves-non-empty}, \eqref{eq:tree-contain-all-heavies}, and \eqref{eq:estimates-are-fine} are maintained. 

Once the size of set $\sident$ grows sufficiently, the algorithm estimates the values of the marked frequencies. More precisely, at some point, the size of $\sident$ will be comparable to the maximum weight of the leaves it contains, and when this happens, the values of all marked frequencies can be estimated cheaply. Hence, when $\sident$ is a cheap to estimate set of leaves, our algorithm esimates those frequencies in a batch up to an average error of $\frac{\mu}{20b}$, updates $\wh \chi_v$ accordingly and removes all estimated ($\sident$) leaves from $T$. This ensures that invariants~\eqref{eq:leaves-non-empty}, \eqref{eq:tree-contain-all-heavies}, and \eqref{eq:estimates-are-fine} are maintained. The estimated leaves are illustrated in Figure~\ref{fig:robust-promise} using red thin dashed subtrees. We also demontrate the subtrees of $T$ that contain $\head$ element and are yet to be explored by our algorithm using gray cones and dashed edges in Figure~\ref{fig:robust-promise}. The gray cone means that there are heavy elements in that frequency cone that need to be identified as that node has not reached the bottom level yet.

Finally, the algorithm keeps tabs on the runtime it spends and ensures that even if the input signal does not satisfy the preconditions for successful recovery, in particular if $|\head \cap \subtree_{\Path}(v)| > b$, the runtime stays bounded. Additionally, the algorithm performs a quality control by running a \textsc{HeavyTest} on the residual and if the recovered signal is not correct due to violation of some preconditions, it reflects this in its output.

\begin{algorithm}[!t]
	\caption{The Inner Loop of Sparse FFT Algorithm}\label{alg:ksparsefft-inner}
	\begin{algorithmic}[1]
	\Procedure{RobustPromiseSFT}{$x, \wh \chi_{in}, \Path, v, b, k, \mu$} 
	\State \Comment{$\mu$: upper bound on tail norm $\|\eta\|_2$}	
	\State $\wh \chi_{out} \gets \{ 0 \}^{n^d}$	
	\Comment{Sparse vector to approximate $(\wh{x}- \wh{\chi}_{in})_{\subtree_\Path(v)}$} 
	
	\State{$\sident \gets \emptyset$} 
	\Comment{Set of marked nodes to be estimated later}

	\State{Let $T$ denote the subtree of $\Path$ rooted at $v$ -- i.e., $T \gets \{v\}$}
	\Repeat
		\If {$|\leaves(T)| + \|\wh \chi_v\|_0 > b$} \label{a9l6}
		\State \textbf{return} $\left( \false, \{0\}^{n^d} \right)$ \label{a9l7} \Comment{Exit because budget of $v$ is wrong}
		\EndIf
		
		\If {$\sident \neq \emptyset$ and $\frac{ \left|\sident\right|}{ \max_{u\in \sident} 2^{w_{T}(u)} } \geq \frac{1}{4+2\log b}$ }  \label{a9l8}
			\State \Comment{The set of marked frequencies that are cheap to estimate on average}
			\State $\left\{ \wh H_u \right\}_{u \in \sident} \gets \textsc{Estimate}\left(x, \wh \chi_{in} + \wh \chi_{out}, \Path \cup T, \sident, \frac{368 b}{ |\sident|} \right)$ \label{a9l10}
			\For {$u \in \sident$}
				\State $\wh \chi_{out}(\ff_u) \gets \wh H_{u}$ 
				\State Remove node $u$ from $T$
			\EndFor
			\State $\sident \gets \emptyset$
			
		\State \textbf{continue}
		\EndIf	
		
		\State $z \gets \argmin_{u \in \leaves(T)\setminus \sident } w_{T}(u)$ \label{a9l16} 
		\Comment{Find the minimum weight unmarked leaf in $T$}
		
		\If{$z \in \leaves(\tfull_N)$} \label{a9l17}
		\Comment{Frequency $\ff_z$ and leaf $z$ are fully identified}
		
		\State $\sident \gets \sident \cup \{ z \}$\label{a9l18}
		
		\Else
		
		\State $z_\lef :=$ left child of $z$ and $z_\righ :=$ right child of $z$
		
		\State $T'\gets T \cup \left\{ z_\lef, z_\righ \right\}$ \label{a9l20} \Comment{Explore children of $z$}
		
		\State{$\text{Heavy}_\ell \gets \textsc{HeavyTest}\left(x, \wh\chi_{in} + \wh \chi_{v} , \Path \cup T', z_\lef , O(b\log^3 N), 6 \mu^2 \right)$} \label{a9l21}
		\State{$\text{Heavy}_r \gets\textsc{HeavyTest}\left(x, \wh\chi_{in} + \wh \chi_{v} , \Path \cup T', z_\righ , O(b\log^3 N), 6 \mu^2 \right)$} \label{a9l22}
		
		\If{$\text{Heavy}_\ell$} 
		
		\State $\text{Add } z_\lef \text{ as the left child of } z \text{ to tree } T$
		
		\EndIf
		
		\If{$\text{Heavy}_r$}
		\State $\text{Add } z_\righ \text{ as the right child of } z \text{ to tree } T$
		
		\EndIf
		
		\If{$z\neq v$ and both $\text{Heavy}_\ell$ and $\text{Heavy}_r $ are $\false$} \label{a9l27}
		\State \textbf{return} $\left( \false, \{0\}^{n^d} \right)$ \Comment{Exit because budget of $v$ is wrong}
		\EndIf
		
		\EndIf 
		
		\Until{$T$ has no leaves besides $v$}\label{a9l29}
	
	\If{$\textsc{HeavyTest}\left(x, \wh\chi_{in} + \wh \chi_{v} , \Path, v , O(k\log^3 N), 6 \mu^2 \right)$} \label{a9l30}
		\State \Comment{The number of heavy coordinates in $\subtree_{\Path}(v)$ is more than $b$}
		\State \textbf{return} $\left( \false, \{0\}^{n^d} \right)$
	\Else
		\State \textbf{return} $\left( \true, \wh \chi_{out} \right)$ 
	\EndIf

	\EndProcedure
	\end{algorithmic}
	
\end{algorithm}

\paragraph{Overview of Algorithm~\ref{alg:sfftrobust}:} Consider an invocation of \textsc{RobustSparseFT}$(x, k, \epsilon, \mu)$. Suppose that $\wh x$ is a signal in the $k$-high SNR regime, i.e., $\wh x$ has $k$ heavy frequencies and the value of each such heavy frequency is at least $3$ times higher than the tail's norm. More formally, let $\head \subseteq[n]^d$ denote the set of heavy (head) frequencies of $\wh x$ and suppose that $|\head| \le k$, and the tail norm of $\wh x$ satisfies $\| \wh x - \wh x_{\head}\|_2 \le \mu$ and additionally suppose that $|\wh x(\ff)| \ge 3\mu$ for every $\ff \in \head$. The primitive \textsc{RobustSparseFT} finds a signal $\wh \chi$ such that $\|\wh x - \wh \chi\|_2^2 \le (1+\epsilon)\mu^2$.

Algorithm~\ref{alg:sfftrobust} recovers heavy frequencies of the input signal $\wh x$, i.e., $\head$, by iteratively exploring the tree that captures the heavy frequencies, which we denote by $\frontier$, and simultaneously updating the proxy signal $\wh \chi$. At the begining of the procedure, tree $\frontier$ only consists of a root and will be dynamically changing throughout the execution of our algorithm. Moreover, $\wh \chi$ is initially zero.
The algorithm also maintains a subset of leaves denoted by $\sident$ that contains the leaves of $\frontier$ that are fully identified, that is the set of leaves that are at the bottom level and hence there is no ambiguity in their frequency content (there is exactly one element in frequency cone of marked leaves).
Tree $\frontier$, in all iterations of our algorithm, maintains the invariant that the frequency cone of each of its leaves contain at least one head element and furthermore the frequency cone of each of its \emph{unmarked} leaves contain at least $b+1$ head element, where $b=k^{1/3}$, i.e.,
\begin{equation}\label{invariant:freq-cone-load}
|\subtree_{\frontier}(v) \cap \head| \ge \begin{cases}
1 & \text{ for every }v \in \sident\\
b+1 & \text{ for every }v \in \leaves(\frontier)\setminus \sident
\end{cases}.
\end{equation}
Additionally, at every iteration of the algorithm, the union of all frequency cones of tree $\frontier$ captures all heavy frequencies that are not recovered yet, i.e., 
\begin{equation}\label{invariant:frontier-captures-support}
	\head \setminus \supp{(\wh \chi)} \subseteq \supp{(\frontier)}.
\end{equation}
The set of fully recovered leaves (frequencies), i.e., $\supp{(\wh \chi_v)}$, are subtracted from the residual signal $\wh x - \wh \chi$ by our algorithm and their corresponding leaves get removed from $\frontier$, as well.
Moreover, the estimated value of every frequency that is recovered so far, is accurate up to an average error of $\sqrt{\frac{\epsilon}{k}} \cdot \mu$. More precisely, in every iteration of the algorithm the following property is maintained,
\begin{equation}\label{invariant:estimation-error}
\frac{\sum_{\ff \in \supp{(\wh \chi)}}| \wh x(\ff) - \wh \chi(\ff) |^2}{|\supp{(\wh \chi)}|} \le \frac{\epsilon}{k} \cdot \mu^2.
\end{equation}

At the start of the procedure, $\frontier$ is initialized to only contain a root, i.e., $\frontier=\{\text{root}\}$. Moreover, we initialize $\wh\chi \equiv 0$. Trivially, these initial values satisfy \eqref{invariant:freq-cone-load}, \eqref{invariant:frontier-captures-support}, and \eqref{invariant:estimation-error}. 
Also the set of fully identified leaves $\sident$ is initially empty.
The algorithm explores $\frontier$ by picking the \emph{unmarked} leaf that has the smallest weight, let us call it $v$. Then the algorithm explores the children of this node by running \textsc{RobustPromiseSFT} on them to recover the heavy frequencies that lie in their frequency cone. We denote by $v_\lef$ and $v_\righ$ the left and right children of $v$. Let us consider exploration of the left child $v_\lef$, the right child is exactly the same. If the number of heavy frequencies in the frequency cone of $v_\lef$ is bounded by $b=k^{1/3}$, i.e., $|\head \cap \subtree_{\frontier\cup \{v_\lef,v_\righ\}}(v_\lef)| \le b$, then \textsc{RobustPromiseSFT} recovers every frequency in the set $\head \cap \subtree_{\frontier\cup \{v_\lef,v_\righ\}}(v_\lef)$ up to average error $\frac{\mu}{\sqrt{20b}}$. Note that this everage estimation error is not sufficient for achieving the invariant \eqref{invariant:estimation-error}, hence, instead of directly using the values that \textsc{RobustPromiseSFT} recovered and update $\wh \chi$ at the newly recovered heavy frequencies, our algorithm adds the leaves corresponding to the recovered set of frequencies, i.e., $\head \cap \subtree_{\frontier\cup \{v_\lef,v_\righ\}}(v_\lef)$, at the bottom level of $\frontier$ and marks them as fully identified (adds them to $\sident$). For achieving maximum efficinecy we employ a new \emph{lazy estimation} scheme, that is, the estimation of values of marked leaves is delayed until there is a large number of marked leaves and thus there exists a subset of them that is cheap to estimate. On the other hand, if the number of head elements in frequency cone of $v_\lef$ is more than $b$ then \textsc{RobustPromiseSFT} detects this and notifies our algorithms about it and our algorithm adds node $v_\lef$ to $\frontier$. These operations ensure that the invariants~\eqref{invariant:freq-cone-load}, \eqref{invariant:frontier-captures-support}, and \eqref{invariant:estimation-error} are maintained. 

Once the size of set $\sident$ grows sufficiently such that it contains a subset that is cheap to estimate, our algorithm estimates the values of the cheap frequencies. More precisely, at some point, $\sident$ will contains a non-empty subset $\scheap$ such that the values of all frequencies in $\scheap$ can be estimated cheaply and subsequently, our algorithm esimates those frequencies in a batch up to an average error of $\sqrt{\frac{\epsilon}{k}} \cdot \mu$, updates $\wh \chi$ accordingly and removes all estimated ($\scheap$) leaves from $\frontier$ and $\sident$. This ensures that invariants~\eqref{invariant:freq-cone-load}, \eqref{invariant:frontier-captures-support}, and \eqref{invariant:estimation-error} are maintained.

\begin{algorithm}[!t]
	\caption{Robust High-dimensional Sparse FFT Algorithm}\label{alg:sfftrobust}
	\begin{algorithmic}[1]
		\Procedure{RobustSparseFT}{$x, k, \epsilon, \mu$} 
		\Comment{$\mu$ is an upper bound on tail norm $\|\eta\|_2$}
		
		\State $\frontier \gets \{\text{root}\}$, $\ff_{\text{root}}\gets 0$
		\State $b \gets \lceil k^{1/3} \rceil$
		\State $\wh \chi \gets \{ 0 \}^{n^d}$
		\State $\sident \gets \emptyset$ \Comment{Set of fully identified leaves (frequencies)}

		\Repeat
			
			\If {$\sum_{u \in \sident} 2^{-w_\frontier(u)} \ge \frac{1}{2}$} \label{a10l7}
				\State{$\scheap \gets \textsc{ExtractCheapSubset}\left( \frontier, \sident \right)$} \label{a10l8}
				\State \Comment{Lazy estimation: We extract from the batch of marked leaves a subset that is cheap to estimate on average}
				\State $\left\{ \wh H_u \right\}_{u \in \scheap} \gets \textsc{Estimate}\left(x, \wh \chi , \frontier, \scheap, \frac{32 k}{\epsilon \cdot |\scheap|} \right)$ \label{a10l10}
				\For{$u \in \scheap$}
					\State{$\wh \chi(\ff_u) \gets \wh H_{u}$}
					\State Remove node $u$ from tree $\frontier$
				\EndFor
				\State $\sident \gets \sident \setminus \scheap$
				\State \textbf{continue}
			\EndIf
			
			\State $v \gets \argmin_{u\in \leaves(\frontier) \setminus \sident} w_\frontier(u)$ \label{a10l16}
			\State \Comment{pick the minimum weight leaf in $\frontier$ which is not in $\sident$}
			
			\State $v_\lef \leftarrow$ left child of $v$ and $v_\righ \leftarrow$ right child of $v$	
			\State $T \gets \frontier \cup \{v_\lef, v_\righ \}$ \label{a10l17}
			\State $(\correct_\lef, \wh{\chi}_\lef  )  \gets \textsc{RobustPromiseSFT} \left(x, \wh \chi, T, v_\lef, b, k, \mu\right)$ \label{a10l18}
			\State $(\correct_\righ, \wh{\chi}_\righ ) \gets \textsc{RobustPromiseSFT} \left(x, \wh \chi, T, v_\righ, b, k, \mu\right)$ \label{a10l19}

			\If{$\correct_\lef$} \label{a10l20} 
				\State $\forall \bm f \in \supp(\wh{\chi}_\lef)$, add the unique leaf corresponding to $\bm f$ to $\frontier$ and $\marked$
			\Else 
				\State Add $v_\lef$ to $\frontier$
			\EndIf	
			\If{$\correct_\righ$} 
			\small
				\State $\forall \bm f \in \supp(\wh{\chi}_\righ)$, add the unique leaf corresponding to $\bm f$ to $\frontier$ and $\marked$
			\normalsize
			\Else 
				\State Add $v_\righ$ to $\frontier$
			\EndIf

			\If{$\correct_\lef$ and $\correct_\righ$}
				\State Remove $v$ from $\frontier$
			\EndIf

		\Until{$\frontier$ has no leaves besides root} \label{a10l30}

		\State \textbf{return} $\wh{\chi}$
		
		\EndProcedure
		
		\Procedure{ExtractCheapSubset}{$T,S$}
		\State $L \gets \emptyset$
		\While{$|L| \cdot \left( 8 + 4 \log |S| \right) < {\max_{v\in L} 2^{w_T(v)}}$}
			\State $L \gets L \cup \left\{\argmin_{\substack{u \in S \setminus L}} w_{{T}}(u)\right\}$
		\EndWhile
		\State Return $L$
		\EndProcedure

	\end{algorithmic}
	
\end{algorithm}

\paragraph{Analysis of \textsc{RobustPromiseSFT}.}
First we analyze the runtime and sample complexity of primitive \textsc{RobustPromiseSFT} in the following lemma.

\begin{lemma}[\textsc{RobustPromiseSFT} -- Time and Sample Complexity]
	\label{promise_correctness-runtime}
	Consider an invocation of \textsc{RobustPromiseSFT} $(x, \wh \chi_{in}, \Path, v, b, \mu)$, where $\Path$ is a subtree of $\tfull_N$, $v$ is some leaf of $T$, $k$ and $b$ are integers with $k > b$, $\mu\ge0$, and $x,\wh \chi_{in} : [n]^d \to \C$. Then

\begin{itemize}
\item The running time of primitive is bounded by 
\[\widetilde{O}\left( \|\wh \chi_{in}\|_0 \cdot \left(b^2 + k \right) + bk + 2^{w_\Path(v)}\cdot \left(b^3 + k \right) \right).\]

\item The number of accesses it makes on $x$ is always bounded by 
\[\widetilde{O}\left(2^{w_\Path(v)}\cdot \left(b^3 + k\right) \right).\]
\end{itemize}
Furthermore, the output signal $\wh{\chi}_v$ always satisfies $\|\wh{\chi}_v\|_0 \le b$ and $\supp(\wh{\chi}_v) \subseteq \subtree_\Path(v)$.
\end{lemma}

\begin{proof}
	First we prove that Algorithm~\ref{alg:ksparsefft-inner} terminates after a bounded number of iterations. In order to bound the number of iterations of \textsc{RobustPromiseSFT}, we use a potential function argument. Let $\wh \chi_v^{(t)}$ denote the signal $\wh \chi_v$ at the end of iteration $t$ of the algorithm. Furthermore, let $T^{(t)}$ denote the subtree $T$ at the end of $t^{th}$ iteration. Additionally, let $\sident^{(t)}$ and $\idnt^{(t)}$ denote the set $\sident$ (defined in Algorithm~\ref{alg:ksparsefft-inner}) at the end of iteration $t$.  

We prove that the algorithm always terminates after $O\left( b \cdot {\log N} \right)$ iterations. We prove this by contradiction. For any integer $t$, define the following potential function 
\[\phi_t := {\left|\sident^{(t)} \right|} + 2\log N \cdot \left\| \wh \chi_v^{(t)} \right\|_0 + \sum_{u\in \leaves\left(T^{(t)}\right)} l_{T^{(t)}}(u).\]
Towards contradiction, suppose that Algorithm~\ref{alg:ksparsefft-inner} does not terminate after ${4b \log  N}$ iterations. We show that the above potential function increases by at least $1$ at every iteration $2 \le t \le {4 b \log  N}$, i.e., $\phi_t \ge \phi_{t-1} + 1$. This is enough to conclude the termination of the algorithm because the if-statement in line~\ref{a9l6} ensures that $\left|\sident^{(t)} \right| \le \left|\leaves\left(T^{(t)}\right) \right| \le b$ and also $\left\| \wh \chi_v^{(t)} \right\|_0 \le b$, thus, $\phi_t = O(b \log N)$ for any $t$, which proves that algorithm terminates after $O(b\log N)$ iterations.

At any given iteration $t$ of the algorithm, there are 3 possibilities that can happen. We show that if any of these possibilities happen, then the potential function $\phi_t$ increases by at least $1$.
\paragraph{Case 1 -- the if-statement in line~\ref{a9l8} of Algorithm~\ref{alg:ksparsefft-inner} is True.}
In this case, the algorithm constructs $T^{(t)}$ by removing all leaves that are in the set $\sident^{(t-1)}$ from tree $T^{(t-1)}$ and leaving the rest of the tree unchanged.
Furthermore, the algorithm sets $\sident^{(t)} \gets \emptyset$. By construction, the level of the leaves that are in $\sident^{(t-1)}$ is at most $\log N$, thus
\[\sum_{u\in \leaves\left(T^{(t)}\right)} l_{T^{(t)}}(u) \ge \sum_{u\in \leaves\left(T^{(t-1)}\right)} l_{T^{(t-1)}}(u) - \log N \cdot \left|\sident^{(t-1)} \right| \] 

Additionally, in this case, the algorithm computes $\{\wh H_u\}_{u \in \sident^{(t-1)}}$ by running the procedure \textsc{Estimate} in line~\ref{a9l10} and then updates $\wh \chi_v^{(t)}(\ff_u) \gets \wh H_u$ for every $u \in \sident^{(t-1)}$ and $\wh \chi_v^{(t)}(\bm \xi) = \wh \chi_v^{(t-1)}(\bm \xi)$ at every other frequency $\bm \xi$. Therefore, $\left\| \wh \chi_v^{(t)} \right\|_0 = \left\| \wh \chi_v^{(t)} \right\|_0 + \left|\sident^{(t-1)} \right|$. Also, ${\left|\sident^{(t)} \right|} = 0$. Hence,
\[\phi_t - \phi_{t-1} \ge (\log N-1) \cdot \left| \sident^{(t-1)} \right| \ge 1,\]
where the inequality above holds because the if-statement in line~\ref{a9l8} of the algorithm is $\true$, ensuring that $ \sident^{(t-1)} \neq \emptyset $.

\paragraph{Case 2 -- the if-statement in line~\ref{a9l8} is False and if-statement in line~\ref{a9l17} is True.}
In this case, in line~\ref{a9l18}, the algorithm updates $\sident$ by adding the leaf $z$ to this set, i.e., $\sident^{(t)} \gets \sident^{(t-1)} \cup \{z\}$. Additionally, tree $T$ and signal $\wh \chi_v$ stay unchanged, i.e., $\wh \chi_v^{(t)} = \wh \chi_v^{(t-1)} $ and $T^{(t)} = T^{(t-1)} $. Therefore, in this case, $\phi_{t+1} - \phi_t = 1$. 

\paragraph{Case 3 -- both if-statements in lines~\ref{a9l8} and \ref{a9l17} are False.} 
In this case, either the algorithm terminates by the if-statement in line~\ref{a9l27}, which is exactly what we have assumed towards a contradiction that did not happen, or $\sum_{u\in \leaves\left(T^{(t)}\right)} l_{T^{(t)}}(u) \ge \sum_{u\in \leaves\left(T^{(t-1)}\right)} l_{T^{(t-1)}}(u) + 1$, while ${\left|\sident^{(t)} \right|} = \left|\sident^{(t-1)} \right|$ and $\left\| \wh \chi_v^{(t)} \right\|_0 = \left\| \wh \chi_v^{(t-1)} \right\|_0$ (since we assumed $t \ge 2$ and hence $z\neq v$). Thus, $\phi_{t+1} - \phi_t \ge 1$.

So far we have showed that at every iteration, under the {\bf cases} {\bf 1}, {\bf 2}, and {\bf 3}, the potential function $\phi_t$ increases by at least one. Now we show that, at every iteration, exactly one of these three cases happens and hence the algorithm never stalls. For the sake of contradiction suppose that at iteration $t$, the algorithm stalls. For this to happen, we must have that all leaves of $T^{(t-1)}$ are in the set $\sident^{(t-1)}$. By the if-statement in line~\ref{a9l6} of Algorithm~\ref{alg:ksparsefft-inner}, we are guaranteed that $|\sident^{(t-1)}| \le b$. Therefore, by Lemma~\ref{lemma:kraft-ISCheap}, there must exist a subset $\emptyset \neq L \subset \sident^{(t-1)}$ such that ${|L|} \ge \frac{1}{4+ 2\log b} \cdot \max_{u \in L} 2^{w_{T^{(t-1)}}(u)}$. 
Hence, it follows from the way our algorithm explores the nodes of the tree in an increasing order of weights, that there must exist some $t'<t$ such that $\emptyset \neq \sident^{(t'-1)} \subseteq \sident^{(t-1)}$ such that the if-statement in line~\ref{a9l8} becomes $\true$ on $\sident^{(t'-1)}$. Therefore, {\bf case 1} must have happened at iteration $t'$, resulting in emptying the set of identified frequencies, i.e., $\sident^{(t')} \gets \emptyset$. This would have resulted in $\sident^{(t'-1)} \nsubseteq \sident^{(t-1)}$ which is the contradiction we wanted. Therefore the algorithm never stalls and always exactly one of {\bf case 1}, {\bf 2}, and {\bf 3} happen.

We proved that $\phi_t$ must increase by at least $1$ at every iteration. Since $\phi_1 \ge 0$ and we assumed that the algorithm did not terminate after $q = {4b \log  N}$ iterations, this potential will have a value of at least $4b \log N -1$:
\[ \phi_q \ge 4b \log N - 1, \text{ where } q = {4b \log  N}. \]
On the other hand, since the if-statement in line~\ref{a9l6} ensures that the number of leaves of $T^{(t)}$ is always bounded by $b - \left\|\wh \chi_v^{(t)} \right\|_0$, the sum $\sum_{u\in \leaves\left(T^{(t)}\right)} l_{T^{(t)}}(u)$ is always bounded by $\left(b - \left\|\wh \chi_v^{(t)} \right\|_0\right) \cdot \log N$. Also, the size of the set $\sident^{(t)}$, which is a subset of $\leaves(T^{(t)})$, is always bounded by $b - \left\|\wh \chi_v^{(t)} \right\|_0$. 
This means that we must have $\phi_q \le b \cdot (\log N + 1) + (\log N - 1) \cdot \left\|\wh \chi_v^{(q)} \right\|_0$. The if-statement in line~\ref{a9l6} also ensures that $\left\|\wh \chi_v^{(q)} \right\|_0 \le b$ which implies that $\phi_q \le 2 b \cdot \log N$ which contradicts $\phi_q \ge 4b \cdot \log N -1$. This proves that the number of iterations of the algorithm must be bounded by $O\left( b \cdot {\log N}\right)$, guaranteeing termination of \textsc{RobustSparseFT}.
The termination quarantee along with the way our algorithm constructs $\wh \chi_v$ and the if-staement in line~\ref{a9l6}, imply that the output signal $\wh \chi_v$ always satisfies $\|\wh{\chi}_v\|_0 \le b$ and $\supp(\wh{\chi}_v) \subseteq \subtree_\Path(v)$.
Now we bound the running time and sample complexity of the algorithm.

\paragraph{Sample Complexity and Runtime:} First recall that we proved $\left\|\wh \chi_{v}^{(t)} \right\|_0 \le b$ for every iteration $t$. Additionally, the weight of the node $z$ at every iteration of the algorithm is bounded by $w_{T^{(t)}}(z) \le \log (2b)$. To see this, note that if at some iteration $t$, the set of identified frequencies (or leaves) that our algorithm keeps, $\sident^{(t)}$, is such that there exists a leaf $u \in \sident^{(t)}$ with $w_{T^{(t)}}(u) > \log (2b)$, then by Lemma~\ref{lemma:kraft-ISCheap}, $\sident^{(t)}$ contains a non-empty subset that is cheap to estimate. Thus, at some iteration $t'<t$, where $\emptyset \neq \sident^{(t')} \subset \sident^{(t)}$ holds, it must have been the case that the if-statement in line~\ref{a9l8} became $\true$ on $\sident^{(t')}$. If this happened, our algorithm would have estimated $\sident^{(t')}$ at iteration $t'$ and so we would have $\sident^{(t')} \cap \sident^{(t)} = \emptyset$ which is a contradiction.

Given the above inequalities, by Lemma~\ref{lem:guarantee_heavy_test}, time and sample complexities of every invocation of \textsc{HeavyTest} in lines~\ref{a9l21} and \ref{a9l22} of Algorithm~\ref{alg:ksparsefft-inner} are bounded by $\widetilde{O}\left( \|\wh \chi_{in}\|_0 \cdot b + 2^{w_\Path(v)}\cdot b^2 \right)$ and $\widetilde{O}\left( 2^{w_\Path(v)}\cdot b^2 \right)$, respectively. Also, since $\|\wh \chi_v\|_0 \le b$, the runtime and sample complexity of the \textsc{HeavyTest} in line~\ref{a9l30} of the algorithm are bounded by $\widetilde{O}\left( \|\wh \chi_{in}\|_0 \cdot k + bk + 2^{w_\Path(v)}\cdot k \right)$ and $\widetilde{O}\left( 2^{w_\Path(v)}\cdot k \right)$, respectively.
Thus, total sample and time complexity of all invocations of \textsc{HeavyTest} throughout the execution of our algorithm are bounded by $\widetilde{O}\left( 2^{w_\Path(v)}\cdot (b^3 + k) \right)$ and $\widetilde{O}\left( \|\wh \chi_{in}\|_0 \cdot (b^2 + k) + bk + 2^{w_\Path(v)}\cdot (b^3 + k) \right)$, respectively

Additionally, by Lemma~\ref{est-inner-lem}, the sample and time complexity of every invocation of \textsc{Estimate} in line~\ref{a9l10} of our algorithm are bounded by $\widetilde{O}\left( \frac{b \cdot 2^{w_\Path(v)}}{ \left| \sident^{(t-1)} \right|} \cdot \sum_{u \in \sident^{(t-1)}} 2^{w_{T^{(t-1)}}(u)}\right)$ and $\widetilde{O}\left( \frac{b \cdot 2^{w_\Path(v)}}{ \left| \sident^{(t-1)} \right|} \cdot \sum_{u \in \sident^{(t-1)}} 2^{w_{T^{(t-1)}}(u)} + b \cdot \|\wh \chi\|_0 \right)$, respectively. Because we run \textsc{Estimate} only when the if-statement in line~\ref{a9l8} holds true, the runtime and sample complexity of \textsc{Estimate} can be further upper bounded by $\widetilde{O}\left( \left| \sident^{(t-1)} \right| \cdot b \cdot 2^{w_\Path(v)} + b \cdot \|\wh \chi\|_0 \right)$ and $\widetilde{O}\left( \left| \sident^{(t-1)} \right| \cdot b \cdot 2^{w_\Path(v)} \right)$, respectively. Using the fact that 
\[\sum_{t: \text{ if-statement in line~\ref{a9l8} is }\true }\left| \sident^{(t-1)} \right| = \left\|\wh \chi_{v} \right\|_0 \le b,\] 
the total runtime and sample complexity of all invocations of \textsc{Estimate} in all iterations can be upper bounded by $\widetilde{O}\left( 2^{w_\Path(v)} \cdot {b}^2 + {b}^2 \cdot \|\wh \chi\|_0 \right)$ and $\widetilde{O}\left( 2^{w_\Path(v)} \cdot {b}^2 \right)$, respectively.
Therefore, by adding up the above contributions we can upper bound the total runtime and sample complexity by $\widetilde{O}\left( \|\wh \chi_{in}\|_0 \cdot \left(b^2 + k \right) + bk + 2^{w_\Path(v)}\cdot \left(b^3 + k \right) \right)$ and $\widetilde{O}\left( 2^{w_\Path(v)}\cdot \left(b^3 + k\right) \right)$
which completes the proof of the lemma.

\end{proof}

We are now in a position to present the main invariant of primitive \textsc{RobustPromiseSFT}.

\begin{lemma}[\textsc{RobustPromiseSFT} - Invariants]
\label{promise_correctness-invariants}
Consider the preconditions of Lemma~\ref{promise_correctness-runtime}. Let $\wh y := \wh  x - \wh \chi_{in}$ and $S := \subtree_\Path(v) \cap \head_\mu(\wh y)$, where $\head_\mu(\cdot)$ is defined as per \eqref{def:head}. If i) $ \head_\mu(\wh y) \subseteq \supp{(\Path)}$, ii) $\| \wh y - \wh y_{\head_\mu(\wh y)} \|_2^2 \le \frac{11 \mu^2}{10}$, and iii) $\left| S \right| \leq k$, then with probability at least $1 - \frac{1}{N^4}$, the output $\left( \mathrm{Budget}, \wh{\chi}_v \right)$ of Algorithm~\ref{alg:ksparsefft-inner} satisfies the following,
\begin{enumerate}
\item If $\left| S \right| \leq b$ then $\mathrm{Budget} =\true$, $\supp{(\wh \chi_v)} \subseteq S$, and $\left\|\wh{y}_{S} - \wh{\chi}_v\right\|_2^2 \leq \frac{\mu^2}{20} $;
\item If $\left| S \right| > b$ then $\mathrm{Budget}  =\false$ and $\wh \chi_v \equiv \{0\}^{n^d}$.
\end{enumerate}	

\end{lemma}

\begin{proof}
 We first analyze the algorithm under the assumption that the primitives \textsc{HeavyTest} and \textsc{Estimate} are replaced with more powerful primitives that succeeds deterministically. Hence, we assume that \textsc{HeavyTest} correctly tests the ``heavy'' hypothesis on its input signal with probability $1$ and also \textsc{Estimate} achieves the estimation guarantee of Lemma~\ref{est-inner-lem} deterministrically. With these assumptions in place, we prove that the lemma holds deterministically (with probability 1).
	We then establish a coupling between this idealized execution and the actual execution of our algorithm, leading to our result.
	
We prove the first statement of lemma by induction on the \emph{Repeat-Until loop} of the algorithm. Let $\wh \chi_v^{(t)}$ denote the signal $\wh \chi_v$ at the end of iteration $t$ of the algorithm. Furthermore, let $T^{(t)}$ denote the subtree $T$ at the end of $t^{th}$ iteration. Additionally, let $\sident^{(t)}$ denote the set $\sident$ (defined in Algorithm~\ref{alg:ksparsefft-inner}) at the end of iteration $t$.  
We prove that if the precondition of statement 1 (that is $|S| \le b$) together with i, ii and iii hold, then at every iteration $t=0,1,2, \dots$ of Algorithm~\ref{alg:ksparsefft-inner}, the following properties are maintained,
\begin{description}
	\item[$P_1(t)$] $S \setminus \supp\left(\wh \chi_v^{(t)}\right)  \subseteq \supp{\left(T^{(t)}\right)} := \bigcup_{u\in \leaves\left(T^{(t)}\right)} \subtree_{\Path \cup T^{(t)}}(u)$;
	\item[$P_2(t)$] For every leaf $u \neq v$ of subtree $T^{(t)}$, $\head_\mu(\wh y)\cap \subtree_{\Path \cup T^{(t)}}(u) \neq \emptyset$; 
	\item[$P_3(t)$] $\left\|\wh{y}_{S^{(t)}}-\wh{\chi}_v^{(t)}  \right\|_2^2 \leq \frac{\left|S^{(t)} \right|}{20b} \cdot  \mu^2$, where $S^{(t)} := \supp\left(\wh \chi_v^{(t)} \right) $;
	\item[$P_4(t)$] $S^{(t)} \subseteq S$ and $S^{(t)}  \cap \left(\bigcup_{\substack{u \in \leaves\left(T^{(t)}\right) \\ u \neq v}} \subtree_{\Path \cup T^{(t)}}(u)\right) = \emptyset$;
\end{description} 

The {\bf base of induction} corresponds to the zeroth iteration ($t=0$), at which point $T^{(0)}=\{v\}$ is a subtree of $\Path$ that solely consists of node $v$. Moreover, $\wh \chi_v^{(0)} \equiv 0$. Thus, statement $P_1(0)$ trivially holds by definition of set $S$. The statement $P_2(0)$ holds since there exists no leaf $u \neq v$ in $T^{(0)}$. Statements $P_3(0)$ and $P_4(0)$ hold because of the fact that $\wh \chi_v^{(0)} \equiv 0$.

We now prove the {\bf inductive step} by assuming that the inductive hypothesis, $P(t-1)$ is satisfied for some iteration $t-1$ of Algorithm~\ref{alg:ksparsefft-inner}, and then proving that $P(t)$ holds. 
First, we remark that if inductive hypotheses $P_2(t-1)$ and $P_4(t-1)$ hold true, then by the precondition of statement 1 of the lemma (that is $|S| \le b$) the if-statement in line~\ref{a9l6} of Algorithm~\ref{alg:ksparsefft-inner} is $\false$ and hence lines~\ref{a9l6} and \ref{a9l7} of the algorithm can be ignored in our analysis. 
We proceed to prove the induction by considering the three cases that can happen in iteration $t$:

\paragraph{Case 1 -- the if-statement in line~\ref{a9l8} of Algorithm~\ref{alg:ksparsefft-inner} is True.} 
In this case, the algorithm computes $\{\wh H_u\}_{u \in \sident^{(t-1)}}$ by running the procedure \textsc{Estimate} in line~\ref{a9l10} and then updates $\wh \chi_v^{(t)}(\ff_u) \gets \wh H_u$ for every $u \in \sident^{(t-1)}$ and $\wh \chi_v^{(t)}(\bm \xi) = \wh \chi_v^{(t-1)}(\bm \xi)$ at every other frequency $\bm \xi$. Therefore, if we let $L:= \left\{ \ff_u: u \in \sident^{(t-1)} \right\}$, then $S^{(t)} \setminus S^{(t-1)} = L$, by inductive hypothesis $P_4(t-1)$. By $P_3(t-1)$ along with Lemma~\ref{est-inner-lem} (its deterministic version that succeeds with probability 1), we find that
\small
\begin{align}
\left\| (\wh \chi_v^{(t)} - \wh y)_{S^{(t)}}\right\|_2^2 &= \left\| (\wh \chi_v^{(t)} - \wh y)_{S^{(t-1)}}\right\|_2^2 + \left\| (\wh \chi_v^{(t)} - \wh y)_{S^{(t)}\setminus S^{(t-1)}}\right\|_2^2\nonumber\\
&=\left\| (\wh \chi_v^{(t-1)} - \wh y)_{S^{(t-1)}}\right\|_2^2 + \left\| (\wh \chi_v^{(t)} - \wh y)_{L}\right\|_2^2\nonumber\\
&\le \frac{\left|S^{(t-1)} \right|}{20b} \mu^2 + \frac{\left|L \right|}{23b} \sum_{\bm\xi \in [n]^d \setminus \supp{\left(\Path \cup T^{(t-1)}\right)} } \left| \left(\wh{y}-\wh{\chi}_v^{(t-1)}\right)({\bm{\xi}}) \right|^2.\label{estimate-error}
\end{align}
\normalsize
Now we bound the second term above,
\small
\begin{align}
&\sum_{\bm\xi \in [n]^d \setminus \supp{\left(\Path \cup T^{(t-1)}\right)} } \left| \left(\wh{y}-\wh{\chi}_v^{(t-1)}\right)({\bm{\xi}}) \right|^2\nonumber\\ 
&\qquad= \sum_{\bm\xi \in [n]^d \setminus \supp{(\Path)} } \left| \wh{y}({\bm{\xi}}) \right|^2 + \sum_{\bm\xi \in \subtree_{\Path}(v) \setminus \supp{\left(T^{(t-1)}\right)}} \left| \left(\wh{y}-\wh{\chi}_v^{(t-1)}\right)({\bm{\xi}}) \right|^2\nonumber\\
&\qquad = \sum_{\bm\xi \in [n]^d \setminus \supp{(\Path)} } \left| \wh{y}({\bm{\xi}}) \right|^2 \nonumber\\
&\qquad\qquad + \sum_{\bm\xi \in \subtree_{\Path}(v) \setminus \left(\supp{\left(T^{(t-1)}\right)} \cup S^{(t-1)}\right)} \left| \wh{y}({\bm{\xi}}) \right|^2 + \left\|\wh{y}_{S^{(t-1)}} -\wh{\chi}_v^{(t-1)} \right\|_2^2\nonumber\\
&\qquad = \sum_{\bm\xi \in [n]^d \setminus \left(\supp{\left(\Path \cup T^{(t-1)}\right)} \cup S^{(t-1)} \right) } \left| \wh{y}({\bm{\xi}}) \right|^2 + \left\|\wh{y}_{S^{(t-1)}} -\wh{\chi}_v^{(t-1)} \right\|_2^2 \nonumber\\
&\qquad\le \sum_{\bm\xi \in [n]^d \setminus \head_\mu(\wh y)} \left| \wh{y}({\bm{\xi}}) \right|^2 + \left\|\wh{y}_{S^{(t-1)}} -\wh{\chi}_v^{(t-1)} \right\|_2^2 \text{~~~~~~~~~(by $P_1(t-1)$, precondition i and definition of $S$)}\nonumber\\
&\qquad \le \frac{23}{20} \cdot \mu^2 \text{~~~~~~~~~~~~~~~~~~~~~~~~~~~~~~~~~~~~~~~~~~~~~~~~~(by $P_3(t-1)$ and $P_4(t-1)$ and precondition $|S| \le b$)}.\nonumber
\end{align}
\normalsize
Therefore, by plugging the above bound back to \eqref{estimate-error} we find that,
\[\left\| (\wh \chi_v^{(t)} - \wh y)_{S^{(t)}}\right\|_2^2 \le \frac{\left|S^{(t-1)} \right|}{20b} \cdot \mu^2 + \frac{\left|L \right|}{23b} \cdot \left( \frac{23}{20} \mu^2 \right) = \frac{\left|S^{(t)} \right|}{20b} \cdot \mu^2,
\]
which proves the inductive claim $P_3(t)$. Moreover, $P_2(t-1)$ implies that $L \subseteq S$. Thus, the fact $S^{(t)} = S^{(t-1)} \cup L$ together with inductive hypothesis $P_4(t-1)$ as well as the construction of $T^{(t)}$ ($T^{(t)}$ is constructed by removing leaves of $\sident^{(t-1)}$ from tree $T^{(t-1)}$), imply $P_4(t)$. The construction of $T^{(t)}$ together with the fact that $|\subtree_{\Path \cup T^{(t-1)}}(u)| = 1 $ for every $u \in \sident^{(t-1)}$ give $P_1(t)$ and $P_2(t)$.

We now consider the other two cases. Let $z\in \leaves\left(T^{(t-1)}\right)$ be the smallest weight leaf chosen by the algorithm in line~\ref{a9l16}. 

\paragraph{Case 2 -- the if-statement in line~\ref{a9l8} is False and if-statement in line~\ref{a9l17} is True.}
In this case, in line~\ref{a9l18}, the algorithm updates $\sident$ by adding the leaf $z$ to this set, i.e., $\sident^{(t)} \gets \sident^{(t-1)} \cup \{z\}$. Additionally, in this case the tree $T$ and signal $\wh \chi_v$ stay unchanged, i.e., $\wh \chi_v^{(t)} = \wh \chi_v^{(t-1)} $ and $T^{(t)} = T^{(t-1)} $. Therefore, $P_1(t)$, $P_2(t)$, $P_3(t)$, and $P_4(t)$ all trivially hold because of the inductive hypothesis $P(t-1)$. 

\paragraph{Case 3 -- both if-statements in lines~\ref{a9l8} and \ref{a9l17} are False.}
In this case, the algorithm constructs tree $T'$ by adding leaves $z_\righ$ and $z_\lef$ to tree $T^{(t-1)}$ as right and left children of $z$ in line~\ref{a9l20}. Then we compute $\text{Heavy}_\ell$ and $\text{Heavy}_r$ in lines~\ref{a9l21} and \ref{a9l22} by running the primitive \textsc{HeavyTest} with inputs $\left(x, \wh \chi_{v}^{(t-1)} + \wh\chi_{in} , \Path \cup T', z_\lef , O(b \log^3N), 6 \mu^2 \right)$ and $\left(x, \wh \chi_{v}^{(t-1)} + \wh\chi_{in} , \Path \cup T', z_\righ , O(b \log^3N), 6 \mu^2 \right)$, respectively. There are two possibilities that can happen to each of $\text{Heavy}_\ell$ and $\text{Heavy}_r$. In the following we focus on analyzing $\text{Heavy}_\ell$, but $\text{Heavy}_r$ can be analyzed exactly the same way.

{\bf Possibility 1)} $\subtree_{\Path \cup T'}(z_\lef) \cap \head_\mu(\wh y) = \emptyset$. Note that, by construction of $T'$ we have
\small
\[\subtree_{\Path \cup T^{(t-1)}}(z) = \subtree_{\Path \cup T'}(z_\lef) \cup \subtree_{\Path \cup T'}(z_\righ).\]
\normalsize
Hence, by inductive hypothesis $P_4(t-1)$ we have,
\small
\begin{align*}
&\sum_{\bm\xi \in [n]^d \setminus \supp{(\Path \cup T')} } \left| \left(\wh{y}-\wh{\chi}_v^{(t-1)}\right)({\bm{\xi}}) \right|^2\\ 
&\qquad= \sum_{\bm\xi \in [n]^d \setminus \supp{(\Path)}} \left| \wh{y}({\bm{\xi}}) \right|^2\\ 
&\qquad\qquad+ \sum_{\bm\xi \in \subtree_{\Path}(v) \setminus \supp{\left( T^{(t-1)} \right)} } \left| \left(\wh{y}-\wh{\chi}_v^{(t-1)}\right)({\bm{\xi}}) \right|^2\\
&\qquad = \sum_{\bm\xi \in [n]^d \setminus \supp{(\Path)} } \left| \wh{y}({\bm{\xi}}) \right|^2 \\
&\qquad\qquad + \sum_{\bm\xi \in \subtree_{\Path}(v) \setminus \left(\supp{\left( T^{(t-1)} \right)} \cup S^{(t-1)}\right)} \left| \wh{y}({\bm{\xi}}) \right|^2 + \left\|\wh{y}_{S^{(t-1)}} -\wh{\chi}_v^{(t-1)} \right\|_2^2\\
&\qquad \le \sum_{\bm\xi \in [n]^d \setminus \left(\supp{\left(\Path \cup T' \right)} \cup S^{(t-1)}\right)} \left| \wh{y}({\bm{\xi}}) \right|^2 + \frac{\mu^2}{20},
\end{align*}
\normalsize
where the last inequality above follows by inductive hypotheses $P_3(t-1)$ and $P_4(t-1)$ and precondition $|S| \le b$.
Therefore, if $\wh G_\ell$ is a $(z_\lef,\Path\cup T')$-isolating filter as per the construction in Lemma~\ref{lem:isolate-filter-highdim}, then by Corollary~\ref{infty-norm-bound-islating-filter} along with the above inequality, we have
\small
\begin{align*}
\left\| \left(\wh y -\wh \chi_v^{(t-1)}\right) \cdot \wh G_\ell \right\|_2^2 &\le \left\| \wh y_{\subtree_{\Path \cup T'}(z_\lef)} \right\|_2^2 + \sum_{\bm\xi \in [n]^d \setminus \supp{(\Path \cup T')} } \left| \left(\wh{y}-\wh{\chi}_v^{(t-1)}\right)({\bm{\xi}}) \right|^2 \\
&\le \left\| \wh y_{\subtree_{\Path \cup T'}(z_\lef)} \right\|_2^2 + \sum_{\bm\xi \in [n]^d \setminus \left(\supp{(\Path \cup T')} \cup S^{(t-1)}\right)} \left| \wh{y}({\bm{\xi}}) \right|^2 + \frac{\mu^2}{20}\\
&\le \sum_{\bm\xi \in [n]^d \setminus \head_\mu(\wh y)} \left| \wh{y}({\bm{\xi}}) \right|^2 + \frac{\mu^2}{20}\\
&\le \frac{23}{20} \cdot \mu^2
\end{align*}
\normalsize
where the third line above follows from the assumption that $\subtree_{\Path \cup T'}(z_\lef) \cap \head_\mu(\wh y) = \emptyset$, inductive hypothesis $P_1(t-1)$, precondition i of the lemma together with the definition of set $S$. This proves that the precondition of the second claim of Lemma~\ref{lem:guarantee_heavy_test} holds and therefore by invoking this lemma (the deterministic version of it that succeeds with probability 1), we have that $\text{Heavy}_\ell$ in line~\ref{a9l21} of the algorithm is $\false$. 
Using a similar argument, if $\subtree_{\Path \cup T'}(z_\righ) \cap \head_\mu(\wh y)= \emptyset$, then $\text{Heavy}_r$ is $\false$. 

{\bf Possibility 2)} Suppose that $\subtree_{\Path \cup T'}(z_\lef) \cap \head_\mu(\wh y) \neq \emptyset$. If filter $\wh G_\ell$ is a $(z_\lef,\Path \cup T')$-isolating filter constructed in Lemma~\ref{lem:isolate-filter-highdim}, then by Corollary~\ref{infty-norm-bound-islating-filter} along with inductive hypothesis $P_4(t-1)$, 
\small
\begin{align*}
\left\| \left(\left( \wh y -\wh \chi_v^{(t-1)} \right) \cdot \wh G_\ell\right)_{[n]^d \setminus S} \right\|_2^2 &= \left\| \left(\wh y \cdot \wh G_\ell\right)_{[n]^d \setminus S} \right\|_2^2\\
&\le \left\| \wh y_{\subtree_{\Path \cup T'}(z_\lef) \setminus S} \right\|_2^2 + \sum_{\bm\xi \in [n]^d \setminus \left(\supp{(\Path \cup T')} \cup S\right) } \left| \wh{y}({\bm{\xi}}) \right|^2\\
&\le \left\| \wh y - \wh y _{\head_\mu(\wh y)} \right\|_2^2 \le \frac{11}{10} \cdot \mu^2. \text{~~~~~~~~(precondition ii)}
\end{align*}
\normalsize
Additionally,
\begin{align*}
\left\| \left(\left( \wh y -\wh \chi_v^{(t-1)} \right) \cdot \wh G_\ell\right)_{S} \right\|_2^2 &\ge \left\| \left( \wh y -\wh \chi_v^{(t-1)} \right)_{\subtree_{\Path \cup T'}(z_\lef) \cap S} \right\|_2^2\\ 
&= \left\| \wh y_{\subtree_{\Path \cup T'}(z_\lef) \cap S} \right\|_2^2 \ge 9\mu^2,
\end{align*}
which follows by the assumption $\subtree_{\Path \cup T'}(z_\lef) \cap \head_\mu(\wh y)\neq \emptyset$ along with the definition of $S$ and $\head_\mu(\cdot)$.
Hence, by the above inequalities and the precondition $|S| \le b$, we can invoke Lemma~\ref{lem:guarantee_heavy_test} to conclude that $\text{Heavy}_\ell$ in line~\ref{a9l21} of the algorithm is $\true$.
Using a similar argument, if $\subtree_{\Path \cup T'}(z_\righ) \cap \head_\mu(\wh y) \neq \emptyset$ then $\text{Heavy}_r$ is $\true$. 

Based on the above arguments, according to the values of $\text{Heavy}_\ell$ and $\text{Heavy}_r$, there are various cases that can happen. First, it cannot happen that $\text{Heavy}_\ell$ and $\text{Heavy}_r$ are both $\false$ unless $z = v$, by the inductive hypothesis $P(t-1)$. If $\text{Heavy}_\ell = \text{Heavy}_r = \false$ and $z = v$, the algorithm returns $\wh \chi_v^{(t)} \equiv \{0\}^{n^d}$ which satisfies all properties in $P(t)$. The second case corresponds to $\text{Heavy}_\ell =\false$ and $\text{Heavy}_r= \true$. In this case, tree $T^{(t)}$ is obtained from $T^{(t-1)}$ by adding $z_\righ$ as the right child of $z$. Therefore, by inductive hypothesis $P(t-1)$, all properties in $P(t)$ immediately hold. 
One can show that $P(t)$ holds in the case of $\text{Heavy}_\ell =\true$ and $\text{Heavy}_r=\false$ in exactly the same fashion. Finally, if both of $\text{Heavy}_\ell $ and $\text{Heavy}_r$ are $\true$, then tree $T^{(t)}$ is obtained by adding leaves $z_\righ$ and $z_\lef$ as right and left children of $z$ to tree $T^{(t-1)}$. It follows straightforwardly from the inductive hypothesis $P(t-1)$ that $P(t)$ holds. 

So far we have showed that under {\bf cases} {\bf 1}, {\bf 2}, and {\bf 3}, the property $P(t)$ is maintained. Recall that in the proof of Lemma~\ref{promise_correctness-runtime} we showed that, at every iteration, exactly one of these three cases happen and hence the algorithm never stalls.
This completess the induction and proves that properties $P(t)$ are maintained throughout the execution of Algorithm~\ref{alg:ksparsefft-inner}, assuming that preconditions i, ii, and iii of the lemma along with the precondition $|S| \le b$ hold. 

In Lemma~\ref{promise_correctness-runtime} we showed that Algorithm~\ref{alg:ksparsefft-inner} must terminate after some $q$ iterations. When the algorithm terminates, the condition of the \emph{Repeat-Until} loop in line~\ref{a9l29} of the algorithm must be True. Thus, when the algorithm terminates, at $q^{th}$ iteration, there is no leaf in subtree $T_v^{(q)}$ besides $v$ and as a consequence the set $\sident^{(q)}$ must be empty. This, together with $P_1(q)$ imply that the signal $\wh \chi_v^{(q)}$ satisfies,
\[ \supp{\left(\wh{\chi}_v^{(q)}\right)} = S = \subtree_\Path(v) \cap \head_\mu(\wh y) .\] 
Moreover, $P_3(q)$ together with precondition $|S|\le b$ imply that 
$$\left\|\wh{y}_S - \wh{\chi}_v^{(q)}\right\|_2^2 \leq  \frac{\left|S \right|}{20b} \cdot \mu^2 \le \frac{\mu^2}{20}.$$

Now we analyze the if-statement in line~\ref{a9l30} of the algorithm. The above equalities and inequalities on $\wh \chi_v^{(q)}$ imply that,
\begin{align*}
\left\| \left(\wh y - \wh \chi_v^{(q)} \right)_{\subtree_{\Path}(v)}\right\|_2^2 &= \left\| \wh y_{\subtree_{\Path}(v) \setminus S}\right\|_2^2 + \left\| \left(\wh y - \wh \chi_v^{(q)} \right)_S\right\|_2^2\\
&\le \left\| \wh y_{\subtree_{\Path}(v) \setminus \head_\mu(\wh y)}\right\|_2^2 + \frac{\mu^2}{20}.
\end{align*}
Therefore, if $\wh G_v$ is a Fourier domain $(v,\Path)$-isolating filter constructed in Lemma~\ref{lem:isolate-filter-highdim}, then by Corollary~\ref{infty-norm-bound-islating-filter} along with the above inequality, we have
\small
\begin{align*}
\left\| \left(\wh y - \wh \chi_v^{(q)} \right) \cdot \wh G_v \right\|_2^2 &\le \sum_{\bm\xi \in [n]^d \setminus \supp{(\Path)} } \left| \wh{y}({\bm{\xi}}) \right|^2 + \left\| \left(\wh y - \wh \chi_v^{(q)} \right)_{\subtree_{\Path}(v)}\right\|_2^2 \\ 
& \le \sum_{\bm\xi \in [n]^d \setminus \supp{(\Path)} } \left| \wh{y}({\bm{\xi}}) \right|^2 + \left\| \wh y_{\subtree_{\Path}(v) \setminus \head_\mu(\wh y)}\right\|_2^2 + \frac{\mu^2}{20} \\
& \le \left\| \wh y - \wh y_{\head_\mu(\wh y)}\right\|_2^2 + \frac{\mu^2}{20} \le \frac{23}{20}\cdot \mu^2.
\end{align*}
\normalsize
Thus, the preconditions of the second claim of Lemma~\ref{lem:guarantee_heavy_test} hold. So, we can invoke this lemma to conclude that the if-statement in line~\ref{a9l30} of the algorithm is $\false$ and hence the algorithm outputs $\left( \true , \wh \chi_v^{(q)} \right)$. This proves statement 1 of the lemma.

Now we prove the second statement of lemma. Suppose that preconditions i, ii, iii along with the precondition of statement 2 (that is $|S| > b$) hold. Lemma~\ref{promise_correctness-runtime} proved that the signal $ \wh \chi_v$ always satisfies $\supp{(\wh \chi_v)} \subseteq \subtree_{\Path}(v)$ and $\|\wh \chi_v\|_0 \le b$. Therefore, $S \setminus \supp{(\wh \chi_v)} \neq \emptyset$. Consequently, if $\wh G_v$ is a Fourier domain $(v,\Path)$-isolating filter constructed in Lemma~\ref{lem:isolate-filter-highdim}, then by definition of isolating filters we have
\begin{align*}
\left\| \left(\left( \wh y -\wh \chi_v \right) \cdot \wh G_v\right)_{S\cup \supp{(\wh \chi_v)}} \right\|_2^2 \ge \left\| \left( \wh y -\wh \chi_v \right)_{S\cup \supp{(\wh \chi_v)}} \right\|_2^2 \ge \left\| \wh y_{S \setminus \supp{(\wh \chi_v)}} \right\|_2^2 \ge 9\mu^2,
\end{align*}
which follows from the definition of $S$ and $\head_\mu(\cdot)$. On the other hand,
\begin{align*}
\left\| \left(\left( \wh y -\wh \chi_v\right) \cdot \wh G_v\right)_{[n]^d \setminus (S\cup \supp{(\wh \chi_v)})} \right\|_2^2 &= \left\| \left(\wh y \cdot \wh G_v\right)_{[n]^d \setminus (S\cup \supp{(\wh \chi_v)})} \right\|_2^2\\
&\le \left\| \left(\wh y \cdot \wh G_v\right)_{[n]^d \setminus S} \right\|_2^2\\
&\le \left\| \wh y_{\subtree_{\Path}(v) \setminus S} \right\|_2^2 + \sum_{\bm\xi \in [n]^d \setminus \supp{(\Path)} } \left| \wh{y}({\bm{\xi}}) \right|^2\\
&\le \left\| \wh y - \wh y _{\head_\mu(\wh y)} \right\|_2^2 \le \frac{11}{10} \cdot \mu^2. \text{~~~~~~~~~~~~~~~~(precondition ii)}
\end{align*}
Additionally note that $\left| S \cup \supp{(\wh \chi_v)} \right| \le k + b \le 2k$ by preconditions of the lemma and property of $\supp{(\wh \chi_v)}$ that we have proved. Hence, by invoking the first claim of Lemma~\ref{lem:guarantee_heavy_test}, the if-statement in line~\ref{a9l30} of the algorithm is $\true$ and hence the algorithm outputs $\left( \false , \{0\}^{n^d} \right)$. This proves statement 2 of the lemma.

Finally, observe that throughout this analysis we have assumed that Lemma~\ref{lem:guarantee_heavy_test} holds with probability 1 for all the invocations of \textsc{HeavyTest} by our algorithm. Moreover, we assumend that \textsc{Estimate} successfully works with probability 1. In reality, we have to take the fact that  these primitives are randomized into acount of our analysis. 

The first source of randomness is the fact that \textsc{HeavyTest} only succeeds with some high probability. In fact, Lemma~\ref{lem:guarantee_heavy_test} tells us that every invocation of \textsc{HeavyTest} succeeds with probability at least $1-1/N^5$.	
Our analysis in proof of Lemma~\ref{promise_correctness-runtime} shows that \textsc{RobustPromiseSFT} makes at most $O\left(b \log N \right)$ calls to \textsc{HeavyTest}. Therefore, by a union bound, the overall failure probability of all invocations of $\textsc{HeavyTest}$ is bounded by $O\left( \frac{b \log N}{N^{5}} \right)$. 

The second source of randomness is the fact that \textsc{Estimate} only succeeds with some high probability. Lemma~\ref{est-inner-lem} tells us that every invocation of \textsc{Estimate} on a set $\sident$, succeeds with probability $1-|\sident|/N^8$. Therefore if the algorithm invokes \textsc{Estimate} at iterations $t_1, t_2, \dots$, then, by union bound, the total failure probability of all invocations of this primitive will be bounded by $\sum_{i} \frac{\left|\sident^{(t_i)}\right|}{N^8} = \frac{\left| \supp{(\wh \chi_v}) \right|}{N^8} \le \frac{b}{N^8}$.

Finally, by another application of union bound, the overall failure probability of Algorithm~\ref{alg:ksparsefft-inner}, is bounded by $\frac{1}{N^4}$.
This proves that the lemma holds. 
\end{proof}

\paragraph{Analysis of \textsc{RobustSparseFT}.}
Now we present the invariants of \textsc{RobustSparseFT}.

\begin{lemma}[Invariant of \textsc{RobustSparseFT}: Signal Containment and Energy Control]
\label{lem:invariant} For every integer $t\ge 0$, let $\wh \chi^{(t)}$ and $\sident^{(t)}$ denote the signal $\wh \chi$ and the set $\sident$ at the end of iteration $t$ of Algorithm~\ref{alg:sfftrobust}, respectively. Furthermore, let $\frontier^{(t)}$ denote the tree $\frontier$ at the end of $t^{th}$ iteration and let $\Estimated^{(t)}$ denote the set of ``estimated frequencies'' so far, i.e., $\Estimated^{(t)} : = \supp{\left( \wh \chi^{(t)} \right)}$. Additionaly, for every leaf $v$ of $\frontier^{(t)}$, let $L_v^{(t)}$ denote the ``{unestimated}'' frequencies in support of $\wh x$ that lie in frequency cone of $v$, i.e., $L_v^{(t)} := \subtree_{\frontier^{(t)}}(v) \cap \head_\mu(\wh{x})$, where $\head_\mu(\cdot)$ is defined as per \eqref{def:head}.
If $|\head_\mu(\wh{x})| \le k$ and $ \left\| \wh x - \wh x_{\head_\mu(\wh{x})} \right\|_2 \le \mu $, then for every non-negative integer $t$ the following properties are maintained at the end of $t^{th}$ iteration of Algorithm~\ref{alg:sfftrobust}, with probability at least $1 - \frac{4 t}{N^4}$,
\begin{enumerate}
\item[$P_1(t)$] $\head_\mu(\wh{x}) \setminus \Estimated^{(t)}  \subseteq \supp{\left(\frontier^{(t)}\right)}$;
\item[$P_2(t)$] For every leaf $u \neq \mathrm{root}$ of tree $\frontier^{(t)}$, $\left| L_u^{(t)} \right| \ge 1$. Additionally, if $u \notin \sident^{(t)}$, then $\left| L_u^{(t)} \right| > {b}$;
\item[$P_3(t)$] $\left\|\wh{x}_{\Estimated^{(t)}} - \wh{\chi}^{(t)} \right\|_2^2 \leq \epsilon \cdot \frac{ \left| \Estimated^{(t)} \right|}{k} \cdot \mu^2$;
\item[$P_4(t)$] $\Estimated^{(t)} \subseteq \head_\mu(\wh{x})$ and $\Estimated^{(t)} \cap \supp{\left(\frontier^{(t)}\right)} = \emptyset$;
\item[$P_5(t)$] In every iteration $t>1$, if the if-statement in line~\ref{a10l7} of Algorithm~\ref{alg:sfftrobust} is $\false$, then the following potential function decreases by at least ${b}$. Additionally, when the if-statement in line~\ref{a10l7} is $\true$, the potential decreases by at least $\log N$. Furthermore, the potential does not increase at iteration $t=1$.
$$\phi_t := \sum_{u \in \leaves\left(\frontier^{(t)}\right)} \left( 2\log N - l_{\frontier^{(t)}}(u) \right) \cdot \left| L_u^{(t)}\right| ;$$
\end{enumerate} 
\end{lemma}

\begin{proof}
The proof is by induction on the \emph{Repeat-Until loop} of the algorithm. The {\bf base of induction} corresponds to the zeroth iteration ($t=0$), at which point $\frontier^{(0)}=\{\text{root}\}$ is a tree that solely consists of a root and has no other leaves. Moreover, $\wh \chi^{(0)} \equiv 0$. The statement $P_1(t)$ trivially holds because $\subtree_{\frontier^{(0)}}(r) = [n]^d$. The statement $P_2(t)$ holds since there exists no leaf $u \neq$ root in $\frontier^{(0)}$. The statements $P_3(t)$ and $P_4(t)$ hold because of the facts $\wh \chi^{(0)} \equiv 0$ and $\Estimated^{(0)} = \emptyset$.

We now prove the {\bf inductive step} by assuming that the inductive hypotheses, i.e property $P(t-1)$ is satisfied for some iteration $t-1$ of Algorithm~\ref{alg:sfftrobust} with probability a least $1 - \frac{4(t-1)}{N^4}$, and then proving that property $P(t)$ holds at the end of iteration $t$ with probabiliy at least $1 - \frac{4 t}{N^4}$. We also show that the value of the quantity $\phi_t$ defined in $P_5(t)$, satisfies $\phi_t - \phi_{t-1} \le - {b}$ if the if-statement in line~\ref{a10l7} of the algorithm is $\false$ in iteration $t>1$ and $\phi_t - \phi_{t-1} \le - \log N$ if the if-statement in line~\ref{a10l7} is $\true$ in iteration $t$ and also $\phi_1 - \phi_0 \le 0$.
At any given iteration $t$ of the algorithm, there are two possibilities that can happen. We proceed to prove the induction by considering any of the two possibilities:
\paragraph{Case 1 -- the if-statement in line~\ref{a10l7} of Algorithm~\ref{alg:sfftrobust} is True.} 
In this case, we have that $\sum_{u \in \sident^{(t-1)}} 2^{-w_{\frontier^{(t-1)}}(u)} \ge \frac{1}{2}$. As a result, by Claim~\ref{claim:findCheap}, the set $\scheap \subseteq \sident^{(t-1)}$ that the algorithm computes in line~\ref{a10l8} by running the primitive \textsc{ExtractCheapSubset} satisfies the property that $\left| \scheap \right| \cdot \left( 8+4\log |\sident^{(t-1)}| \right) \ge \max_{u \in \scheap} 2^{w_{\frontier^{(t-1)}}(u)}$. Clearly $\scheap \neq \emptyset$, by Claim~\ref{claim:findCheap}. 
Then the algorithm computes $\{\wh H_u\}_{u \in \scheap}$ by running the procedure \textsc{Estimate} in line~\ref{a10l10} and then updates $\wh \chi^{(t)}(\ff_u) \gets \wh H_u$ for every $u \in \scheap$ and $\wh \chi^{(t)}(\bm \xi) = \wh \chi^{(t-1)}(\bm \xi)$ at every other frequency $\bm \xi$.
Therefore, if we let $L:= \left\{ \ff_u: u \in \scheap \right\}$, then $\Estimated^{(t)} \setminus \Estimated^{(t-1)} = L$, by inductive hypothesis $P_4(t-1)$. By $P_3(t-1)$ along with Lemma~\ref{est-inner-lem}, we find that with probability at least $1 - \frac{|\scheap|}{N^8} \ge 1 - \frac{1}{N^7}$ the following holds,
\begin{align}
\left\| \wh \chi^{(t)} - \wh x_{\Estimated^{(t)}}\right\|_2^2 
&= \left\| \wh \chi^{(t-1)} - \wh x_{\Estimated^{(t-1)}}\right\|_2^2 + \left\| \left( \wh \chi^{(t)} - \wh x \right)_{L}\right\|_2^2\nonumber\\
&\le \frac{\epsilon |\Estimated^{(t-1)} |}{k} \mu^2 + \frac{\epsilon \left|L \right|}{2k}  \sum_{\bm\xi \in [n]^d \setminus \supp{\left(\frontier^{(t-1)}\right)}} \left| \left( \wh \chi^{(t-1)} - \wh{x} \right)({\bm{\xi}}) \right|^2.\label{estimate-error-invariant-lemma}
\end{align}
Now we bound the second term above,
\small
\begin{align}
&\sum_{\bm\xi \in [n]^d \setminus \supp{\left(\frontier^{(t-1)}\right)}} \left| \left(\wh{x}-\wh{\chi}^{(t-1)}\right)({\bm{\xi}}) \right|^2\nonumber\\ 
&\qquad = \sum_{\bm\xi \in  [n]^d \setminus \left(\supp{\left(\frontier^{(t-1)}\right)} \cup \Estimated^{(t-1)}\right)} \left| \wh{x}({\bm{\xi}}) \right|^2 + \left\|\wh{x}_{\Estimated^{(t-1)}} -\wh{\chi}^{(t-1)} \right\|_2^2\nonumber\\
&\qquad\le \sum_{\bm\xi \in [n]^d \setminus \head_\mu(\wh{x})} \left| \wh{x}({\bm{\xi}}) \right|^2 + \left\|\wh{x}_{\Estimated^{(t-1)}} -\wh{\chi}^{(t-1)} \right\|_2^2 \text{~~~~~~~~~~~~~~~~~~~~~~~~~~~~~~~~~~~~~~~~~~~~(by $P_1(t-1)$)}\nonumber\\
&\qquad \le 2 \mu^2 \text{~~~~~~~~~~~~~~~~~~~~~~~~~~~~~~~~(by $P_3(t-1)$ and $P_4(t-1)$, preconditions of lemma and $\epsilon \le 1$)}.\nonumber
\end{align}
\normalsize
Therefore, by plugging the above bound back to \eqref{estimate-error-invariant-lemma} we find that,
\[\left\| \wh \chi^{(t)} - \wh x_{\Estimated^{(t)}}\right\|_2^2 \le \epsilon \cdot \frac{|\Estimated^{(t-1)} |}{k} \cdot \mu^2 + \epsilon \cdot \frac{\left|L \right|}{2k} \cdot \left( 2 \mu^2 \right) = \epsilon \cdot  \frac{|\Estimated^{(t)} |}{k} \cdot \mu^2,
\]
which proves the inductive claim $P_3(t)$. Moreover, $P_2(t-1)$ implies that $L \subseteq \head_\mu(\wh{x})$. Thus, the fact that $\Estimated^{(t)} = \Estimated^{(t-1)} \cup L$ together with inductive hypothesis $P_4(t-1)$ as well as the construction of $\frontier$ ($\frontier^{(t)}$ is constructed by removing leaves of $\scheap$ from tree $\frontier^{(t-1)}$), imply $P_4(t)$. The construction of $\frontier^{(t)}$ together with the fact that $|\subtree_{\frontier^{(t-1)}}(u)| = 1 $ for every $u \in \scheap$ give $P_1(t)$ and $P_2(t)$. Additionally, we have,
\begin{align*}
\phi_t - \phi_{t-1} &= - \sum_{u \in \scheap} (2\log N - l_{\frontier^{(t-1)}}(u)) \cdot \left| L_u^{(t-1)} \right|\\
&= - \sum_{u \in \scheap} \log N \cdot \left| L_u^{(t-1)} \right| \\
&= - \sum_{u \in \scheap} \log N \le - \log N,
\end{align*}
where the last inequality follows from the fact that $\scheap \neq \emptyset$. This proves $P_5(t)$.

\paragraph{Case 2 -- the if-statement in line~\ref{a10l7} is False.}
Let $v\in \leaves(\frontier^{(t-1)}) \setminus \sident^{(t-1)}$ be the smallest weight leaf chosen by the algorithm in line~\ref{a10l16}. 
The algorithm constructs tree $T$ by adding leaves $v_\righ$ and $v_\lef$ to tree $\frontier^{(t-1)}$ as right and left children of $v$, in line~\ref{a10l17}.
Then, the algorithm runs \textsc{RobustPromiseSFT} with inputs $(x,  \wh \chi^{(t-1)}, T, v_\lef, b, k, \mu )$ and $(x,  \wh \chi^{(t-1)}, T, v_\righ, b, k, \mu)$ in lines~\ref{a10l18} and \ref{a10l19} respectively.
In the following we focus on analyzing $\left(\correct_\lef, \wh{\chi}_\lef\right)$ but $\left(\correct_\righ, \wh{\chi}_\righ\right)$ can be analyzed exactly the same way. There are two possibilities that can happen:

{\bf Possibility 1)} $\left|\subtree_{T}(v_\lef) \cap \head_\mu(\wh{x}) \right| \le b$. In this case, the inductive hypothesis $P_4(t-1)$ implies that $|\Estimated^{(t-1)}| \le k$ and hence inductive hypothesis $P_3(t-1)$ along with the assumption $\epsilon \le \frac{1}{10}$ gives
\begin{equation}\label{eq:est-error-bound}
\left\| \wh x_{\Estimated^{(t-1)}} - \wh \chi^{(t-1)} \right\|_2^2 \le {\epsilon} \mu^2 \le \frac{\mu^2}{10},
\end{equation}
hence, $\head_\mu\left(\wh x - \wh \chi^{(t-1)}\right) = \head_\mu(\wh{x}) \setminus \Estimated^{(t-1)}$.
Consequently, if we let $\wh y : = \wh x - \wh \chi^{(t-1)}$, then: i) $\head_\mu(\wh{y}) \subseteq \supp{(T')}$, by $P_1(t-1)$, ii) $\| \wh y - \wh y_{\head_\mu(\wh{y})} \|_2^2 \le \frac{11\mu^2}{10}$, by precondition of the lemma along with \eqref{eq:est-error-bound}, and iii) $\left|\subtree_{T}(v_\lef) \cap \head_\mu(\wh{y}) \right| \le b$, by the assumption that $\left|\subtree_{T}(v_\lef) \cap \head_\mu(\wh{x}) \right| \le b$. 
Therefore, all preconditions of the first statement of Lemma~\ref{promise_correctness-invariants} hold, and thus, by invoking this lemma we have that, with probability at least $1 - \frac{1}{N^4}$, $\correct_\lef = \true$, and $\supp{ (\wh \chi_\lef)} \subseteq \subtree_{T}(v_\lef) \cap \head_\mu(\wh{y})$, and $\left\| \wh y_{\subtree_{T}(v_\lef) \cap \head_\mu(\wh{y})} - \wh \chi_\lef \right\|_2^2 \le \frac{\mu^2}{20}$. This together with inductive hypothesis $P_4(t-1)$ imply that, with probability at least $1 - \frac{1}{N^4}$,  $\correct_\lef = \true$ and $\supp{ (\wh \chi_\lef)} = \subtree_{T}(v_\lef) \cap \head_\mu(\wh{x})$. 

So, the if-statement in line~\ref{a10l20} of the algorithm is $\true$  and consequently the algorithm adds all leaves that correspond to frequencies in $\subtree_{T}(v_\lef) \cap \head_\mu(\wh{x})$ to $\frontier^{(t-1)}$ and also updates
\small
\[\sident^{(t)} \gets \sident^{(t-1)} \cup \left\{ u \in \leaves(\frontier): \ff_u \in \subtree_{T}(v_\lef) \cap \head_\mu(\wh{x}) \right\}.\]
\normalsize
By a similar argument, if $\left|\subtree_{T}(v_\righ) \cap \head_\mu(\wh{x}) \right| \le b$, then, with probability at least $1 - \frac{1}{N^4}$, the algorithm adds all leaves corresponding to frequencies in $\subtree_{T}(v_\righ) \cap \head_\mu(\wh{x})$ to $\frontier^{(t-1)}$ and updates
\small
\[\sident^{(t)} \gets \sident^{(t-1)} \cup \left\{ u \in \leaves(\frontier): \ff_u \in \subtree_{T}(v_\righ) \cap \head_\mu(\wh{x}) \right\}.\]
\normalsize

{\bf Possibility 2)} $\left|\subtree_{T}(v_\lef) \cap \head_\mu(\wh{x}) \right| > b$. 
Same as in {\bf possibility 1}, the inductive hypothesis $P_4(t-1)$ implies that $|\Estimated^{(t-1)}| \le k$ and hence inductive hypothesis $P_3(t-1)$ along with the assumption $\epsilon \le \frac{1}{10}$ gives \eqref{eq:est-error-bound}.
Hence, $\head_\mu\left(\wh x - \wh \chi^{(t-1)} \right) = \head_\mu(\wh{x})\setminus \Estimated^{(t-1)}$.
Consequently, if we let $\wh y : = \wh x - \wh \chi^{(t-1)}$, then it holds that: i) $\head_\mu(\wh{y})\subseteq \supp{(T)}$, by $P_1(t-1)$, ii) $\| \wh y - \wh y_{\head_\mu(\wh{y})} \|_2^2 \le \frac{11\mu^2}{10}$, by precondition of the lemma along with \eqref{eq:est-error-bound}, and iii) $\left|\subtree_{T}(v_\lef) \cap \head_\mu(\wh{y}) \right| \le \left| \head_\mu(\wh{x}) \right| \le k$, by precondition of the lemma. 
Additionally, by $P_4(t-1)$, we find that 
\[\left|\subtree_{T}(v_\lef) \cap \head_\mu(\wh{y}) \right| = \left|\subtree_{T}(v_\lef) \cap \head_\mu(\wh{x}) \right| > b.\]
Therefore, all preconditions of the second statement of Lemma~\ref{promise_correctness-invariants} hold, and thus, by invoking this lemma we have that, with probability at least $1 - \frac{1}{N^4}$, $\correct_\lef = \false$, and $\wh \chi_\lef \equiv 0$. 
So, the if-statement in line~\ref{a10l20} of the algorithm is $\false$  and consequently the algorithm adds leaf $v_\lef$ as the left child of $v$ to tree $\frontier^{(t-1)}$.
By a similar argument, if $\left|\subtree_{T}(v_\righ) \cap \head_\mu(\wh{x}) \right| > b$, then, with probability $1 - \frac{1}{N^4}$, the algorithm adds leaf $v_\righ$ as the left child of $v$ to tree $\frontier^{(t-1)}$.

Based on the above arguments, according to the values of $\correct_\lef$ and $\correct_\righ$, there are various cases that can happen. From the way tree $\frontier^{(t)}$ and set $\sident^{(t)}$ are obtained from $\frontier^{(t-1)}$ and $\sident^{(t-1)}$, it follows that in any case, the first 4 properties of $P(t)$ are maintained with probability at least $1 - \frac{2}{N^4}$. Furthermore, the way tree $T^{(t)}$ is constructed implies that,
\[ \sum_{u\in \leaves\left(\frontier_v^{(t)}\right)} \left| L_u^{(t)} \right| = \left| L_v^{(t-1)} \right|. \]
Therefore, for every $t>1$, by inductive hypothesis $P_2(t-1)$, the change in potential is bounded as follows,
\small
\begin{align*}
\phi_t - \phi_{t-1} &= \sum_{u \in \leaves\left(\frontier_v^{(t)}\right)} \left( 2\log N - l_{\frontier^{(t)}}(u) \right) \cdot \left| L_u^{(t)}\right| - \left( 2\log N - l_{\frontier^{(t-1)}}(v) \right) \cdot \left| L_v^{(t-1)}\right|\\
&\le - \left| L_v^{(t-1)}\right| < -b.
\end{align*}
\normalsize
Moreover, if $t=1$ then the change in potential satisfies $\phi_1 - \phi_{0} \le - \left| L_v^{(t-1)}\right| \le 0$ (because in this case $v=$ root). This proves the inductive claim $P_5(t)$. 

We have proved that for every $t$, if the inductive hypothesis $P(t-1)$ is satisfied then the property $P(t)$ is maintained with probability at least $1 - \frac{2}{N^4} - \frac{1}{N^7} \ge 1 - \frac{4}{N^4}$. Therefore, using the inductive hypothesis that $\Pr[P(t-1)] \ge 1 - \frac{4(t-1)}{N^4}$, by using union bound we find that
\[ \Pr[P(t)] \ge \Pr[ P(t) | P(t-1) ] \cdot \Pr[ P(t-1) ] \ge 1 - \frac{4 t }{N^4}. \]
This complets the proof of the lemma.
\end{proof}

Now we are in a position to prove the main result of this section.
\paragraph{Proof of Theorem~\ref{thm:core}.}
The proof basically follows by invoking Lemma~\ref{lem:invariant} and then analyzing the runtime and sample complexity of Algorithm~\ref{alg:sfftrobust}. If we let $\mu := \|\eta\|_2$ then because $x$ is a signal in the $k$-high SNR regime, we have that $\left| \head_\mu(\wh{x}) \right| \le k$ and $\left\| \wh x - \wh x_{\head_\mu(\wh{x})} \right\|_2 \le \mu$. Therefore, if we run the procedure \textsc{RobustSparseFT} (Algorithm~\ref{alg:sfftrobust}) with inputs $(x, k, \epsilon, \mu)$, then the preconditions of Lemma~\ref{lem:invariant} hold and hence by invoking this lemma we conclude that all the invariants $P_1(t)$ through $P_5(t)$, defined in Lemma~\ref{lem:invariant}, hold throughout the execution of Algorithm~\ref{alg:sfftrobust} for every non-negative integer $t$. 

Using this, we first prove the termination of the algorithm. Let $q = O\left( k + \frac{k\log N}{b}  \right)$ be some large enough integer. We show that the algorithm must terminate in $q$ iterations. Note that the probability that the properties $P(t)$ hold for all iterations $t \in \{ 0,1, \dots q\}$ of algorithm \textsc{RobustSparseFFT} is at least $1 - \frac{4(q+1)}{N^4} \ge 1 - \frac{1}{N^3}$, by Lemma~\ref{lem:invariant}. From now on, we condition on the event corresponding to $P(t)$ holding for all iterations $t \in \{ 0,1, \dots q\}$, which holds with probability at least $1 - \frac{1}{N^3}$. Conditioned on this event we prove that the algorithm terminates in less than $q$ iterations. 

Note that, the potential function $\phi_t$ defined in $P_5(t)$ is non-negative for every $t$.	Moreover, at the zeroth iteration of the algorithm $T^{(0)} = \{\text{root}\}$ and hence $L_\text{root}^{(0)} = \head_\mu(\wh{x})$, thus
\[ \phi_0 \le 2 k \log N. \]
Therefore, it follows from $P_5(t)$ that Algorithm~\ref{alg:sfftrobust} must terminate in at most $q = O\left(k + \frac{k\log N}{b} \right)$ iterations.

When the algorithm terminates, the condition of the \emph{Repeat-Until} loop in line~\ref{a10l30} of the algorithm must be True. Thus, when the algorithm terminates, there is no leaf in tree $T^{(q)}$ besides the root. Cosequently, by invariants $P_1(q)$ and $P_3(q)$, the output of the algorithm satisfies,
$\head_\mu(\wh{x}) \subseteq \supp(\wh{\chi})$ and $\left\|\wh{x}_{\Estimated} - \wh{\chi}\right\|_2^2 \leq \frac{\epsilon \left|\Estimated \right|}{k} \cdot \mu^2$, where $\Estimated=\supp(\wh \chi)$. Using the invariant $P_4(q)$, the latter can be Further upper bounded as $\left\|\wh{x}_{\Estimated} - \wh{\chi}\right\|_2^2 \leq \epsilon \cdot \mu^2$. This together with the $k$-high SNR assumption of the theorem gives the approximation guarantee of the theorem $\left\|\wh{x} - \wh{\chi}\right\|_2^2 \leq (1+\epsilon) \cdot \|\eta\|_2^2$.

{\bf Runtime and Sample Complexity.} The expensive components of the algorithm are primitive \textsc{Estimate} in line~\ref{a10l10} and primitive \textsc{RobustPromiseSFT} in lines~\ref{a10l18} and \ref{a10l19} of the algorithm. 
We first bound the time and sample complexity of invoking \textsc{Estimate} in line~\ref{a10l10}. We remark that, at any iteration $t$, the algorithm runs primitive \textsc{Estimate} only if {\bf case 1} that we mentioned earlier in the proof happens. Therefore, in this case, the set $\emptyset \neq \scheap^{(t)} \subseteq \sident^{(t-1)}$ that our algorithm computes in line~\ref{a10l8} by running the primitive \textsc{ExtractCheapSubset} satisfies the property that $\left| \scheap^{(t)} \right| \cdot \left( 8+4\log \left| \sident^{(t-1)} \right| \right) \ge \max_{u \in \scheap^{(t)}} 2^{w_{\frontier^{(t-1)}}(u)}$. By $P_2(t-1)$, and $k$-high SNR assumption, this implies that $\left| \scheap^{(t)} \right| \cdot \left( 8+4\log k \right) \ge \max_{u \in \scheap^{(t)}} 2^{w_{\frontier^{(t-1)}}(u)}$.

Thus, by Lemma~\ref{est-inner-lem}, the runtime and sample complexity of every invocation of \textsc{Estimate} in line~\ref{a10l10} of our algorithm are bounded by $\widetilde{O}\left( \frac{k}{ \epsilon \cdot \left| \scheap^{(t)} \right|}  \sum_{u \in \scheap^{(t)}} 2^{w_{\frontier^{(t-1)}}(u)} + \frac{k}{\epsilon}  \|\wh \chi^{(t-1)} \|_0 \right)$ and $\widetilde{O}\left( \frac{k}{ \epsilon \cdot \left| \scheap^{(t)} \right|}  \sum_{u \in \scheap^{(t)}} 2^{w_{\frontier^{(t-1)}}(u)}\right)$, respectively. Using $P_4(t-1)$, the runtime and sample complexity of \textsc{Estimate} can be further upper bounded by $\widetilde{O}\left( \frac{k}{ \epsilon} \cdot \left| \scheap^{(t)} \right| + \frac{k^2}{\epsilon} \right)$ and $\widetilde{O}\left( \frac{k}{ \epsilon} \cdot \left| \scheap^{(t)} \right| \right)$, respectively.
By property $P_5(t)$ we find that the total number of iterations in which {\bf case 1} happens, and hence number of times we run \textsc{Estimate} in line~\ref{a10l10} of the algorithm, is bounded by $O(k)$.
Using this together with the fact that $\sum_{t: \text{ if-statement in line~\ref{a10l7} is }\true }\left| \scheap^{(t)} \right| = \left\|\wh \chi \right\|_0 \le k$, the total runtime and sample complexity of all invocations of \textsc{Estimate} in all iterations can be upper bounded by $\widetilde{O}\left( \frac{k^3}{\epsilon} \right)$ and $\widetilde{O}\left( \frac{k^2}{\epsilon} \right)$, respectively.

Now we bound the runtime and sample complexity of invoking \textsc{RobustPromiseSFT} in lines~\ref{a10l18} and \ref{a10l19} of the algorithm. Note that at any iteration $t$, the algorithm runs \textsc{RobustPromiseSFT} in lines~\ref{a10l18} and \ref{a10l19} only if {\bf case 2} that we mentioned earlier in the proof happens. Since we pick leaf $v$ in line~\ref{a10l16} of the algorithm with smallest weight, and since the number of leaves that are not in the set $\sident^{(t-1)}$ are bounded by $\frac{k}{b}$ (by invariant $P_2(t-1)$), we have $w_{T^{(t-1)}}(v) \le \log \frac{k}{b}$. Also note that $\left\| \wh \chi^{(t-1)} \right\|_0 \le k$ by invariant $P_4(t-1)$ and the $k$-high SNR assumption.

Therefore, by Lemma~\ref{promise_correctness-runtime}, the runtime and sample complexity of each invokation of \textsc{RobustPromiseSFT} by our algorithm are bounded by $\widetilde{O} \left( k \cdot (b^2 + k) + \frac{k}{b} \cdot (b^3 + k) \right)$ and $\widetilde{O} \left( \frac{k}{b} \cdot (b^3 + k) \right)$. By property $P_5(t)$ we find that the total number of iterations in which {\bf case 2} happens, and hence the number of times we run \textsc{RobustPromiseSFT} in lines~\ref{a10l18} and \ref{a10l19} of the algorithm, is bounded by $O\left( \frac{k \log N}{b} \right)$. Therefore, by using $b \approx {k}^{1/3}$, we find that the total runtime and sample complexity of all invocations of \textsc{RobustPromiseSFT} are bounded by $\widetilde{O} \left( {k^{8/3} } \right)$ and $\widetilde{O} \left( k^{7/3} \right)$, respectively. Hence, the total time and sample complexity of the algorithm are bounded by $\widetilde{O} \left( \frac{k^3}{\epsilon} \right)$ and $\widetilde{O} \left( k^{7/3} + \frac{k^2}{\epsilon} \right)$, respectively.

\subsection{Proving the Correctness of our Computational Primitives.}

\label{sec:prove_correct_primitives}

In this subsection, we shall prove Lemmas~\ref{lem:guarantee_heavy_test}, \ref{est-inner-lem}, and Claim~\ref{claim:findCheap}. We proceed by proving them in the aforementioned order.

\begin{proofof}{Lemma~\ref{lem:guarantee_heavy_test}}

	By convolution-multiplication theorem, $h_{\Delta}^{z}$ computed in line~\ref{a7l6} of Algorithm~\ref{alg:zero-test} satisfies $h_{\Delta}^{z}= N \cdot \left(\chi \star G_v\right)(\Delta)$, and thus
	\begin{align*}
	H^{z} &= \frac{1}{|\textsc{RIP}_m^z|} \sum_{\Delta \in \textsc{RIP}_m^z} \left| N\cdot \sum_{\jj \in [n]^d} G_v(\Delta - \jj) \cdot x({\jj})- h^{z}_{\Delta} \right|^2\\ 
	&=  \frac{N^{2}}{|\textsc{RIP}_m^z|} \sum_{\Delta \in \textsc{RIP}_m^z} \left| \left((x - \chi) \star G_v\right)(\Delta) \right|^2.
	\end{align*}	
	Therefore, by the convolution-multiplication duality and using the definition $\wh y:= \left(\wh x - \wh \chi \right) \cdot \wh G_v$, if we let $y$ be the inverse Fourier transform of $\wh y$, we find that for every $z \in [32\log N]$,
	\[H^{z} = \frac{N^2}{|\textsc{RIP}_m^z|} \sum_{\Delta \in \textsc{RIP}_m^z} \left| y(\Delta) \right|^2. \]
	We first prove the {\bf first claim of the Lemma}. Let us write $\wh{y} = \wh{y}_S + \wh{y}_{\bar{S}}$, where $\wh{y}_S \in \C^{n^d}$ is defined as $\wh{y}_S(\ff) := \wh{y}(\ff) \cdot  \mathbbm{1}_{\{\ff \in S\}}$ and $\wh{y}_{\bar{S}} \in \C^{n^d}$ is defined as $\wh{y}_{\bar{S}}(\ff) := \wh{y} (\ff) \cdot  \mathbbm{1}_{\{\ff \notin S\}}$.  By the assumption of lemma $\|\wh y_S\|_2^2 > \frac{11 \theta}{10}$. Let $y_S$ and $y_{\bar{S}}$ denote the inverse Fourier transform of $\wh{y}_S$ and $\wh{y}_{\bar{S}}$ respectively. We have $y = y_S + y_{\bar{S}}$. Thus we find that,
	\begin{align*}
	\frac{1}{|\textsc{RIP}_m^z|} \sum_{\Delta \in \textsc{RIP}_m^z} \left| y(\Delta) \right|^2 &= \frac{1}{m} \sum_{\Delta \in \textsc{RIP}_m^z} |y_S(\Delta) + y_{\bar{S}}(\Delta)|^2\\
	&=\frac{1}{m} \sum_{\Delta \in \textsc{RIP}_m^z} |y_S(\Delta)|^2 + |y_{\bar{S}}(\Delta)|^2 + 2 \Re\left\{ y_S(\Delta)^* \cdot y_{\bar{S}}(\Delta) \right\}\\
	&\ge \frac{1}{m} \sum_{\Delta \in \textsc{RIP}_m^z} |y_S(\Delta)|^2 + 2 \Re\left\{ y_S(\Delta)^* \cdot y_{\bar{S}}(\Delta) \right\}
	\end{align*}
	First note that since $\wh{y}_S$ is $|S|$-sparse and because we assumed $m = \Omega\left( |S|\log^2|S| \log N \right)$ and because $\Delta$'s are i.i.d. uniform samples from $[n]^d$, by Theorem~\ref{RIP-thrm}, 
	\begin{equation}\label{RIP}
	\Pr \left[\frac{1}{m}\sum_{\Delta \in \textsc{RIP}_m^z} |y_S(\Delta)|^2 \ge 0.99 \cdot \frac{\|\wh{y}_S\|_2^2}{N^2}\right] \ge 1 - \frac{1}{N^2}.
	\end{equation}
	Now it suffices to bound the term $\frac{1}{m} \sum_{\Delta \in \textsc{RIP}_m^z} 2 \Re\left\{ y_S(\Delta)^* \cdot y_{\bar{S}}(\Delta) \right\}$. First, note that 
	\begin{align*}
	\E\left[\frac{1}{m} \sum_{\Delta \in \textsc{RIP}_m^z} 2 \Re\left\{ y_S(\Delta)^* \cdot y_{\bar{S}}(\Delta) \right\} \right] &= \frac{1}{m} \sum_{\Delta \in \textsc{RIP}_m^z} \E\left[ y_S(\Delta)^* \cdot y_{\bar{S}}(\Delta) \right] + \E\left[ y_S(\Delta) \cdot y_{\bar{S}}(\Delta)^* \right]\\
	&= \frac{1}{m} \sum_{\Delta \in \textsc{RIP}_m^z} \frac{1}{N}\langle y_S, y_{\bar{S}} \rangle + \frac{1}{N}\langle y_{\bar{S}}, y_S \rangle\\
	&= \frac{1}{m} \sum_{\Delta \in \textsc{RIP}_m^z} \frac{1}{N^2}\langle \wh{y}_S, \wh{y}_{\bar{S}} \rangle + \frac{1}{N^2} \langle \wh{y}_{\bar{S}}, \wh{y}_S \rangle\\
	&= 0,
	\end{align*}
	where the last line follows because the support of $\wh{y}_{\bar{S}}$ and $\wh{y}_S$ are disjoint.
	We proceed by bounding the second moment of the quantity $\frac{1}{m} \sum_{\Delta \in \textsc{RIP}_m^z} 2 \Re\left\{ y_S(\Delta)^* \cdot y_{\bar{S}}(\Delta) \right\}$ as follows,
	\begin{align*}
	\E\left[ \left| \frac{1}{m} \sum_{\Delta \in \textsc{RIP}_m^z} 2 \Re\left\{ y_S(\Delta)^* \cdot y_{\bar{S}}(\Delta) \right\} \right|^2 \right] &\le \E\left[ \left| \frac{2}{m} \sum_{\Delta \in \textsc{RIP}_m^z}  y_S(\Delta)^* \cdot y_{\bar{S}}(\Delta)  \right|^2 \right]\\
	&= \frac{4}{m}  \E\left[ \left| y_S(\Delta)^* \cdot y_{\bar{S}}(\Delta)\right|^2 \right] \text{~~~(By independence of $\Delta$'s)}\\
	&\le \frac{4}{m} \E\left[ \|y_S\|_\infty^2 \left|y_{\bar{S}}(\Delta)\right|^2 \right]\\
	&= \frac{4}{m} \|y_S\|_\infty^2 \E\left[ \left|y_{\bar{S}}(\Delta)\right|^2 \right]\\
	&= \frac{4}{m} \|y_S\|_\infty^2  \frac{\|\wh{y}_{\bar{S}}\|_2^2}{N^2} 
	\end{align*}
	
	By Chebyshev's inequality we have the following,
	\small
	\begin{align*}
	\Pr\left[ \left| \frac{1}{m} \sum_{\Delta \in \textsc{RIP}_m^z} 2 \Re\left\{ y_S(\Delta)^* \cdot y_{\bar{S}}(\Delta) \right\} \right| \ge 1/20 \cdot \frac{\|\wh{y}_S\|_2^2}{N^2} \right] &\le \frac{1600 N^2 \|y_S\|_\infty^2  \|\wh{y}_{\bar{S}}\|_2^2 }{m \|\wh{y}_S\|_2^4}\\
	&\le \frac{1600 \|\wh{y}_S\|_1^2  \|\wh{y}_{\bar{S}}\|_2^2 }{m \|\wh{y}_S\|_2^4}\\
	&\le \frac{1600 |S| \cdot \|\wh{y}_S\|_2^2  \|\wh{y}_{\bar{S}}\|_2^2 }{m \|\wh{y}_S\|_2^4} \text{~~~(Cauchy-Schwarz)}\\
	&= \frac{1600 |S| \cdot \|\wh{y}_{\bar{S}}\|_2^2 }{m \|\wh{y}_S\|_2^2}.
	\end{align*}
	\normalsize
	Therefore because we assumed that $m = \Omega \left( |S| \frac{\|\wh{y}\|_2^2}{\|\wh{y}_S\|_2^2} \right)$, the following holds,
	\[ \Pr\left[ \left| \frac{1}{m} \sum_{\Delta \in \textsc{RIP}_m^z} 2 \Re\left\{ y_S(\Delta)^* \cdot y_{\bar{S}}(\Delta) \right\} \right| \ge 1/20 \cdot \frac{\|\wh{y}_S\|_2^2}{N^2} \right] \le 1/10.\]
	Combining the above inequality with \eqref{RIP} using union bound gives,
	\[ \Pr\left[ H^{z} \le 0.94 \cdot \|\wh{y}_S\|_2^2 \right] \le 1/8.\]
	
	Since in line~\ref{a7l8} of the algorithm we compare $\textsc{Median}_{z \in [32 \log N]} \left\{H^{z}\right\}$ to $\theta$, using the fact that $\|\wh{y}_S\|_2^2 > \frac{11\theta}{10}$, we have the following,
	\begin{align*}
	\Pr \left[ \textsc{HeavyTest} = \false \right] &\le \Pr \left[ \textsc{Median}_{z \in [32 \log N]} \left\{H^{z} \right\} \le 10/11 \cdot \|\wh{y}_S\|_2^2 \right]\\ 
	&\le {32\log N \choose 16\log N} \frac{1}{8^{16\log N}}\\ 
	&\le \frac{2^{32\log N}}{8^{16\log N}} = \frac{1}{N^{16}}.
	\end{align*}
	This completes the proof of the first claim.
	
	The proof of the {\bf second claim of the lemma} is more straightforward. The expected value of $H^{z}$ is,
	\[ \E [H^{z}] = \frac{N^2}{|\textsc{RIP}_m^z|} \sum_{\Delta \in \textsc{RIP}_m^z} \E \left[\left| y(\Delta) \right|^2\right] = \|\wh y\|_2^2. \]
	Therefore by Markov's inequality we find that for every $z \in [32\log N]$,
	\[ \Pr \left[ H^{z} \ge 5 \|\wh y\|_2^2 \right] \le 1/5. \]
	The assumption of the lemma in this case is that $\|\wh y\|_2^2 \le \theta/5$, thus we have,
	\begin{align*}
	\Pr \left[ \textsc{HeavyTest} = {\true} \right] &\le \Pr \left[ \textsc{Median}_{z \in [32 \log N]} \left\{H^{z}_{\ff}\right\} > 5 \cdot \|\wh{y}_S\|_2^2 \right]\\ 
	&\le {32\log N \choose 16\log N} \frac{1}{5^{16\log N}}\\ 
	&\le \frac{2^{32\log N}}{5^{16\log N}} = \frac{1}{N^{5}}.
	\end{align*}
	This completes the proof of the second claim of the lemma.

	\paragraph{Sample Complexity and Runtime:} Computing the filters $(G_v,\wh G_v)$ uses $O\left( 2^{w_T(v)} + \log N \right)$ runtime, by Lemma~\ref{lem:isolate-filter-highdim}. Given filter $\wh G_v$, computing the quantities $h_\Delta^{z}$ for all $\Delta$ and $z$ in line~\ref{a7l6} of the algorithm uses $O\left( \|\wh \chi\|_0 \cdot \sum_{z}|\textsc{RIP}_m^z| \right) = O\left( \|\wh \chi\|_0 \cdot m \log N \right)$ time. Given filter $G_v$ with $|\supp (G_v)| = 2^{w_T(v)}$, computing the quantity $H^z$ for all $z$ requires $O\left( 2^{w_T(v)} \cdot \sum_{z}|\textsc{RIP}_m^z| \right) = O\left( 2^{w_T(v)} \cdot m\log N \right)$ accesses to the signal $x$ and $O\left( 2^{w_T(v)} \cdot m\log N \right)$ runtime. Therefore, the total sample complexity of the algorithm is $O\left( 2^{w_T(v)} \cdot m\log N \right)$ and the total runtime of  the algorithm is $O\left( 2^{w_T(v)} \cdot m\log N + \|\wh \chi\|_0 \cdot m \log N \right)$
	
\end{proofof}


\begin{proofof}{Lemma~\ref{est-inner-lem}}
	Note that the algorithm constructs $(v, T)$-isolating filters $(G_v,\wh G_v)$ for every leaf $v \in S$. By Lemma~\ref{lem:isolate-filter-highdim}, constructing filters $G_v$ and $\wh G_v$ takes time $O \left( 2^{w_{T}(v)} + \log N \right)$. Moreover, Lemma~\ref{lem:isolate-filter-highdim} tells us that filter $G_v$ has support size $|\supp (G_v)| =2^{w_{T}(v)}$ and $\wh G_v$ can be accessed at any frequency using $O(\log N)$ operations.\\	
	Therefore, for every fixed $v \in S$, computing  $h^z_v = \sum_{\Delta \in \textsc{RIP}_m^z} e^{-2\pi i \frac{\ff^\top\Delta}{n}} \sum_{\bm\xi \in [n]^d}  e^{2\pi i \frac{\bm\xi^T \Delta}{n}} \cdot \wh \chi_{\bm\xi} \cdot \wh G_v(\bm\xi)$ in line~\ref{a8l7} of Algorithm~\ref{alg:high-dim-Est-robust} can be done in total time $O\left(|\textsc{RIP}_m^z|  \log N \cdot \|\wh \chi\|_0  \right) = O\left(B \log N \cdot \|\wh \chi\|_0 \right)$ for all $z$. By convolution-multiplication duality theorem, $h^z_v$ satisfies $h^z_v = N \cdot \sum_{\Delta \in \textsc{RIP}_m^z} e^{-2\pi i \frac{\ff^\top \Delta}{n}} \left(\chi \star G_v\right)(\Delta)$, and thus, for every leaf $v \in S$:
	\begin{align*}
	H^{z}_v&= \frac{1}{|\textsc{RIP}_m^z|} \cdot\left( N \cdot \sum_{\Delta \in \textsc{RIP}_m^z} \left(e^{-2\pi i \frac{\ff^\top \Delta}{n}} \sum_{\jj \in [n]^d} G_v(\Delta -\jj) \cdot x({\jj})\right) - h^z_v \right)\\ 
	&=  \frac{N}{|\textsc{RIP}_m^z|} \sum_{\Delta \in \textsc{RIP}_m^z} e^{-2\pi i \frac{\ff^\top \Delta}{n}} \left( \left(x - \chi \right) \star G_v \right)(\Delta).
	\end{align*}
	To simplify the notation, let us use $y_v := \left(x - \chi\right) \star G_v$. 
	Because $G_v$ is $\left( v, T \right)$-isolating, by Definition~\ref{def:v-t-isolating-highdim}, we have that $\wh{y}_v({\bm \xi}) = 0$ for every ${\bm \xi} \in \bigcup_{\substack{u\in \leaves(T) \\ u \neq v}} \subtree_T(u)$ and also $\wh{y}_v({\ff}) = (\wh{x - \chi})({\ff})$, where $\ff := \ff_v$ is the frequency label of the leaf $v$. Using these facts together with the above equality and the assumption of the lemma on $\textsc{IsIdentified}(T,v) = \true$, we can write,
	\begin{align*}
	H^{z}_v &= \frac{N}{|\textsc{RIP}_m^z|} \sum_{\Delta \in \textsc{RIP}_m^z} e^{-2\pi i \frac{\ff^\top\Delta}{n}} y_v(\Delta)\\
	&= \wh{y}_v(\ff) + \frac{1}{|\textsc{RIP}_m^z|} \sum_{\Delta \in \textsc{RIP}_m^z} \sum_{{\bm \xi} \in [n]^d\setminus \supp{(T)} } e^{2\pi i \frac{({\bm \xi} - \ff)^\top\Delta}{n}} \cdot \wh{y}_v({\bm \xi}).
	\end{align*}
	We continue by computing the expectation of the above quantity. Since $\ff \in \subtree_T(v)$, $\bm{\xi}-\ff \neq 0$ for every $\bm{\xi} \in [n]^d\setminus \supp{(T)}$, which in turn implies that,
	\[\E\left[H^{z}_{v}\right] = \wh{y}_v(\ff) + \frac{1}{|\textsc{RIP}_m^z|} \sum_{\Delta \in \textsc{RIP}_m^z} \sum_{{\bm \xi} \in [n]^d\setminus \supp{(T)}} \E_\Delta\left[e^{2\pi i \frac{({\bm \xi} - \ff)^\top\Delta}{n}}\right] \wh{y}_v({\bm \xi}) = \wh{y}_v(\ff). \]
	In the above expectation we used the fact that $\Delta$ is distributed uniformly on $[n]^d$. Next we compute the second moment of $H^{z}_{v}$. We have,
	\begin{align*}
	\E\left[ \left|H^{z}_{v} - \wh{y}_v(\ff)\right|^2 \right] &= \frac{1}{|\textsc{RIP}_m^z|^2} \sum_{\Delta \in \textsc{RIP}_m^z} \E\left[\left|\sum_{{\bm \xi} \in [n]^d\setminus \supp{(T)} } e^{2\pi i \frac{({\bm \xi} - \ff)^\top\Delta}{n}} \wh{y}_v({\bm \xi})\right|^2\right] \text{~~(by independence of $\Delta$'s)}\\
	&= \frac{1}{|\textsc{RIP}_m^z|} \sum_{{\bm \xi} \in [n]^d\setminus \supp{(T)} } |\wh{y}_v({\bm \xi})|^2 \text{~~~~~~~(since $\Delta$ is uniform over $[n]^d$ and $\bm{\xi}-\ff \neq 0$)}\\
	&= \frac{1}{B} \sum_{{\bm \xi} \in [n]^d\setminus \supp{(T)} } \left|(\wh{x -\chi})({\bm{\xi}}) \cdot \wh{G}_v({\bm\xi})\right|^2. \text{~~~~~~~~~~~~~~~~~~~~~~~~~ (by definition of $y$)}
	\end{align*}
	In the final line above we used the fact that the multiset $\textsc{RIP}_m^z$ defined in Algorithm~\ref{alg:high-dim-Est-robust} has size $m$. Therefore, Markov's inequality implies that for every $z \in [16 \log N]$,
	\[ \Pr \left[ \left|H^{z}_{v} - \wh{y}_v(\ff) \right|^2 \ge \frac{8}{m} \cdot \sum_{{\bm \xi} \in [n]^d\setminus \supp{(T)} } \left|(\wh{x -\chi})({\bm{\xi}}) \cdot \wh{G}_v({\bm\xi})\right|^2 \right] \le \frac{1}{8}. \]
	Since in line~\ref{a8l9} of Algorithm~\ref{alg:high-dim-Est-robust} we set $\wh{H}_v = \textsc{Median}_{z \in [16 \log N]} \left\{H^{z}_{v}\right\}$, where the median of real and imaginary parts are computed separately, we find that
	\begin{align*}
	\Pr \left[ \left|\wh{H}_v - \wh{y}_v(\ff) \right|^2 \ge \frac{16}{m} \cdot \sum_{{\bm \xi} \in [n]^d\setminus \supp{(T)} } \left|(\wh{x -\chi})({\bm{\xi}}) \cdot \wh{G}_v({\bm\xi})\right|^2 \right] & \le {16 \log N \choose 8 \log N} \frac{1}{8^{8\log N}}\\ 
	&\le \frac{2^{16\log N}}{8^{8\log N}} = \frac{1}{N^{8}}.
	\end{align*}
	By recalling that $\wh{y}_v(\ff) = (\wh{x-\chi})({\ff_v})$ for every $v \in S$ and applying union bound we find that,
	\begin{equation} \label{eq:error-emt-inner}
	\Pr \left[ \sum_{v\in S} \left|\wh{H}_v - (\wh{x-\chi})({\ff_v}) \right|^2 \ge \frac{16}{m} \cdot \sum_{v\in S} \sum_{{\bm \xi} \in [n]^d\setminus \supp{(T)} } \left|(\wh{x -\chi})({\bm{\xi}}) \cdot \wh{G}_v({\bm\xi})\right|^2 \right] \le \frac{|S|}{N^{8}}. 
	\end{equation}
	In the last step, we bound the quantity $\sum_{v\in S} \sum_{{\bm \xi} \in [n]^d\setminus \supp{(T)} } \left|(\wh{x -\chi})({\bm{\xi}}) \cdot \wh{G}_v({\bm\xi})\right|^2$ as follows,
	\small
	\begin{align*}
	\sum_{v\in S} \sum_{{\bm \xi} \in [n]^d\setminus \supp{(T)} } \left|(\wh{x -\chi})({\bm{\xi}}) \cdot \wh{G}_v({\bm\xi})\right|^2 &= \sum_{{\bm \xi} \in [n]^d\setminus \supp{(T)} } \left|(\wh{x -\chi})({\bm{\xi}})\right|^2 \cdot \sum_{v\in S} \left|\wh{G}_v({\bm\xi})\right|^2\\
	&\le \sum_{{\bm \xi} \in [n]^d\setminus \supp{(T)} } \left|(\wh{x -\chi})({\bm{\xi}})\right|^2 \cdot \sum_{v\in \leaves(T)} \left|\wh{G}_v({\bm\xi})\right|^2\\
	&= \sum_{{\bm \xi} \in [n]^d\setminus \supp{(T)} } \left|(\wh{x -\chi})({\bm{\xi}})\right|^2, \text{~~~~~~~~(By Lemma~\ref{filter-robust-multidim})}
	\end{align*}
\normalsize
	hence, plugging the above bound into \eqref{eq:error-emt-inner} gives,
	\small
	\[ \Pr \left[ \sum_{v\in S} \left|\wh{H}_v - (\wh{x-\chi})({\ff_v}) \right|^2 \ge \frac{16}{m} \cdot \sum_{{\bm \xi} \in [n]^d\setminus \supp{(T)}} \left|(\wh{x -\chi})({\bm{\xi}})\right|^2 \right] \le \frac{|S|}{N^{8}}.  \]
	\normalsize
\end{proofof}

Lastly, we prove the correctness of \textsc{ExtractCheapSubset}, and in particular Claim~\ref{claim:findCheap}.

\begin{proofof}{Claim~\ref{claim:findCheap}}
	First let $S':= \left\{ u \in S : 2^{w_T(u)} \le 4 |S| \right\}$. It easily follows that $\sum_{u \in S'} 2^{-w_T(u)} \ge \frac{1}{4}$.
	For every $j = 0, 1, \dots \lfloor \log (4|S|) \rfloor$, let $L_j$ denote the subset of $S'$ defined as $L_j:= \{ u: u\in S', w_T(u) = j\}$. We can write,
	\[ \sum_{u \in S'} 2^{-w_T(u)} = \sum_{j=0}^{\lfloor \log (4|S|) \rfloor} \frac{|L_j|}{2^j}\]
	Therefore, by the fact that $\sum_{u \in S'} 2^{-w_T(u)} \ge \frac{1}{4}$, we have that there must exist an integer $j \in \{0, 1, \dots \lfloor \log (4|S|) \rfloor\}$ such that $\frac{|L_j|}{2^j} \ge \frac{1}{4\lfloor \log (4|S|) \rfloor}$. Hence, there must exist a set $L \subseteq S$ such that $|L| \cdot ( 8 + 4 \log |S|) \ge \max_{v\in L} 2^{w_T(v)}$. The primitive \textsc{ExtractCheapSubset} finds this set $L$ efficiently.
\end{proofof}

\newpage

\section{Robust Sparse Fourier Transform II.}
\label{sec:robust_sec}
In this section we present an algorithm that can compute a $1+\epsilon$ approximation to the Fourier transform of a singnal in the $k$-high SNR regime using a sample complexity that is nearly quadratic in $k$ and a runtime that is cubic in $k$, fully making use of techniques I-IV.

Formally we prove the following theorem,
\recursivesfftrobust*

We first present a recursive procedure in Algorithm~\ref{alg:nearquad-robust} that is the main computational component of achieving the abovementioned theorem for a constant value of $\epsilon = \frac{1}{20}$. Any sparse $\wh \chi$ that satisfies the approximation guarantee of Theorem~\ref{thrm:near-quad-robust} for constant $\epsilon$, by the $k$-high SNR assumption, must recover all the \emph{head} elements of $\wh x$ correctly. Once we have the set of heavy frequencies of $\wh x$ we can estimate the head vlaues to a higher $\epsilon$ precision for arbitrarily small $\epsilon$ using a simple algorithm. We present the procedure that achieves such $1+\epsilon$ approximation and thus achieves the guarantee of Theorem~\ref{thrm:near-quad-robust} in Algorithm~\ref{high-dim-nearquad-main-alg}. We demonstrate the execution of primitive \textsc{RecursiveRobustSFT} (Algorithm~\ref{alg:nearquad-robust}) in Figure~\ref{fig:robust-recursive}.

\begin{figure}[t]
	\centering
	\scalebox{.73}{
		\begin{tikzpicture}[level/.style={sibling distance=80mm/#1,level distance = 1.5cm}]
			\node [arn] (z){}
			child {node [arn] {}edge from parent [photon]
				child[draw=white]
				child{node [arn]{}}
			}
			child { node [arn] {} edge from parent [photon]
				child[sibling distance=50mm] {node [arn_l] (v) {}
					child[sibling distance=70mm]{node [arn] {}edge from parent [solidline]
						child[sibling distance=55mm]{node [arn] {}
							child[sibling distance=30mm]{ node [arn]{}
								child[draw=white]
								child{node [arn] {}
									child{node [draw=blue]{}}
									child[draw=white]
								}
							}
							child{ node [arn] (e1) {}
								child{node [arn] {}edge from parent [dashed]
									child{node [arn] {}}
									child{node [arn] {}}
								}
								child[draw=white]
							}
						}
						child[sibling distance=37mm]{node [arn] {}
							child{node [arn] {}
								child{node [arn]{}
									child{node [draw=blue] {}}
									child{node [draw=blue] {}}
								}
								child[draw=white]
							}
							child{node [arn] {}
								child{node [arn]{}
									child{node [draw=blue] {}}
									child[draw=white]
								}
								child{node [arn]{}
									child{node [draw=blue] {}}
									child[draw=white]
								}
							}
						}
					}
					child[sibling distance=55mm]{node [arn] {}edge from parent [solidline]
						child[sibling distance=35mm]{node [circle, white, draw=red, fill=red, inner sep = 1.4] {}edge from parent [draw=red, thin, dashed]
							child[sibling distance=25mm]{node [circle, white, draw=red, fill=red, inner sep = 1.4]{} 
								child[draw=white]
								child{node [circle, white, draw=red, fill=red, inner sep = 1.4]{}
									child{node [circle, white, draw=red, fill=red, inner sep = 1.4]{}}
									child{node [circle, white, draw=red, fill=red, inner sep = 1.4]{}}
								}
							}
							child[sibling distance=30mm]{node [circle, white, draw=red, fill=red, inner sep = 1.4]{}
								child{node [circle, white, draw=red, fill=red, inner sep = 1.4] (r1){}
									child{node [circle, white, draw=red, fill=red, inner sep = 1.4] {}}
									child[draw=white]
								}
								child{node [circle, white, draw=red, fill=red, inner sep = 1.4] (r2){}
									child{node [circle, white, draw=red, fill=red, inner sep = 1.4] {}}
									child[draw=white]
								}
							}
						}
						child[sibling distance=50mm]{node [arn] {}
							child{node [arn] {}
								child{node [arn] (e2) {}
									child{node [arn] {} edge from parent [dashed]}
									child{node [arn] {} edge from parent [dashed]}
								}
							child[draw=white]
							}
							child{node [arn] (e3) {}
								child{node [arn] {} edge from parent [dashed]
									child[draw=white]
									child[sibling distance=10mm]{node [arn] {} edge from parent [dashed]}
								}
								child{node [arn] {} edge from parent [dashed]
									child[sibling distance=10mm]{node [arn] {} edge from parent [dashed]}
									child{node [arn] {} edge from parent [dashed]}
								}
							}
						}
					}
				}
				child [sibling distance=50mm]{node [arn] {}
					child[sibling distance=20mm]{node [arn]{}}
					child[sibling distance=20mm]{node [arn]{}}
				}
			};
			
			\node []	at (v.north)	[label=left: \LARGE{$v$}]	{};
			
			\node [] at (-6.9,0.1) [label=right:{\Large{\frontier}}]	{};
			\draw[draw=black,very  thick, ->] (-4,0) -- (-1.8,-0.5);

			\node [] at (-5,-3.5) [label=right:{\Large{subtree $T$}}]	{};
			\draw[draw=black,very  thick, ->] (-2.5,-3.6) -- (-1,-4);
			
			\node [] at (e3) [label=right:{\Large{yet to be explored}}]	{};
			\node [trrr] at (e1.north) [] {};

			\node [trr] at (e2.north) [] {};
			
			\node [trrr] at (e3.north) [] {};
			
			\node [] at (-9,-10.5) [label=north:\Large{$\sident$ leaves}]	{}
			edge[->, thick, bend right=35] (-6.3,-10.8)
			edge[->, thick, bend right=35] (-2.3,-10.8)
			edge[->, thick, bend right=35] (-1.2,-10.8)
			edge[->, thick, bend right=37] (-0.65,-10.8)
			edge[->, thick, bend right=37] (0.7,-10.8);
			
			\node [] at (-6,-13.5) [label=north:\Large{recovered \& subtracted}] {};
			\node [] at (-3.5,-13.8) [label=left:\Large{leaves (frequencies)}]	{}
			edge[->, bend right=20, thick, dashed, draw=red] (1.3,-10.8)
			edge[->, bend right=20, thick, dashed, draw=red] (2.3,-10.8)
			edge[->, bend right=20, thick, dashed, draw=red] (2.75,-10.8)
			edge[->, bend right=20, thick, dashed, draw=red] (4,-10.8);

			\draw[draw=black,thick] (-6.7,-10.75) rectangle ++(16,0.55);
			\foreach \x in {-5.5,-4.6, -2.5, -1.65, -0.82, 0.1 , 1.1, 1.9, 2.62, 3.5, 4.4, 5.25, 6, 7.55, 8.25} \draw[draw=black,very thick] (\x,-10.2) -- (\x,-10.75);
			\node []	at (10,-13)	[label=left:\Large $\head \cap \subtree_\frontier(v)$] {}
			edge[->, bend right=50, very thick] (9.35,-10.45);
			
		\end{tikzpicture}
	}
	\par
	
	\caption{Illustration of an instance of \textsc{RecursiveRobustSFT} (Algorithm~\ref{alg:nearquad-robust}). This procedure takes in a tree $\frontier$ (shown with thin edges) together with a leaf $v \in \leaves(\frontier)$ and adaptively explores/constructs the subtree $T$ rooted at $v$ to find all heavy frequencies that lie in $\subtree_\frontier(v)$. If $\head$ denotes the set of heavy frequencies, then the algorithm finds $\head \cap \subtree_\Path(v)$ by exploring $T$. Once the identity of a leaf is fully revealed, the algorithm adds that leaf to the set $\sident$. When the number of marked leaves grows to the point where there exists a subset of marked frequencies that can be estimated cheaply, our algorithm estimates the $\scheap$ subset in a batch, subtracts off the estimated signal, and removes all corresponding leaves from $T$ and $\sident$.}\label{fig:robust-recursive}
\end{figure}

\paragraph{Overview of \textsc{RecursiveRobustSFT} (Algorithm~\ref{alg:nearquad-robust}):}
Consider an invocation of \textsc{RecursiveRobustSFT}$(x, \wh\chi_{in}, \frontier, v, k, \alpha, \mu)$. Suppose that $\wh y := \wh x - \wh \chi_{in}$ is a signal in the high SNR regime, i.e., the value of each heavy frequency of signal $\wh y$ is at least $3$ times higher than the tail's norm. More formally, let $\head \subseteq[n]^d$ denote the set of heavy (head) frequencies of $\wh y$ and suppose that the tail norm of $\wh y$ satisfies $\| \wh y - \wh y_{\head}\|_2 \le \mu$ and additionally suppose that $|\wh y(\ff)| \ge 3\mu$ for every $\ff \in \head$. If $\frontier$ fully captures the heavy frequencies of $\wh y$, i.e., $\head \subseteq \supp{(\frontier)}$, and the number of heavy frequencies in frequency cone of node $v$ is bounded by $k$, i.e., $|\head \cap \subtree_{\frontier}(v)| \le k$, then \textsc{RecursiveRobustSFT} finds a signal $\wh \chi_v$ such that $\supp{(\wh\chi_v)} = \head \cap \subtree_{\frontier}(v) := S$ and $\|\wh y_S - \wh \chi_{v}\|_2^2 \le \frac{\mu^2}{40\log^2_{\frac{1}{\alpha}}k}$. An example of the input tree $\frontier$ is illustrated in Figure~\ref{fig:robust-recursive} with thin solid black edges. Additionally, one can see node $v$ which is a leaf of $\frontier$ in this figure.

Algorithm~\ref{alg:nearquad-robust} recovers heavy frequencies of signal $\wh y$ that lie in the subree of $v$, i.e., set $S = \head \cap \subtree_{\frontier}(v)$, by iteratively exploring the subtree of $\frontier$ rooted at $v$, which we denote by $T$, and simultaneously updating the proxy signal $\wh \chi_v$. We show an example of subtree $T$ at some iteration of our algorithm in Figure~\ref{fig:robust-recursive} with thick solid edges. The algorithm also maintains a subset of leaves denoted by $\sident$ that contains the leaves of $\frontier$ that are fully identified, that is the set of leaves that are at the bottom level and hence there is no ambiguity in their frequency content (there is exactly one element in frequency cone of marked leaves). We show the set of marked leaves in Figure~\ref{fig:robust-recursive} using blue squares.
Subtree $T$, in all iterations of our algorithm, maintains the invariant that the frequency cone of each of its leaves contain at least one head element and furthermore the frequency cone of each of its \emph{unmarked} leaves contain at least $b+1$ head element, where $b=\alpha k$, i.e.,
\begin{equation}\label{invariant:freq-cone-load-recursive}
	|\subtree_{\frontier\cup T}(u) \cap \head| \ge \begin{cases}
		1 & \text{ for every }u \in \sident\\
		b+1 & \text{ for every }u \in \leaves(T)\setminus \sident
	\end{cases}.
\end{equation}
We demonstrate, in Figure~\ref{fig:robust-recursive}, the leaves that correspond to set $S = \head \cap \subtree_{\frontier}(v)$ via leaves at bottom level of the subtree rooted at $v$.
Assuming that for the example shown in this figure $b=\alpha k = 2$, one can easily verify \eqref{invariant:freq-cone-load-recursive} by noting that the frequency cone of each leaf of $T$ contains at least one element from the set $\head$ and frequency cones of \emph{unmarked} leaves contain at least two element of $\head$.
Additionally, at every iteration of the algorithm, the union of all frequency cones of subtree $T$ captures all heavy frequencies that are not recovered yet, i.e., 
\begin{equation}\label{eq:tree-contain-all-heavies-recursive}
	S \setminus \supp{(\wh \chi_v)} \subseteq \supp{\left(\frontier\cup T\right)}.
\end{equation}
In Figure~\ref{fig:robust-recursive}, we show the set of fully recovered leaves (frequencies), i.e., $\supp{(\wh \chi_v)}$, using red thin dashed subtrees. These frequencies are subtracted from the residual signal $\wh y - \wh \chi_v$ and their corresponding leaves are removed from subtree $T$, as well. One can verify that condition~\ref{eq:tree-contain-all-heavies-recursive} holds in the example depicted in Figure~\ref{fig:robust-recursive}.
Moreover, the estimated value of every frequency that is recovered so far, is accurate up to an average error of $\frac{\mu}{\sqrt{40k}\cdot \log_{\frac{1}{\alpha}} k}$. More precisely, in every iteration of the algorithm the following property is maintained,
\begin{equation}\label{eq:estimates-recursive}
	\frac{\sum_{\ff \in \supp{(\wh \chi_v)}}| \wh y(\ff) - \wh \chi_v(\ff) |^2}{|\supp{(\wh \chi_v)}|} \le \frac{\mu^2}{40k\cdot \log_{\frac{1}{\alpha}}^2 k}.
\end{equation}

At the begining of the procedure, subtree $T$ is initialized to be the leaf $v$, i.e., $T=\{v\}$, and will be dynamically changing throughout the execution of our algorithm. Moreover, we initialize $\wh\chi_v \equiv 0$. Trivially, these initial values satisfy \eqref{invariant:freq-cone-load-recursive}, \eqref{eq:tree-contain-all-heavies-recursive}, and \eqref{eq:estimates-recursive}. 

The algorithm operates by picking the \emph{unmarked} leaf of $T$ that has the smallest weight. Then the algorithm explores the children of this node by recursively running \textsc{RecursiveRobustSFT} on them with a \emph{reduced budget} to recover the heavy frequencies that lie in their frequency cones. To be more precise, let us call the \emph{unmarked} leaf of $T$ that has the smallest weight $z$. We denote by $z_\lef$ and $z_\righ$ the left and right children of $z$. Let us consider exploration of the left child $z_\lef$, the right child is exactly the same. 
If the number of heavy frequencies in the frequency cone of $z_\lef$ is bounded by $b=\alpha k$, i.e., $|\head \cap \subtree_{\frontier\cup \{z_\lef,z_\righ\}}(z_\lef)| \le b$, then \textsc{RecursiveRobustSFT}$(x, \wh \chi_{in}+\wh \chi_v,\frontier \cup T \cup\{z_\lef,z_\righ\}, z_\lef, b, \alpha, \mu)$ recovers every frequency in the set $\head \cap \subtree_{\frontier\cup \{z_\lef,z_\righ\}}(z_\lef)$ up to an average error of $\frac{\mu}{\sqrt{40b} \cdot \log_{\frac{1}{\alpha}}b}$. 
Note that this everage estimation error is not sufficient for achieving the invariant \eqref{eq:estimates-recursive}, hence, instead of directly using the values that the recursive call of \textsc{RecursiveRobustSFT} recovered to update $\wh \chi_v$ at the newly recovered heavy frequencies, our algorithm adds the leaves corresponding to the recovered set of frequencies, i.e., $\head \cap \subtree_{\frontier\cup \{v_\lef,v_\righ\}}(v_\lef)$, at the bottom level of $T$ and marks them as fully identified (adds them to $\sident$). It can be seen in Figure~\ref{fig:robust-recursive} that all marked leaves are at the bottom level of the tree. 
For achieving maximum efficinecy we employ a new \emph{lazy estimation} scheme, that is, the estimation of values of marked leaves is delayed until there is a large number of marked leaves and thus there exists a subset of them that is cheap to estimate. 
On the other hand, if the number of head elements in frequency cone of $z_\lef$ is more than $b$ then our algorithm detects this and subsequently adds node $z_\lef$ to $T$. These operations ensure that the invariants~\eqref{invariant:freq-cone-load-recursive}, \eqref{eq:tree-contain-all-heavies-recursive}, and \eqref{eq:estimates-recursive} are maintained. 

Once the size of set $\sident$ grows sufficiently such that it contains a subset that is cheap to estimate, our algorithm estimates the values of the cheap frequencies. More precisely, at some point, $\sident$ will contains a non-empty subset $\scheap$ such that the values of all frequencies in $\scheap$ can be estimated cheaply and subsequently, our algorithm esimates those frequencies in a batch up to an average error of $O\left(\frac{\mu}{\sqrt{k}\cdot \log N}\right)$, updates $\wh \chi$ accordingly and removes all estimated ($\scheap$) leaves from $\frontier$ and $\sident$.
This ensures that invariants~\eqref{invariant:freq-cone-load-recursive}, \eqref{eq:tree-contain-all-heavies-recursive}, and \eqref{eq:estimates-recursive} are maintained. The estimated leaves are illustrated in Figure~\ref{fig:robust-recursive} using red thin dashed subtrees. We also demontrate the subtrees of $T$ that contain $\head$ element and are yet to be explored by our algorithm using gray cones and dashed edges in Figure~\ref{fig:robust-recursive}. The gray cone means that there are heavy elements in that frequency cone that need to be identified as that node has not reached the bottom level yet.

Finally, the algorithm keeps tabs on the runtime it spends and ensures that even if the input signal does not satisfy the preconditions for successful recovery, in particular if $|\head \cap \subtree_{\frontier}(v)| > k$, the runtime stays bounded. Additionally, the algorithm performs a quality control by running \textsc{HeavyTest} on the residual and if the recovered signal is not correct due to violation of some preconditions, it will be reflected in the output of our algorithm.

\begin{algorithm}[!t]
	\caption{A Recursive Robust High-dimensional Sparse FFT Algorithm}\label{alg:nearquad-robust}
	\begin{algorithmic}[1]
		\Procedure{RecursiveRobustSFT}{$x, \wh \chi_{in}, \frontier, v, k, \alpha, \mu$}
		\Comment{ $\mu$: upper bound on tail norm $\|\eta\|_2$}
		\If{$k \le \frac{1}{\alpha}$}
		\textbf{return} $\textsc{PromiseSparseFT} \left(x, \wh \chi_{in}, \frontier, v, k, \lceil\frac{k}{\alpha}\rceil, \mu\right)$ \label{a11l3}
		\EndIf
		
		\State{Let $T $ denote the subtree of $\frontier$ rooted at $v$ -- i.e. $T \gets \{v\}$}
		
		\State $\wh \chi_v \gets \{ 0 \}^{n^d}$	\Comment{ Sparse vector to approximate $(\wh{x}- \wh{\chi}_{in})_{\subtree_\frontier(v)}$}
		\State{$b \gets \lceil \alpha k \rceil$, $\sident \gets \emptyset$} \Comment $\sident$: set of fully identified leaves (frequencies)
		
		\Repeat
		\If {$(b +1)\cdot \left|\leaves(T_v) \setminus \sident\right| + |\sident| +  \|\wh \chi_v \|_0 > k$ } \label{a11l9}
		
		\State \textbf{return} $\left( \false, \{0\}^{n^d} \right)$ \label{a11l10} \Comment{Exit because budget of $v$ is wrong}
		\EndIf
		
		\If {$\sum_{u \in \sident} 2^{-w_T(u)} \ge \frac{1}{2}$} \label{a11l11}
		
		\State{$\scheap \gets \textsc{FindCheapToEstimate}\left( T, \sident \right)$} \label{a11l12}
		\State \Comment{Lazy estimation: We extract from the batch of marked leaves a subset that is cheap to estimate on average}
		\State $\left\{ \wh H_u \right\}_{u \in \scheap} \gets \textsc{Estimate}\left(x, \wh \chi_{in} + \wh \chi_v , \frontier \cup T, \scheap, \frac{736 k \cdot \log^2 N}{ |\scheap|} \right)$ \label{a11l14}
		\For{$u \in \scheap$}
		\State{$\wh \chi_v(\ff_u) \gets \wh H_{u}$}
		\State Remove node $u$ from subtree $T$
		\EndFor
		\State $\sident \gets \sident \setminus \scheap$
		\State \textbf{continue}
		\EndIf
		
		\State $z \gets \argmin_{u\in \leaves(T) \setminus \sident} w_{T}(u)$\label{a11l20}
		\Comment{pick the minimum weight leaf in subtree $T$ which is not in $\sident$}
		
		\State $z_\lef :=$ left child of $z$ and $z_\righ :=$ right child of $z$
		
		\State $T'\gets T \cup \left\{ z_\lef, z_\righ \right\}$ \label{a11l21} \Comment{Explore children of $z$}
		
		\small
		\State $(\correct_\lef, \wh{\chi}_\lef  ) \gets \textsc{RecursiveRobustSFT} \left(x, \wh \chi_{in} + \wh \chi_v, \frontier \cup T', z_\lef, b, \alpha, \mu\right)$ \label{a11l22}
		
		\State $(\correct_\righ, \wh{\chi}_\righ  ) \gets \textsc{RecursiveRobustSFT} \left(x, \wh \chi_{in} + \wh \chi_v, \frontier \cup T', z_\righ, b, \alpha, \mu\right)$ \label{a11l23}
		\normalsize
		\If{$\correct_\lef$ and $\correct_\righ$ and $z \neq v$ and $\|\wh \chi_\lef\|_0 + \|\wh \chi_\righ\|_0 \le b$} \label{a11l24}
		\State \textbf{return} $\left( \false, \{0\}^{n^d} \right)$ \label{a11l25}\Comment{Exit because budget of $v$ is wrong}
		\EndIf
		
		\If{$\correct_\lef$}
		
		\State $\forall \bm f \in \supp(\wh{\chi}_\lef)$, add the unique leaf corresponding to $\bm f$ to subtree $T$ and $\marked$
		
		\Else 
		\State $\text{Add } z_\lef$ to subtree $T$
		
		\EndIf
		
		\If{$\correct_\righ$}
		\small
		\State $\forall \ff \in \supp(\wh{\chi}_\righ)$, add the unique leaf corresponding to $\ff$ to subtree $T$ and $\marked$
		\normalsize
		\Else 
		\State $\text{Add } z_\righ$ to subtree $T$
		
		\EndIf

		\Until{$T$ has no leaves besides $v$}\label{a11l36}

		\If{$\textsc{HeavyTest}\left(x, \wh\chi_{in} + \wh \chi_{v} , \frontier, v , O\left(\frac{k}{\alpha} \log^3N \right), 6 \mu^2 \right)$} \label{a11l37} 
		\State \Comment{The number of heavy coordinates in $\subtree_{\frontier}(v)$ is more than $k$}
		\State \textbf{return} $\left( \false, \{0\}^{n^d} \right)$
		\Else
		\State \textbf{return} $\left( \true, \wh \chi_{v} \right)$ 
		\EndIf
		
		\EndProcedure
		
	\end{algorithmic}
\end{algorithm}

\begin{algorithm}[!t]
	\caption{Robust High-dimensional Sparse FFT with $\widetilde{O}(k^3)$ Time and $\widetilde{O}\left(k^{2+o(1)} \right)$ Samples}\label{high-dim-nearquad-main-alg}
	\begin{algorithmic}[1]
		
		\Procedure{RobustSFT}{$x, k, \epsilon, \mu$}
		
		\State $\alpha \gets 2^{-\sqrt{\log k \cdot \log(2\log N)}}$
		
		\State $(\correct , \wh \chi) \gets \textsc{RecursiveRobustSFT}\left(x , \{0\}^{n^d}, \{ \text{root}\}, \text{root}, k, \alpha, \mu \right)$ \label{a12l3}
		
		\State Let ${T}$ be the splitting tree corresponding to the set $\supp{(\wh \chi)}$ \label{a12l4}
		
		\State $\wh \chi_\epsilon \gets \{0\}^{n^d}$

		\While{tree $T$ has a leaf besides its root}
		
		\State{$\scheap \gets \textsc{FindCheapToEstimate}\left( T, \leaves(T) \right)$} \label{a12l7}
		\State \Comment{The set of frequencies that are cheap to estimate on average}
		\State $\left\{ \wh H_u \right\}_{u \in \scheap} \gets \textsc{Estimate}\left(x, \wh \chi_\epsilon , T, \scheap, \frac{32k}{ \epsilon \cdot |\scheap|} \right)$ \label{a12l9} 
		\For{$u \in \scheap$}
		\State{$\wh \chi_\epsilon(\ff_u) \gets \wh H_{u}$}
		\State Remove node $u$ from tree $T$
		\EndFor
		
		\EndWhile
		
		\State \textbf{return} $ \wh \chi_\epsilon$
		
		\EndProcedure

	\end{algorithmic}
	
\end{algorithm}

\paragraph{Analysis of \textsc{RecursiveRobustSparseFT}.}
Frirst we analyze the runtime and sample complexity of \textsc{RecursiveRobustSparseFT} in the following lemma.

\begin{lemma}[\textsc{RecursiveRobustSFT} -- Time and Sample Complexity]
	\label{recursive-robust-runtime}
	For every subtree $\frontier$ of $\tfull_N$, every leaf $v$ of $\frontier$, positive integer $k$, every $\alpha = o\left(\frac{1}{\log N}\right)$ and $\mu\ge0$, and every signals $x,\wh \chi_{in} : [n]^d \to \C$, consider an invocation of primitive \textsc{RecursiveRobustSFT} (Algorithm~\ref{alg:nearquad-robust}) with inputs $(x, \wh \chi_{in}, \frontier, v, k, \alpha, \mu)$. Then,
	\begin{itemize}
	\item The running time of primitive is bounded by
	\[\widetilde{O}\left( \left(\frac{k^2}{\alpha} \cdot 2^{w_\frontier(v)} + \frac{k}{\alpha} \cdot \|\wh \chi_{in}\|_0 \right) \cdot (2\log N)^{\log_{\frac{1}{\alpha}} k} + k^2 \cdot \|\wh \chi_{in}\|_0 + k^3 \right).\]
	\item The number of accesses it makes on $x$ is always bounded by
	\[\widetilde{O}\left(\frac{k^2}{\alpha} \cdot 2^{w_\frontier(v)} \cdot (2\log N)^{\log_{\frac{1}{\alpha}} k} \right).\]
	\end{itemize}	
	Moreover, the output signal $\wh{\chi}_v$ always satisfies $\supp(\wh{\chi}_v) \subseteq \subtree_\frontier(v)$ and $\|\wh{\chi}_v\|_0 \le k$.
	
\end{lemma}
\begin{proof}
	The proof is by induction on parameter $k$. The {\bf base of induction} corresponds to $k \le \frac{1}{\alpha}$. For every $k \le \frac{1}{\alpha}$, Algorithm~\ref{alg:nearquad-robust} simply runs \textsc{PromiseSparseFT}$(x, \wh \chi_{in}, \frontier, v, k, \lceil\frac{k}{\alpha} \rceil, \mu)$ in line~\ref{a11l3}. Therefore, by Lemma\ref{promise_correctness-runtime}, the runtime and sample complexity of our algorithm are bounded by $\widetilde{O} \left( \frac{k}{\alpha}\cdot \|\wh \chi_{in}\|_0 + \frac{k^2}{\alpha} \cdot 2^{w_\frontier(v)} \right)$ and $\widetilde{O} \left( \frac{k^2}{\alpha} \cdot 2^{w_\frontier(v)} \right)$, respectively. Moreover,  by Lemma\ref{promise_correctness-runtime}, the output signal $\wh \chi_v$ satisfies $\|\wh{\chi}_v\|_0 \le k$ as well as $\supp(\wh{\chi}_v) \subseteq \subtree_\frontier(v)$. This proves that the inductive hypothesis holds for every integer $k \le \frac{1}{\alpha}$, hence the base of induction holds.
	
	To prove the {\bf inductive step}, suppose that the lemma holds for every $k \le m - 1$ for some integer $m \ge \lfloor \frac{1}{\alpha} \rfloor + 1$. Assuming the inductive hypothesis, we prove that the lemma holds for $k=m$.
	First, we prove that Algorithm~\ref{alg:nearquad-robust} terminates after a bounded number of iterations. For the purpose of having a tight analysis of the runtime and sample complexity, we need to have tight upper bounds on the number of times our algorithm invokes primitive \textsc{Estimate} in line~\ref{a11l14} as well as the number of times our algorithm recursively calls itself in lines~\ref{a11l22} and \ref{a11l23}. First, we show that the number of iterations in which the if-staement in line~\ref{a11l11} is $\true$, and hence the number of times we invoke \textsc{Estimate} in line~\ref{a11l14}, is bounded by $O(k)$. The reason is, everytime the if-staement in line~\ref{a11l11} becomes $\true$ the sparsity of $\wh \chi_v$, i.e., $\|\wh \chi_v\|_0$, increases by $|\scheap| \ge 1$, because the if-staement in line~\ref{a11l11} ensures that preconditions of Claim~\ref{claim:findCheap} hold, hence, by invoking this claim, $\scheap \neq \emptyset$. On the other hand, we can see from the way our algorithm operates that the sparity of $\wh \chi_v$ does not decrease in any of the iterations of our algorithm. Therefore, because the if-statement in line~\ref{a11l9} of the algorithm makes sure that $\|\wh \chi_v\|_0$ does not exceed $k$, we conclude that the total number of iterations in which the if-statement in line~\ref{a11l11} is $\true$ is bounded by $O(k)$. Hence, the number of times our algorithm calls \textsc{Estimate} in line~\ref{a11l14} is $O(k)$.
	
	In order to bound the number of iterations of our algorithm in which the if-statement in line~\ref{a11l11} is $\false$, we use a potential function. Let $\wh \chi_v^{(t)}$ denote the signal $\wh \chi_v$ at the end of iteration $t$ of the algorithm. Furthermore, let $T^{(t)}$ denote the subtree $T$ at the end of $t^{th}$ iteration. Additionally, let $\sident^{(t)}$ denote the set $\sident$ (defined in Algorithm~\ref{alg:nearquad-robust}) at the end of iteration $t$.  
	We prove that the number of iterations in which the if-statement in line~\ref{a11l11} of our algorithm is $\false$ is bounded by $O\left( \frac{\log N}{\alpha} \right)$ using the following potential function, defined for non-negative integer $t$:
	$$\phi_t := (\log N + 1) \cdot {|\sident^{(t)}|} + 2\log N \cdot \| \wh \chi_v^{(t)} \|_0 + b \cdot \sum_{u\in \leaves\left(T^{(t)}\right) \setminus \sident^{(t)}} l_{T^{(t)}}(u).$$
	We prove that assuming the algorithm does not terminate in $q$ iterations, for some integer $q$, then in every positive iteration $t \le q$, if the if-statement in line~\ref{a11l11} of Algorithm~\ref{alg:nearquad-robust} is $\false$, then the above potential function increases by at least ${b}$, i.e., $\phi_t \ge \phi_{t-1} + b$. Additionally, when the if-statement in line~\ref{a11l11} is $\true$, the potential increases by at least $\log N-1$, i.e., $\phi_t \ge \phi_{t-1} + \log N - 1$.
	We show that at any given iteration $t$ of the algorithm the potential function $\phi_t$ increases in the abovementioned fashion.
	
	\paragraph{Case 1 -- the if-statement in line~\ref{a11l11} of Algorithm~\ref{alg:nearquad-robust} is True.}
	In this case, we have that $\sum_{u \in \sident^{(t-1)}} 2^{-w_{T^{(t-1)}}(u)} \ge \frac{1}{2}$. As a result, by Claim~\ref{claim:findCheap}, the set $\scheap^{(t)} \subseteq \sident^{(t-1)}$ that the algorithm computes in line~\ref{a11l12} by running the primitive \textsc{FindCheapToEstimate} is non-empty. 
	Then, the algorithm constructs $T^{(t)}$ by removing all leaves that are in the set $\scheap^{(t)}$ from tree $T^{(t-1)}$ and leaving the rest of the tree unchanged.
	Furthermore, the algorithm updates the set $\sident^{(t)}$ by subtracting $\scheap^{(t)}$ from $\sident^{(t-1)}$. Additionally, in this case, the algorithm computes $\{\wh H_u\}_{u \in \scheap^{(t)}}$ by running the procedure \textsc{Estimate} in line~\ref{a11l14} and then updates $\wh \chi_v^{(t)}(\ff_u) \gets \wh H_u$ for every $u \in \scheap^{(t)}$ and $\wh \chi_v^{(t)}(\bm \xi) = \wh \chi_v^{(t-1)}(\bm \xi)$ at every other frequency $\bm \xi$. Therefore, $\| \wh \chi_v^{(t)} \|_0 = \| \wh \chi_v^{(t)} \|_0 + |\scheap^{(t)}|$. Thus,
	\[\phi_t - \phi_{t-1} =  (\log N -1 )\cdot |\scheap^{(t)}| \ge \log N - 1,\]
	where the inequality follows from $\scheap^{(t)} \neq \emptyset$. This proves the potential increase that we wanted.
	
	\paragraph{Case 2 -- the if-statement in line~\ref{a11l11} is False.}
	In this case, either the algorithm terminates by the if-statement in line~\ref{a11l24}, which contradicts with our assumption that the algorithm does not terminate after $q \ge t$ iterations, or the following holds,
	\begin{align*}
	&{|\sident^{(t)}|} + b \cdot \sum_{u\in \leaves\left(T^{(t)}\right) \setminus\sident^{(t)}} l_{T^{(t)}}(u)\\ 
	&\qquad \ge {|\sident^{(t-1)}|} + b \cdot \sum_{u\in \leaves\left(T^{(t-1)}\right)\setminus\sident^{(t-1)}} l_{T^{(t-1)}}(u) + b,
	\end{align*}
	while ${|\sident^{(t)}|} \ge |\sident^{(t-1)} |$ and $\| \wh \chi_v^{(t)} \|_0 = \| \wh \chi_v^{(t-1)} \|_0$. Thus, in this case, $\phi_{t+1} - \phi_t \ge b$ which is the potential increase that we wanted to prove.
	
	So far, we proved that $\phi_t$ must increase by at least $\log N - 1$ at every iteration of the algorithm. Moreover, at every iteration of the algorithm where the if-statement in line~\ref{a11l11} is $\false$ the potential increases by at least $b$. Also, the potential function $\phi_t$ is non-negative for every $t$. On the other hand, the if-statement in line~\ref{a11l9} ensures that at any iteration $t \le q$ it must hold that $\phi_t \le 2 k \log N$.
	Therefore, the potential increse that we proved implies that Algorithm~\ref{alg:nearquad-robust} must terminate after at most $q = 2 k \log N$ iterations, where only in $\frac{2 \log N}{\alpha}$ of the iterations the if-statement in line~\ref{a11l11} can be $\false$. Therefore, the total number of times our algorithm recursively invokes itself in lines~\ref{a11l22} and \ref{a11l23} is bounded by $\frac{2 \log N}{\alpha}$.

	Now that we have the termination quarantee, we can use the fact that our algorithm constructs $\wh \chi_v$ by exclusively estimating the values of frequencies that lie in $\subtree_\frontier(v)$ in line~\ref{a11l14}, one can see that the output signal $\wh \chi_v$ always satisfies $\supp(\wh{\chi}_v) \subseteq \subtree_\frontier(v)$. Additionally, the if-staement in line~\ref{a11l9}, ensures that $\|\wh{\chi}_v\|_0 \le k$.
	Now we bound the running time and sample complexity of the algorithm.
	
	\paragraph{Sample Complexity and Runtime:} The expensive components of the algorithm are primitive \textsc{Estimate} in line~\ref{a11l14}, the recursive call of \textsc{RecursiveRobustSFT} in lines~\ref{a11l22} and \ref{a11l23}, and invocation of \textsc{HeavyTest} in line~\ref{a11l37} of the algorithm. 
	
	We first bound the time and sample complexity of invoking \textsc{Estimate} in line~\ref{a11l14}. We remark that, at any iteration $t$, the algorithm runs primitive \textsc{Estimate} only if {\bf case 1} that we mentioned earlier in the proof happens. Therefore, by Claim~\ref{claim:findCheap}, the set $\emptyset \neq \scheap^{(t)} \subseteq \sident^{(t-1)}$ that our algorithm computes in line~\ref{a11l12} by running the primitive \textsc{FindCheapToEstimate} satisfies the property that $| \scheap^{(t)} | \cdot \left( 8+4\log| \sident^{(t-1)} | \right) \ge \max_{u \in \scheap^{(t)}} 2^{w_{T^{(t-1)}}(u)}$. 
	By the if-statement in line~\ref{a11l9} of the algorithm, this implies that $| \scheap^{(t)} | \cdot \left( 8+4\log k \right) \ge \max_{u \in \scheap^{(t)}} 2^{w_{T^{(t-1)}}(u)}$.
	Thus, by Lemma~\ref{est-inner-lem}, the time and sample complexity of every invocation of \textsc{Estimate} in line~\ref{a11l14} of our algorithm are bounded by 
	\[ \widetilde{O}\left( \frac{k}{ | \scheap^{(t)} |}  \sum_{u \in \scheap^{(t)}} 2^{w_{\frontier \cup T^{(t-1)}}(u)} + {k}\cdot \left\|\wh \chi_v^{(t-1)} + \wh \chi_{in} \right\|_0 \right) \]
	and $\widetilde{O}\left( \frac{k}{ | \scheap^{(t)} |}  \sum_{u \in \scheap^{(t)}} 2^{w_{\frontier \cup T^{(t-1)}}(u)}\right)$, respectively. Using the fact that $\|\wh \chi_v^{(t-1)} \|_0 \le k$, these time and sample complexities are further upper bounded by 
	\[\widetilde{O}\left( {k} \cdot \left( 2^{w_\frontier(v)} \cdot | \scheap^{(t)} | + \|\wh \chi_{in} \|_0 \right) + {k^2} \right)\] 
	and $\widetilde{O}\left( {k}\cdot 2^{w_\frontier(v)} \cdot | \scheap^{(t)} | \right)$, respectively.
	We proved that the total number of times we run \textsc{Estimate} in line~\ref{a11l14} of the algorithm, is bounded by $O(k)$.
	Using this together with the fact that $\sum_{t: \text{ if-statement in line~\ref{a11l11} is }\true }\left| \scheap^{(t)} \right| = \left\|\wh \chi_v \right\|_0 \le k$, the total runtime and sample complexity of all invocations of \textsc{Estimate} in all iterations can be upper bounded by $\widetilde{O}\left( {k^3} + k^2 (\| \wh \chi_{in}\|_0 + 2^{w_\frontier(v)})\right)$ and $\widetilde{O}\left( {k^2} \cdot 2^{w_\frontier(v)} \right)$, respectively.
	
	Now we bound the runtime and sample complexity of invoking \textsc{RecursiveRobustSFT} in lines~\ref{a11l22} and \ref{a11l23} of the algorithm. Note that at any iteration $t$, our algorithm recursively calls \textsc{RecursiveRobustSFT} only if {\bf case 2} that we mentioned earlier in the proof occurs. As we showed, the total number of times that this happens is bounded by $ \frac{2\log N}{\alpha}$. 
	Since, in line~\ref{a11l20} of the algorithm, we pick leaf $z$ with the smallest weight, and since the number of leaves of subtree $T^{(t-1)}$ that are not in the set $\sident^{(t-1)}$ are bounded by $\frac{k}{b+1}$ (ensured by the if-statement in line~\ref{a11l9}), we have $w_{\frontier \cup T'}(z_\lef) = w_{\frontier \cup T'}(z_\righ) \le w_\frontier(v) + \log\frac{k}{b+1} + 1$. Also note that $\| \wh \chi_v^{(t-1)} \|_0 \le k$, ensured by the if-statement in line~\ref{a11l9}. Therefore, by the inductive hypothesis, the time and sample complexities of each recursive invocation of \textsc{RecursiveRobustSFT} by our algorithm are bounded by 
	\[\widetilde{O}\left( \left(\frac{b^2 \cdot 2^{w_\frontier(v)}}{\alpha^2} + \frac{b}{\alpha} \cdot  \|\wh \chi_{in}\|_0\right) \cdot (2\log N)^{\log_{\frac{1}{\alpha}} b} + b^2 \cdot \|\wh \chi_{in}\|_0 + kb^2 \right)\] 
	and $\widetilde{O}\left( \frac{b^2 }{\alpha^2}\cdot 2^{w_\frontier(v)} \cdot (2\log N)^{\log_{\frac{1}{\alpha}} b} \right)$.
	We proved that the total number of iterations in which {\bf case 2} happens, and hence the number of times we run \textsc{RecursiveRobustSFT} in lines~\ref{a11l22} and \ref{a11l23} of the algorithm, is bounded by $\frac{2\log N}{\alpha}$. Therefore, the total time and sample complexity of all invocations of \textsc{PromiseSparseFT} in lines~\ref{a11l22} and \ref{a11l23} are bounded by
	\[\widetilde{O}\left( \left(\frac{k^2}{\alpha} \cdot 2^{w_\frontier(v)} + \frac{k}{\alpha} \cdot \|\wh \chi_{in}\|_0\right) \cdot (2\log N)^{\log_{\frac{1}{\alpha}} k} + \alpha k^2 \cdot \|\wh \chi_{in}\|_0 + \alpha k^3 \right)\]
	and $\widetilde{O}\left( \frac{k^2}{\alpha} \cdot 2^{w_\frontier(v)} \cdot (2\log N)^{\log_{\frac{1}{\alpha}} k} \right)$, respectively.
	
	Finally, we bound the time and sample complexity of invoking \textsc{HeavyTest} in line~\ref{a11l37} of our algorithm. Since $\|\wh \chi_v\|_0 \le k$, by Lemma~\ref{lem:guarantee_heavy_test}, the time and sample complexity of the \textsc{HeavyTest} in line~\ref{a11l37} are bounded by $\widetilde{O}\left( \|\wh \chi_{in}\|_0 \cdot \frac{k}{\alpha} + \frac{k^2}{\alpha} + 2^{w_\frontier(v)}\cdot \frac{k}{\alpha} \right)$ and $\widetilde{O}\left( 2^{w_\frontier(v)}\cdot \frac{k}{\alpha} \right)$, respectively.
	Hence, we find that the total time and sample complexity of our algorithm are bounded by
	\[\widetilde{O}\left( \left(\frac{k^2 \cdot 2^{w_\frontier(v)}}{\alpha}  + \frac{k}{\alpha} \cdot  \|\wh \chi_{in}\|_0\right) \cdot (2\log N)^{\log_{\frac{1}{\alpha}} k} + k^2 \cdot \|\wh \chi_{in}\|_0 + k^3 \right)\]
	and $\widetilde{O} \left( \frac{k^2 }{\alpha} \cdot 2^{w_\frontier(v)} \cdot (2\log N)^{\log_{\frac{1}{\alpha}} k} \right)$, respectively. This proves the inductive step of the proof and consequently completes the proof of our lemma.
\end{proof}

Now we are in a position to present the main invariant of primitive \textsc{RecursiveRobustSFT}.
\begin{lemma}[\textsc{RecursiveRobustSFT} - Invariants]
	\label{recursive-robust-invariants}
	Consider the preconditions of Lemma~\ref{recursive-robust-runtime}. Let $\wh y := \wh x - \wh \chi_{in}$ and $S := \subtree_T(v) \cap \head_\mu(\wh y)$, where $\head_\mu(\cdot)$ is defined as per \eqref{def:head}. If i) $ \head_\mu(\wh y) \subseteq \supp{(\frontier)}$, ii) $\| \wh y - \wh y_{\head_\mu(\wh y)} \|_2^2 \le \frac{21\mu^2}{20} + \frac{\mu^2}{20\log_{\frac{1}{\alpha}}(k/\alpha)} $, and iii) $\left| S \right| \leq \frac{k}{\alpha}$, then with probability at least $1 - O \left(\left(\frac{2 \log N}{\alpha}\right)^{\log_{\frac{1}{\alpha}}k} \cdot N^{-4}\right)$, the output $\left( \mathrm{Budget}, \wh{\chi}_v \right)$ of Algorithm~\ref{alg:nearquad-robust} satisfies the following,
	\begin{enumerate}
		\item If $\left| S \right| \leq k$ then $\mathrm{Budget}=\true$, $\supp{(\wh \chi_v)} \subseteq S$, and $\left\|\wh{y}_{S} - \wh{\chi}_v\right\|_2^2 \leq \frac{\mu^2}{40 \log_{1/\alpha}^2 k} $;
		\item If $\left| S \right| > k$ then $\mathrm{Budget}=\false$ and $\wh \chi_v \equiv \{0\}^{n^d}$.
	\end{enumerate}	
	
\end{lemma}

\begin{proof}
	The proof is by induction on parameter $k$. The {\bf base of induction} corresponds to $k \le \frac{1}{\alpha}$. For every $k \le \frac{1}{\alpha}$, Algorithm~\ref{alg:nearquad-robust} simply runs \textsc{PromiseSparseFT}$\left( x, \wh \chi_{in}, \frontier, v, k, \lceil\frac{k}{\alpha} \rceil, \mu \right)$ in line~\ref{a11l3}. Therefore, by Lemma\ref{promise_correctness-invariants}, the claims of the lemma hold with probability at least $1 - \frac{1}{N^4}$. This proves that the inductive hypothesis holds for every integer $k \le \frac{1}{\alpha}$, hence the base of induction holds.
	
	To prove the {\bf inductive step}, suppose that the lemma holds for every $k \le m - 1$ for some integer $m \ge \lfloor \frac{1}{\alpha} \rfloor + 1$. Assuming the inductive hypothesis, we prove that the lemma holds for $k=m$. To prove the inductive claim, we first analyze the algorithm under the assumption that the primitives \textsc{HeavyTest} and \textsc{Estimate} are replaced with more powerful primitives that succeeds deterministically. Hence, we assume that \textsc{HeavyTest} correctly tests the ``heavy'' hypothesis on its input signal with probability $1$ and also \textsc{Estimate} achieves the estimation guarantee of Lemma~\ref{est-inner-lem} deterministrically. Moreover, we assume that our inductive invocation of \textsc{RecursiveRobustSFT} in lines~\ref{a11l22} and \ref{a11l23} of the algorithm succeed deterministically, hence, we assume that the inductive hypothesis (the lemma) holds with probability $1$. With these assumptions in place, we prove that the lemma holds deterministically (with probability 1).
	We then establish a coupling between this idealized execution and the actual execution of our algorithm, leading to our result.
	
	We prove the first statement of lemma by (another) induction on the \emph{Repeat-Until loop} of the algorithm. Note that we are proving the inductive step of an inductive proof using another induction (two nested inductions). The first (outer) induction was on the integer $k$ and the second (inner) induction is on the iteration number $t$ of the \emph{Repeat-Until loop} of our algorithm. 
	Let $\wh \chi_v^{(t)}$ denote the signal $\wh \chi_v$ at the end of iteration $t$ of the algorithm. Furthermore, let $\frontier^{(t)}$ denote the subtree $T$ at the end of $t^{th}$ iteration. Also, let $\sident^{(t)}$ denote the set $\sident$ (defined in Algorithm~\ref{alg:nearquad-robust}) at the end of iteration $t$. Additionaly, for every leaf $u$ of subtree $T^{(t)}$, let $L_u^{(t)}$ denote the ``{unestimated}'' frequencies in support of $\wh y$ that lie in frequency cone of $u$, i.e., $L_u^{(t)} := \subtree_{\frontier \cup T^{(t)}}(u) \cap \head_\mu(\wh{y})$
	We prove that if preconditions i, ii and iii together with the presondition of statement 1 (that is $|S| \le k$), hold, then at every iteration $t=0,1,2, \dots$ of Algorithm~\ref{alg:nearquad-robust}, the following properties are maintained,
	\begin{description}
		\item[$P_1(t)$] $S \setminus \supp\left(\wh \chi_v^{(t)}\right) \subseteq \supp\left(T^{(t)}\right) := \bigcup_{u\in \leaves\left(T^{(t)}\right)} \subtree_{\frontier \cup T^{(t)}}(u)$;
		\item[$P_2(t)$] For every leaf $u \neq v$ of subtree $T^{(t)}$, $\left| L_u^{(t)} \right| \ge 1$. Additionally, if $u \notin \sident^{(t)}$, then $\left| L_u^{(t)} \right| > {b}$;
		\item[$P_3(t)$] $\left\|\wh{y}_{S^{(t)}}-\wh{\chi}_v^{(t)}  \right\|_2^2 \leq \frac{\left|S^{(t)} \right|}{40k \cdot \log_{1/\alpha}^2k } \cdot  \mu^2$, where $S^{(t)} := \supp\left(\wh \chi_v^{(t)} \right)$;
		\item[$P_4(t)$] $S^{(t)} \subseteq S$ and $S^{(t)}  \cap \left(\bigcup_{\substack{u \in \leaves\left(T^{(t)}\right) \\ u \neq v}} \subtree_{\frontier \cup T^{(t)}}(u)\right) = \emptyset$;
	\end{description} 
	
	The {\bf base of induction} corresponds to the zeroth iteration ($t=0$), at which point $T^{(0)}$ is a subtree that solely consists of node $v$ and has no other leaves. Moreover, $\wh \chi_v^{(0)} \equiv 0$. Thus, statement $P_1(0)$ trivially holds by definition of set $S$. The statement $P_2(0)$ holds since there exists no leaf $u \neq v$ in $T^{(0)}$. The statements $P_3(0)$ and $P_4(0)$ hold because of the fact $\wh \chi_v^{(0)} \equiv 0$.

	We now prove the {\bf inductive step} by assuming that the inductive hypothesis, $P(t-1)$ is satisfied for some iteration $t-1$ of Algorithm~\ref{alg:nearquad-robust}, and then proving that $P(t)$ holds. 
	First, we remark that if inductive hypotheses $P_2(t-1)$ and $P_4(t-1)$ hold true, then by the precondition of statement 1 of the lemma (that is $|S| \le k$) the if-statement in line~\ref{a11l9} of Algorithm~\ref{alg:nearquad-robust} is $\false$ and hence lines~\ref{a11l9} and \ref{a11l10} of the algorithm can be ignored in our analysis. 
	We proceed to prove the induction by considering the two cases that can happen in every iteration $t$ of the algorithm:	
	
	\paragraph{Case 1 -- the if-statement in line~\ref{a11l11} of Algorithm~\ref{alg:nearquad-robust} is True.} 
	In this case, we have that $\sum_{u \in \sident^{(t-1)}} 2^{-w_{T^{(t-1)}}(u)} \ge \frac{1}{2}$. As a result, by Claim~\ref{claim:findCheap}, the set $\scheap \subseteq \sident^{(t-1)}$ that the algorithm computes in line~\ref{a11l12} by running the primitive \textsc{FindCheapToEstimate} satisfies the property that $\left| \scheap \right| \cdot \left( 8+4\log|\sident^{(t-1)}| \right) \ge \max_{u \in \scheap} 2^{w_{T^{(t-1)}}(u)}$. Clearly $\scheap \neq \emptyset$, by Claim~\ref{claim:findCheap}. 
	Then the algorithm computes $\{\wh H_u\}_{u \in \scheap}$ by running the procedure \textsc{Estimate} in line~\ref{a11l14} and then updates $\wh \chi^{(t)}(\ff_u) \gets \wh H_u$ for every $u \in \scheap$ and $\wh \chi^{(t)}(\bm \xi) = \wh \chi^{(t-1)}(\bm \xi)$ at every other frequency $\bm \xi$.
	Therefore, if we let $L:= \left\{ \ff_u: u \in \scheap \right\}$, then $S^{(t)} \setminus S^{(t-1)} = L$, by inductive hypothesis $P_4(t-1)$. By $P_3(t-1)$ along with Lemma~\ref{est-inner-lem} (its deterministic version that succeeds with probability 1), we find that
	\begin{align}
		\left\| \wh \chi_v^{(t)} - \wh y_{S^{(t)}}\right\|_2^2 &= \left\| (\wh \chi_v^{(t)} - \wh y)_{S^{(t-1)}}\right\|_2^2 + \left\| (\wh \chi_v^{(t)} - \wh y)_{S^{(t)}\setminus S^{(t-1)}}\right\|_2^2\nonumber\\
		&=\left\| \wh \chi_v^{(t-1)} - \wh y_{S^{(t-1)}}\right\|_2^2 + \left\| (\wh \chi_v^{(t)} - \wh y)_{L}\right\|_2^2\nonumber\\
		&\le \frac{\left|S^{(t-1)} \right| \cdot \mu^2}{40k \log_{1/\alpha}^2k} + \frac{\left|L \right|}{46 k \log^2N} \sum_{\bm\xi \in [n]^d \setminus \supp{\left(\frontier \cup T^{(t-1)}\right)} } \left| \left(\wh{y}-\wh{\chi}_v^{(t-1)}\right)({\bm{\xi}}) \right|^2.\label{estimate-error-nearquad}
	\end{align}
	Now we bound the second term above,
	\small
	\begin{align}
		&\sum_{\bm\xi \in [n]^d \setminus \supp{\left(\frontier \cup T^{(t-1)}\right)}} \left| \left(\wh{y}-\wh{\chi}_v^{(t-1)}\right)({\bm{\xi}}) \right|^2\nonumber\\ 
		&\qquad= \sum_{\bm\xi \in [n]^d \setminus \supp{(\frontier)} } \left| \wh{y}({\bm{\xi}}) \right|^2 + \sum_{\bm\xi \in \subtree_{\frontier}(v) \setminus \supp{\left(T^{(t-1)}\right)} } \left| \left(\wh{y}-\wh{\chi}_v^{(t-1)}\right)({\bm{\xi}}) \right|^2\nonumber\\
		&\qquad = \sum_{\bm\xi \in [n]^d \setminus \supp{(\frontier)} } \left| \wh{y}({\bm{\xi}}) \right|^2 \nonumber\\
		&\qquad\qquad + \sum_{\bm\xi \in \subtree_{\frontier}(v) \setminus \left(\supp{\left(T^{(t-1)}\right)} \cup S^{(t-1)}\right)} \left| \wh{y}({\bm{\xi}}) \right|^2 + \left\|\wh{y}_{S^{(t-1)}} -\wh{\chi}_v^{(t-1)} \right\|_2^2\nonumber\\
		&\qquad = \sum_{\bm\xi \in [n]^d \setminus \left(\supp{\left(\frontier\cup T^{(t-1)}\right)} \cup S^{(t-1)} \right) } \left| \wh{y}({\bm{\xi}}) \right|^2 + \left\|\wh{y}_{S^{(t-1)}} -\wh{\chi}_v^{(t-1)} \right\|_2^2 \nonumber\\
		&\qquad\le \sum_{\bm\xi \in [n]^d \setminus \head_\mu(\wh y)} \left| \wh{y}({\bm{\xi}}) \right|^2 + \left\|\wh{y}_{S^{(t-1)}} -\wh{\chi}_v^{(t-1)} \right\|_2^2 \text{~~~~~~~~~(by $P_1(t-1)$, precondition i and definition of $S$)}\nonumber\\
		&\qquad \le \frac{21\mu^2 }{20} + \frac{\mu^2}{20\log_{\frac{1}{\alpha}}(k/\alpha)} + \frac{\mu^2 }{40 \log_{\frac{1}{\alpha}}^2k} \text{~~~~~~~~~~~~~~(by $P_3(t-1)$ and $P_4(t-1)$ and precondition $|S| \le b$)}\nonumber\\
		&\qquad \le \frac{23\mu^2 }{20}.\nonumber
	\end{align}
	\normalsize
	Therefore, by plugging the above bound back to \eqref{estimate-error-nearquad} we find that,
	\[\left\| \wh \chi_v^{(t)} - \wh y_{S^{(t)}}\right\|_2^2 \le \frac{\left|S^{(t-1)} \right|}{40k \log_{\frac{1}{\alpha}}^2k} \cdot \mu^2 + \frac{\left|L \right|}{46k \log^2N} \cdot \left( \frac{23}{20} \mu^2 \right) \le \frac{\left|S^{(t)} \right|}{40k \log_{\frac{1}{\alpha}}^2k} \cdot \mu^2,
	\]
	which proves the inductive claim $P_3(t)$. 
	
	Moreover, in this case, the algorithm constructs $T^{(t)}$ by removing all leaves that are in the set $\scheap$ from tree $T^{(t-1)}$ and leaving the rest of the tree unchanged.
	Furthermore, the algorithm updates the set $\sident^{(t)}$ by subtracting $\scheap$ from $\sident^{(t-1)}$.	
	Note that, $P_2(t-1)$ implies that $L \subseteq S$. Thus, the fact $S^{(t)} = S^{(t-1)} \cup L$ together with inductive hypothesis $P_4(t-1)$ as well as the construction of $T^{(t)}$, imply $P_4(t)$. The construction of $T^{(t)}$ together with the fact that $|\subtree_{\frontier \cup T^{(t-1)}}(u)| = 1 $ for every $u \in \sident^{(t-1)}$ give $P_1(t)$ and $P_2(t)$.
	
	\paragraph{Case 2 -- the if-statement in line~\ref{a11l11} is False.} Let $z\in \leaves\left(T^{(t-1)}\right) \setminus \sident^{(t-1)}$ be the smallest weight leaf chosen by the algorithm in line~\ref{a11l20}. In this case, the algorithm constructs tree $T'$ by adding leaves $z_\righ$ and $z_\lef$ to tree $T^{(t-1)}$ as right and left children of $z$ in line~\ref{a11l21}. Then, the algorithm runs \textsc{RecursiveRobustSFT} with inputs $\left(x, \wh \chi_{in} + \wh \chi_v^{(t-1)}, T', z_\lef, b, \alpha, \mu \right)$ and $\left(x,  \wh \chi_{in} + \wh \chi_v^{(t-1)}, T', z_\righ, b, \alpha, \mu \right)$ in lines~\ref{a11l22} and \ref{a11l23} respectively. Now we analyze the output of the recursive invocation of \textsc{RecursiveRobustSFT} in lines~\ref{a11l22} and \ref{a11l23}. In the following we focus on analyzing $\left(\correct_\lef, \wh \chi_\lef\right)$ but $\left(\correct_\righ, \wh \chi_\righ\right)$ can be analyzed exactly the same way. There are two possibilities that can happen:	
	
	{\bf Possibility 1)} $\left|\subtree_{\frontier \cup T'}(z_\lef) \cap \head_\mu(\wh{y}) \right| \le b$. In this case, the inductive hypothesis $P_4(t-1)$ implies that $|S^{(t-1)}| \le k$ and hence inductive hypothesis $P_3(t-1)$ gives
	\begin{equation}\label{eq:est-error-bound-nearquad-robust}
	\left\| \wh y_{S^{(t-1)}} - \wh \chi_v^{(t-1)} \right\|_2^2 \le \frac{\mu^2}{40 \log_{1/\alpha}^2 k},
	\end{equation}
	hence, $\head_\mu\left(\wh y - \wh \chi_v^{(t-1)}\right) = \head_\mu(\wh{y}) \setminus S^{(t-1)}$.
	Consequently, if we let $\wh g : = \wh y - \wh \chi_v^{(t-1)}$, then: i) $\head_\mu(\wh{g}) \subseteq \supp{(\frontier \cup T')}$, by \eqref{eq:est-error-bound-nearquad-robust} along with $P_1(t-1)$, ii) $\| \wh g - \wh g_{\head_\mu(\wh{g})} \|_2^2 \le \frac{21\mu^2}{20} + \frac{\mu^2}{20\log_{\frac{1}{\alpha}}(b/\alpha)}$, by precondition of the lemma along with \eqref{eq:est-error-bound-nearquad-robust}, and iii)
	\[\left|\subtree_{\frontier \cup T'}(z_\lef) \cap \head_\mu(\wh{g}) \right| \le b,\]
	by assumption $\left|\subtree_{\frontier \cup T'}(z_\lef) \cap \head_\mu(\wh{y}) \right| \le b$.
	Therefore, all preconditions of the first statement of Lemma~\ref{recursive-robust-invariants} hold. Since we invoke primitive \textsc{RecursiveRobustSFT} with sparsity $b \le m-1$, by our inducive hypothesis that Lemma~\ref{recursive-robust-invariants} holds for any sparsity parameter $k \le m-1$, we can invoke this lemma (a deterministic version of it that succeeds with probability 1) and conclude that, $\correct_\lef = \true$, and $\supp{ (\wh \chi_\lef)} \subseteq \subtree_{\frontier \cup T'}(z_\lef) \cap \head_\mu(\wh{g})$, and $\left\| \wh g_{\subtree_{\frontier \cup T'}(z_\lef) \cap \head_\mu(\wh{g})} - \wh \chi_\lef \right\|_2^2 \le \frac{\mu^2}{40 \log_{1/\alpha}^2 b} \le \frac{\mu^2}{10}$. This together with inductive hypothesis $P_4(t-1)$ imply that, $\supp{ (\wh \chi_\lef)} = \subtree_{\frontier \cup T'}(z_\lef) \cap \head_\mu(\wh{y})$.
	
	So, if $\left|\subtree_{\frontier \cup T'}(z_\lef) \cap \head_\mu(\wh{y}) \right| \le b$, then the algorithm adds all leaves that correspond to frequencies in $\subtree_{\frontier \cup T'}(z_\lef) \cap \head_\mu(\wh{y})$ to tree $T^{(t-1)}$ as well as set $\sident^{(t-1)}$.
	By a similar argument, if $\left|\subtree_{\frontier \cup T'}(z_\righ) \cap \head_\mu(\wh{y}) \right| \le b$, then the algorithm adds all leaves corresponding to frequencies in $\subtree_{\frontier \cup T'}(z_\righ) \cap \head_\mu(\wh{y})$ to tree $T^{(t-1)}$ and set $\sident^{(t-1)}$.
	
	{\bf Possibility 2)} $\left|\subtree_{\frontier \cup T'}(z_\lef) \cap \head_\mu(\wh{y}) \right| > b$. 
	Same as in {\bf possibility 1}, the inductive hypothesis $P_4(t-1)$ implies that $|S^{(t-1)}| \le k$, hence, inductive hypothesis $P_3(t-1)$ gives \eqref{eq:est-error-bound-nearquad-robust}.
	Hence, $\head_\mu\left(\wh y - \wh \chi_v^{(t-1)} \right) = \head_\mu(\wh{y})\setminus S^{(t-1)}$.
	Consequently, if we let $\wh g : = \wh y - \wh \chi_v^{(t-1)}$, then we find that i) $\head_\mu(\wh{g})\subseteq \supp{(\frontier \cup T')}$, by $P_1(t-1)$, ii) $\| \wh g - \wh g_{\head_\mu(\wh{g})} \|_2^2 \le \frac{21\mu^2}{20} + \frac{\mu^2}{20\log_{\frac{1}{\alpha}}(b/\alpha)}$, by precondition of the lemma along with \eqref{eq:est-error-bound-nearquad-robust}, and iii) 
	\[\left|\subtree_{\frontier \cup T'}(z_\lef) \cap \head_\mu(\wh{g}) \right| \le \left| S \right| \le k,\]
	by precondition of statement 1 of the lemma. 
	Additionally, by $P_4(t-1)$, we find that 
	\[\left|\subtree_{\frontier \cup T'}(z_\lef) \cap \head_\mu(\wh{g}) \right| = \left|\subtree_{\frontier \cup T'}(z_\lef) \cap \head_\mu(\wh{y}) \right| > b.\]
	Since we invoke primitive \textsc{RecursiveRobustSFT} with sparsity $b \le m-1$, by our inducive hypothesis that Lemma~\ref{recursive-robust-invariants} holds for any sparsity parameter $k \le m-1$, we can invoke this lemma (a deterministic version of it that succeeds with probability 1) and conclude that, $\correct_\lef = \false$, and $\wh \chi_\lef \equiv 0$. 
	
	We remark that since 
	\[\left|\subtree_{\frontier \cup T'}(z_\lef) \cap \head_\mu(\wh{y}) \right| + \left|\subtree_{\frontier \cup T'}(z_\righ) \cap \head_\mu(\wh{y}) \right| =\left|L_z^{(t-1)} \right|,\]
	 the inductive hypothesis $P_2(t-1)$ along with the above arguments imply that the if-statement in line~\ref{a11l24} of our algorithm cannot be $\true$ and hence in the rest of our analysis we can ignore lines~\ref{a11l24} and \ref{a11l25} of the algorithm.
	Furthermore, in this case the algorithm adds leaf $z_\lef$ as the left child of $v$ to tree $T^{(t-1)}$.
	By a similar argument, if $\left|\subtree_{\frontier \cup T'}(z_\righ) \cap \head_\mu(\wh{y}) \right| > b$, then the algorithm adds leaf $z_\righ$ as the left child of $v$ to tree $T^{(t-1)}$.
	
	Based on the above arguments, according to the values of $\correct_\lef$ and $\correct_\righ$, there are various cases that can happen. From the way tree $T^{(t)}$ and set $\sident^{(t)}$ are obtained from $T^{(t-1)}$ and $\sident^{(t-1)}$, it follows that in any case all 4 properties of $P(t)$ are maintained.
	We have proved that for every $t$, if the inductive hypothesis $P(t-1)$ is satisfied then the property $P(t)$ is maintained. This completess the induction (i.e., the inner induction, recall that we have nested inductions) and proves that properties $P(t)$ is maintained throughout the execution of Algorithm~\ref{alg:nearquad-robust}, assuming that preconditions i, ii, and iii of the lemma along with the precondition $|S| \le k$ of statement 1 of the lemma hold. 
	
	Lemma~\ref{recursive-robust-runtime} proves that Algorithm~\ref{alg:nearquad-robust} must terminate after some $q$ iterations. When the algorithm terminates, the condition of the \emph{Repeat-Until} loop in line~\ref{a11l36} of the algorithm must be $\true$. Thus, when the algorithm terminates, at $q^{th}$ iteration, there is no leaf in subtree $T^{(q)}$ besides $v$ and as a consequence the set $\sident^{(q)}$ must be empty. This, together with $P_1(q)$ imply that the signal $\wh \chi_v^{(q)}$ satisfies,
	\[ \supp{\left(\wh{\chi}_v^{(q)}\right)} = S = \subtree_{\frontier}(v) \cap \head_\mu(\wh y) .\] 
	Moreover, $P_3(q)$ together with precondition $|S|\le k$ imply that 
	\[\left\|\wh{y}_S - \wh{\chi}_v^{(q)}\right\|_2^2 \leq  \frac{\left|S \right|}{40k \log_{1/\alpha}^2k} \cdot \mu^2 \le \frac{\mu^2}{40 \log_{1/\alpha}^2k}.\]
	
	Now we analyze the if-statement in line~\ref{a11l37} of the algorithm. The above equalities and inequalities on $\wh \chi_v^{(q)}$ imply that,
	\begin{align*}
		\left\| \left(\wh y - \wh \chi_v^{(q)} \right)_{\subtree_{\frontier}(v)}\right\|_2^2 &= \left\| \wh y_{\subtree_{\frontier}(v) \setminus S}\right\|_2^2 + \left\| \left(\wh y - \wh \chi_v^{(q)} \right)_S\right\|_2^2\\
		&\le \left\| \wh y_{\subtree_{\frontier}(v) \setminus \head_\mu(\wh y)}\right\|_2^2 + \frac{\mu^2}{40}.
	\end{align*}
	Therefore, if $\wh G_v$ is a Fourier domain $(v,\frontier)$-isolating filter constructed in Lemma~\ref{lem:isolate-filter-highdim}, then by Corollary~\ref{infty-norm-bound-islating-filter} along with the above inequality, we have
	\begin{align*}
		\left\| \left(\wh y - \wh \chi_v^{(q)} \right) \cdot \wh G_v \right\|_2^2 &\le \sum_{\bm\xi \in [n]^d \setminus \supp{(\frontier)} } \left| \wh{y}({\bm{\xi}}) \right|^2 + \left\| \left(\wh y - \wh \chi_v^{(q)} \right)_{\subtree_{\frontier}(v)}\right\|_2^2 \\ 
		& \le \sum_{\bm\xi \in [n]^d \setminus \supp{(\frontier)} } \left| \wh{y}({\bm{\xi}}) \right|^2 + \left\| \wh y_{\subtree_{\frontier}(v) \setminus \head_\mu(\wh y)}\right\|_2^2 + \frac{\mu^2}{40} \\
		& \le \left\| \wh y - \wh y_{\head_\mu(\wh y)}\right\|_2^2 + \frac{\mu^2}{40} \le \frac{11}{10}\cdot \mu^2.
	\end{align*}
	Thus, the preconditions of the second claim of Lemma~\ref{lem:guarantee_heavy_test} hold. So, we can invoke this lemma to conclude that the if-statement in line~\ref{a11l37} of the algorithm is $\false$ and hence the algorithm outputs $\left( \true , \wh \chi_v^{(q)} \right)$. This completes the inductive proof of statement 1 of the lemma.
	
	Now we proceed with the inductive step towards proving the second statement of lemma. Suppose that preconditions i, ii, iii along with the precondition of statement 2 (that is $|S| > k$) hold. Lemma~\ref{recursive-robust-runtime} proved that the signal $ \wh \chi_v$ always satisfies $\supp{(\wh \chi_v)} \subseteq \subtree_{\frontier}(v)$ and $\|\wh \chi_v\|_0 \le k$. Therefore, $S \setminus \supp{(\wh \chi_v)} \neq \emptyset$. Consequently, if $\wh G_v$ is a Fourier domain $(v,\frontier)$-isolating filter constructed in Lemma~\ref{lem:isolate-filter-highdim}, then by definition of isolating filters we have
	\begin{align*}
		\left\| \left(\left( \wh y -\wh \chi_v \right) \cdot \wh G_v\right)_{S\cup \supp{(\wh \chi_v)}} \right\|_2^2 \ge \left\| \left( \wh y -\wh \chi_v \right)_{S\cup \supp{(\wh \chi_v)}} \right\|_2^2 \ge \left\| \wh y_{S \setminus \supp{(\wh \chi_v)}} \right\|_2^2 \ge 9\mu^2,
	\end{align*}
	which follows from the definition of $S$ and $\head_\mu(\cdot)$. On the other hand,
	\begin{align*}
		\left\| \left(\left( \wh y -\wh \chi_v\right) \cdot \wh G_\ell\right)_{[n]^d \setminus (S\cup \supp{(\wh \chi_v)})} \right\|_2^2 &= \left\| \left(\wh y \cdot \wh G_\ell\right)_{[n]^d \setminus (S\cup \supp{(\wh \chi_v)})} \right\|_2^2\\
		&\le \left\| \left(\wh y \cdot \wh G_\ell\right)_{[n]^d \setminus S} \right\|_2^2\\
		&\le \left\| \wh y_{\subtree_{\frontier}(v) \setminus S} \right\|_2^2\\ 
		&\qquad + \sum_{\bm\xi \in [n]^d \setminus \supp{(\frontier)} } \left| \wh{y}({\bm{\xi}}) \right|^2\\
		&\le \left\| \wh y - \wh y _{\head_\mu(\wh y)} \right\|_2^2 \le \frac{11}{10} \cdot \mu^2. \text{~~~~~~~~(precondition ii)}
	\end{align*}
	Additionally note that $\left| S \cup \supp{(\wh \chi_v)} \right| \le k/\alpha + k \le 2k/\alpha$ by preconditions of the lemma and property of $\supp{(\wh \chi_v)}$ that we have proved. Hence, by invoking the first claim of Lemma~\ref{lem:guarantee_heavy_test}, the if-statement in line~\ref{a11l37} of the algorithm is $\true$ and hence the algorithm outputs $\left( \false , \{0\}^{n^d} \right)$. This proves statement 2 of the lemma.

	Finally, observe that throughout this analysis we have assumed that Lemma~\ref{lem:guarantee_heavy_test} holds with probability 1 for all the invocations of \textsc{HeavyTest} by our algorithm. Moreover, we assumend that \textsc{Estimate} successfully works with probability 1. Also we assumed that the inductive hypothesis (that is Lemma~\ref{recursive-robust-invariants} for sparsity parameters $k \le m-1$) holds deterministically. In reality, we have to take the fact that these primitives are randomized into acount of our analysis. 
	
	The first source of randomness is the fact that \textsc{HeavyTest} only succeeds with some high probability. In fact, Lemma~\ref{lem:guarantee_heavy_test} tells us that every invocation of \textsc{HeavyTest} succeeds with probability at least $1-1/N^5$.

	The second source of randomness is the fact that \textsc{Estimate} only succeeds with some high probability. Lemma~\ref{est-inner-lem} tells us that every invocation of \textsc{Estimate} on a set $\scheap$, succeeds with probability $1-\frac{|\scheap|}{N^8} \ge 1 - \frac{1}{N^7}$. Since, our analysis in proof of Lemma~\ref{recursive-robust-runtime} shows that \textsc{RecursiveRobustSFT} makes at most $k$ recursive calls to \textsc{Estimate}, by a union bound, the overall failure probability of all invocations of this primitive will be bounded by $\frac{k}{N^7}$.
	
	The third and last source of randomness in our algorithm is the recursive invocations of \textsc{RecursiveRobustSFT} in lines~\ref{a11l22} and \ref{a11l23} of our algorithm. By the inductive hypothesis (statement of Lemma~\ref{recursive-robust-invariants}), the invocation of this primitive succeeds with probability $1 - O \left( \left(\frac{2 \log N}{\alpha}\right)^{\log_{1/\alpha}b} \cdot N^{-4}\right)$. Our analysis in proof of Lemma~\ref{recursive-robust-runtime} shows that \textsc{RecursiveRobustSFT} makes at most $\frac{2\log N}{\alpha}$ recursive calls to \textsc{RecursiveRobustSFT}. Therefore, by a union bound, the overall failure probability of all invocations of \textsc{RecursiveRobustSFT} is bounded by $O \left( \left(\frac{2 \log N}{\alpha}\right)^{\log_{1/\alpha}k} \cdot N^{-4}\right)$. 
	
	Finally, by another application of union bound, the overall failure probability of Algorithm~\ref{alg:nearquad-robust}, is bounded by $O \left( \left(\frac{2 \log N}{\alpha}\right)^{\log_{1/\alpha}k} \cdot N^{-4}\right)$.
	This completes the proof of the lemma. 
	
\end{proof}

Now we are ready to present our main robust sparse Fourier transform algorithm that achieves the guarantee of Theorem~\ref{thrm:near-quad-robust} for any $\epsilon$ using a number of samples that is near quadratic in $k$ and a runtime that is cubic and prove the main result of this section.

\begin{proofof}{Theorem~\ref{thrm:near-quad-robust}}
	The procedure that achieves the guarantees of the theorem is presented in Algorithm~\ref{high-dim-nearquad-main-alg}. The correctness proof basically follows by invoking Lemma~\ref{recursive-robust-invariants} and the runtime and sample complexity follows from Lemma~\ref{recursive-robust-runtime}. If we let $\mu := \|\eta\|_2$ then because $x$ is a signal in the $k$-high SNR regime, we have that $\left| \head_\mu(\wh{x}) \right| \le k$ and $\left\| \wh x - \wh x_{\head_\mu(\wh{x})} \right\|_2 \le \mu$. Therefore, the signal $\wh \chi$ that we computed in line~\ref{a12l3} of Algorithm~\ref{high-dim-nearquad-main-alg} by running procedure \textsc{RecursiveRobustSFT} (Algorithm~\ref{alg:nearquad-robust}) with inputs $\left( x, \{0\}^{n^d}, \{ \text{root}\}, \text{root}, k, \alpha, \mu \right)$, then all preconditions of Lemma~\ref{recursive-robust-invariants} hold and hence by invoking the first statement of this lemma we conclude that, with probability at least $1 - \frac{1}{2N^3}$, $\wh \chi$ satisfies the following properties:
	\[ \| \wh x - \wh \chi \|_2^2 \le \frac{\mu^2}{40} \text{~~~~and~~} \supp(\wh \chi) \subseteq \head_\mu(\wh x). \]
	This together with the $k$-high SNR assumption imply that, with probability at least $1 - \frac{1}{2N^3}$, $\supp(\wh \chi) = \head_\mu(\wh x)$.
	Therefore, tree $T$ that we construct in line~\ref{a12l4} of Algorithm~\ref{high-dim-nearquad-main-alg} is in fact the spliting tree of the set $\head_\mu(\wh x)$, that is, $\supp{(T)} = \head_\mu(\wh x)$ and $|\leaves(T)| = |\head_\mu(\wh x)|$.
	
	In the rest of the correctness proof we condition on the event that tree $T$ is the spliting tree of the set $\head_\mu(\wh x)$ and analyze the evolution of singal $\wh \chi_\epsilon$ and tree $T$ in every iteration $t=0,1,2,...$ of the \emph{while loop} in Algorithm~\ref{high-dim-nearquad-main-alg}.
	Let $\wh \chi_\epsilon^{(t)}$ denote the signal $\wh \chi_\epsilon$ at the end of iteration $t$, and let ${T}^{(t)}$ denote the tree ${T}$ at the end of iteration $t$. 
	In every iteration $t$, Algorithm~\ref{high-dim-nearquad-main-alg} computes a subset $\scheap^{(t)}$ of leaves of the tree $T^{(t-1)}$ by running the primitive \textsc{FindCheapToEstimate} in line~\ref{a12l7} of the algorithm. By Claim~\ref{claim:findCheap}, the set $\scheap^{(t)} \subseteq \leaves\left(T^{(t-1)}\right)$ satisfies the property that $\left| \scheap^{(t)} \right| \cdot \left( 8+4\log k \right) \ge \max_{u \in \scheap^{(t)}} 2^{w_{T^{(t-1)}}(u)}$. Clearly $\scheap^{(t)} \neq \emptyset$, by Claim~\ref{claim:findCheap}. 
	Then the algorithm computes $\{\wh H_u\}_{u \in \scheap^{(t)}}$ by running the procedure \textsc{Estimate} in line~\ref{a12l9} and then updates $\wh \chi_\epsilon^{(t)}(\ff_u) \gets \wh H_u$ for every $u \in \scheap^{(t)}$ and $\wh \chi_\epsilon^{(t)}(\bm \xi) = \wh \chi_\epsilon^{(t-1)}(\bm \xi)$ at every other frequency $\bm \xi$. Moreover, the algorithm updates the tree $T^{(t)}$ by removing every leaf that is in the set $ \scheap$ from tree $T^{(t-1)}$.
	Hence, one can readily see that since at each iteration of the while loop, tree $T$ looses at least one of its leaves, the algorithm terminates after at most $\left|\leaves\left( T^{(0)} \right) \right| = k$ iterations, since initially the number of leaves of $T^{(0)}$ equals $|\head_\mu(\wh x)| = k$. 
	
	If we denote by $S^{(t)}$ the set $\supp{\left( \wh\chi_\epsilon^{(t)} \right)}$ for every $t$, then we claim that the following holds,
	\[ \Pr \left[\left\|\wh{x}_{S^{(t)}}-\wh{\chi}_\epsilon^{(t)}  \right\|_2^2 \leq \frac{\epsilon\left|S^{(t)} \right|}{k} \cdot  \mu^2\right] \ge 1 - \frac{\left|S^{(t)} \right|}{N^8}. \]
	We prove the above claim by induction on iteration number $t$ of the \emph{while loop} of our algorithm. One can see that the base of induction trivially holds for $t=0$ because $\wh \chi_\epsilon^{(0)} \equiv 0$. To prove the inductive step, suppose that the inductive hypothesis holds for $t-1$, that is,
	\[ \Pr \left[\left\|\wh{x}_{S^{(t-1)}}-\wh{\chi}_\epsilon^{(t-1)}  \right\|_2^2 \leq \frac{\epsilon\left|S^{(t-1)} \right|}{k} \cdot  \mu^2\right] \ge 1 - \frac{\left|S^{(t-1)} \right|}{N^8}. \]
	If we let $L:= \left\{ \ff_u: u \in \scheap^{(t)} \right\}$, then one can see from the way our algorithm updates signal $\wh \chi_\epsilon^{(t)}$ and tree $T^{(t)}$ that $S^{(t)} \setminus S^{(t-1)} = L$ for every iteration $t$. Furthermore, by Lemma~\ref{est-inner-lem} and union bound, we find that with probability at least $1 - \frac{|S^{(t-1)}|}{N^8} - \frac{|\scheap^{(t)}|}{N^8} = 1 - \frac{|S^{(t-1)}|}{N^8}$ the following holds
	\begin{align}
		\left\| \wh x_{S^{(t)}} - \wh \chi_\epsilon^{(t)}\right\|_2^2 &= \left\| ( \wh x - \wh \chi_\epsilon^{(t)})_{S^{(t-1)}}\right\|_2^2 + \left\| ( \wh x - \wh \chi_\epsilon^{(t)})_{S^{(t)}\setminus S^{(t-1)}}\right\|_2^2\nonumber\\
		&=\left\| \wh x_{S^{(t-1)}} - \wh \chi_\epsilon^{(t-1)} \right\|_2^2 + \left\| (\wh x - \wh \chi_\epsilon^{(t)} )_{L}\right\|_2^2\nonumber\\
		&\le \frac{\epsilon |S^{(t-1)}| \mu^2}{k} + \frac{\epsilon \left|L \right|}{2k} \sum_{\bm\xi \in [n]^d \setminus \supp{\left(T^{(t-1)}\right)} } \left| \left(\wh{x}-\wh{\chi}_\epsilon^{(t-1)}\right)({\bm{\xi}}) \right|^2.\label{estimate-error-mainthrm}
	\end{align}
Now we bound the second term above,
\small
\begin{align}
	&\sum_{\bm\xi \in [n]^d \setminus \supp{\left(T^{(t-1)}\right)} } \left| \left(\wh{x}-\wh{\chi}_\epsilon^{(t-1)}\right)({\bm{\xi}}) \right|^2\nonumber\\ 
	&\qquad = \sum_{\bm\xi \in  [n]^d \setminus \left(\supp{\left(T^{(t-1)}\right)} \cup S^{(t-1)}\right)} \left| \wh{x}({\bm{\xi}}) \right|^2 + \left\|\wh{x}_{S^{(t-1)}} -\wh{\chi}_\epsilon^{(t-1)} \right\|_2^2\nonumber\\
	&\qquad\le \sum_{\bm\xi \in [n]^d \setminus \head_\mu(\wh{x})} \left| \wh{x}({\bm{\xi}}) \right|^2 + \left\|\wh{x}_{S^{(t-1)}} -\wh{\chi}_\epsilon^{(t-1)} \right\|_2^2 \text{~~~~($T$ was initially the splitting tree of $\head_\mu(\wh x)$)}\nonumber\\
	&\qquad \le 2 \mu^2 \text{~~~~~~~~~~~~~~~~~~~~~~~~~~~~~~~~~~~~~~~~~~~~~~~~~~~~~~~~~~~~~~~~~~~~~~~~~~~~~~(by the inductive hypothesis)}.\nonumber
\end{align}
\normalsize
Therefore, by plugging the above bound back to \eqref{estimate-error-mainthrm} we find that,
\[\Pr \left[\left\|\wh{x}_{S^{(t)}}-\wh{\chi}_\epsilon^{(t)}  \right\|_2^2 \leq \frac{\epsilon\left|S^{(t)} \right|}{k} \cdot  \mu^2\right] \ge 1 - \frac{\left|S^{(t)} \right|}{N^8},
\]
which proves the inductive claim. Therefore, by another application of union bound, with probability at least $1 - \frac{1}{N^3}$, the output of the algorithm $\wh \chi_\epsilon$ satisfies $\left\|\wh{x} -\wh{\chi}_\epsilon \right\|_2^2 \leq (1 + \epsilon) \cdot  \mu^2$. This proves the correctness of Algorithm~\ref{high-dim-nearquad-main-alg}.

\paragraph{Runtime and Sample Complexity.} By Lemma~\ref{recursive-robust-runtime}, the running time and sample complexity of invoking primitive \textsc{RecursiveRobustSFT} in line~\ref{a12l3} of the algorithm are bounded by $\widetilde{O}(k^3)$ and $\widetilde{O}\left( k^2 \cdot 2^{2 \sqrt{\log k \cdot \log(2\log N)}} \right)$, respectively. Additionally, by Lemma~\ref{est-inner-lem}, the runtime and sample complexity of every invocation of \textsc{Estimate} in line~\ref{a12l9} of our algorithm are bounded by $\widetilde{O}\left( \frac{k}{ \epsilon| \scheap^{(t)} |}  \sum_{u \in \scheap^{(t)}} 2^{w_{T^{(t-1)}}(u)} + \frac{k}{\epsilon}\cdot \|\wh \chi_\epsilon^{(t-1)}\|_0 \right)$ and $\widetilde{O}\left( \frac{k}{ \epsilon| \scheap^{(t)} |}  \sum_{u \in \scheap^{(t)}} 2^{w_{T^{(t-1)}}(u)}\right)$, respectively. Using the fact that $| \scheap^{(t)} | \cdot \left( 8+4\log k \right) \ge \max_{u \in \scheap^{(t)}} 2^{w_{T^{(t-1)}}(u)}$ together with $\|\wh \chi_\epsilon^{(t-1)} \|_0 \le k$, these time and sample complexities are further upper bounded by $\widetilde{O}\left( \frac{k| \scheap^{(t)} |}{\epsilon} + \frac{k^2}{\epsilon} \right)$ and $\widetilde{O}\left( \frac{k}{\epsilon} \cdot | \scheap^{(t)} | \right)$, respectively.
We proved that the total number of iterations, and hence number of times we run \textsc{Estimate} in line~\ref{a12l9} of the algorithm, is bounded by $k$.
Using this together with the fact that $\sum_{t}\left| \scheap^{(t)} \right| = \left\|\wh \chi_\epsilon \right\|_0 = |\head_\mu(\wh x)| \le k$, the total runtime and sample complexity of all invocations of \textsc{Estimate} in all iterations can be upper bounded by $\widetilde{O}\left( \frac{k^3}{\epsilon}\right)$ and $\widetilde{O}\left( \frac{k^2}{\epsilon} \right)$, respectively.
Therefore the total time and sample complexities of our algorithm are bounded by $\widetilde{O}\left( \frac{k^3}{\epsilon}\right)$ and $\widetilde{O}\left( \frac{k^2}{\epsilon} + k^2 \cdot 2^{2 \sqrt{\log k \cdot \log(2\log N)}} \right)$, respectively.
\end{proofof}

\section{Experiments.}
\label{sec:experiments}

In this section, we empirically show that our FFT backtracking algorithm for high dimensional sparse signals is extremely fast and can compete with highly optimized software packages such as the FFTW~\cite{frigo1999fast, fftw}. Our experiments mainly focus on a modification of  Algorithm~\ref{alg:exact_sparse_simple} which exploits only one level of FFT backtracking and runs in $\widetilde{O}(k^{2.5})$ time. One of the baselines that we compare our algorithm to is the vanilla FFT tree pruning of \cite{kapralov2019dimension}, in order to demonstrate the speed gained by our backtracking technique. Furthermore, we compare our method against the SFFT 2.0~\cite{hikp12b, sfft2}, which is optimized for $1$-dimensional signals, and show that our method's performance for small sparsity $k$ is comparable to that of the SFFT 2.0 even in dimension one.

In a subset of our experiments, we exploit a technique introduced in \cite{GHIKPS} to speed up the high-dimensional Sparse FFT algorithms. This method works as follows. By fixing one of the coordinates of a $d$-dimensional signal we get a $(d-1)$-dimensional signal whose Fourier transform corresponds to projecting (aliasing) the Fourier transform of the original signal along the coordinate that was fixed in time domain. Thus we can effectively \emph{project} the Fourier spectrum into a $(d-1)$-dimensional plane by computing a $(d-1)$-dimensional FFT. 
Using a small number of measurements (projections with different values of the fixed coordinate) we can figure out which frequencies are projected without collision and recover them. We use this trick to recover the frequencies that get isolated under the projection and then run our algorithm on the residual signal. Since the residual signal is likely to have a smaller sparsity than the original one, this \emph{projection technique} can speed up our Sparse FFT algorithms.

\paragraph{Sparse signal classes:} In our experiments, we benchmark all methods on the following classes of $k$-sparse signals:
\begin{enumerate}
\item {\bf Random support with overtones:} The Fourier spectrum of this signal class is the superposition of a set of random frequencies and a set of overtones of these frequencies. Specifically, the support of this signal is $\supp(\wh x)=S_{\textsc{random}} \cup S_{\textsc{overtone}}$, which are defined as follows,
\[ S_{\textsc{random}}:=\left\{ \ff_1, \ff_2, \ldots \ff_{k/(d+1)} \sim \text{i.i.d. } \textsc{Unif}(\ZZ_n^d) \right\}, \]
\[ S_{\textsc{overtone}}:= \left\{ \ff + (n/2)\cdot {\bf e}_i : \forall \ff \in S_{\textsc{random}}, i \in [d] \right\}, \]
where ${\bf e}_i$ is the standard basis vector along coordinate $i$ in dimension $d$.
Note that every $\ff \in S_{\textsc{random}}$ will collide with at least one overtone under projection along any coordinate, thus, $S_{\textsc{random}}$ cannot be recovered using the projection trick. We added the overtones precisely for this reason, i.e., to ensure that the projection trick does no recover the signal entirely and there will be something left for the Sparse FFT to recover.

\item {\bf Randomly shifted $d$-dimensional Dirac Comb:} The Fourier support of a Dirac Comb (without shift) is the following,
\[ S_{\textsc{comb}} := \left\{ \left( i_1 \cdot \frac{n}{k^{1/d}}, i_2 \cdot \frac{n}{k^{1/d}}, \ldots i_d \cdot \frac{n}{k^{1/d}} \right) : i_1,i_2, \ldots i_d \in [k^{1/d}] \right\}. \]
We generate a random frequency shift $\tilde{\ff} \sim \textsc{Unif}(\ZZ_n^d)$ and a random phase shift $\tilde{\tt} \sim \textsc{Unif}(\ZZ_n^d)$ then define the $k$-sparse $\wh{x}$  as,
\[ \wh{x}_{\ff} := \sum_{\jj \in S_{\textsc{comb}}} e^{2\pi i \frac{ \ff^\top \tilde{\tt}}{n}} \cdot \mathbbm{1}_{\{ \ff = \jj + \tilde{\ff} \}} .\]
Note that the projection trick will not help at all on this signal and thus it is a good test case for the Sparse FFT algorithms. Additionally, this signal in time domain is also a randomly shifted Dirac Comb with sparsity $N/k$ and thus distinguishing it from zero with constant probability would require $\Omega(k)$ samples. This makes the Dirac Comb a hard test case for our tree exploration algorithms which heavily rely on the \textsc{ZeroTest} primitive to distinguish a sparse signal from a zero signal.
 
\item {\bf Superposition of a $k/2$-sparse signal with random support and a $d$-dimensional Dirac Comb of sparsity $k/2$:} This signal is a mixture of instances defined in {\bf (1)} and {\bf (2)}

 \item {\bf Superposition of two randomly shifted $d$-dimensional Dirac Combs of sparsity $k/2$:} This signal is a mixture of two independent instances of the randomly shifted Dirac Comb defined in {\bf (2)}.
\end{enumerate}

\subsection{FFT Backtracking vs Vanilla FFT Tree Pruning.}
\begin{figure}[t]
	\centering
	\begin{subfigure}{0.328\textwidth}
		\captionsetup{justification=centering,margin=0.2cm}
		\includegraphics[width=1.02\textwidth]{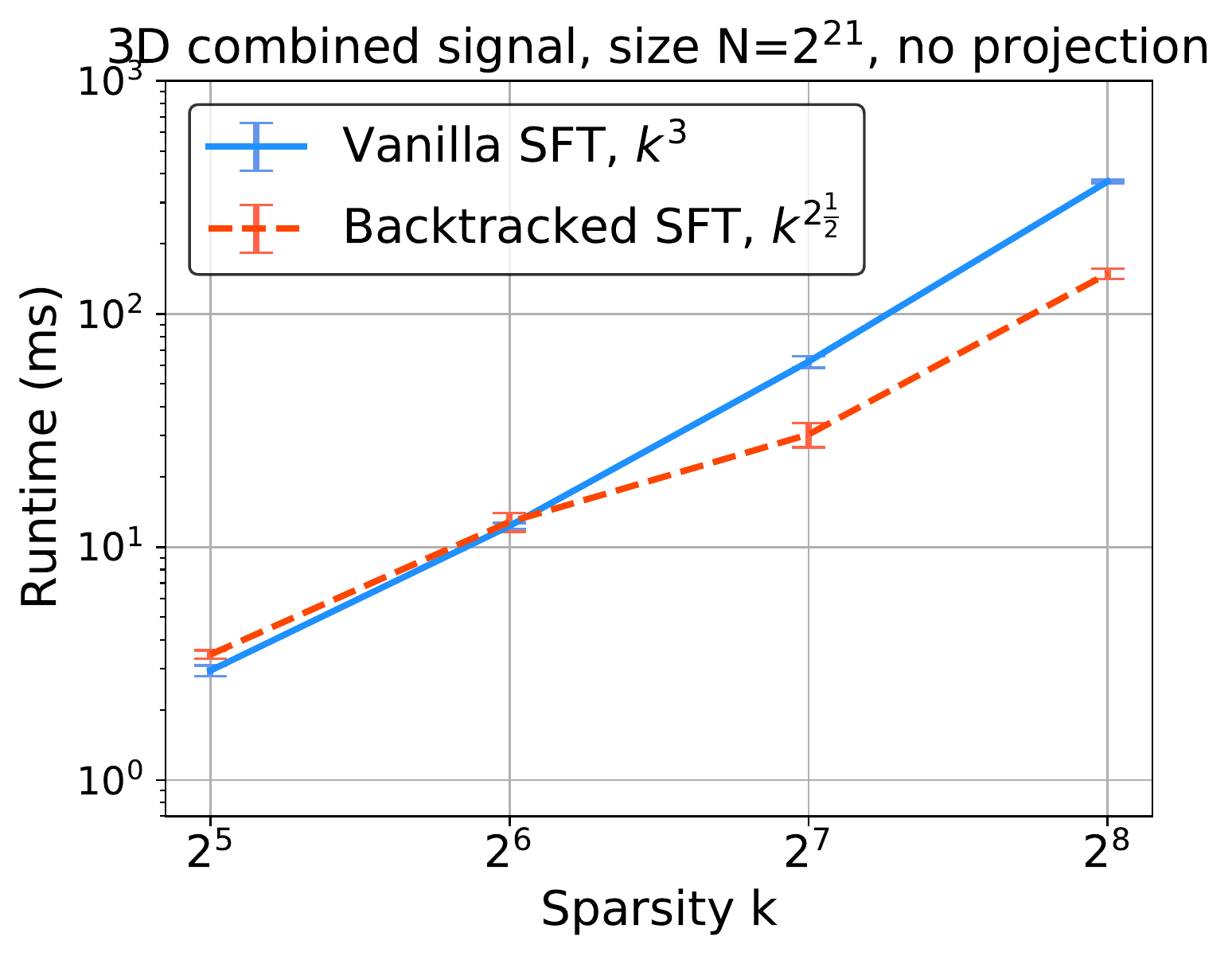}
		
		\caption{Mixture of random support and a $3$D Dirac Comb}\label{fig:combined_vs_vanilla}
	\end{subfigure}
	\begin{subfigure}{0.328\textwidth}
		\captionsetup{justification=centering,margin=0.2cm}
		\includegraphics[width=1.02\textwidth]{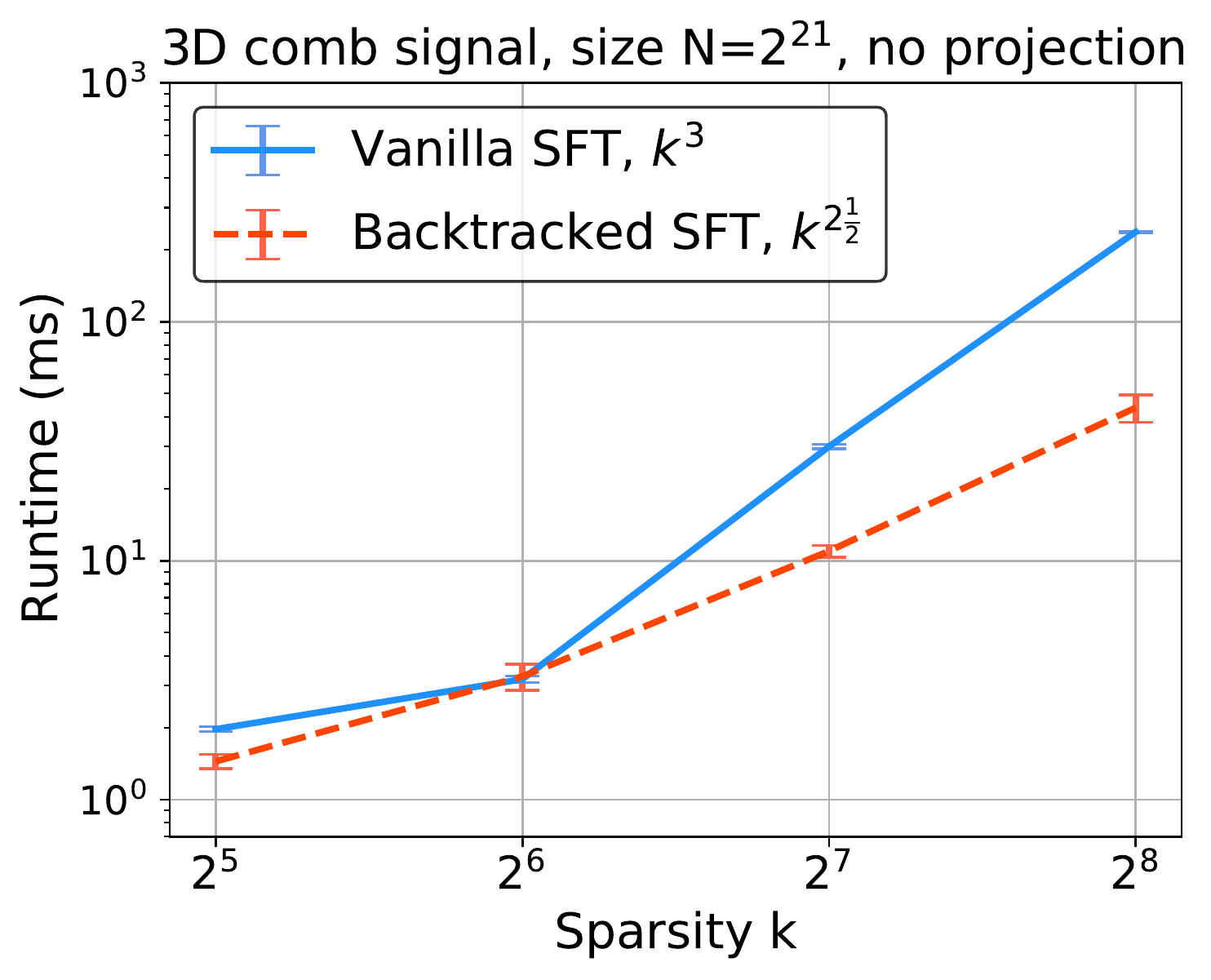}
		
		\caption{Randomly shifted $3$D Dirac Comb}\label{fig:comb_vs_vanilla}
	\end{subfigure}
	\begin{subfigure}{0.328\textwidth}
		\captionsetup{justification=centering,margin=0.2cm}
		\includegraphics[width=1.02\textwidth]{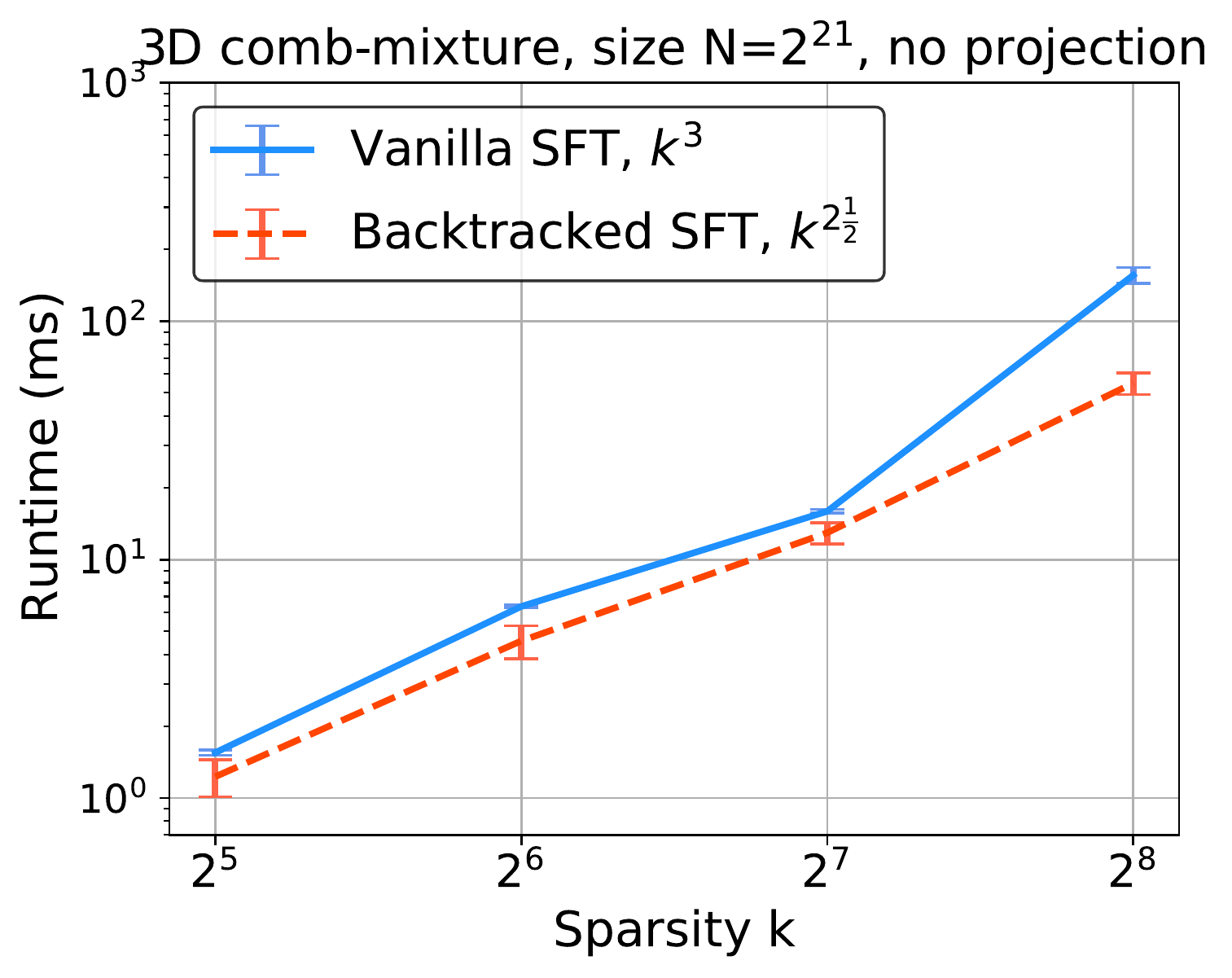}
		
		\caption{Mixture of two randomly shifted $3$D Dirac Combs}\label{fig:twocomb_vs_vanilla}
	\end{subfigure}
	
	\caption{The runtime of recovering: {\bf (a)} superposition of a $k/2$-sparse signal with random support and a $3$D Dirac Comb of sparsity $k/2$, {\bf (b)} a randomly shifted $3$D Dirac Comb with sparsity $k$, and {\bf (c)} mixture of two randomly shifted $3$D Dirac Combs of sparsities $k/2$.} \label{fig:Backtrack_vs_Vanilla}
	
	\vspace{-7pt}
\end{figure}

We first show that our backtracking technique highly improves the runtime of FFT tree pruning and compare our algorithm against the vanilla tree exploration of Kapralov et al.~\cite{kapralov2019dimension} as a baseline.
We run both algorithms on a variety of sparse signals of size $N=2^{21}$ in dimension $d=3$.
We tune the parameters of both algorithms to achieve success probabilities of higher than $90\%$ over 100 independent trials with different random seeds. Projection recovery~\cite{GHIKPS} is turned off for both algorithms to fairly demonstrate the effect of our backtracking technique.
In Figure~\ref{fig:Backtrack_vs_Vanilla}, we benchmark our methods on 3 different classes of $k$-sparse signals and observe that our Backtracked Sparse FFT algorithm consistently achieves a faster runtime and also scales slower as a function of sparsity $k$ compared to the Vanialla Sparse FFT Tree Pruning of \cite{kapralov2019dimension}.

\subsection{Sparse FFT Backtracking vs FFTW.}

\begin{figure}[t]
	\centering
	\begin{subfigure}{0.97\textwidth}
		\captionsetup{justification=centering,margin=0.5cm}
		\includegraphics[width=0.47\textwidth]{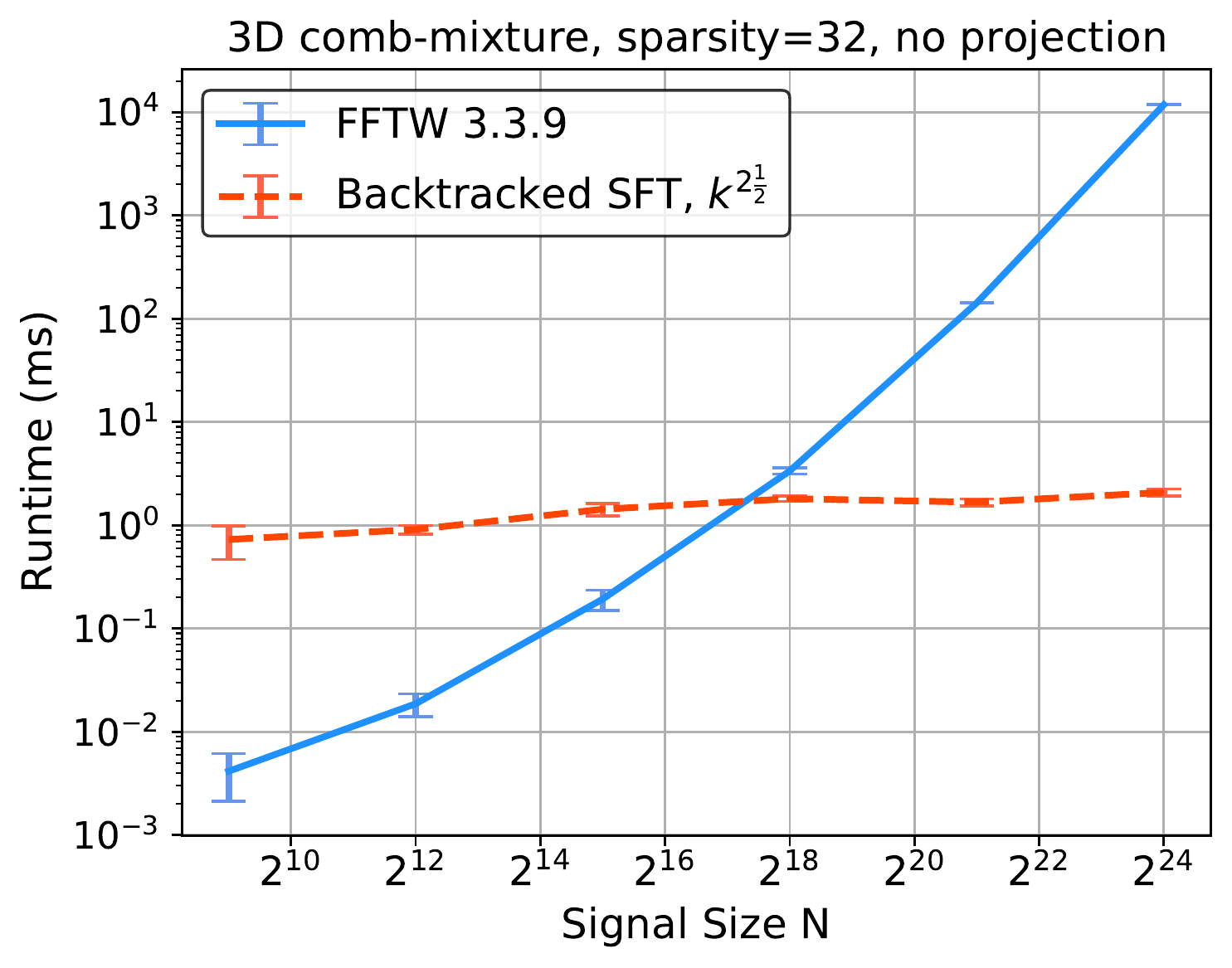}
		\includegraphics[width=0.47\textwidth]{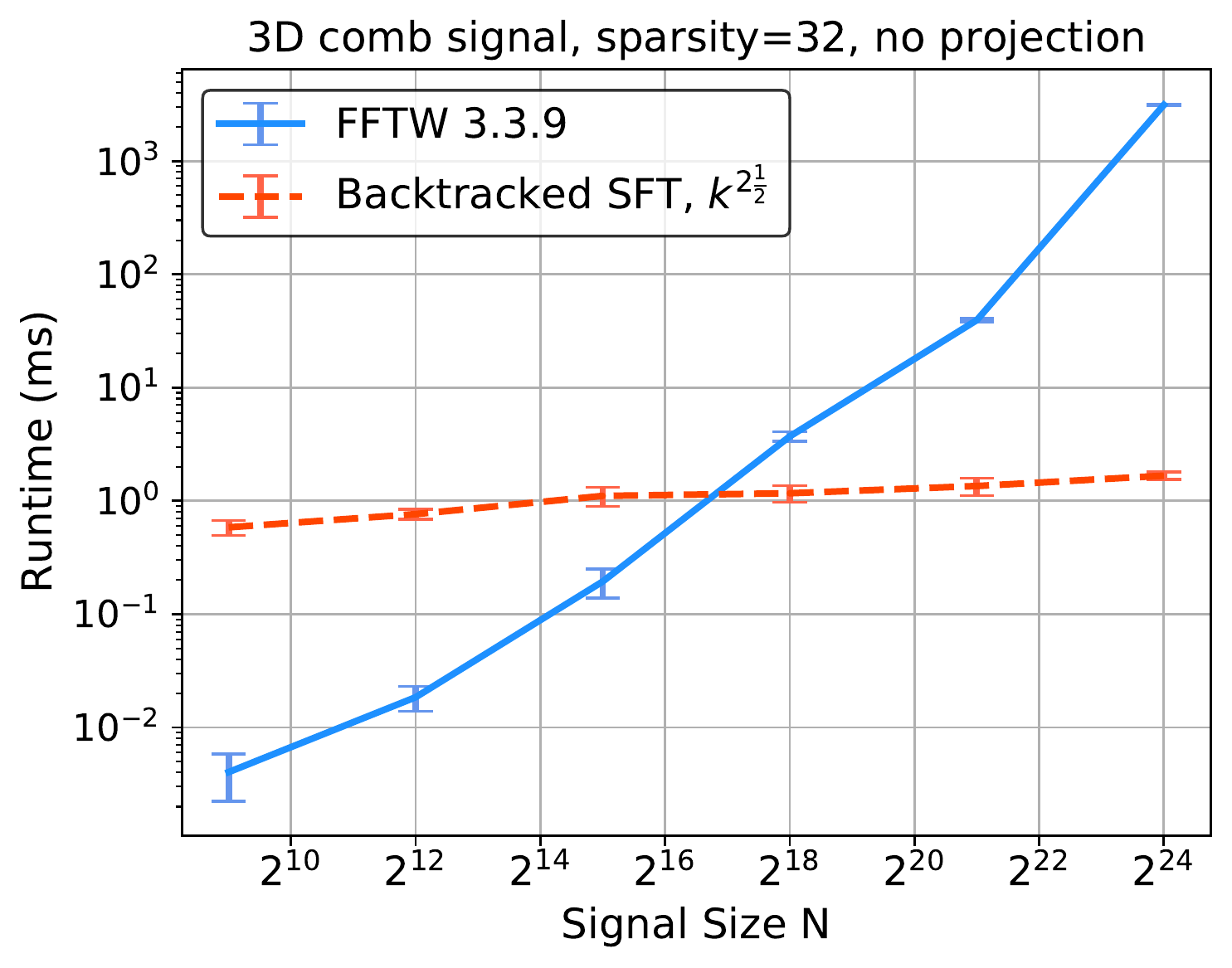}
		\caption{The input signal classes are: (Left) mixture of two randomly shifted $3$D Dirac Combs and (Right) a randomly shifted $3$D Dirac Comb}\label{fig:no_proj}
	\end{subfigure}
	\begin{subfigure}{0.97\textwidth}
		\captionsetup{justification=centering,margin=0.5cm}
		\includegraphics[width=0.47\textwidth]{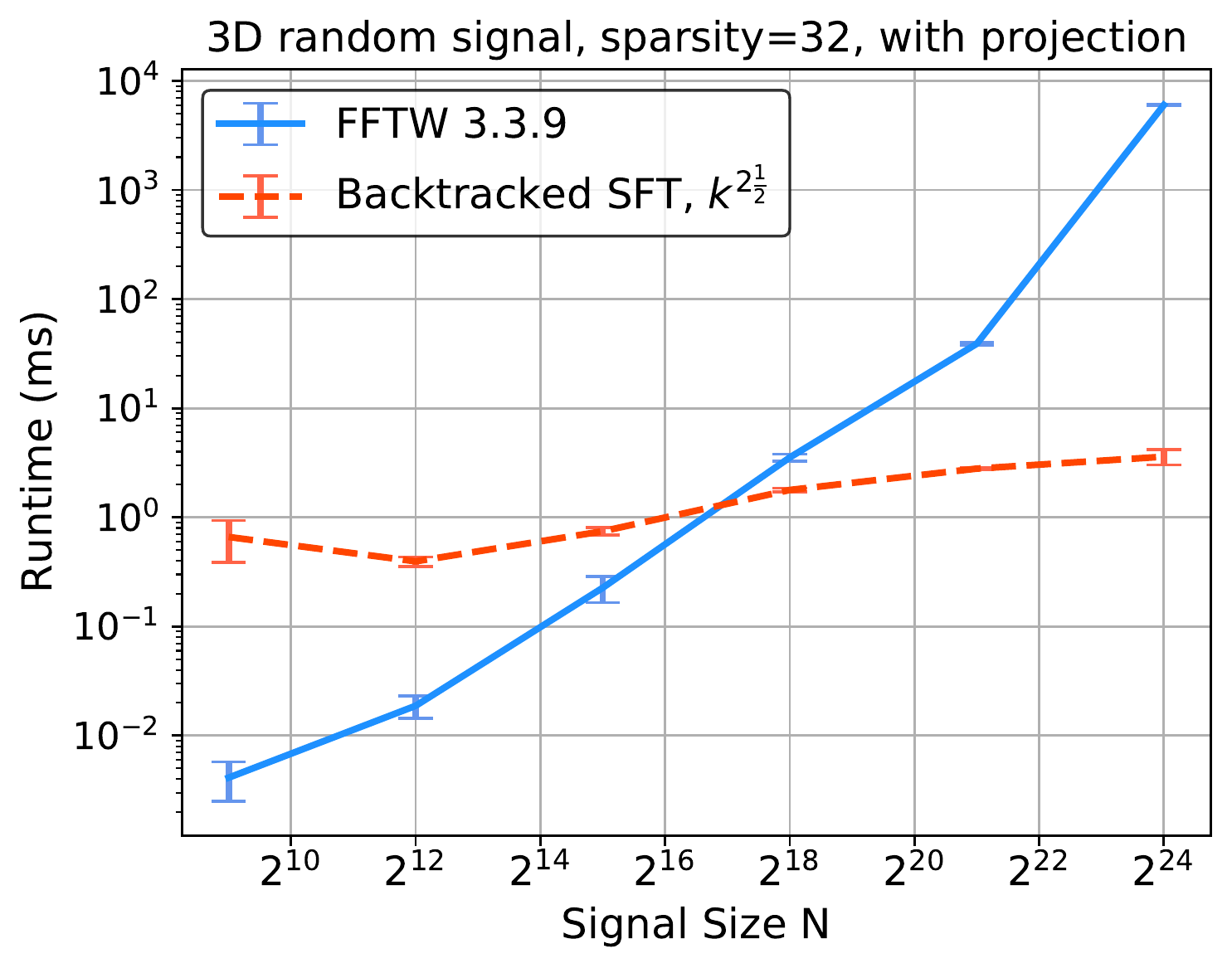}
		\includegraphics[width=0.47\textwidth]{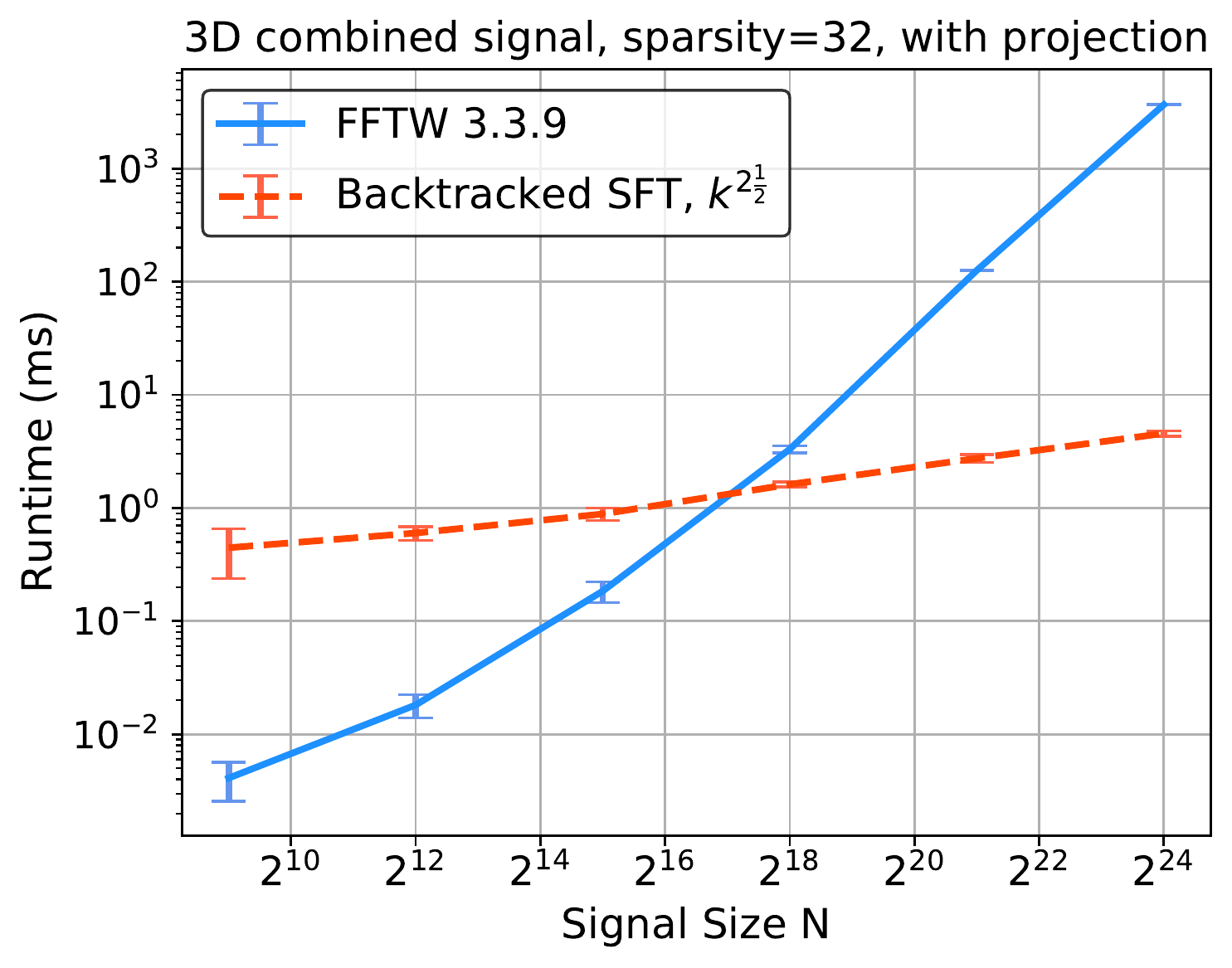}
		\caption{The input signal classes are: (Left) a random support signal with overtones and (Right) mixture of random support and a randomly shifted $3$D Dirac Comb}\label{fig:with_proj}
	\end{subfigure}

	\caption{The runtime of recovering various signal classes with sparsity $k=32$. We consider two variants of our Backtracked Sparse FFT: {\bf (a)} purely modified Algorithm~\ref{alg:exact_sparse_simple} with no prefiltering or projection tricks, {\bf (b)} enhanced version of modified Algorithm~\ref{alg:exact_sparse_simple} which first applies the projection trick.} \label{fig:Backtrack_vs_fftw}
	
	\vspace{-7pt}
\end{figure}
Next we compare our algorithm against the highly optimized FFTW 3.3.9 software package and show that our algorithm outperforms FFTW by a large margin when the signal size $N$ is large. We run both algorithms on a variety of signals of sparsity $k=32$ in dimension $d=3$. As in previous set of experiments, the parameters of our algorithm is tuned to succeed in over $90\%$ of instances.
In Figure~\ref{fig:Backtrack_vs_fftw}, we benchmark our method and the FFTW on 4 different classes of $k$-sparse signals and observe that in all cases the runtime of our Backtracked Sparse FFT algorithm scales very weakly with signal size $N$, particularly, our runtime grows far slower than that of FFTW. Consequently our algorithm is orders of magnitude faster than FFTW for any $N\ge 2^{18}$.

\subsection{Comparison to SFFT 2.0 in Dimension One.}
Finally, in this set of experiments we compare our modified Algorithm~\ref{alg:exact_sparse_simple} against the SFFT software package~\cite{sfft2} which is highly optimized for $1$-dimensional sparse signals and show that we can achieve comparable performance even in dimension one. 
We run both algorithms on two classes of signals with sparsity $k=32$ in dimension $d=1$.
We remark that the runtime of SFFT, which is implemented based on \cite{hikp12b}, will certainly scale badly in high dimensions due to filter support increasing. However, since there is no optimized code available for SFFT in high dimensions, we feel that it is more informative to compare our optimized code to their optimized code in $1$D rather than have a weak extension of their approach as a benchmark.

The SFFT package includes two versions: 1.0 and 2.0. The difference is that SFFT 2.0 adds a Comb prefiltering heuristic to improve the runtime. The idea of this heuristic is to apply the aliasing filter, which is very efficient and has no leakage, to restrict the locations of the large coefficients according to their values mod some number $B=O(k)$. The heuristic, in a preprocessing stage, subsamples the signal at rate $1/B$ and then takes the FFT of the subsampled signal. 

\begin{figure}[t]
	\centering
	\begin{subfigure}{0.48\textwidth}
		\captionsetup{justification=centering,margin=0.2cm}
		\includegraphics[width=0.99\textwidth]{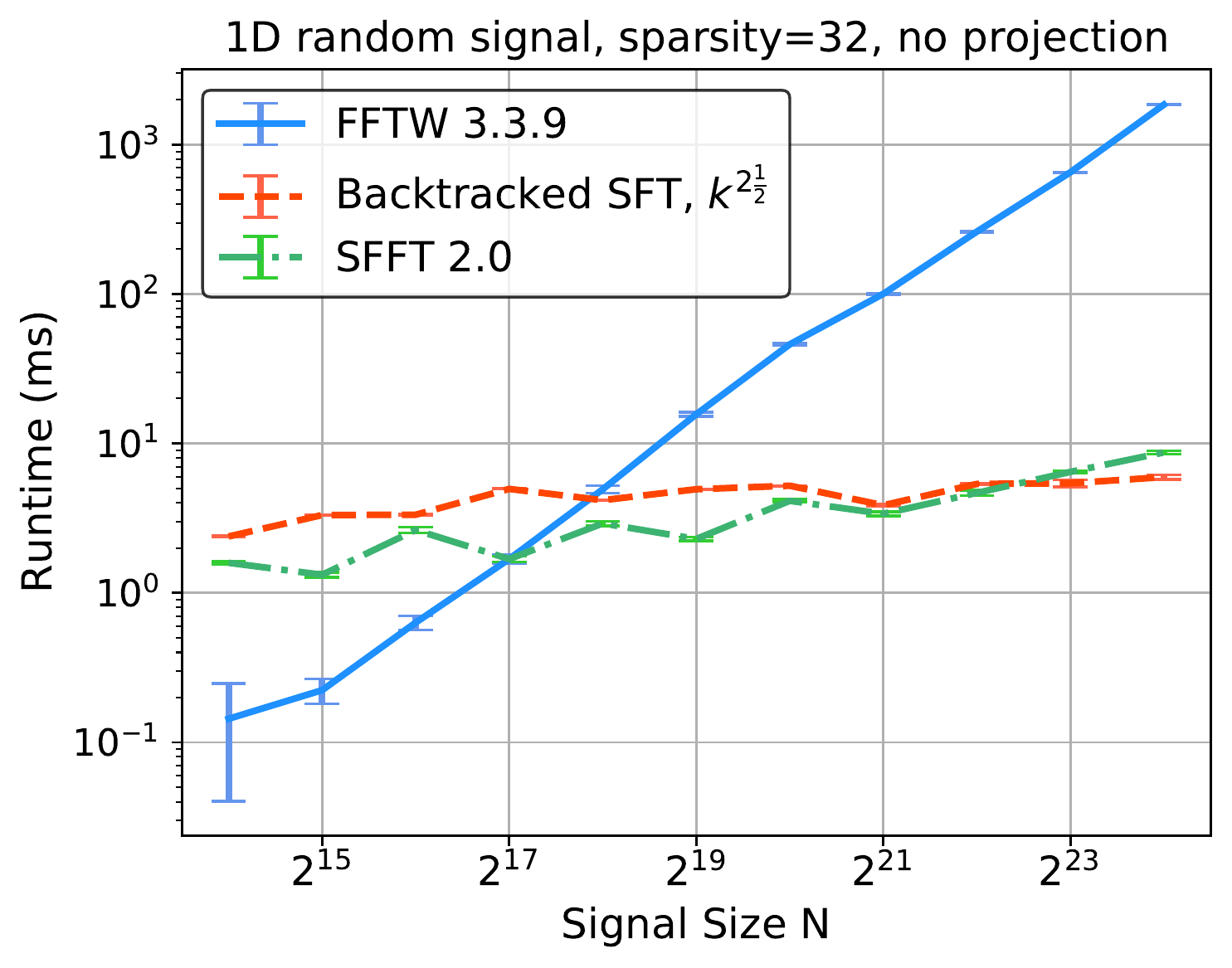}
		
		\caption{Random Fourier support}\label{fig:sfft_random}
	\end{subfigure}
	\begin{subfigure}{0.48\textwidth}
		\captionsetup{justification=centering,margin=0.2cm}
		\includegraphics[width=0.99\textwidth]{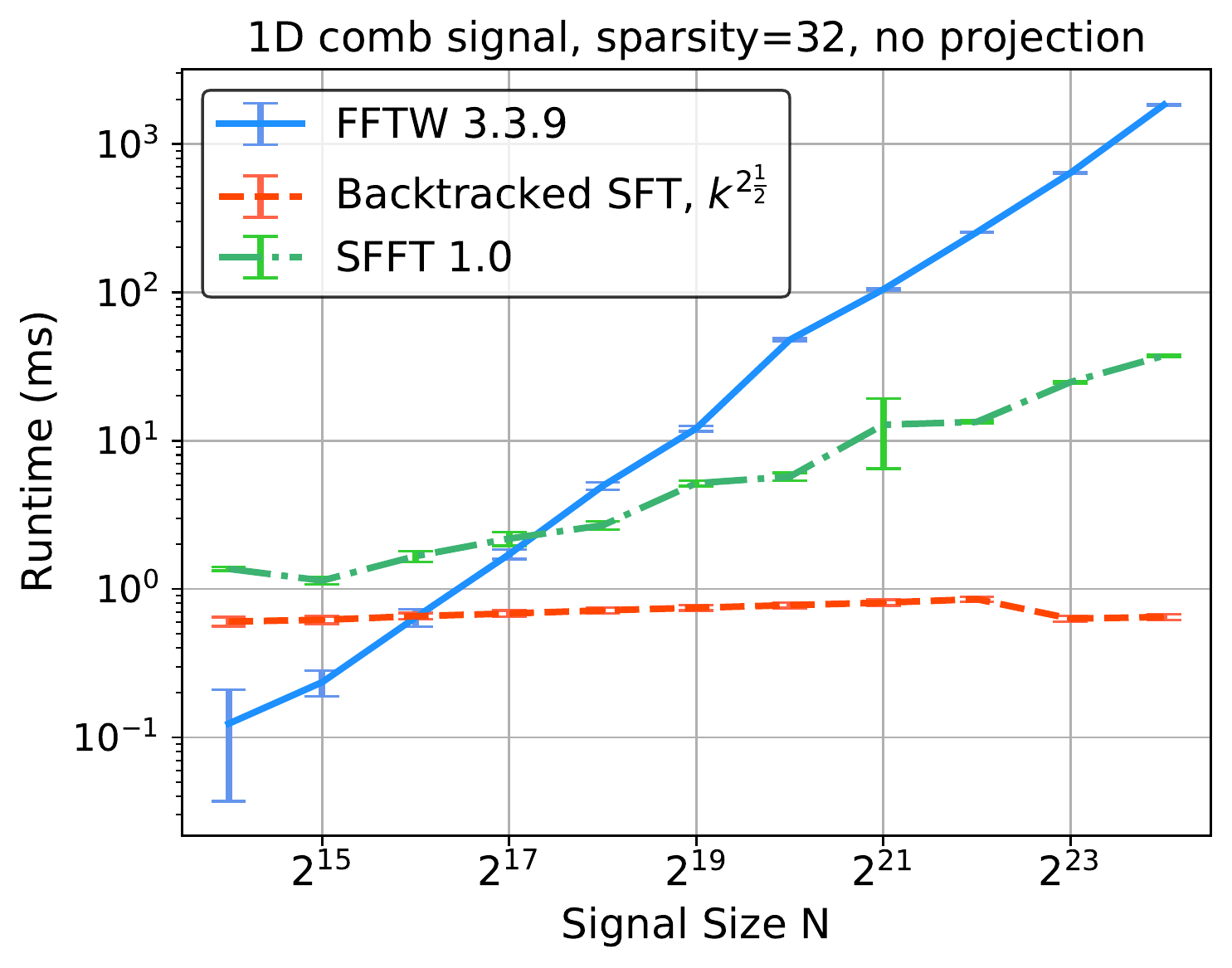}
		
		\caption{Randomly shifted Dirac Comb}\label{fig:sfft_comb}
	\end{subfigure}
	
	\caption{The runtime of recovering: {\bf (a)} $k$-sparse signal with random support and {\bf (b)} a randomly shifted Dirac Comb with sparsity $k$.} \label{fig:sfft}
	
	\vspace{-7pt}
\end{figure}

In Figure~\ref{fig:sfft}, we benchmark our method and SFFT (1.0 and 2.0) on 2 different classes of $k$-sparse signals and observe that the runtime of our Backtracked Sparse FFT algorithm is comparable to that of SFFT. 
In Fig.~\ref{fig:sfft_random} we run the algorithms on a signal with random Fourier support and observe that SFFT 2.0 runs slightly faster. Since the support is random, the heuristic trick used in SFFT 2.0 can recover a large portion of the frequencies and thus SFFT 2.0 owes much of its speed to the heuristic trick.
On the other hand, in Fig.~\ref{fig:sfft_comb}, we run the algorithms on a randomly shifted Dirac Comb and observe that our method outperforms SFFT 1.0. Note that since the Comb prefiltering heuristic used in SFFT 2.0 completely fails on a Dirac Comb input, we used SFFT 1.0 in this experiment instead. This result demonstrates that for signals with small sparsity $k$, our algorithm can run even faster than SFFT when the input's support is a multiplicative subgroup of $\ZZ_n$, such as the Dirac Comb.

\section{Acknowledgements.}
Michael Kapralov, Mikhail Makarov and Amir Zandieh have received funding from the European Research Council (ERC) under the European Unions Horizon 2020 research and innovation programme (grant agreement No. 759471) for the project SUBLINEAR.  Amir Zandieh was supported by the Swiss NSF grant No. P2ELP2\_195140. Karl Bringmann and Vasileios Nakos have received funding from the European Research Council (ERC) under the European Unions Horizon 2020 research and innovation programme (grant agreement No. 850979) for the project TIPEA.

\appendix
\section{Analysis of the Cubic Time Tree Exploration Algorithm.}
\label{appndx-slow-tree-explor}

This section is devoted to proving the correctness and runtime of \cref{alg:ksparsefft_simple}.

The idea behind \cref{alg:ksparsefft_simple} is to recover all non-zero leaves in the subtree of $\tfull_N$ rooted at $v$, given that $|\hleaves(v)| \leq b$ and $v$ is isolated by $\found$ and $\excl$.
\Cref{alg:ksparsefft_simple} is essentially a slightly modified version of~\cite{kapralov2019dimension}, but since in this paper we work with an abstracted problem, we still present the proof of its correctness and runtime.
One of the useful tools for this algorithm is \cref{lem:kraft_averaging}, which states that for any tree $T$, the minimum weight $w_T(\ell)$ of a leaf $\ell$ in $T$ is at most $\log L$, where $L$ is the number of leaves of $T$.

\begin{algorithm}[t]
	\caption{$\textsc{SlowExactSparseRecovery}( \found, \excl, v, b)$}
	\label{alg:ksparsefft_simple}
	\begin{algorithmic}[1]
	\State $\frontier_v \leftarrow \{v\}$ ,  $\steps \leftarrow 1$ ,  $\found_{out} \leftarrow \emptyset$
	\Repeat
		\If {$\steps > 6\cdot b\log N$} \label{line:promise_nonodes_simple}
			\Comment{Have explored more than the expected sparsity}
			\State \Return $\emptyset$
		\EndIf
		\State $z:=$ vertex in $\frontier_v$ with the minimum weight with respect to $\frontier_v$. \label{line:promise-min-z}
		\State $\frontier_v = \frontier_v \setminus \{z\}$,  $\steps \leftarrow \steps +1$ \label{line:remove_from_frontier}
		\State $\excl' \gets \excl \cup \frontier_v$,  $\found' \gets \found + \found_{out}$
		\If {$z$ is a leaf in $\tfull$}
			\State $\found_{out}(z) \leftarrow \textsc{Estimate}(\found', \excl', z)$ 
		\ElsIf {$ \textsc{ZeroTest}(\found', \excl', z, b) = \false $} 
		\Comment{$\leaves(z)$ contains heavy leaves}
			\State $z_{\lef} \gets $ left child of $z$
			\State $z_{\righ} \gets $ right child of $z$
			\State $\frontier_v \leftarrow \frontier_v \cup \{z_{\lef},z_{\righ}\}$		
		\EndIf
		\Comment{$\leaves(z)$ has no heavy leaves we simply remove it from $\frontier$, (see line~\ref{line:remove_from_frontier}).}
	\Until {$\frontier_v = \emptyset$}
	\State \Return $\found_{out}$
	\end{algorithmic}\label{alg:cubic-exact}
	
\end{algorithm}

\begin{theorem}[Theorem~\ref{thm:kexactrecovery_simple_corr}, restated]
If $|\hleaves(v)| \leq b$ and $v$ is isolated by $\found$ and $\excl$, then the procedure $\textsc{SlowExactSparseRecovery}$ returns the correct estimates for all $\hleaves(v)$.
\end{theorem}
\begin{proof}
	Assume for the moment that the check in line \ref{line:promise_nonodes_simple} is not made.
	We will show inductively that the following invariant holds at the end of each {\bf repeat} loop: all estimated values in $\found_{out}$ are correct, and $\hleaves(v) \setminus \found_{out} \subset \leaves(\frontier_v)$. Then the correctness would follow from the fact that $\frontier_v = \emptyset$ at the end of the procedure.
	
	It is easy to see that invariants hold before the loop starts. Now, suppose that at the start of an arbitrary iteration the invariants hold for current sets $\found_{out}$ and $\frontier_v$. Because the invariants hold, the node $z$ picked in line~\ref{line:promise-min-z} is isolated by $\found + \found_{out}$ and $\excl \cup \frontier_v$.
	If $z$ is a leaf in $\tfull_N$, because $z$ is isolated, $\Estimate$ produces the correct estimate for $z$, therefore, the invariants still hold at the end of the loop.
	If $z$ is not a leaf in $\tfull_N$, then because $z$ is isolated and $\hleaves(z) \subseteq \hleaves(v)$, the prerequisites for calling $\ZeroTest$ are fulfilled, and its output is correct. If it says True, we can just delete $z$ from $\frontier_v$ without violating the invariants, and otherwise we add both it's children instead. Since $\hleaves(z_\lef) \cup \hleaves(z_\righ) = \hleaves(z)$, the invariants still hold.
	
	Finally, even if we do the check in line \ref{line:promise_nonodes_simple}, the execution will still be the same, since we do at most $3 b \log N$ iterations. Notice that each node gets added at most once to $\frontier_v$, and on each iteration one node is removed from $\frontier_v$. Also notice that if for the node $z$, $\hleaves(z) = \emptyset$, it's children will not be added to $\frontier_v$, because $\ZeroTest$ would return True. Since there is at most $b \log N$ vertices $z$ such that $\hleaves(z) \neq \emptyset$, and each of them have only $2$ children, there is at most $3 b \log N$ vertices that can be added to $\frontier_v$, hence the maximum number of iterations is $6 b \log N$.
\end{proof}

\begin{theorem}[Theorem~\ref{thm:kspase_simple_runtime}, restated]
	If  $\leaves(v) \cap \leaves(\excl) = \emptyset$, the running time of 
	$\textsc{SlowExactSparseRecovery}$ is upper bounded by 
	\[
	\wt{O} \left( |\found| \cdot b^2 + 2^{w_{\excl}(v)} \cdot b^3 \right).
	\]
	
\end{theorem}
\begin{proof}
	First, notice that since $\leaves(v) \cap \leaves(\excl) = \emptyset$, for all $z$ in the subtree of $v$ at any iteration, $w_{\excl \cup \frontier_v} (z) = w_{\frontier_v}(z) + w_{\excl}(v)$.
	Because of the check in line \ref{line:promise_nonodes_simple}, the algorithm runs for at most $6 b \log N$ iterations, and, consequently, $|\frontier_v| \leq 6 b \log N$. By \cref{lem:kraft_averaging}, for each chosen $z$, $2^{w_{\frontier_v}(z)} \leq 6 b \log N$. Similarly, $|\found_{out}| \leq 6 b \log N$.
	Therefore, each call to $\ZeroTest$ uses 
	\[
	\wt{O}( 2^{w_{\excl \cup \frontier_v}(z)} \cdot b + |\found + \found_{out}|\cdot b) = \wt{O}( 2^{w_{\excl}(v)} \cdot b^2 + |\found|\cdot b + b^2)
	\]
	operations.
	Summing over all iterations, we find that the total runtime of all calls to $\ZeroTest$ is bounded by $\wt{O}( 2^{w_{\excl}(v)} b^3 + |\found|\cdot b^2)$.
	
	Similarly, each call of $\Estimate$ spends $\wt{O}(2^{w_{\excl \cup \frontier_v}(z)} + |\found + \found_{out}|)$ operations, which accumulates to
	$\wt{O}( 2^{w_{\excl}(v)} \cdot b^2 + |\found|\cdot b)$ across all iteration.
	
	Finally, notice that to maintain the set $\excl'$ and its tree $T(\excl)$, we first need to copy it from $\excl$ and add $v$. However, because we only work inside the subtree of $v$, we can reduce the tree $T(\excl)$ and, respectively, it's set to only contain the path to $v$ and all of the children of the vertices in this path. This can be done in time $O(\log N + w_{\excl}(v))$.
	Then, on each iteration, $\excl'$ is only modified by removing one vertex from it, which can be done in time $O(\log N)$. Hence, the total time spent on maintaining $\excl'$ is $\wt{O}(b + w_{\excl}(v))$.
	
	Since runtime of each iteration is dominated by the calls to $\ZeroTest$ and/or $\Estimate$, the total running time is $\wt{O}( 2^{w_{\excl}(v)} \cdot b^3 + |\found|\cdot b^2)$.
\end{proof}

\section{Proof of \Cref{lem:weight_subadditive}.} \label{appndx-weight-subadditive}
\begin{proof}
    Observe that the set of edges of tree $T=T(S_1 \cup S_2 \cup \{ v\})$ is the union of the sets of edges of trees $T_1 = T(S_1 \cup \{v \})$ and $T_2 = T(S_2 \cup \{ v\})$. Consider the set of all ancestors of $v$ with two children. For ancestor $u$, one of those children, say $z$, lies on the path from the $v$ to the root, and, therefore, the edge $(u, z)$ is a part of both of the trees $T_1, T_2$. Now consider the other child of $u$, $z'$. Because edge $(u, z')$ exists in $T$, it also exists in one of the trees $T_1$ or $T_2$. But then $u$ also has two children in that tree. Therefore, each ancestor of $v$ with two children has two children in $T_1$ or $T_2$ as well. Therefore, by definition of $w_T$, we get that $w_T(v) \leq w_{T_1}(v) + w_{T_2}(v)$, which is equivalent to the desired inequality.
\end{proof}
\bibliographystyle{alpha}
\bibliography{paper}

\end{document}